%% file: PRL.tex
\PassOptionsToPackage{dvipsnames,table}{xcolor}
\documentclass[reprint,floatfix,amsmath,amssymb,aps,prl,superscriptaddress]{revtex4-2}
\usepackage{comment}
\usepackage{tikz}
\usepackage{amsmath}
\usepackage{amssymb}
\usepackage{soul}
\usepackage{amsthm}
\usepackage{amsfonts}
\usepackage{graphicx}
\usepackage{physics}
\usepackage{bbm}
\usepackage{url}
\usepackage{stmaryrd}
\usepackage[colorlinks,citecolor={blue},urlcolor={blue}]{hyperref}
\usepackage{varioref}
\usepackage[capitalise]{cleveref}
\usepackage{braket}
\usepackage{ upgreek }
\usepackage{dsfont}
\usepackage{array} 
\usepackage[dvipsnames,table]{xcolor}
 
\usepackage{fancyref} 
\usepackage[utf8]{inputenc}
\usepackage[T1]{fontenc}
\usepackage{algorithm}
\usepackage{algpseudocode}

\usetikzlibrary{decorations,calligraphy}

\usepackage{subcaption}

\usepackage{wasysym}

\makeatletter
  \@booleantrue\secnumarabic@sw
\makeatother

\renewcommand\thesection      {\arabic{section}}
\renewcommand\thesubsection   {\thesection.\arabic{subsection}}
\newcommand{\identity}{\ensuremath{\mathds{1}}}

\usepackage{float}
\newfloat{algorithm}{t}{lop}

\newcommand{\bigzero}{\mbox{\normalfont\Large\bfseries 0}}
\newcommand{\bigI}{\mbox{\normalfont\Large\bfseries I}}
\newcommand{\rvline}{\hspace*{-\arraycolsep}\vline\hspace*{-\arraycolsep}}

\usepackage{hyperref}

\usepackage{tikz}
\usepackage{ifthen}
\usepackage{pgfplots}
\pgfplotsset{compat=1.9}
\usetikzlibrary{shapes,arrows}
\usetikzlibrary{positioning}
\usetikzlibrary{shapes.geometric}

\crefname{subsection}{Section}{Sections}
\Crefname{subsection}{Section}{Sections}
\makeatletter
\renewcommand\p@subsection{}%
\makeatother

\theoremstyle{plain}
\newtheorem{theorem}{Theorem}

\newtheorem{conj}{Conjecture}

\newtheorem{lem}{Lemma}
\newtheorem{cor}{Corollary}

\newtheorem{defi}{Definition}

\theoremstyle{remark}
\newtheorem{rem}{Remark}
\newcommand{\spec}{\operatorname{Spec}}
\newcommand{\supp}{\operatorname{supp}}

\newcommand*{\fancyrefthmlabelprefix}{thm}
\newcommand*{\fancyreflemlabelprefix}{lem}
\newcommand*{\fancyrefcorlabelprefix}{cor}
\newcommand*{\fancyrefdefilabelprefix}{defi}
\frefformat{plain}{\fancyreflemlabelprefix}{lemma\fancyrefdefaultspacing#1}
\Frefformat{plain}{\fancyreflemlabelprefix}{Lemma\fancyrefdefaultspacing#1}
\frefformat{plain}{\fancyrefthmlabelprefix}{theorem\fancyrefdefaultspacing#1}
\Frefformat{plain}{\fancyrefthmlabelprefix}{Theorem\fancyrefdefaultspacing#1}
\frefformat{plain}{\fancyrefcorlabelprefix}{corollary\fancyrefdefaultspacing#1}
\Frefformat{plain}{\fancyrefcorlabelprefix}{Corollary\fancyrefdefaultspacing#1}
\frefformat{plain}{\fancyrefdefilabelprefix}{definition\fancyrefdefaultspacing#1}
\Frefformat{plain}{\fancyrefdefilabelprefix}{Definition\fancyrefdefaultspacing#1}
\newcommand*{\fancyrefalglabelprefix}{alg}
\newcommand*{\frefalgname}{algorithm}
\newcommand*{\Frefalgname}{Algorithm}
\frefformat{plain}{\fancyrefalglabelprefix}{%
  \frefalgname\fancyrefdefaultspacing#1%
}%
\Frefformat{plain}{\fancyrefalglabelprefix}{%
  \Frefalgname\fancyrefdefaultspacing#1%
}%

\newcommand*{\fancyrefapplabelprefix}{app}
\newcommand*{\frefappname}{appendix}
\newcommand*{\Frefappname}{Appendix}
\frefformat{plain}{\fancyrefapplabelprefix}{%
  \frefappname\fancyrefdefaultspacing#1%
}%
\Frefformat{plain}{\fancyrefapplabelprefix}{%
  \Frefappname\fancyrefdefaultspacing#1%
}%

\definecolor{Green}{HTML}{00AD69}  

\def\beq{\begin{equation}}
\def\eeq{\end{equation}}
\def\bq{\begin{quote}}
\def\eq{\end{quote}}
\def\ben{\begin{enumerate}}
\def\een{\end{enumerate}}
\def\bit{\begin{itemize}}
\def\eit{\end{itemize}}

\def\l|{\left|}
\def\r|{\right|}

\newcommand\C{\mathbbm{C}}

\newcommand\Z{\mathbbm{Z}}

\newcommand\cB{\mathcal{B}}

\newcommand{\cD}{\mathcal{D}}

\newcommand{\cL}{\mathcal{L}}
\newcommand{\cU}{\mathcal{U}}
\newcommand{\cC}{\mathcal{C}}

\newcommand{\Ad}{\operatorname{Ad}}

\newcommand{\ZZ}{\mathbb{Z}}

\usepackage{enumitem}

\newcommand{\id}{\operatorname{id}}

\newcommand{\cS}{\mathcal{S}}

\newcommand{\cO}{\mathcal{O}}

\newcommand{\cH}{\mathcal{H}}
\newcommand{\RR}{\mathbb{R}}
\newcommand{\NN}{\mathbb{N}}

\newcommand{\eps}{\epsilon}

\newcommand{\CX}{\mathit{CX}}
\newcommand{\CZ}{\mathit{CZ}}

\definecolor{Green}{HTML}{00AD69}  
\definecolor{coolblue}{RGB}{0,102,102}
\definecolor{coolgreen}{RGB}{0,153,153}
\definecolor{shinyblue}{RGB}{0,255,255}
\definecolor{lightblue}{RGB}{102,210,255}
\definecolor{lightpurple}{RGB}{140,30,255}
\definecolor{lightpink}{RGB}{204,0,204}
\definecolor{midblue}{RGB}{0,102,204}
\definecolor{midpink}{RGB}{153,0,153}
\definecolor{darkblue}{RGB}{0,0,153}
\definecolor{cyan}{RGB}{0,204,204}
\definecolor{lightgreen}{RGB}{0,255,128}
\definecolor{midgreen}{RGB}{0,204,0}
\definecolor{darkgreen}{RGB}{0,102,51}
\definecolor{midyellow}{RGB}{204,204,0}
\definecolor{darkyellow}{RGB}{153,153,0}
\definecolor{darkpurple}{RGB}{102,0,102}
\definecolor{darkred}{RGB}{102,0,0}
\definecolor{neworange}{RGB}{255,153,51}

\usepackage[dvipsnames]{xcolor}
\usepackage{tikz}
\usetikzlibrary{decorations.pathmorphing, decorations.markings, arrows, arrows.meta, shapes, patterns}
\tikzset{baseline={([yshift=-.5ex]current bounding box.center)}}
\tikzset{every path/.style={ line width=0.5pt, line cap=round }}
\colorlet{Virtual}{RedOrange}
\tikzstyle{Bevel} = [ preaction = { draw, white, line width=2pt,  line cap = round } ]
\tikzstyle{Bevel_wide} = [ preaction = { draw, white, line width=4pt,  line cap = round } ]
\tikzstyle{Mark} = [ draw=black, fill=black, line width=0.2pt, inner sep=1.5pt ]
\tikzstyle{Mark_large} = [ inner sep=2.1pt ]
\tikzstyle{Mark_small} = [ inner sep=1pt   ]
\tikzstyle{Mark_medium} = [ inner sep=1.3pt ]
\tikzstyle{Mark_tiny} = [ inner sep=0.8pt ]
\tikzstyle{Mark_fdisk} = [ circle ]
\tikzstyle{Mark_disk} = [ circle ]
\tikzstyle{Mark_square} = [ rectangle ]
\tikzstyle{Mark_fsquare} = [ rectangle ]
\newcommand{\myArrowStyle}{line width=0.4pt,length=3pt,width=3.5pt}
\tikzstyle{->-} = [ decoration={ markings, mark = at position 0.50*\pgfdecoratedpathlength+0.6*3pt with \arrow{>[\myArrowStyle]} }, postaction={decorate} ]
\tikzstyle{-<-} = [ decoration={ markings, mark = at position 0.50*\pgfdecoratedpathlength+0.4*3pt with \arrow{<[\myArrowStyle]} }, postaction={decorate} ]
\tikzstyle{->-25} = [ decoration={ markings, mark = at position 0.25*\pgfdecoratedpathlength+0.6*3pt with \arrow{>[\myArrowStyle]} }, postaction={decorate} ]
\tikzstyle{-<-25} = [ decoration={ markings, mark = at position 0.25*\pgfdecoratedpathlength+0.4*3pt with \arrow{<[\myArrowStyle]} }, postaction={decorate} ]
\tikzstyle{->-75} = [ decoration={ markings, mark = at position 0.75*\pgfdecoratedpathlength+0.6*3pt with \arrow{>[\myArrowStyle]} }, postaction={decorate} ]
\tikzstyle{-<-75} = [ decoration={ markings, mark = at position 0.75*\pgfdecoratedpathlength+0.4*3pt with \arrow{<[\myArrowStyle]} }, postaction={decorate} ] 

\begin{document}

\title{Efficient and simple Gibbs state preparation of the 2D toric code\\
via duality to classical Ising~chains}

\author{Pablo Páez-Velasco}
\email{pablopaez@ucm.es}
\affiliation{Departamento de An\'{a}lisis Matemático y Matemática Aplicada, Universidad Complutense de Madrid, 28040 Madrid, Spain}
\affiliation{Instituto de Ciencias Matemáticas, 28049 Madrid, Spain}
\author{Niclas Schilling}
\email{niclas.schilling@student.uni-tuebingen.de}
\affiliation{Fachbereich Physik, Universität Tübingen, 72076 Tübingen, Germany}
\author{Samuel O. Scalet}
\email{sos25@cam.ac.uk}
\affiliation{Department of Applied Mathematics and Theoretical Physics, University of Cambridge, Wilberforce Road, Cambridge, CB3 0WA, United Kingdom}
\affiliation{IBM Quantum, IBM T. J. Watson Research Center, Yorktown Heights, NY, 10598, USA}
\author{Frank Verstraete}
\email{fv285@cam.ac.uk}
\affiliation{Department of Applied Mathematics and Theoretical Physics, University of Cambridge, Wilberforce Road, Cambridge, CB3 0WA, United Kingdom}
\affiliation{Department of Physics and Astronomy, Ghent University, Krijgslaan 281, 9000 Gent, Belgium}
\author{{\'A}ngela Capel}
\email{ac2722@cam.ac.uk}
\affiliation{Department of Applied Mathematics and Theoretical Physics, University of Cambridge, Wilberforce Road, Cambridge, CB3 0WA, United Kingdom}
\affiliation{Fachbereich Mathematik, Universität Tübingen, 72076 Tübingen, Germany}

\begin{abstract}
We introduce the notion of polynomial-depth duality transformations, which relates two sets of operator algebras through a conjugation by a poly-depth quantum circuit, and make use of this to construct efficient Gibbs samplers for a variety of interesting quantum Hamiltonians as they are poly-depth dual to classical Hamiltonians. This is for example the case for the 2D toric code, which is demonstrated to be poly-depth dual to two decoupled classical Ising spin chains for any system size, and we give evidence that such dualities hold for a wide class of stabilizer Hamiltonians.  Additionally, we extend the above notion of duality to Lindbladians in order to show that mixing times and other quantities such as the spectral gap or the modified logarithmic Sobolev inequality are preserved under duality.
\end{abstract}
\maketitle

\section{Introduction}

Given a Hamiltonian $H$ and some inverse temperature $0 < \beta < \infty$, its associated Gibbs state is defined as
\begin{equation}
    \sigma(\beta) = \frac{e^{-\beta H}}{\Tr[e^{-\beta H}]} \, .
\end{equation}
The Gibbs state of a Hamiltonian describes its equilibrium properties at a given inverse temperature \cite{alhambraquantum}, and has been a subject of interest in the community for a long time \cite{molnar.2015,landoncardinal.2013,riera.2012,muller.2015}. One of the main ambitions in the field is being able to prepare Gibbs states efficiently with a quantum computer. There are multiple algorithms for preparing Gibbs states \cite{Temme2010Metropolis,Poulin.2009,Chowdhury.2017,Harrow.2020,Scalet.2025, Fawzi.2024}, but an avenue that has gained special attraction is that of algorithms based on dissipation \cite{Kastoryano2016GibbsSamplersCommuting,Ding2024EfficientQuantumGibbs,Gilyen.2024,Bardet.2023,Capel2024Gibbs}.  In order to prepare Gibbs states---also known as \textit{Gibbs sampling}---efficiently using dissipation, there are two conditions that must hold: the Lindbladian associated to the model must thermalize \textit{fast} and it must be efficiently implementable on a quantum computer. 

\textbf{In this Letter, we prove that efficient Gibbs state preparation is preserved under conjugation by polynomial-depth quantum circuits}. This fact lies at the core of the subsequent results obtained and motivates the definition of \textit{poly-depth dual families of Hamiltonians}, which will be the main focus of the text. Indeed, as any Hamiltonian composed of commuting Pauli strings can be diagonalized using a poly-depth Clifford circuit \cite{Dehaene.2003,van2020circuit}, it is possible to perform efficient Gibbs sampling for any such Hamiltonian by sampling its classical poly-depth dual counterparts, provided these can be efficiently sampled. A similar approach is studied in \cite{hwang2024gibbsstatepreparationcommuting}. Furthermore, Clifford circuits can be simulated efficiently on classical computers, so the Gibbs sampling algorithms we propose in this work can be realized entirely without using quantum circuits. Even in settings where equilibrium properties like free energies can be efficiently computed classically, quantum Gibbs sampling can be useful to simulate more complicated properties such as the behavior under quantum quenches.

In order to study explicit examples, in \cref{sec:dualitytoricising} we provide a formal proof of the 2D toric code being poly-depth dual to two decoupled classical Ising chains. In particular, as these can be efficiently sampled, this allows us to provide an \textbf{efficient Gibbs sampling algorithm for the 2D toric code for any system size whose runtime is independent of $\beta$}. Note that up to date it was only proven that the Davies generator associated to the 2D toric code was gapped \cite{Alicki_2009,ding2025polynomialtimepreparationlowtemperaturegibbs} and, more recently, had a positive modified log-Sobolev inequality (MLSI) \cite{stengele2025modifiedlogarithmicsobolevinequalities}. Therefore, we obtain a more efficient and practical sampler. 

Furthermore, we provide \textbf{explicit circuits mapping several well-known quantum models, such as the 3D toric code, Haah's code or the X-cube, to simple classical models}. The results obtained, which are summarized in \cref{tab:diagonalizationresults} and discussed in \cref{sec:comm_Pauli_oper}, are backed by a computer-assisted proof up to a large, finite system size.
Should these poly-depth dualities hold for every system size, our methods would provide, for the first time, a constructive method for performing efficient Gibbs sampling for these models. Therefore, the Gibbs state associated with the models could be efficiently sampled for every $\beta < \beta_0$, where $\beta_0 < \infty$ only in the 3D toric code case and $\beta_0=\infty$ in the remaining cases. 

Lastly, the notion of poly-depth duality can be extended to the context of Lindbladians. In this context, dual Lindbladians arise as an alternative to well-known Lindbladians such as the Davies generator \cite{davies_GeneratorsDynamicalSemigroups_1979} or the CKG Lindbladian \cite{Chen2023EfficientExact}. Indeed, if $\cL$ is fast/rapid mixing, then any Lindbladian which is \textit{dual} to it---as defined in \cref{sec:preservationmixingtimes}---is also fast/rapid mixing. Furthermore, if $\cL$ is efficiently implementable in a quantum computer, every dual to it will be efficiently implementable, irrespective of its locality, as long as the duality is given with respect to a poly-depth quantum circuit. 

In \cref{sec:preservationmixingtimes} it is shown that fundamental properties such as uniqueness and full-rankness of fixed points, as well as estimates on the mixing times of the dynamics, are preserved under conjugation by unitaries. This in particular provides us with a plethora of examples of Lindbladians satisfying rapid/fast mixing (which is implied by an MLSI/gap, respectively). Furthermore, given an initial sampler in terms of a Lindbladian with MLSI, we show that when restricting to dualities with poly-log circuits, we obtain an even more efficient sampling. 

\section{Dualities and efficient sampling}

In this work, we will consider a possibly infinite graph $V$ and increasing families of finite subgraphs $\{ \Lambda\}_{\Lambda\Subset V} $, with corresponding Hamiltonians $\{ H_\Lambda\}_{\Lambda\Subset V}$. We will typically consider square lattices and the spins of the system will be placed at the vertices or the midpoints of the edges of the lattices.

\begin{defi}[Poly-depth dual Hamiltonians, dual Lindbladians]\label{def:dual}
Two families of Hamiltonians $\{ H_\Lambda^1\}_{\Lambda\Subset V}$ and $\{ H_\Lambda^2\}_{\Lambda\Subset V}$ are \emph{poly-depth dual} if for every $\Lambda \Subset V$ there exists a unitary matrix $U_\Lambda$ which can be implemented by a circuit of polynomial depth in $|\Lambda|$ consisting of two-local gates, and such that $H^1_\Lambda = U_\Lambda H^2_\Lambda U_\Lambda^\dagger$.

\end{defi}

Note that our definition of duality is not standard; in \cite{weinstein2019universality} two Hamiltonians are defined to be dual if their partition functions are proportional, and thus the sum of their eigenvalues is preserved. In particular, our definition of duality, which preserves all the eigenvalues---and the number of interactions---implies that of \cite{weinstein2019universality}. On the other hand, the notion of duality introduced in \cite{HastingsWen.2005} is restricted to poly-log-depth circuits and only requires the preservation of the ground state energy---and degeneracy---of the Hamiltonian. 

Poly-depth dualities do not preserve many physical properties, such as topological order or any notion of locality. In particular, our duality sits between that studied in \cite{weinstein2019universality} and the one from \cite{HastingsWen.2005}, and is best characterized as a duality that \textit{preserves polynomial-time guarantees}. 
The lack of physicality of our duality is a crucial feature as it allows us to leverage efficient sampling algorithms for physically much simpler models to obtain algorithms for topologically ordered Hamiltonians.

The main tool that will be used throughout the rest of the paper is the following.

\begin{lem}
\label{lem:efficientgibbssampling}
    Fix $\beta<\infty$. If $\{H_\Lambda \}_{\Lambda\Subset V}$ and $\{ \widetilde{H}_\Lambda \}_{\Lambda\Subset V}$ are poly-depth dual families of Hamiltonians and there is an efficient quantum Gibbs sampler for $ \{ \sigma_\beta(H_\Lambda) \}_{\Lambda\Subset V}$, then there is another efficient quantum Gibbs sampler for $\{\sigma_\beta(\widetilde{H}_\Lambda)\}_{\Lambda\Subset V}$. 
\end{lem}
More specifically, for any $\Lambda$, if $H_\Lambda$ and $\widetilde{H}_\Lambda$ are related by a poly-depth unitary $U_\Lambda$, i.e. $H_\Lambda= U_\Lambda \widetilde{H}_\Lambda U_\Lambda^\dagger$ and $\sigma_\beta(H_\Lambda)$ is efficiently sampled with $\mathcal{C}_\Lambda$, then  $\sigma_\beta(\widetilde{H}_\Lambda)$ is efficiently sampled with $U_\Lambda \mathcal{C}_\Lambda$. 
Note that $\mathcal{C}_\Lambda$ acts on the physical and additional ancillary qubits, where the latter are used in order to generate the randomness as we start from an initial pure state and generate a mixed state.
In this work, we will find dualities to \emph{classical} models, which means the entire circuit $\mathcal C$ and randomness generation can be performed classically.
See \cref{fig:gibbssampler} for a better understanding.

From a complexity-theoretic point of view the above also rules out poly-depth dualities between computationally hard and easy Hamiltonians \cite{sly2012,Bergamaschi_2024}. It remains an open question whether such examples exist for local Hamiltonians with the same spectrum. Nevertheless, even in this case, finding explicit circuit dualities remains a nontrivial task that we address in the following sections.

\begin{figure}
    \centering
    \input{Figures/gibbssampler.tikz}
    \caption{Visual representation of the setting described in \cref{lem:efficientgibbssampling}. Note that the number of qubits is $|\Lambda| = n$, whilst the depth of $U_\Lambda$ is polynomial in $n$. $\mathcal{C}_\Lambda$ is assumed to be an efficient Gibbs sampler for $\sigma_\beta(H)$ taking $a$ ancilla qubits.}
    \label{fig:gibbssampler}
\end{figure}
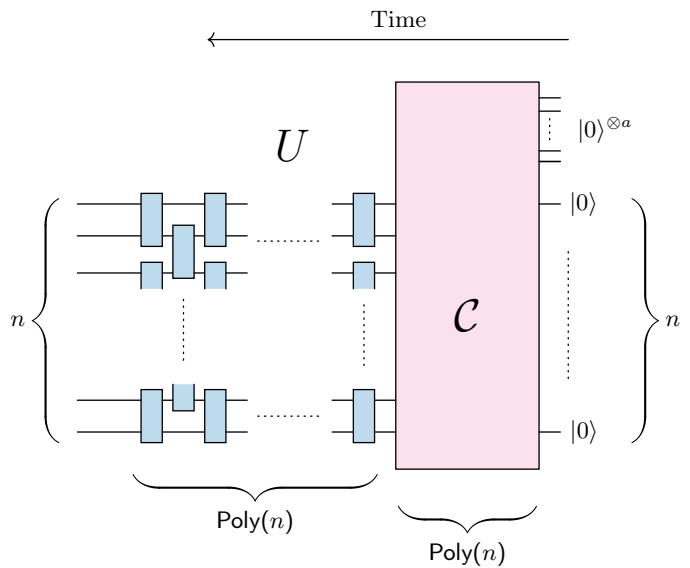

The proof of \cref{lem:efficientgibbssampling} is included in \cref{sec:appendix_technical}. This is in fact very powerful, as we will show that several well-known Hamiltonians---including some whose ground state(s) have non-trivial topological order---can be mapped via a poly-depth quantum circuit to very simple classical Hamiltonians. Thus, efficient sampling for the Gibbs states of those Hamiltonians will follow immediately from the efficient sampling of the Gibbs states of their classical counterparts.

Let us now show where the motivation for our definition of duality comes from. Let $H$ be a Hamiltonian
\begin{equation}
\label{eq:Hamiltonian}
H = \sum_{i = 1}^m \alpha_i H_i,    
\end{equation}
where $\{H_i\}_{i = 1}^m$ is a set of mutually commuting \textit{Pauli operators} and $\{\alpha_i\}_{i = 1}^m$ are some scalars. A Pauli operator is given by the tensor product of the Pauli matrices $\sigma_x, \sigma_y, \sigma_z$ and the identity.

It is known that the terms $\{H_i\}_{i = 1}^m$ can be simultaneously diagonalized, as they commute. Furthermore, this diagonalization can be performed using an explicit one-dimensional quantum circuit of quadratic depth in the number of spins of the system, as proven in \cite{van2020circuit}, following an algorithm originally presented in \cite{aaronson2004simulation}.

One key property of the aforementioned algorithm is its classical complexity of $\cO(n^2\max(m,n))$, where $n$ is the number of spins of the system and $m$ is the number of terms of $H$. Furthermore, the resulting circuit is a Clifford circuit of depth $\cO(n^2)$ \cite{van2020circuit}.

In this paper, we study several well-known quantum models of the form shown in \cref{eq:Hamiltonian}, and make use of this algorithm to obtain classical poly-depth dual Hamiltonians for each of them, which in general do not have a clear structure. For this reason, we apply a simple algorithm, which we call \textit{pseudo-Gaussian elimination algorithm} \ref{alg:gaussian}---which has a classical complexity of $\cO(n^2m^2)$, and obtains a quantum circuit of depth at most $\cO(mn^2)$---that allows us to further restructure, via conjugation by $\CX$ gates, the classical Hamiltonian obtained in order to identify its structure.

\section{Duality between 2D toric code and classical Ising chains}
\label{sec:dualitytoricising}

One of the main results of this work is the poly-depth duality between the two-dimensional toric code and two decoupled classical Ising chains. A duality between these models in the sense of having the same partition function had already been established in \cite{nussinov2009topological}, but here we extend this to our stronger definition of poly-depth duality, i.e. conjugation by unitaries that are additionally of poly-depth. The proof is deferred to \cref{sec:2Dtoriccode_proof}. Nevertheless, in this section we will provide a brief introduction to the model along with a discussion of the result.

The two-dimensional toric code \cite{kitaev2003toric,nussinov2008toric} is defined on a square lattice with periodic boundary conditions. We define these lattices following the same notation as in \cite{Lucia_Pérez-García_Pérez-Hernández_2023}; indeed, let $L \in \NN$, and define $\mathbb{S}_L$ as $\RR \slash \!\!\! \sim$, where the quotient is taken with respect to the relation $x \sim x + L$ for every $x \in \RR$. Without loss of generality, one can identify $\mathbb{S}_L$ with $[0, L)$. Let $(V_L, \mathcal{E}_L)$ be the square lattice on the torus $\mathbb{S}_L \times \mathbb{S}_L$ with vertices on the integer coordinates. Let $\Lambda_L$ be the set of spins of the system. In this case, one spin is located at the midpoint of every edge in $\mathcal{E}_L$ (see \cref{fig:toric3x3}).

\begin{figure}
    \centering
    \input{Figures/toric3x3.tikz}
    \caption{Visual representation of a $3 \times 3$ two-dimensional toric code. We have used pink crosses and blue squares to represent star and plaquette operators, respectively.}
    \label{fig:toric3x3}
\end{figure}

In order to define the 2D toric code Hamiltonian, we use the following notation; given a vertex of the lattice $v \in V_L$, we denote by $\partial v$ the set of the four spins that lay in the edges adjacent to $v$. We also define $p$ as any four-spin set such that its associated edges---which by a slight abuse of notation we also denote by $p \subset \mathcal{E}_L$---form a square. 

Thus, the Hamiltonian associated to this model is given by
\begin{equation}
\label{eq:HamiltonianTC1}
H_{\mathit{TC}} = - \sum_{v \in V_L} J_v A_v - \sum_{p \subset \mathcal{E}_L} J_p B_p,   
\end{equation}
where $J_v, J_p \in \RR$ for every $v \in V_L$ and every $p \subset \mathcal{E}_L$,
\[
A_v := \bigotimes_{i \in \partial v } \sigma_x^i,\quad B_p := \bigotimes_{i \in p} \sigma_z^i.
\]
The above operators are known as star and plaquette operators, respectively. See \cref{fig:toric3x3} for a visual representation of both of them. 

In this context, we prove the following theorem: 
\begin{theorem}[Theorem \ref{thm:result}, informal version]
\label{thm:informalduality2TC}
Let $H_{TC}$ be the two-dimensional toric code Hamiltonian defined on an $L \times L$ lattice, written as in \cref{eq:HamiltonianTC1}. Then, there exists a quantum circuit $C$ composed of $\cO(L^3)$ $\CX$ gates and $\cO(L^2)$ Hadamard gates such that 
\[
C \Big(\sum_{v \in V_L} J_v A_v\Big)C^\dagger
\]
is an Ising chain Hamiltonian and 
\[
C \Big(\sum_{p \subset \mathcal{E}_L} J_p B_p\Big)C^\dagger
\]
corresponds to another Ising chain Hamiltonian with disjoint support from the first.
\end{theorem}

The proof, which we defer to \cref{sec:2Dtoriccode_proof}, is inspired by the circuits produced by the algorithm from \cite{van2020circuit} and \cref{alg:gaussian}. 
From analyzing such finite-sized circuits, we find an explicit description of the circuit $C$ needed to diagonalize the Hamiltonian of the toric code for any lattice size $L \times L$. 

In fact, the circuit $C$ can be decomposed into two circuits, namely $V_1$ and $V_2$. When conjugating the toric code Hamiltonian by $V_1$, its operators get restricted to only act on a subset of their support, i.e.,
\[
V_1 A_v(V_1)^\dagger= \hspace{-0.3cm}\bigotimes_{i\in (\partial v)'\subset \partial v}\hspace{-0.3cm}\sigma^i_x,\quad \text{and}\quad V_1 B_p (V_1)^\dagger = \bigotimes_{i \in p' \subset p} \sigma_z^i,
\]
for every $v \in V_L$ and every $p \subset \mathcal{E}_L$. 

This results in two decoupled simple and classical systems, whose interaction graph is a tree. Finally, for each decoupled system we will construct $V_2$, which maps the resulting interaction terms in each decoupled system to Ising interactions (see \cref{sec:Isingtononinter}) thus resulting in two Ising chains. 

See \cref{fig:toricising} for a visual representation of the final two Ising chains that are dual to the toric code model in a $3 \times 3$ lattice. Notice that, as we will prove in \cref{sec:2Dtoriccode_proof}, the resulting Ising chains are of length $L^2-1$ and do not include every spin of the system, but rather there will always be two non-interacting spins, which are related to the four-dimensional ground space of the toric code. These two non-interacting spins can also be seen as the logical subspace of the 2D toric code.

Since our proof provides an explicit construction of the circuit realizing the duality, \cref{thm:informalduality2TC} has an additional application: Our algorithm allows to explicitly prepare Gibbs states within the logical sector \cite{bergamaschi2025rapidmixinggibbsstates}. An interesting open question that can be addressed experimentally with this is whether the logical information of Gibbs states within the logical sector could be recovered using active error correction routines.

The following theorem gives our main application of the duality in this section, the efficient ground and Gibbs state preparation for the toric code.
\begin{theorem}\label{thm:efficient_tc}
    The ground and Gibbs states of the 2D toric code can be prepared with a gate complexity of $O(L^3)$ for any $0\le\beta\le\infty$.
\end{theorem}

Note that our circuit model assumes gates between arbitrary pairs of qubits. In particular, if a more restrictive circuit connectivity is assumed, the size of the circuit may increase up to a $O(L^2)$ factor.

While we include the ground state preparation in the theorem, more efficient algorithms for this part are known \cite{bravyi2022adaptiveconstantdepthcircuitsmanipulating,Chen2024quantumcircuits,bravyi2006,yang2023acceleratinginexacthypergradientdescent,PRXQuantum.3.040315,10.21468/SciPostPhys.6.3.029}.
Note that the only error dependence is due to the sampling error of a \emph{single classical} bit with a given binary distribution, while the quantum circuit is inherently error-free for a gate set containing noiseless $\CX$ and $H$ gates.

The proof of this theorem only involves sampling from a classical Gibbs state of two Ising chains and then applying the explicit circuit from \cref{thm:informalduality2TC}. Sampling from a one-dimensional Ising chain can be achieved by a simple iterative procedure that we describe in \cref{sec:2Dtoriccode_proof}.

Prior approaches to efficient sampling exist based on Davies generators \cite{Alicki_2009} and modifications thereof \cite{ding2025polynomialtimepreparationlowtemperaturegibbs} achieving fast mixing with a polynomial scaling in $\beta$, which come with cumbersome and less efficient implementations based on implementing the Lindbladian evolution. Moreover, for the Gibbs state preparation, we achieve a circuit complexity of $\mathcal{O}(N^{3/2})$ on $N=|\Lambda| = \mathcal{O}(L^2)$, independently of $\beta$. To the best of our knowledge, this improves the results known up to date, which relied on dissipative Gibbs sampling algorithms using the Davies generator, and required a quantum circuit of complexity $\widetilde{\mathcal{O}}(N^{3}\exp(\beta))$ \cite{Alicki_2009} or $\widetilde{\mathcal{O}}(N^5\beta)$ \cite{ding2025polynomialtimepreparationlowtemperaturegibbs}, or more recently $\widetilde{\mathcal{O}}(N^{2}\exp(\beta))$  \cite{stengele2025modifiedlogarithmicsobolevinequalities}. Recently, an alternative approach studied in \cite{schmidhuber2025hamiltoniandecodedquantuminterferometry} obtains a classical algorithm running in $\mathcal{O}(N^4)$ time. Lastly, a depth of $\widetilde{\mathcal{O}}(N^{2})$ was obtained for the defected 2D toric code in \cite{hwang2024gibbsstatepreparationcommuting}. A more detailed discussion is deferred to \Cref{sec:implementation_davies,sec:previous_Gibbs_sampling}.

\begin{figure}
    \centering
    \input{Figures/toricIsing.tikz}
    \caption{Visual representation of the dual model obtained for a $3 \times 3$ two-dimensional toric code, which corresponds to two decoupled one-dimensional Ising chains. Note that the final dual Hamiltonian acts trivially on both spins $3$ and $16$. Pink and blue colors represent the two decoupled final Ising chains. Each bar denotes a $\sigma_z \otimes \sigma_z$ interaction, and every circle represents a $\sigma_z$ magnetic field.}
    \label{fig:toricising}
\end{figure}

\section{Duality between commuting Pauli operators and classical models}\label{sec:comm_Pauli_oper}

As indicated in the previous section, the explicit circuit obtained in \cref{thm:informalduality2TC} was inspired by a computational procedure, which we explain here together with a plethora of analogous poly-depth dualities.
We provide an algorithm, which takes as input a commuting Pauli Hamiltonian and outputs---deterministically in time polynomial in the number of sites and terms---a polynomially sized circuit that maps the Hamiltonian to a classical model (see \cite{github2025}). Before even considering asymptotic scaling, this can be seen as a practical tool to preprocess Hamiltonians for quantum state preparation tasks. We break down the sampling problem into an often simpler state preparation problem followed by postprocessing by an explicit efficient quantum circuit.

While the procedure, described in more detail in \cref{sec:appendixtableaus}, provably diagonalizes every given input Hamiltonian by performing operations that correspond to gates on a tableau containing the Hamiltonian coefficients, the locality of the final Hamiltonian cannot be predetermined. However, with a suitable design of \cref{alg:gaussian}, we were able to identify patterns in the algorithmic output corresponding to the eight interaction models given in the left column of \cref{tab:diagonalizationresults}. Tracing back these patterns allowed us to empirically establish poly-depth dualities in the sense of \cref{def:dual} between the initial lattice models and well-understood interaction terms such as the 1D Ising model. This results in the following conjecture. 

\begin{conj}\label{conj:dualities}
The poly-depth dualities listed in \cref{tab:diagonalizationresults} hold in the sense of \cref{def:dual}.
\end{conj}
Our evidence consists of computing the output of the above algorithm for system sizes up to $L=90$ (2D models) or $L=20$ (3D models), which constitutes a proof for such systems sizes and leaves us with little doubt about its correctness in arbitrary sizes.
If the claimed structure of the output is correct, the polynomial asymptotic circuit size provably follows from runtime bounds on the algorithm.
However, we omit a formal proof that would involve a similar analysis to the one in the previous section, as it involves a separate tedious description for each model.

Henceforth adopting this conjecture, we state the following application to state preparation. 

\begin{theorem}\label{thm:Gibbs_samplers_Pauli_models}
    Assuming~\cref{conj:dualities}, there are polynomial-time Gibbs sampling algorithms for the models in \cref{tab:diagonalizationresults} at any $0\le\beta<\infty$ except for the 3D toric code for which there is a polynomial-time Gibbs sampler for $\beta<\beta_0$ for some critical temperature $1/\beta_0$.
\end{theorem}

The proof of this theorem, analogously to the previous section, consists of applying the circuit to the respective classical Gibbs sampler, and hence only requires such a classical algorithm. These are given by the previously described Ising chains which are again sampled in the same way as described in the section before with the only exception being the 3D toric code.

For the 3D toric code, the dual system consists of local interactions with bounded degree interaction graphs. For these systems, efficient Gibbs sampling follows from results on the mixing time of Glauber dynamics at high temperature \cite{caputo2015approximatetensorizationentropyhigh}.
At low temperature, phase transitions in the 3D toric code have been studied in \cite{castelnovo2008threetoric}, such that we do not expect efficiency in that regime as thermal phase transitions are commonly linked to computational ones \cite{Harrow.2020}. It is not clear whether the critical temperature of the phase transition coincides with the limit of computational efficiency.

Note that the selection of models is motivated by the findings in \cite{weinstein2019universality}, which proves this duality in a weaker sense: The factorization of the partition function into partition functions of the dual models. Their results imply that computing the 3D toric code partition function is as hard as computing the partition function of the 3D Ising model, explaining the breakdown of the efficiency of our algorithms for this model at low temperatures.
\begin{table*}
    \centering
    \includegraphics[scale = 0.88]{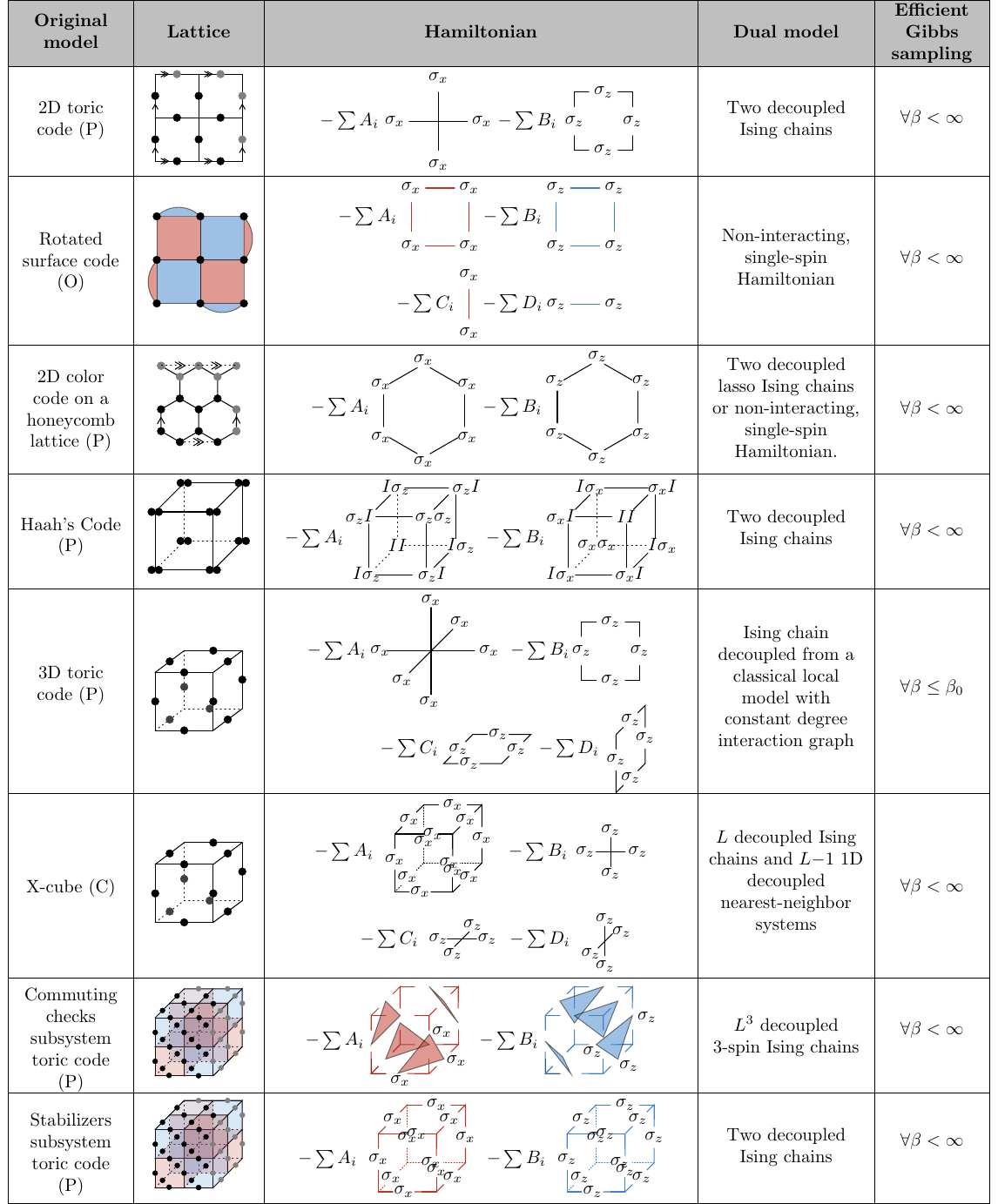}
    \caption{Summary of the results obtained for the different models considered. The boundary conditions considered depend on the model; most models have been studied under periodic (P) boundary conditions, whilst others have been studied with open (O) or cylindrical (C) boundary conditions. Except for the 2D toric code, for which we incorporate a formal proof of the duality, all of the above findings are verified for system sizes up to $90 \times 90$ in the two-dimensional models, and $20 \times 20 \times 20$ in the three-dimensional models. Efficient Gibbs sampling follows by \cref{lem:efficientgibbssampling} from the poly-depth duality of the models. Refer to \cref{sec:lasso} for a proof of the existence of an efficient Gibbs sampling method for the lasso Ising chains, while the other models follow from well-known results, see \cref{sec:appendixtableaus} for more details.}
    \label{tab:diagonalizationresults}
\end{table*}

\section{Preservation of mixing times under conjugation by unitaries}
\label{sec:preservationmixingtimes}

In this section, we apply the notion of duality explored above in the context of Lindbladians. Given a state $\rho\in \cS(\cH)$, let us recall that a Lindbladian is given by 
 \begin{equation} \label{eq GKSL Lindblad form (rep in unitary invariance)}
        \cL(\rho) = -i[H, \rho] + \sum_{k} \gamma_k \left[L_k\rho L_k^\dagger - \frac{1}{2}\{L_k^\dagger L_k, \rho\} \right],
    \end{equation}
with positive constants $\gamma_k$, a Hamiltonian $H=H^\dagger$, and (bounded) jump-operators $\{L_k\} \subseteq \cB(\cH)$. For an arbitrary unitary $U \in \cU(\cH)$ define the map $\widetilde{\cL}$ as
    \begin{equation} \label{eq rotated Lindbladian}
        \widetilde{\cL}:=\Ad_{U}\circ\cL\circ\Ad_{U^\dagger} \; 
    \end{equation}
with adjoint action $\Ad_{U}(X):=UXU^\dagger.$ Then, we say that $\cL$ and $\widetilde{\cL}$ are \textit{dual Lindbladians}. 

The following properties are preserved under duality of Lindbladians (see \cref{sec:mixing_times} for a precise formulation of these properties). 

\begin{theorem}[\Cref{thm:preservation_mixing_times}, informal version]\label{thm:preservation_mixing_times_informal}
     Let  $\cL$, $U$ and $\widetilde{\cL}$ be defined as above.  Then, the following statements hold:
    
    \begin{enumerate}
    
    \item If $\sigma$ is the unique fixed point of $\cL$, $\widetilde{\sigma} = U \sigma U^\dagger$ is the unique fixed point of $\widetilde{\cL}$.
    
    \item
    The spectral gap, MLSI and mixing time of $\cL$ coincide with those of $\widetilde{\cL}$.

\end{enumerate}
\end{theorem}
The proof of this result is deferred to \cref{sec:preservation_mixing_proof}. Note that it can be immediately extended to families of Lindbladians in order to show that rapid and fast mixing are preserved under dualities: $\{\widetilde{\cL}_\Lambda\}_{\Lambda \Subset V}$ achieves fast/rapid mixing, or has MLSI/spectral gap if and only if $\{\cL_\Lambda\}_{\Lambda \Subset V}$ does (cf. \cref{cor:unitary_invariance_families}). This immediately gives us a plethora of examples of systems satisfying a positive MLSI/spectral gap or having rapid/fast mixing, as explicitly written in  \cref{cor:examples_MLSI_gap}. In particular, if \cref{conj:dualities} holds, we can find a Lindbladian that satisfies a positive MLSI at every positive temperature for all these models except for the 3D toric code, for which such a Lindbladian will satisfy a positive MLSI at high enough temperature (cf. \cref{thm:Gibbs_samplers_Pauli_models}).

While we note that the above theorem makes no further assumption on the unitary, in the context of this work we are most interested in the case where $U$ is a poly-depth circuit. 
In \cite{Bergamaschi_2024}, the same type of duality has been observed to relate two Davies generators of a classical non-interacting and a commuting low-locality Hamiltonian, respectively.
Consequently, the two corresponding dual Davies generators come with an efficient implementation due to their low locality.
We improve upon their result, Lemma~2.3, by showing that the MLSI does in fact not decay with the circuit depth when conjugating the Davies Lindbladian as well.

Furthermore, we can consider polynomial-sized circuits as the unitaries relating strictly local Lindbladians to their duals as in the rest of this paper.
While the dual Lindbladian is no longer local, we notice that standard simulation techniques like \cite[Theorem 1]{Li2023SimulatingMarkovianOQS} based on block encodings yield an efficient implementation of the dual Lindbladian: Simply conjugating each application of the block encoding of Lindblad operators and unitary terms by the circuit yields an efficient implementation of the dual Lindbladian \footnote{While this demonstrates a family of Lindbladians with efficiently implementable duals, note that an implementation of the original Lindbladian and post-processing of the sample by application of the circuit is in general more efficient.}.

Let us conclude by mentioning a straightforward application of \cref{lem:efficientgibbssampling} and \cref{thm:preservation_mixing_times_informal}. By conjugation with poly-depth unitaries, efficient sampling is preserved, as shown in \cref{lem:efficientgibbssampling}. Additionally, the same conjugation in the Lindbladian preserves the mixing time (\cref{thm:preservation_mixing_times_informal}), and thus, by restricting the depth of the unitaries, if we sample with known Lindbladians we can improve the efficiency of the corresponding dual samplers (cf. \cref{cor:examples_MLSI_gap}). 

\section{Outlook}

In this paper, we introduce a new notion of duality based on conjugation by poly-depth circuits, and we show that efficient sampling is preserved under such a duality. We prove analytically that the 2D toric code is dual to two decoupled classical 1D Ising models for any system size. In addition, we give computer-assisted proofs of multiple other dualities between Hamiltonians composed of commuting Pauli operators and classical Hamiltonians up to a finite system size. We conjecture that these dualities extend to arbitrary system sizes. We leave it as an open question to formally prove the dualities in \cref{tab:diagonalizationresults}, which we expect to be achievable by a proof analogous to the case of the toric code. Another direction to explore in future works is the extension of the results of the current paper to higher-dimensional Paulis, as well as more general models.

A very similar notion of duality to that explored in this paper was already studied in \cite{Verstraete2009} to disentangle Hamiltonians, i.e. to map them to non-interacting models. They show that the circuit to disentangle the XY model is of depth $\mathcal{O}(n \log n)$, and mention that something similar could be done for the honeycomb and stabilizer states.  As the disentangling of the XY model goes through free fermions, we expect that the results of the current manuscript can be extended to fermions, and more generally to coherent states \cite{Perelomov.1986}, but leave this for future work. 

A natural question from our duality relations is to identify further Hamiltonians that are dual to ``easy'' Hamiltonians. More specifically, given a Hamiltonian whose associated Gibbs state can be efficient sampled (e.g. a non-interacting one), which Hamiltonians can we reach from it with a poly-depth circuit?

Lastly, while scaling to large system sizes will require fault-tolerance as our circuits are not shallow-depth, testing our approach for Gibbs state preparation for small systems on noisy quantum computers in the near future appears within reach. Indeed, the sampling part of the computation can be performed entirely on a classical computer whereas the quantum circuit comes with a very moderate gate count, when restricting to small system sizes. In particular, for the 2D toric code, the number of gates used in the duality circuit associated to a $3\times3$ lattice is $84$, and the gate count remains below $1000$ for a $7\times7$ lattice, which should make the implementation possible given error rates in the order of $10^{-3}$ achieved by current hardware. A detailed discussion including depth estimates may depend on the connectivity and is left for future work.

\section*{Acknowledgments}

P.P.V. and A.C. are grateful to Yuhan Liu for fruitful discussions. The authors are grateful to Sofyan Iblisdir for the useful comments and suggestions. 
S.O.S. acknowledges support from the UK Engineering and Physical Sciences Research Council (EPSRC) under grant number EP/W524141/1. N.S. and A.C. acknowledge the support of the Deutsche Forschungsgemeinschaft (DFG, German Research Foundation) - Project-ID 470903074 - TRR 352. F.V. acknowledges support by the UKRI grant EP/Z003342/1. This project was funded within the QuantERA II Programme which has received funding from the EU’s H2020 research and innovation programme under the GA No 101017733. P.P.V acknowledges support of the Spanish Ministry of Science and Innovation MCIN/AEI/10.13039/501100011033 (CEX2023-001347-S, CEX2019-000904-S, CEX2019-000904-S-21-2, PID2020-113523GB-I00, PID2023-146758NB-I00), Comunidad de Madrid (TEC-2024/COM-84-QUITEMAD-CM), Universidad Complutense de Madrid (FEI-EU-22-06), and the CSIC Quantum Technologies Platform PTI-001. This work has been financially supported by the Ministry for Digital Transformation and the Civil Service of the Spanish Government through the QUANTUM ENIA project call – Quantum Spain project, and by the European Union through the Recovery, Transformation and Resilience Plan – NextGenerationEU within the framework of the Digital Spain 2026 Agenda.

\sloppy
\bibliography{references}

\clearpage

\appendix
\renewcommand{\thesubsection}{\thesection.\arabic{subsection}}

\section{Notation and basic results}
\label{sec:setting}

\subsection{Notation}

Let us consider a possibly infinite graph $G = (V, E) $  with vertices $V$ and edges $E$. This graph will be endowed with the metric $d: V \times V \rightarrow \mathbb{R}_+$ given by the shortest path in the graph, and for $X,Y \subset V$,
\begin{equation}
d(X,Y):= \underset{x \in X}{\inf} \underset{y \in Y}{\inf} d(x,y) \, .
\end{equation}
We will denote by $\Lambda \Subset V$ that $\Lambda$ is a finite subset of $V$. In this paper, we will generally take $V \equiv \mathbb{Z}^D$ for $D$ any dimension, although some of our results can be extended to more general graphs. 

In order to describe the Hilbert space associated to the quantum spin system, we set at each site $x\in V$ a local Hilbert space $\mathcal{H}_x \equiv \mathbb{C}^d$, and thus the global Hilbert space associated to $\Lambda \Subset V$ is $\mathcal{H}_\Lambda = \underset{x \in \Lambda}{\bigotimes} \mathcal{H}_x $ of dimension $d^{|\Lambda|}$. In some of the examples presented in this paper, we associate the spins to the edges, rather than the vertices of the graph, but the definition of the Hilbert space is totally analogous. Moreover, whenever we consider a square lattice and place the spins either in the vertices or the edges, we denote by $\Lambda_L$ the set of spins defined in a cubic lattice of size $L$. We denote the algebra of bounded operators in $\Lambda$ by $\mathcal{B}_\Lambda \equiv \mathcal{B}(\mathcal{H}_\Lambda)$, its subset of Hermitian operators by $\mathcal{A}_\Lambda \equiv \mathcal{A}(\mathcal{H}_\Lambda)$, and the set of density matrices by  $\mathcal{S}_\Lambda \equiv \mathcal{S}(\mathcal{H}_\Lambda)$. For any two finite subsets $\Lambda \subset \Lambda'$, we identify $\mathcal{B}_\Lambda \subset \mathcal{B}_{\Lambda'}$ via the canonical linear isometry between $\mathcal{B}_\Lambda$ and $\mathcal{B}_{\Lambda'}$ given by $X \mapsto X \otimes \identity_{\Lambda' \setminus \Lambda}$. From this we can define the algebra of local observables as $\mathcal{B}_{\text{loc}}:= \underset{\Lambda \Subset V}{\bigcup} \mathcal{B}_\Lambda$, and its closure, which is the algebra of quasi-local observables.

In the following, we adopt the setting of a finite-dimensional Hilbert space $\cH$ with the \textit{Hilbert-Schmidt} (HS) inner product, given by $\langle X, Y \rangle_{\text{HS}}:= \Tr[X^\dagger Y]$, where $X^\dagger$ represents the transpose conjugate of $X$. Dual operators with respect to the HS inner product will be denoted with $^*$. Additionally, given a full-rank state $\sigma \in \cS(\cH)$, we will make use of the \textit{Kubo-Martin-Schwinger} (KMS) inner product, given by
\begin{equation}
    \langle X, Y \rangle^{\text{KMS}}_\sigma:= \Tr[\sigma^{1/2}X^\dagger \sigma^{1/2} Y], 
\end{equation}
for every pair of matrices $X, Y \in \cB (\cH)$. We will furthermore consider Schatten $p$-norms in these spaces, given for $X \in \cB (\cH)$ and any $p \in [1, \infty)$ by $\norm{X}_p:= \Tr[ |X|^p]^{1/p}$, with extension to the operator norm, the $\infty$-norm, by $\norm{X}_\infty:= \text{lim}_{p \rightarrow \infty }\norm{X}_p$. Finally, we will frequently consider the coupling of systems with arbitrarily large ancillas. In such a case, we will employ the \textit{diamond}-norm, given by $\norm{\mathcal{L}_X}_{\diamond}=\|\mathcal{L}_X\|_{1\to1,\text{cb}}$, which is the completely bounded $1\to1$ norm, i.e. $$\|\mathcal{L}_X\|_{1\to1,\text{cb}} := \sup_{n\in\mathbb{N}}\sup_{\rho\in\mathcal{S}(\mathbb{C}^n\otimes\mathcal{H})}\|(\id_n\otimes\mathcal{L}_X)(\rho)\|_1 .$$

\subsection{Basic operations}
\label{sec:appendixbasicops}
Throughout the text, we will frequently use the Pauli matrices on $\mathbb{C}^{2\times 2} $, which we will denote by $\{ \identity , \sigma_x , \sigma_y , \sigma_z\}$. Moreover, we recall the definition and basic properties of the quantum gates that we will use in the circuits constructed in this paper. All quantum gates act on one or two qubits. We will use in the rest of the paper the following 1-qubit gates: 
\begin{equation}
    \begin{array}{ccc}
         \textbf{Phase $\frac{\pi}{2}$ gate} & \phantom{aaaa} & \textbf{Hadamard gate} \\[1mm]
         S = \begin{pmatrix}
            1 & 0\\[-1mm]
            0 & i
        \end{pmatrix} \; , & \phantom{aaaa} &H = \frac{1}{\sqrt{2}}
        \begin{pmatrix}
            1 & \; 1 \\[-1mm]
            1 & -1
        \end{pmatrix} \; .
    \end{array}
\end{equation}
In two qubits, we will mainly use $\CX$ gates and $\CZ$ gates. They are respectively given by the following expressions:
\begin{equation}
    \begin{array}{cc}
         \textbf{$\CX$ gate} & \textbf{$\CZ$ gate} \\[1mm]
        \CX = 
     \begin{pmatrix}
            1 & 0 & 0 & 0 \\[-1mm]
            0 & 1 & 0 & 0 \\[-1mm]
            0 & 0 & 0 & 1 \\[-1mm]
            0 & 0 & 1 & 0  
        \end{pmatrix} \; , &\hspace{0.5cm}\CZ = \begin{pmatrix}
            1 & 0 & 0 & 0 \\[-1mm]
            0 & 1 & 0 & 0 \\[-1mm]
            0 & 0 & 1 & 0 \\[-1mm]
            0 & 0 & 0 & -1  
        \end{pmatrix} \; .
    \end{array}
\end{equation}
Note that in the last two examples we define the first qubit as the control qubit and the second one as the target qubit. For general $\CX$ and $\CZ$ gates with control spin $i$ and target spin $j$, we write $\CX(i, j)$ (rep. $\CZ(i, j)$). 

Let us recall the following important relations between $\CX$ gates and the Pauli matrices, which will be used later in the text:
\begin{align}
\label{eq:effectCNOT}
\begin{split}
\CX (\sigma_x \otimes \mathds{1}) \CX^\dagger &=\sigma_x \otimes \sigma_x,\\
\CX (\sigma_x \otimes \sigma_x) \CX^\dagger &= \sigma_x \otimes \mathds{1},\\
\CX (\mathds{1} \otimes \sigma_x) \CX^\dagger &= \mathds{1} \otimes \sigma_x,\\
\CX (\sigma_z \otimes \mathds{1}) \CX^\dagger &= \sigma_z \otimes \mathds{1},\\
\CX (\sigma_z \otimes \sigma_z) \CX^\dagger &= \mathds{1} \otimes \sigma_z,\\
\CX (\mathds{1} \otimes \sigma_z) \CX^\dagger &= \sigma_z \otimes \sigma_z.
\end{split}
\end{align}

\section{From 2D toric code to classical Ising chains}

Before proceeding to the main proof of the paper, which will be shown in \cref{sec:2Dtoriccode_proof}, let us present a simple yet illustrative example. Using an explicit circuit, we conjugate an already classical Hamiltonian, namely a classical Ising chain with open boundary conditions and no magnetic field, by $\CX$ gates to obtain a non-interacting classical Hamiltonian where each term acts on a single spin.

\subsection{From classical Ising chain to non-interacting Hamiltonian}
\label{sec:Isingtononinter}

Consider a classical Ising chain of length $L$ with open boundary conditions and no magnetic field;
\[
H = -\sum_{i = 1}^{L-1} J_i \sigma_z^i \sigma_z^{i+1},
\]
where $J_i \in \RR$ for every $i \in \{1, \dotsc, L-1\}$, and $\sigma^i_z$ denotes a $\sigma_z$ Pauli matrix acting on the $i$-th spin of the lattice. 

\begin{figure}
    \centering
     \input{Figures/1Dising.tikz}
    \caption{Visual representation of the different steps in the transformation from an Ising model to a decoupled model in a chain with $5$ spins. Each bar represents a $\sigma_z \otimes \sigma_z$ interaction, and each circle denotes a single-site $\sigma_z$ term.}
    \label{fig:1DIsing}
\end{figure}

The idea is to conjugate $H$ by $\CX$ gates in order to obtain a simpler classical Hamiltonian. We choose the gates so that they have their control and target qubits situated in adjacent spins. In particular, we consider
\[
U := \CX(1, 2)\CX(2, 3)\cdots\CX(L-1, L).
\]
Note that for every $i \in \{1, \dotsc, L-1\}$, the gate $\CX(i, i+1)$ will only affect two interactions of the Hamiltonian obtained by conjugation with the previous gates. These interactions are $\sigma_z^{i}\sigma_z^{i+1}$ and $\sigma_z^{i-1}\sigma_z^{i}$. In fact, the latter will remain unchanged, whilst the former will get transformed to $\sigma_z^{i+1}$ (cf. \cref{eq:effectCNOT}). Thus, the resulting Hamiltonian after conjugation by $U$ is
\[
U H U^\dagger = -\sum_{i = 2}^L J_{i-1}\sigma_z^i.
\]
See \cref{fig:1DIsing} for a visual representation of the intermediate steps for an Ising chain with $5$ spins. In particular, it is clear that the depth of the circuit obtained is linear in the system size.
Note that the two-fold ground-state degeneracy of $H$ carries over to the degeneracy of the first spin, which does not lie in the support of $UHU^\dagger$.

Using the above circuit to sample from the Gibbs state of an Ising chain is in fact equivalent to simulating the corresponding Markov chain, where the non-interacting systems correspond to a source of randomness (with appropriate weight) and with the $\CX$ gates applying the conditional probabilities to the consecutive spins.

\subsection{Duality between 2D toric code and classical Ising chains; a formal proof}
\label{sec:2Dtoriccode_proof}

We recall the notation from the main text for completeness.

The two-dimensional toric code \cite{kitaev2003toric,nussinov2008toric} is defined on a square lattice $(V_L, \mathcal{E}_L)$ on the torus $\mathbb{S}_L \times \mathbb{S}_L$ with vertices on the integer coordinates. Let $\Lambda_L$ be the set of spins of the system, where one spin is located at the midpoint of every edge in $\mathcal{E}_L$ (see \cref{fig:2dtoriclattice}). 

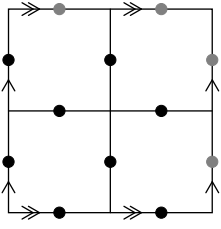
\begin{figure}[]
    \centering
    \input{Figures/latticetoric2D.tikz}
    \caption{Example of a $2 \times 2$ two-dimensional toric code lattice. We have shown the periodic boundary conditions using single and double arrows.}
    \label{fig:2dtoriclattice}
\end{figure}

For simplicity, we will label the spins of the lattice as $(i,j,k)\in \mathbb Z_L^2 \times \{h,v\}$ with $(i,j,v)$ and $(i,j,h)$ corresponding to the spins below and right of the $(i,j)$ vertex of the lattice (see \cref{fig:2DToricCoords}). Furthermore, $i$ and $j$ grow to the bottom and to the right, respectively.

\begin{figure}
    \centering
    \input{Figures/toriccoords.tikz}
    \caption{Visual representation of the coordinate system that will be used throughout the section.}
    \label{fig:2DToricCoords}
\end{figure}
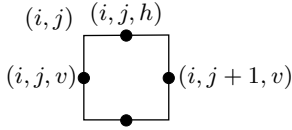

The 2D toric code Hamiltonian can be defined as 
\begin{equation}
\label{eq:HamiltonianTC}
H_{\mathit{TC}} = - \sum_{(i, j) \in \ZZ_L^2 } J_{(i, j)} A_{(i, j)} - \sum_{(i, j) \in \ZZ_L^2} \tilde{J}_{(i, j)} B_{(i, j)}, 
\end{equation}
$J_{(i, j)}, \tilde{J}_{(i, j)} \in \RR$ for every $(i, j) \in \ZZ_L^2$. $A_{(i, j)}$ and $B_{(i, j)}$ denote the star and plaquette operators, respectively, which can be written as
\begin{align*}
A_{(i,j)}&= {\sigma_x}^{(i,j,h)}{\sigma_x}^{(i,j,v)}{\sigma_x}^{(i-1,j,v)}{\sigma_x}^{(i,j-1,h)},\\
B_{(i,j)}&= \sigma_z^{(i,j,h)}\sigma_z^{(i,j,v)}\sigma_z^{(i+1,j,h)}\sigma_z^{(i,j+1,v)}, 
\end{align*}
where $\sigma^i_x$ (resp. $\sigma^i_z$) denotes a $\sigma_x$ (resp. $\sigma_z$) Pauli matrix acting on the $i$-th spin of the lattice, with $i \in \ZZ_L^2 \times \{h, v\}$. See \cref{fig:2dtoricinteractions} for a visual representation of both of them. 

See \cite{iblisdir2010} for a thorough analysis of the eigenstates, partition function, and eigenenergies of this model. While these insights could allow for the preparation of its Gibbs state, exploring this lies beyond the scope of the present work.

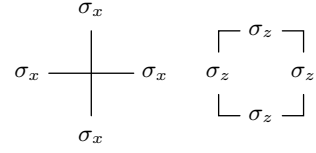
\begin{figure}[]
    \centering
    \input{Figures/2Dtoric.tikz}
    \caption{Visual representation of the star and plaquette interactions $A_{(i, j)}$ and $B_{(i, j)}$---left and right, respectively---of a 2D toric code.}
    \label{fig:2dtoricinteractions}
\end{figure}

In the proof, we will provide a bipartition $\Lambda_A\cup\Lambda_B=\Lambda_L$ of the lattice and an explicit quantum circuit $C$ with the following properties:
\begin{itemize}
    \item The circuit is of depth $\mathcal \cO(L^3)$ for every lattice of size $L \times L$, and only consists of $\CX$ and Hadamard gates.
    
    \item For every star operator $A_v$ and every plaquette operator $B_p$, $\supp(C A_v C^\dagger)\subset\Lambda_A$, and $\supp(C B_p C^\dagger)\subset\Lambda_B$, i.e, the sum of star and the sum of plaquette operators are mapped onto two classical Hamiltonians with  no interaction between them.

    \item The component $\Lambda_B$ includes two spins $s_1, s_2$ which are not contained in the support of the final Hamiltonian. 

    \item The final Hamiltonian associated to each component $\Lambda_A$ and $\Lambda_B \setminus \{s_1, s_2\}$ corresponds to an Ising chain with a parallel magnetic field at each end. 
\end{itemize}

In particular, we will prove the following statement: 

\begin{theorem}
\label{thm:result}
Let $H_{\mathit{TC}}$ be the Hamiltonian of an $L \times L$ two-dimensional toric code system. There exists a quantum circuit $C$ composed of $\cO(L^3)$ $\CX$ and $\cO(L^2)$ Hadamard gates such that $C H_{\mathit{TC}}\, C^\dagger$ is diagonal. 

Furthermore, let $H_{\mathit{TC}} = H_1 + H_2$, where $H_1$ (resp. $H_2$) denotes the star (resp. plaquette) operators of $H_{\mathit{TC}}$. Then $CH_1C^\dagger$ corresponds to a Hamiltonian with an all-to-all $\sigma_z$ interaction supported on $\Lambda_A$, and a parallel magnetic field on each site of $\Lambda_A$. Similarly, $C H_2 C^\dagger$ corresponds to an all-to-all $\sigma_z$ interaction on every spin in $\Lambda_B$ except for two spins, and a parallel magnetic field on each site in $\Lambda_B$ except for the aforementioned two spins. 
\end{theorem}

Note that the above statement does not mention the Ising chains we described previously, but rather describes a different system with an all-to-all interaction. Nevertheless, both systems are equivalent: 

\begin{lem}
\label{lem:equivalencelemma}
Let $\Lambda = \{s_i\}_{i= 1}^N$ be some set of spins and let $H_1$, $H_2$ be two Hamiltonians supported in $\Lambda$, defined as
\[
H_1 := -\sum_{i = 1}^N J_ i \sigma_z^i - \bigotimes_{s = 1}^N J_{N+1} \sigma_z^i,
\]
and 
\[
H_2 := -\sum_{i = 1}^{N-1} J_{i+1} \sigma_z^i \sigma_z^{i+1} - J_1 \sigma_z^1 - J_{N+1} \sigma_z^L,
\]
where $J_i \in \RR$, for every $i \in \{1, \dotsc, N+1\}$. Then, there exists some quantum circuit $U$ composed of $N-1$ $\CX$ gates such that
\[
U H_1 U^\dagger = H_2.
\]

Furthermore, the statement remains true up to any reordering of the spins in $\Lambda$.
\end{lem}
Note that $H_1$ corresponds to an all-to-all $\sigma_z$ interaction and a parallel magnetic field at each spin, whilst $H_2$ corresponds to an Ising chain with a parallel magnetic field at each end.

Although the mapping goes from a 2D to a 1D model, note that the conjugation procedure preserves the number of free parameters; each Hamiltonian term of the Ising chains is in one-to-one correspondence to a Hamiltonian term of the two-dimensional toric code. In particular, each Ising chain Hamiltonian has $L^2$ terms.

\begin{proof}
We make use of the relations \cref{eq:effectCNOT} and recall their action on the terms at present. Note that conjugation is a linear operation, thus the coefficients $J_i$ do not get modified.  

We start by conjugating $H_1$ by $\CX(1, 2)$. This gate will affect two of its interactions, namely
$\bigotimes_{s = 1}^N \sigma_z^i$ and $\sigma_z^2$, since $\CX(1, 2)\sigma_z^1\CX(1, 2)^\dagger = \sigma_z^1$. 

Following the rules presented in \cref{eq:effectCNOT}, we conclude that 
\[
\CX(1, 2)\Big(\bigotimes_{s = 1}^N \sigma_z^i\Big) \CX(1, 2)^\dagger = \bigotimes_{s = 2}^N \sigma_z^i,
\]
and
\[
\CX(1, 2)\sigma_z^2 \CX(1, 2) = \sigma_z^1 \sigma_z^2.
\]

Conjugating now the resulting Hamiltonian by $\CX(2, 3)$, we transform $\bigotimes_{s = 2}^N \sigma_z^i$ into $\bigotimes_{s = 3}^N \sigma_z^i$, and transform $\sigma_z^3$ into $\sigma_z^2 \sigma_z^3$. 

Repeating this process until we apply $\CX(N-1, N)$, we will transform 
\[
\bigotimes_{s = 1}^N \sigma_z^i
\]
into 
\[
\sigma_z^N.
\]
Furthermore, for every $i \in \{2, \dotsc, N\}$, the magnetic field $\sigma_z^i$ will be mapped to $\sigma_z^{i-1}\sigma_z^i$. Note that the only interaction that remains untouched is $\sigma_z^1$. 

This way, the final circuit is given as 
\[
U := \CX(N-1, N)\CX(N-2, N-1)\cdots \CX(1, 2),
\]
and $U H_1 U^\dagger = H_2$. 
\end{proof}

Equipped with the proof for the above auxiliary lemma, let us briefly discuss the notation and the main structure of the proof of \cref{thm:result}. 

Moreover, we include figures throughout the section which correspond to a $4 \times 4$ toric code lattice. The size has been chosen for readability. Nevertheless, the proof is general and does not depend on the specific lattice size, as long as it is an $L \times L$ square with periodic boundary conditions. 

We will study the star and plaquette interactions separately. In order to make the proof easier to follow, we decompose the circuit $C$ into three sub-parts: 
\begin{itemize}
    \item A circuit $\tilde C$ composed of commuting $\CX$ gates with their control qubits contained in $\Lambda_A$ and their target qubits contained in $\Lambda_B$. 
    \item A single layer of Hadamard gates acting on every spin of $\Lambda_A$. 
    \item A final circuit $\hat C$ composed of non-commuting $\CX$ gates, with their control and target qubits either  both in $\Lambda_A$ or both in $\Lambda_B$. 
\end{itemize}

Moreover, we will decompose $\tilde C$ into four sub-circuits. In order to define each sub-circuit, let us define the set of gates that compose them. Note that in general this is not the way to define a quantum circuit, as the order in which the gates are applied can change its outcome. Nevertheless, in this case, the gates considered are always mutually commuting---as the control qubits of the gates in the set never act as target qubits and vice-versa---so the sets uniquely identify the circuit. 

The circuit $\tilde C$ will allow us to decouple the initial Hamiltonian $H_{\mathit{TC}}$ into two non-interacting systems, whilst the circuit $\hat C$ will allow us to identify each non-interacting system as a classical Ising chain. 

\begin{proof}[Proof of \cref{thm:result}]
We begin by studying the effect of $\tilde C$ onto the initial Hamiltonian. Let us define its sub-circuits. First, consider a horizontal line of $\CX$ gates acting on the first row of the lattice; let $(0, L-1, h)$ be the last spin in the top row of the lattice. We will consider the gates  
\[
C_1 := \{\CX((0, i, h), (0, L-1, h)) : 0 \leq i < L-1\}. 
\]

We will also consider vertical lines of $\CX$ gates acting on every column of the lattice. Thus, for every spin in the last row of the lattice $(L-1, i, v)$, $0 \leq i \leq L-1$, we consider
\[
C_2^i := \{\CX((k, i, v), (L-1, i, v)) : 0 \leq k < L-1\}. 
\]

See \cref{fig:linegates} for a visual representation of the sets $C_1$ and $C^i_2$, $0 \leq i \leq L-1$. 

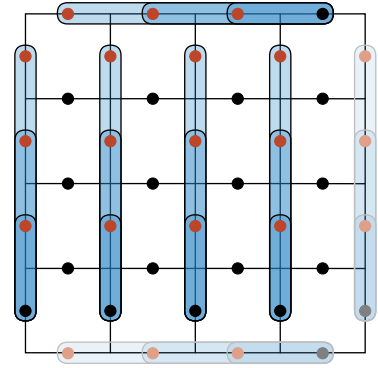
\begin{figure}
    \centering
    \input{Figures/gates1.tikz}
    \caption{Visual representation of the $\CX$ gates from the sets $C_1$ (horizontal) and $C_2^i$ (vertical), $i \in \{0,\dotsc ,3\}$. The control qubits are marked in red. }
    \label{fig:linegates}
\end{figure}

Next, for every site $(i, j, h)$ with $1 \leq i \leq  L-1$, and $0 \leq j < L-1$ we consider the set 
\begin{align*}
C_3^{(i, j)} &:= \bigcup_{0 \leq k < i}\{\CX((k, j, v),(i, j, h))\} 
\\&\quad \bigcup_{0 \leq k < i}\{\CX((k, j+1, v),(i, j, h))\}
\\&\quad \cup \{\CX((0, j, h),(i, j, h))\}.
\end{align*}
See \cref{fig:gatesC3} for a visual representation of the sets $C^{(i, j)}_3$. 

\begin{figure*}
  \centering
  \begin{subfigure}[b]{0.29\textwidth}
    \centering
    \input{Figures/gates2.tikz}
    \caption{Visual representation of the sets $C_3^{(1, j)}$}
    \label{fig:gates1}
  \end{subfigure}
  \hfill
  \begin{subfigure}[b]{0.29\textwidth}
    \centering
    \input{Figures/gates3.tikz}
    \caption{Visual representation of the sets $C_3^{(2, j)}$}
    \label{fig:gates2}
  \end{subfigure}
  \hfill
  \begin{subfigure}[b]{0.29\textwidth}
    \centering
    \input{Figures/gates4.tikz}
    \caption{Visual representation of the sets $C_3^{(3, j)}$}
    \label{fig:gates3}
  \end{subfigure}

  \caption{Visual representation of the sets $C_3^{(i, j)}$ in a $4 \times 4$ lattice. Again, the control qubits have been marked in red. }
  \label{fig:gatesC3}
\end{figure*}
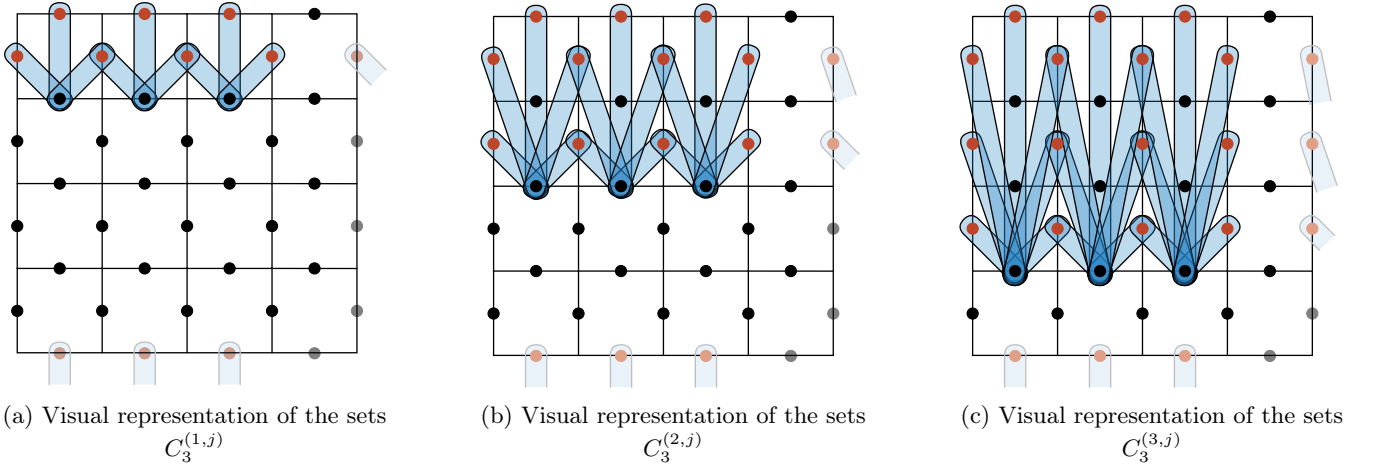

Lastly, for every site $(i, L-1, h)$, $1 \leq i \leq L-1$ we consider the set of gates
\begin{align*}
C_4^i &:= \bigcup_{0 \leq k < L-1}\{\CX((0, k, h), (i, L-1, h))\}\\
&\quad \bigcup_{0 \leq k <i} \{\CX((k, 0, v),(i, L-1, h))\}\\
&\quad \bigcup_{0 \leq k < i} \{\CX((k, L-1, v),(i, L-1, h))\}. 
\end{align*}
See \cref{fig:gatesC4} for a visual representation of the sets $C_4^i$. 

\begin{figure*}
  \centering
  \begin{subfigure}[b]{0.3\textwidth}
    \centering
    \input{Figures/gates5.tikz}
    \caption{Visual representation of the set $C^1_4$.}
    \label{fig:gates5}
  \end{subfigure}
  \hfill
  \begin{subfigure}[b]{0.3\textwidth}
    \centering
    \input{Figures/gates6.tikz}
    \caption{Visual representation of the set $C^2_4$.}
    \label{fig:gates6}
  \end{subfigure}
  \hfill
  \begin{subfigure}[b]{0.3\textwidth}
    \centering
    \input{Figures/gates7.tikz}
    \caption{Visual representation of the set $C^3_4$.}
    \label{fig:gates7}
  \end{subfigure}

  \caption{Visual representation of the sets $C^i_4$ in a $4 \times 4$ lattice. The control qubits are marked in red.}
  \label{fig:gatesC4}
\end{figure*}
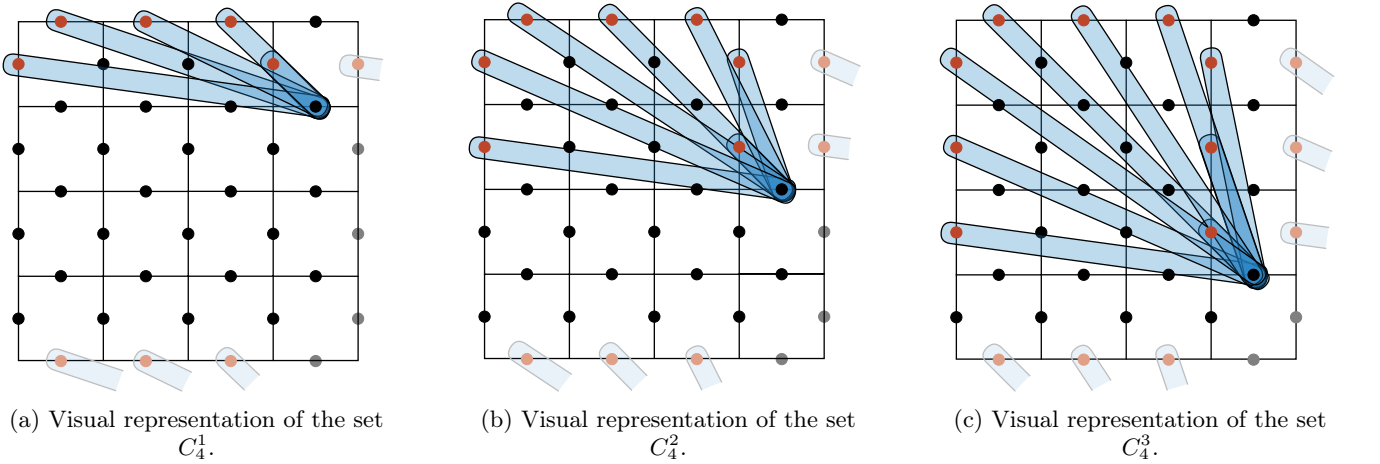

The above sets allow us to define the four sub-circuits, and the final circuit $\tilde C$. By a slight abuse of notation, we define
\begin{align*}
C_1 &:= \prod_{U \in C_1} U \, ,\\
C_2 &:= \prod_{0\leq i \leq L-1} \; \prod_{U \in C_2^i}U \, ,\\
C_3 &:= \prod_{\substack{1 \leq i \leq L-1\\0 \leq j < L-1}} \; \prod_{U \in C_3^{(i, j)}} U \, ,\\
C_4 &:= \prod_{1 \leq i \leq L-1} \;\prod_{U \in C_4^i} U \, ,\\
\tilde C &:= C_4 C_3 C_2 C_1 \, .
\end{align*}

As previously mentioned, the circuit $\tilde C$ allows us to partition the set of spins of the model, $\Lambda_L$, into two mutually non-interacting spin systems. We will denote the spins of each system as $\Lambda_A$ and $\Lambda_B$:
\begin{align*}
\Lambda_A&=\{(0,j,h) : \;0\le j< L-1\}
\\&\quad \cup\{(i,j,v) :  0 \le i<L-1, j \in \ZZ_L\},\\
\Lambda_B&=\Lambda_L \setminus \Lambda_A.
\end{align*}

We will prove that, after the conjugation by $\tilde C$, the star operators act on $\Lambda_A$ only while the plaquette operators act on $\Lambda_B$ only,  see \cref{fig:twocombs} for a visual representation of the two decoupled systems. 

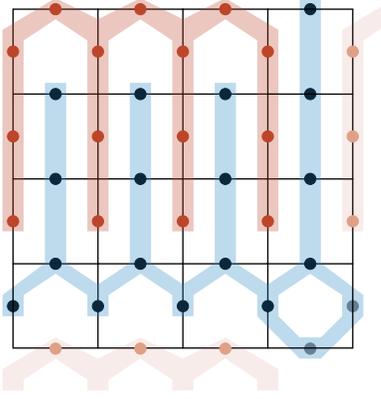
\begin{figure}
    \centering
    \input{Figures/combstoric.tikz}
    \caption{Visual representation of the two decoupled systems in a $4 \times 4$ lattice. The two systems are highlighted in red and blue, respectively.}
    \label{fig:twocombs}
\end{figure}

Let us first observe that for every gate $\CX(v,w)$ in $\tilde C$, the control qubit $v$ lies in $\Lambda_A$, while the target qubit $w$ lies in $\Lambda_B$.

As we noted earlier, we will study the star and plaquette interactions separately. Let us first focus on the plaquette operators. We will consider the action of $\tilde C$ on the plaquette terms $B_{(i,j)}$, $(i, j) \in \ZZ_L \times \ZZ_L$.

Let $B$ be any $\sigma_z$-string operator, i.e., an operator that is a tensor product of $\sigma_z$ operators over all sites in its support.
Let $\CX(v, w)$ be a $\CX$ gate with control qubit $v$ and target spin $w$. Let us recall \cref{eq:effectCNOT}, which can be interpreted as:
\begin{itemize}
    \item If $v$ and $w$ lie in the support of $B$, then $\CX(v, w)B\CX(v, w)^\dagger$ is a $\sigma_z$-string operator with support $\supp(B) \setminus \{v\}$.
    \item If $w$ lies in the support of $B$ but $v$ does not, then $\CX(v, w) B \CX(v, w)^\dagger$ is a $\sigma_z$-string with support $\supp(B) \cup \{v\}$. 
    \item If $w$ does not lie in the support of $B$, the interaction remains unchanged after conjugation.
\end{itemize}

The above effect can be summarized as follows: Given a Pauli-$\sigma_z$ string, the effect of conjugation by the gate $\CX(v,w)$ is to multiply by $\sigma_z^v$ if a $\sigma_z$-operator is present at site $w$.

Our goal is to prove the following characterization of the terms $\tilde{C} B_{(i, j)} \tilde{C}^\dagger$:
\begin{enumerate}
    \item For every $(i, j) \in \ZZ_L \times \ZZ_L$, $\tilde{C} B_{(i, j)} \tilde{C}^\dagger$ is a $\sigma_z$-string operator.
    \item For every $(i, j) \in \ZZ_L \times \ZZ_L$, $\supp (\tilde{C} B_{(i, j)} \tilde{C}^\dagger) = \Lambda_B\cap\supp B_{(i, j)}$.
\end{enumerate}
The first property is immediate from the above discussion of the action of $CX$-gates on $\sigma_z$-strings.

We will make a number of case distinctions regarding the index $(i,j)$ of a plaquette interaction. 

\begin{itemize}
\item $0<i<L-1$: The gates in $C_1$ have no overlap with these terms nor do the target qubits of $C_2$, so they do not affect $B_{(i, j)}$. 
\begin{itemize}
    \item $j<L-1$:
    The target qubit of all gates in $C^{(i,j)}_3$ is $(i,j,h)$ and thereby they affect the operators $B_{(i,j)}$ and $B_{(i-1,j)}$ only.
In addition, since the control qubits of the gates in $C_3^{(i,j)}$ and $C_3^{(i+1,j)}$ are identical except for $(i,j,v)$ and $(i,j+1,v)$, and multiplication by $\sigma_z^2=\id$ has no effect, we find
\[
C_3B_{(i,j)}C_3^\dagger = \sigma_z^{(i,j,h)}\sigma_z^{(i+1,j,h)}\,.
\]
The target qubits of $C_4$ are not contained in the supports of the $B_{(i,j)}$ for $j\neq L -1$.
\item $j=L-1$:
The target qubits of $C_3$ are not  contained in the supports of the $B_{(i,L-1)}$.
The action of the $C_4^i$ circuits, analogously to before, only affects $B_{(i,L-1)}$ and $B_{(i-1,j)}$. Also analogously when considering the circuits acting on $B_{(i,L-1)}$, the control qubits of the gates in $C_4^i$ and $C_4^{i-1}$ are identical except for the sites $(i,L-1,v)$ and $(i,0,v)$ and as such we find
\[
C_4B_{(i,L-1)}C_4^\dagger = \sigma_z^{(i,L-1,h)}\sigma_z^{(i+1,L-1,h)}\,.
\]
\end{itemize}
Overall we confirmed that 
\[
\tilde CB_{(i,j)}\tilde C^\dagger = \sigma_z^{(i,j,h)}\sigma_z^{(i+1,j,h)}\,,
\]
for every $0 < i < L-1$ and every $0 \leq j \leq L-1$.

\item $i=0$: 
The target qubit of the gates in $C_2$ do not lie in the support of $B_{(0,j)}$.
\begin{itemize}
\item $0 \le j< L-1$:
The target qubit of $C_1$ does not lie in the support of $B_{(0,j)}$.
The circuits $C_3^{(1,j)}$ have their target qubits in the support of $B_{(0,j)}$, yielding
\[
C_3^{(1,j)}B_{(0,j)}C^{(1,j)\dagger}_3 = \sigma_z^{(1,j,h)}.
\]
The other gates in $C_3$ have no target qubits in $B_{(0,j)}$. The $C_4$ gates have no target spins in $B_{(0,j)}$.

\item $j=L-1$: 
The gates in $C_3$ have no target qubits in $B_{(0,L-1)}$. All gates in $C_1$ and $C_4^1$, however, do. Their action cancels on all spins except for $(0,0,v)$ and $(0,L-1,v)$. Therefore, 
\[
C_4C_1B_{(0,L-1)}C_1^\dagger C_4^\dagger = \sigma_z^{(0,L-1,h)}\sigma_z^{(1,L-1,h)}\,.
\]
\end{itemize}
\item $i=L-1$:
\begin{itemize}
\item $j\neq L$: The target qubits of the gates in $C_1$ and $C_4$ do not lie in the supports of $B_{(L-1,j)}$.
$B_{(L-1,j)}$ is only affected by $C_3^{(L-1,j)}$, $C_2^j$, and $C_2^{j+1}$, whose action cancels except for the multiplication by $\sigma_z^{(0,j,h)}$. Therefore
\[
C_3C_2 B_{(L-1,j)}C_2^\dagger C_3^\dagger=\sigma_z^{(L-1,j,h)}\sigma_z^{(L-1,j,v)}\sigma_z^{(L-1,j+1,v)}.
\]
\item $j=L-1$:
The target qubits of $C_3$ do not lie in the support of $B_{(L-1,L-1)}$.
The action of $C_4^{L-1}$ cancels exactly with the joint action of $C_1$, $C_2^0$ and $C_2^{L-1}$ leaving us with
\[
C_4C_2C_1 B_{(L-1,L-1)}C_1^\dagger C_2^\dagger C_3^\dagger= B_{(L-1,L-1)},
\]
as desired.
\end{itemize}
\end{itemize}

This way, we have proven that for every $(i, j) \in \ZZ_L \times \ZZ_L$, 
\[
\supp(\tilde C B_{(i, j)} \tilde C^\dagger) = \Lambda_B \cap \supp B_{(i, j)}. 
\]
See \cref{fig:finalplaquettes} for a visual representation of the final plaquette operators.

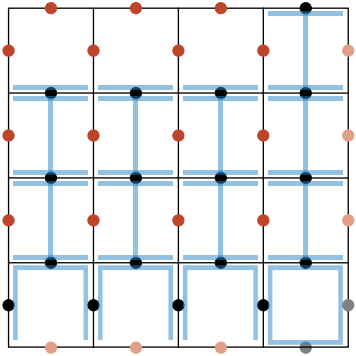
\begin{figure}
    \centering
    \input{Figures/finalplaquettes.tikz}
    \caption{Visual representation of the final plaquette interactions. The two decoupled systems are marked in red and black, respectively. }
    \label{fig:finalplaquettes}
\end{figure}

Let us now study the star terms. As we showed when studying the plaquette operators, let $A$ be any $\sigma_x$-string operator, let $\CX(v, w)$ be a $\CX$ gate with control qubit $v$ and target spin $w$. \cref{eq:effectCNOT} can be interpreted as:
\begin{itemize}
    \item If both spins lie in the support of $A$, then $\CX(v, w) A \CX(v, w)^\dagger$ is a $\sigma_x$-string operator with support $\supp(A)\setminus \{w\}$.
    \item If $v$ lies in the support of $A$ but $w$ does not, then $\CX(v, w) A \CX(v, w)^\dagger$ is a $\sigma_x$-string operator with support $\supp(A)\cup\{w\}$
    \item If $v$ does not lie in the support of $A$, the interaction remains unchanged after conjugation. 
\end{itemize}

From this, it is clear that, in the worst case, after conjugating the star operators with $\tilde C$, their support could be extended to some of the spins in $\Lambda_B$, but never to any of the spins in $\Lambda_A$. In the following paragraphs, we will prove that in fact
\[
\supp (\tilde{C} A_{(i, j)} \tilde{C}^\dagger) = \Lambda_A\cap\supp(A_{(i,j)}),
\]
for every $(i, j) \in \ZZ_L \times \ZZ_L$.

Let us fix some $w \in \Lambda_B$. We define $\mathcal{C}_w$ as the set of $\CX$ gates in $\tilde C$ that have $w$ as target qubit:
\[
\mathcal{C}_w := \{\CX(v,w) : \CX(v,w) \in \tilde C \}.
\]

Furthermore, let us consider the set of star operators containing $w$ in their support:
\[
\mathcal{A}_w := \{A_{(i, j)} : w \in \supp A_{(i, j)}\}. 
\]

By inspecting $C_1$ to $C_4$ we observe that for every $w \in \Lambda_B$ and every $A_{(i, j)} \in \mathcal{A}_w$, there exists exactly one gate in $\mathcal{C}_w$, which we denote by $\CX(v, w)$, with $v \in \supp A_{(i, j)}$. This implies that $w$ will no longer be in the support of $\tilde C A_{(i, j)} \tilde{C}^\dagger$. 

Also, for every $w \in \Lambda_B$ and every $A_{(i, j)} \in \mathcal{A}^c_w$, there exist exactly two gates in $\mathcal{C}_w$ that are supported in $A_{(i,j)}$, which we denote by $\CX(u_1, w)$ and $\CX(u_2, w)$. Again, this implies that $w$ will not be in the support of $\tilde C A_{(i, j)} \tilde C^\dagger$; indeed, $\CX(u_1, w) A_{(i, j)} \CX(u_1, w)^\dagger$ will include $w$ in its support, as we discussed previously, which implies that
\[
\CX(u_2, w) \CX(u_1, w) A_{(i, j)} \CX(u_1, w)^\dagger \CX(u_2, w)^\dagger
\]
will not include $w$ in its support.

In this case, we have also proven that for every $(i, j) \in \ZZ_L \times \ZZ_L$
\[
\supp(\tilde CA_{(i,j)}\tilde{C}^\dagger) = \Lambda_A \cap \supp A_{(i, j)}.
\]
See \cref{fig:finalstars} for a visual representation of the final star operators. 
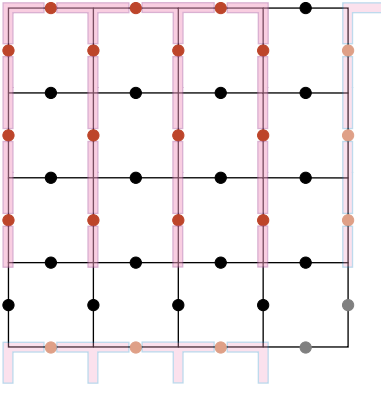
\begin{figure}
    \centering
    \input{Figures/finalstars.tikz}
    \caption{Visual representation of the final star operators. The two decoupled systems are marked in red and black, respectively.}
    \label{fig:finalstars}
\end{figure}

The final Hamiltonian obtained after conjugating by $\tilde C$ corresponds to two decoupled ``combs". See \cref{fig:finalhamiltoniantoric2D} for a visual representation of the two decoupled combs and their interactions. 

\begin{figure}
    \centering
    \input{Figures/finalhamiltonian2D.tikz}
    \caption{Visual representation of the two decoupled systems, where we have reshaped the final interactions shown originally in \cref{fig:finalstars,fig:finalplaquettes}.}
    \label{fig:finalhamiltoniantoric2D}
\end{figure}
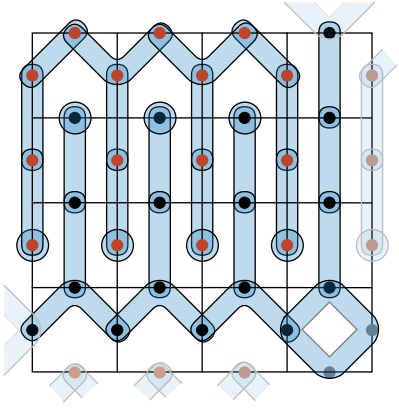

Note that, whilst 
\begin{equation}
\label{eq:conjugatedplaquettes}
\tilde C \Big(\sum_{(i, j) \in \ZZ_L \times \ZZ_L} B_{(i, j)}\Big) \tilde C^\dagger
\end{equation}
is classical, 
\begin{equation}
\label{eq:conjugatedstars}
\tilde C \Big(\sum_{(i, j) \in \ZZ_L \times \ZZ_L} A_{(i, j)}\Big) \tilde C^\dagger
\end{equation}
only contains $\sigma_x$ interactions. Since our aim is to obtain two classical systems, we use the fact that the Hadamard gate acts on $\sigma_x$ as 
\[
H \sigma_x H^\dagger = \sigma_z,
\]
and conjugate the resulting Hamiltonian from \cref{eq:conjugatedstars} by a circuit defined as 
\begin{equation}
\label{eq:circuitHadamard}
\prod_{v \in \Lambda_A} H(a),    
\end{equation}
thus changing every $\sigma_x$ term to a $\sigma_z$, which ultimately leads the resulting Hamiltonian to be classical.

It only remains to study the effect of the final sub-circuit, $\hat C$, which will transform the current interactions in each system into one all-to-all interaction in each comb, and a single magnetic field in each term. From this setting, it follows from \cref{lem:equivalencelemma} that the Hamiltonian is equivalent to a classical Ising chain with magnetic fields at both ends.

Let us focus on the system that results after conjugating the Hamiltonian from \cref{eq:conjugatedstars} by the Hadamard gates shown in \cref{eq:circuitHadamard}. 

The resulting Hamiltonian consists of three different interactions: 
\begin{itemize}
    \item Three-body interactions:
    \[
    \sigma_z^{(0, j-1, h)}\sigma_z^{(0, j, h)}\sigma_z^{(0, j, v)},
    \]
    for every $1 \leq j < L-1$. 
    \item Two-body interactions: most of which are vertical;
    \[
    \sigma_z^{(i-1, j, v)}\sigma_z^{(i, j, v)},
    \]
    for every $j \in \ZZ_L$ and every $1 \leq i < L-1$. Furthermore, there are two extra two-body interactions:
    \[
    \sigma_z^{(0, 0, h)}\sigma_z^{(0, 0, v)},\quad \text{and}\quad \sigma_z^{(0, L-1, v)}\sigma_z^{(0, L-2, h)}.
    \]
    \item Magnetic fields:
    \[
    \sigma_z^{(L-2, j, v)},
    \]
    for every $j \in \ZZ_L$.
\end{itemize}

Note that in the leftmost three-body interaction 
\begin{equation}
\label{eq:leftmostthreebody}
\sigma_z^{(0, 0, h)}\sigma_z^{(0, 1, h)}\sigma_z^{(0, 1, v)},    
\end{equation}
there are two classical Ising chains starting from it; one starting from the spin with coordinates $(0, 0, h)$ and the other starting at $(0, 1, v)$. 

Let us consider the chain starting from $(0, 0, h)$, and let us apply 
\begin{align*}
&\CX((L-2, 0, v),(L-3, 0, v)) \cdots \CX((1, 0, v),(0, 0, v))
\\&\quad \times \CX((0, 0, v), (0, 0, h)). 
\end{align*}

This way, we extend the three body interaction from \cref{eq:leftmostthreebody} to act on every spin in $\Lambda_A$ of the form $(i, 0, v)$, whilst at the same time transforming every two-body interaction of the chain---$\sigma_z^{(0, 0, h)}\sigma_z^{(0, 0, v)}$ and $\sigma_z^{(i-1, 0, v)}\sigma_z^{(i, 0, v)}$, $1 \leq i < L-1$---into a magnetic field on each target spin of the gates used.

Similarly, for the other chain we apply 
\[
\CX((L-2, 0, v),(L-3, 0, v)) \cdots \CX((1, 1, v),(0, 1, v)),
\]
obtaining the same result; the three-body interaction from \cref{eq:leftmostthreebody} is now further extended to the second column of the lattice; i.e. to every spin in $\Lambda_A$ of the form $(i, 1, v)$, whilst transforming the two-body interactions into magnetic fields. 

If we now further apply $\CX((0, 1, h),(0, 2, v))$, we transform the next three-body interaction into a two-body interaction, and further extend the interaction from \cref{eq:leftmostthreebody} to an extra spin, thus creating a new Ising chain starting from it---in the spin situated at $(0, 2, v)$---which we tackle in the same way. See \cref{fig:intermediatestars} for a visual representation of the interactions after applying $\CX((0, 1, h),(0, 2, v))$.

\begin{figure}
    \centering
    \input{Figures/intermediatestars.tikz}
    \caption{Visual representation of the resulting interactions after applying $\CX$ gates on the two Ising chains starting from the leftmost three-body interaction of the system in $\Lambda_A$, and the gate $\CX((0, 1, h),(0, 2, v))$.}
    \label{fig:intermediatestars}
\end{figure}
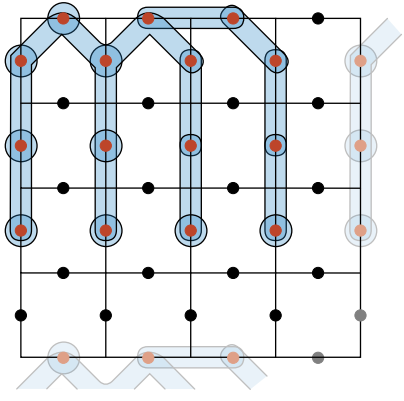

Repeating this process until we reach the last three-body interaction and its corresponding Ising chains will lead to obtaining one interaction term supported on $\Lambda_A$, and a parallel magnetic field on each spin of $\Lambda_A$, which, as we proved in \cref{lem:equivalencelemma}, is equivalent to a classical Ising chain with a magnetic field at each end.

Lastly, for the plaquette part, note that the interactions obtained are the following: 
\begin{itemize}
    \item One four-body interaction: 
    \begin{equation}
    \label{eq:fourbodyinteraction}
    \sigma_z^{(L-1, L-1, h)}\sigma_z^{(L-1, L-1, v)}\sigma_z^{(L-1, 0, v)}\sigma_z^{(0, L-1, h)}.
    \end{equation}
    \item Three-body interactions:
    \[
    \sigma_z^{(L-1, j-1, v)}\sigma_z^{(L-1, j-1, h)}\sigma_z^{(L-1, j, v)},
    \]
    for every $1 \leq j \leq L-1$. 
    \item Two-body interactions: in this case all of which are vertical
    \[
    \sigma_z^{(i-1, j, h)}\sigma_z^{(i, j, h)},
    \]
    for every $j \in \ZZ_L$ and every $2 \leq i \leq L-1$, and 
    \begin{equation}
    \label{eq:twobodyinteraction}
    \sigma_z^{(0, L-1, h)}\sigma_z^{(1, L-1, h)}. 
    \end{equation}
    \item Magnetic fields:
    \[
    \sigma_z^{(1, j, h)},
    \]
    for every $0 \leq j < L-1$.
\end{itemize}

The above interactions are almost the same as those obtained for the star operators, after mirroring them (cf. \cref{fig:finalhamiltoniantoric2D}), with the exception that one magnetic field has been substituted by a two-body interaction, and the two non-vertical two-body interactions have been replaced by a three-body and a four-body interaction, respectively. 

We start by conjugating the resulting Hamiltonian shown in \cref{eq:conjugatedplaquettes} by 
\begin{equation}
\label{eq:gatesplaquettes1}
\prod_{1 \leq j \leq L-1} \CX((L-1, 0, v), (L-1, j, v))    
\end{equation}
and 
\begin{equation}
\label{eq:gatesplaquettes2}
\prod_{1 \leq j \leq L-1} \CX((0, L-1, h), (j, L-1, v)).    
\end{equation}

The above gates will map the four body interaction from \cref{eq:fourbodyinteraction} to
\[
\sigma_z^{(L-1, L-1, h)}\sigma_z^{(L-1, L-1, v)}.
\]
Furthermore, the two-body interaction from \cref{eq:twobodyinteraction} will get mapped to a magnetic field on 
$(1, L-1, h)$, and the three-body interaction 
\[
\sigma_z^{(L-1, 0, v)}\sigma_z^{(L-1, 0, h)}\sigma_z^{(L-1, 1, v)},
\]
to
\[
\sigma_z^{(L-1, 0, h)}\sigma_z^{(L-1, 1, v)}.
\]

Now, the resulting interactions are analogous to those obtained for the stars (see \cref{fig:intermediateplaquettes}), so we can proceed in the same manner. 

\begin{figure}
    \centering
    \input{Figures/intermediateplaquettes.tikz}
    \caption{Visual representation of the interactions obtained after conjugating the final Hamiltonian associated to $\Lambda_B$ (cf. \cref{eq:conjugatedplaquettes}) with the gates shown in \cref{eq:gatesplaquettes1,eq:gatesplaquettes2}}
    \label{fig:intermediateplaquettes}
\end{figure}
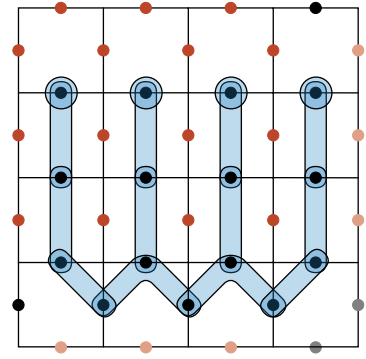

Note that, by applying the gates shown in \cref{eq:gatesplaquettes1,eq:gatesplaquettes2} we are obtaining two free spins---$(0, L-1, h)$ and $(L-1, 0, v)$---in the sense that they are not in the support of either of the two final Hamiltonians.
\end{proof}

\begin{rem}
\label{rem:groundstates}
Note that conjugating by a unitary matrix does not affect the coefficients of the toric code $J_v$ and $J_p$---see \cref{eq:HamiltonianTC}. Therefore, when $J_v = J_p = 1$ for every $v \in V_L$ and every $p \subset \mathcal{E}_L$, the final two Ising chains are ferromagnetic, with a ferromagnetic magnetic field at each end, thus having a unique ground state. In this case, the four possible values for the two free spins are in direct correspondence with the ground states of the toric code.
\end{rem}

Next, let us provide the details of the proof of \cref{thm:efficient_tc} based on the previous diagonalization
\begin{proof}[Proof of \cref{thm:efficient_tc}]
Due to \cref{lem:efficientgibbssampling} and the preceding \cref{thm:result} we are left to provide an efficient Gibbs sampling algorithm for a 1D Ising chain for which we employ a standard iterative procedure.
Denote the spins of an Ising chain by $x_i\in\{-1,+1\}$, $i=1,\ldots,N$ and the Hamiltonian $H=\sum_{i=1}^{N-1} J_i x_i x_{i+1}+
\sum_{i=1}^N h_i x_i$
We can simply make use of the conditional independence of the Gibbs distribution $p(x_1,\ldots,x_n)\propto \exp(-\beta H)$
\[
p(x_i|x_1,\ldots,x_{i-1})=p(x_i|x_{i-1})\,.
\]
We sample the spins consecutively from the distribution $q(x_1):=\exp(-\beta h_1)$, $q(x_{i+1}|x_i):=\exp(-\beta(J_i x_{i+1}x_i+h_{i+1}))$, evidently resulting in the desired distribution. This requires simply to take $N$ samples of binary distributions with the respective weights given by the coefficients and comes without any further sampling error.

Finally, note that this procedure extends to the ground state by deterministically defining consecutive spins by the respective limit of conditional probabilities.
\end{proof}

\section{Diagonalizing Hamiltonians}
\label{sec:appendixtableaus}

Our implementation of the diagonalization algorithm for Hamiltonians composed of commuting scalar multiples of Pauli operators is based on the method originally presented by Aaronson and Gottesman \cite{aaronson2004simulation}, as shown in \cite{van2020circuit}. The code used can be found in \cite{github2025}. 

In the following we provide a short summary of the main ideas of this approach. Wherever we mention the algorithm from \cite{van2020circuit} in this work, we refer to the sequential application of algorithms 1 and 2 as described in the same paper. For an in-depth study, we refer to the aforementioned references.

Let $H$ be a Hamiltonian associated to an $n$-spin system. Assume that $H$ is the sum of $m$ scalar multiples of commuting Pauli operators, as in \cref{eq:Hamiltonian}. In order to diagonalize it, we will construct a \textit{tableau} associated to it. 

The tableau consists of two $m \times n$ matrices, which we denote by $X := \{x_{ij}\}$ and $Z := \{z_{ij}\}$, and an $m \times 1$ column matrix, which we denote by $s := \{s_i\}$. This way, the tableau is an $m \times (2n + 1)$ matrix of the form---in block notation---
\[
\begin{pmatrix}
X & \rvline & Z & \rvline & s
\end{pmatrix}. 
\]

Each row of the tableau represents a term of $H$ in the following way; let $i \in \{1, \dotsc, m\}$ and consider the $i$-th term of $H$, which we denote by $\alpha_i H_i$ (cf. \cref{eq:Hamiltonian}). We set $s_i = 0$ if $\alpha_i > 0$ and $s_i = 1$ otherwise. The values of the $i$-th row of $X$ and $Z$---$\{x_{ij}\}_{j = 1}^n$ and $\{z_{ij}\}_{j =1}^n$, respectively---are chosen according to \cref{tab:operatortableaucorresp}, depending on the $j$-th Pauli matrix in $H_i$. 

As an illustrative example, consider a Hamiltonian $H$ acting on three sites with two terms,
\[
\sigma_z \otimes \sigma_y \otimes \mathds{1} -\, \sigma_x \otimes \mathds{1} \otimes \sigma_y.
\]
Its associated tableau is
\[
\begin{pmatrix}
0 & 1 & 0 & \rvline & 1 & 1 & 0 & \rvline & 0 \\
1 & 0 & 1 & \rvline & 0 & 0 & 1 & \rvline & 1
\end{pmatrix}.
\]

\begin{table}
    \centering
    \begin{tabular}{|c|c|c|}
        \hline
         \textbf{Operator} & $x_{ij}$ & $z_{ij}$\\
         \hline
         $\sigma_x$ & $1$ & $0$\\
         \hline
         $\sigma_z$ & $0$ & $1$\\
         \hline
         $\sigma_y$& $1$ & $1$\\
         \hline
         $\mathds{1}$ & $0$ & $0$ \\
         \hline
    \end{tabular}
    \caption{Representation of each Pauli matrix in the tableau. The above notation assumes that the operator corresponds to the $j$-th Pauli term in the $i$-th operator of $H$. }
    \label{tab:operatortableaucorresp}
\end{table}

Once the tableau associated to the original Hamiltonian $H$ is constructed, the algorithm only regards the tableau in order to diagonalize $H$. The main idea is to encode the effect of conjugating $H$ by a unitary into an operation on the tableau. Note that, while the size of $H$ scales exponentially with respect to the number of spins of the system, the size of its associated tableau only grows linearly in the number of spins and terms of $H$. 

Since the algorithm presented in \cite{van2020circuit} diagonalizes $H$ using only $\CX$, $\CZ$, Hadamard and phase gates, it is sufficient to encode the effect of conjugating $H$ by such gates onto the tableau. See \cref{tab:operationssummary} for more details. For simplicity, we have not included the effect of the $\CZ$ gate in the table, as it can be written as
\[
\CZ = (\mathds{1} \otimes H) \CX (\mathds{1} \otimes H). 
\]
We will use the same notation as in \cite{van2020circuit}. Thus, we denote a Hadamard (resp. phase) gate acting on the $i$-th qubit by $H(i)$ (resp. $S(i)$). And a $\CX$ (resp. $\CZ$) gate acting on the $i$-th and $j$-th qubits---with the $i$-th acting as control---as $\CX(i,j)$ (resp. $\CZ(i, j)$).

Furthermore, as we mentioned earlier, note that the algorithm uses only a quadratic number of gates---in the system size---in order to diagonalize $H$. 

\begin{table*}
    \centering
    \begin{tabular}{|c|c|c|}
        \hline
         \textbf{Gate} & \textbf{Sign update} & \textbf{Column update} \\
         \hline
         $H(a)$ & $s \leftarrow s \oplus (x_a \otimes z_a)$ & swap $x_a$ and $z_a$\\
         \hline
         $S(a)$ & $s \leftarrow s \oplus (x_a \otimes z_a)$ & $z_a \leftarrow z_a \oplus x_a$\\
         \hline
         $\CX(a,b)$ & $s \leftarrow s \oplus (x_a \otimes z_b \otimes (x_b \oplus z_a \oplus 1))$ & $x_b \leftarrow x_b \oplus x_a$, $z_a \leftarrow z_a \oplus z_b$\\
         \hline
    \end{tabular}
    \caption{Summary of the tableau update after conjugating its associated Hamiltonian by the different gates used in the algorithm. We have used $x_i$ (resp. $z_i$) to denote the $i$-th column of the $X$ matrix (resp. $Z$ matrix) in the tableau. The sign update is performed before the column update.}
    \label{tab:operationssummary}
\end{table*}

When the diagonalization algorithm terminates, the $X$ matrix of the tableau is always identically zero, meaning that there are no $\sigma_x$ or $\sigma_y$ interactions in any of the terms of the final Hamiltonian. 

It is important to note that the algorithm shown in \cite{van2020circuit} first performs a series of operations on the tableau which do not correspond to actual conjugation by gates, but rather correspond to spin reordering, term reordering, and multiplication of terms. This way, the algorithm is performed on a modified tableau, which corresponds to a different Hamiltonian than the one given as input. Nevertheless, note that the same circuit will diagonalize both the original and the modified Hamiltonian. This way, by maintaining a record of the actual gates applied to the modified Hamiltonian, one can later apply the same gates to the original one to obtain its diagonalized---thus classical---version (see \cite[Section 4.2]{van2020circuit}). Throughout this section, we will therefore use ``final tableau'' to refer to the one which corresponds to the diagonalized version of the original Hamiltonian, not the one obtained originally from the algorithm. 

Reconstructing the Hamiltonian from the final tableau, we obtain a classical Hamiltonian $\tilde H$ with the same number of terms, such that 
\[
\tilde H = U H U^\dagger,
\]
where $U$ is the circuit explicitly constructed by the algorithm. Although this is enough to find an explicit correspondence between quantum and classical Hamiltonians, it is often not so easy to recognize a clear structure in $\tilde H$---such as translation invariance, dimension of the model or locality. For this reason, we implement a final stage in the algorithm, based on the idea of column-wise Gaussian elimination, which we call \textit{pseudo-Gaussian elimination}. It consists of $\CX$ gates in order to create as many rows with a single nonzero entry as possible, which correspond to single-$\sigma_z$ terms:
\begin{algorithm}[H]
  \caption{Pseudo-Gaussian elimination of \(Z\)}
  \label{alg:gaussian}
  \begin{algorithmic}
    \State $v \gets \{1, \dotsc, n\}$ \Comment{Array with the unused columns}
    \For{$i \in \{1, \dotsc, m\}$}
    \For{$j \in v$}
    \If{$z_{ij} = 1$} 
    \State remove $j$ from $v$
    \For{$k \in \{1, \dotsc, n\} \setminus \{j\}$}
    \If{$z_{ik} = 1$}
    \State apply $\CX(k, j)$
    \EndIf
    \EndFor
    \EndIf
    \EndFor
    \EndFor
  \end{algorithmic}
\end{algorithm}

As mentioned in the main text, this algorithm has a classical complexity of $\cO(m^2n^2)$ and the output is a quantum circuit of depth at most $\cO(mn^2)$. 

By running the above algorithm on the final tableau, we were able to recognize classical Ising chain Hamiltonians arising from the models considered (cf. \cref{tab:diagonalizationresults}). It is important to note that, although \cref{alg:gaussian} is useful for identifying Ising chains, it is less straightforward to recognize 2D or 3D Ising models via this method. Thus, we believe that there should be other algorithms that could be applied at the final stage in order to identify more intricate classical models.

In the following sections, we give a short introduction on every model considered in this work, and discuss how the correspondence shown in \cref{tab:diagonalizationresults} arises by inspecting the final tableau obtained after performing the pseudo-Gaussian elimination. Note that the results that we will discuss correspond to empirical observation of the final tableau associated to the models considered, for sizes up to $90 \times 90$ in 2D models, and $20 \times 20 \times 20$ in 3D models, which correspond to systems with over $10^4$ spins. Also note that the final tableau obtained for a given Hamiltonian depends on the order considered for both its terms and the spins. For this reason, we will always emphasize that the structures that we will present here are always up to column and row reordering. 

Note that the dualities obtained are explicit in the sense that we have been able to obtain---for almost every model considered---a correspondence between the different kinds of interactions in the original Hamiltonian and the different decoupled systems; similarly to the 2D toric code case where the star and the plaquette operators were each dual to a decoupled Ising chain.

Finally, since the tableau associated to a given Hamiltonian $H$ does not contain the scalar prefactor of each term---only their sign---we consider each term of the Hamiltonian to be multiplied by its respective prefactor. Furthermore, as we mentioned previously, after running the diagonalization algorithm, the $X$ matrix in the tableau will always be identically zero. In the same manner, for all considered models, the algorithm does not change any of the values in the $s$ column. Therefore, we will often omit them in the coming sections.

For simplicity, we only consider regular lattices---in the sense that they are always of size $L \times L$ or $L \times L \times L$. Nevertheless, the algorithm used can be easily adapted to study any other non-regular lattices such as $L \times L'$ with $L \neq L'$.
\subsection{2D toric code}
\label{sec:2Dtoriccode}

Let us first discuss the two-dimensional toric code. Although we have given a formal proof of the duality between the 2D toric code and the two decoupled classical Ising chains in \cref{sec:2Dtoriccode_proof}, we include the discussion in this context as well for illustration. 

Recall that the toric code is defined on a square lattice with periodic boundary conditions, where a spin is considered at the midpoint of each edge. This way, the 2D toric code model defined on $\mathbb{S}_L \times \mathbb{S}_L$ contains $2 L^2$ spins. 

After performing the pseudo-Gaussian elimination on the final tableau, the rows and columns of the $Z$ matrix can be rearranged so that it is of the form
\[
\begin{pmatrix}
  \bigI & \rvline & \bigzero & \rvline &  \hspace{3pt} \begin{matrix} 0 \, 0 \\ \vdots \end{matrix}\\
  \cline{1-3}
  1 \cdots 1 & \rvline & 0 \cdots 0 & \rvline & \vdots \\
  \cline{1-3}
  \bigzero & \rvline & \bigI & \rvline &  \vdots \\
  \cline{1-3}
  0 \cdots 0 & \rvline & 1 \cdots 1 & \rvline & 0 \, 0\\
\end{pmatrix},
\]
where the identity matrices are of size $(L^2 - 1) \times (L^2 - 1)$. Furthermore, under this arrangement, the first $L^2$ rows of $Z$ correspond to the star operators and the last $L^2$ correspond to the plaquette operators. 

From the structure of $Z$, the two decoupled models can be distinguished easily; note that there is no row in the $Z$ matrix containing ones in any of the first $L^2-1$ columns and any of the next $L^2-1$ columns at the same time, meaning that both spin subsets do not interact with each other.

Let us now inspect each subsystem in detail. Without loss of generality, we consider the first $L^2 - 1$ columns and the first $L^2$ rows of $Z$, obtaining the matrix
\begin{equation}
\label{eq:identitywithones}
\begin{pmatrix}
I\\
\hline
1 \cdots 1
\end{pmatrix},
\end{equation}
which corresponds to the terms as described in \cref{thm:result}. Equivalently to \cref{lem:equivalencelemma}, let $\tilde H$ be the Hamiltonian associated to the $Z$ matrix from \cref{eq:identitywithones}---taking the $X$ matrix as identically zero. Thus, conjugating $\tilde H$ by
\[
U := \CX(L-2,L-1)\CX(L-3,L-2)\dotsm\CX(1, 2)
\]
we obtain the following $Z$ matrix:
\begin{equation}
\label{eq:matrixising}
\begin{pmatrix}
1 & 0 & 0 & 0 & \cdots & 0 & 0\\
1 & 1 & 0 & 0 & \cdots & 0 & 0\\
0 & 1 & 1 & 0 & \cdots & 0 & 0\\
0 & 0 & 0 & 0 & \cdots & 1 & 1\\
0 & 0 & 0 & 0 & \cdots & 0 & 1\\
\end{pmatrix},
\end{equation}
which corresponds to an Ising chain with a magnetic field on both ends.

This way, we can see that the star interactions get mapped to a classical Ising chain and the plaquette interactions get mapped to another classical Ising chain decoupled from the first one. 

As the sign vector $s$ associated to the original Hamiltonian $H$ remains intact after the algorithm, there are no sign changes in any of its terms. In particular, as discussed in \cref{rem:groundstates}, when $J_v = J_p = 1$ for every $v \in V_L$ and every $p \subset \mathcal{E}_L$, the final two Ising chains are ferromagnetic, thus having a unique ground state. Therefore, the last two all-zero columns of $Z$ show that the ground state degeneracy of the final Hamiltonian is $4$, the same as that of the original toric code \cite{kitaev2003toric,dennis2002quantummemory,Tasaki2020QuantumSystems}. 

\subsection{3D toric code}

We also consider 3D toric code model \cite{nussinov2008toric,castelnovo2008threetoric}, which is defined on a regular cubic lattice $(V_L, \mathcal{E}_L)$ on the three-torus $\mathbb{S}_L \times \mathbb{S}_L \times \mathbb{S}_L$ where the vertices have integer coordinates. Again, one spin is placed at the midpoint of every edge of the lattice (see \cref{fig:3Dtoriclattice}). Thus, each 3D toric code system defined on $\mathbb{S}_L \times \mathbb{S}_L \times \mathbb{S}_L$ contains $3 L^3$ spins.

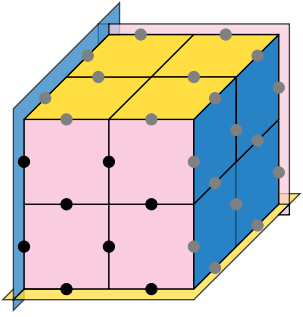
\begin{figure}
    \centering
    \input{Figures/latticetoric3Dperiodic.tikz}
    \caption{Visual representation of a $2 \times 2 \times 2$ toric code lattice. The colors illustrate the periodic boundary conditions.}
    \label{fig:enter-label}
\end{figure}

In the same spirit as for the 2D toric code, for every $v \in V_L$, we denote by $\partial v$ the set of the six spins which lie on the edges adjacent to $v$. We also denote by $p \subset \mathcal{E}_L$ any set of four spins whose corresponding edges form a square in the lattice---note that this time the squares can lie in the $\mathit{xy}$, $\mathit{yz}$ or $\mathit{xz}$ plane. 

The 3D toric code Hamiltonian is given by
\[
H = - \sum_{v \in V_L} J_v A_v - \sum_{p \subset \mathcal{E}_L} J_p B_p
\]
where $J_v, J_p \in \RR$ for every $v \in V_L$ and every $p \subset \mathcal{E}_L$, and
\[
A_v := \bigotimes_{i \in \partial v } \sigma_x^i,\quad B_p := \bigotimes_{i \in p} \sigma_z^i,
\]
are the star and plaquette operators (see \cref{fig:3Dtoricinter}). Thus, the 3D toric code Hamiltonian has $4L^3$ terms. 

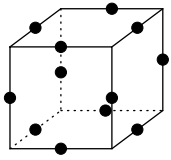
\begin{figure}[]
    \centering
    \input{Figures/latticetoric3D.tikz}
    \caption{Visual representation of a single block in a 3D toric code lattice.}
    \label{fig:3Dtoriclattice}
\end{figure}

\begin{figure}[]
    \centering
    \input{Figures/3Dtoric.tikz}
    \caption{Visual representation of the star operator $A_v$ and the three plaquette operators $B_p$---contained in the $yz$, $xy$ and $xz$ plane, respectively---of a 3D toric code.}
    \label{fig:3Dtoricinter}
\end{figure}
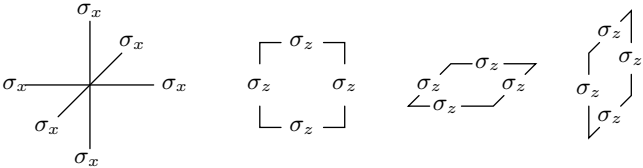

After the pseudo-Gaussian elimination, the rows and columns of $Z$ can be rearranged so that it is of the form 
\[
\begin{pmatrix}
  \bigI & \rvline & \bigzero & \rvline &  \hspace{3pt} \begin{matrix} 0 \, 0 \, 0 \\ \vdots \end{matrix}\\
  \cline{1-3}
  1 \cdots 1 & \rvline & 0 \cdots 0 & \rvline & \vdots \\
  \cline{1-3}
  \bigzero & \rvline & Z' & \rvline & 0 \, 0 \, 0\\
\end{pmatrix},
\]
where the identity matrix is of size $(L^3 -1) \times (L^3-1)$, and the first $L^3$ rows correspond to the star operators.

This structure allows us to map the star operators to an Ising spin chain with $L^3-1$ spins and a magnetic field on each end. 

Although we are able to identify this subsystem, we have systematic description of the terms on the remaining, despite believing that the remaining columns should represent a 3D Ising model, in the spirit of \cite[Table 1]{weinstein2019universality}. Nevertheless, we confirm for all instances up to the sizes we address that, prior to the Gaussian elimination, the resulting classical Hamiltonian for this second subsystem is such that each interaction acts on at most four spins, and for each spin there are at most four interactions acting on it. Thus, this can be efficiently sampled at a high enough temperature \cite{caputo2015approximatetensorizationentropyhigh}. 

As in the 2D toric code case, recall that the sign vector $s$ remains unchanged. Furthermore, there is always exactly three all-zero columns in the final $Z$ matrix, which, as in the 2D toric code case, should be related to the $8$ ground states of the original model when $J_v = J_p = 1$ for every $v \in V_L$ and $p \subset \mathcal{E}_L$ \cite{castelnovo2008threetoric}. 

\subsection{Color code on a honeycomb lattice}
Let us now show our findings for the color code \cite{nussinov2012color,bombin2006color}, which can be defined on any trivalent graph. For the purpose of this work, we consider the color code model on a two-dimensional regular honeycomb lattice $(V_L, \mathcal{E}_L)$ defined on the two-dimensional torus $\mathbb{S}_L \times \mathbb{S}_L$. This time, one spin is considered at each vertex of the lattice (see \cref{fig:colorcodelattice}).

We consider models associated to lattices having $2(L+1) \times 2(L+1)$ honeycombs, for $L \geq 0$, which consist of $8(L+1)^2$ spins. 

Let $\varhexagon \subset \Lambda_L$ denote the set of six spins around a honeycomb plaquette of the lattice. The Hamiltonian of the color code is defined as 
\[
H = - \sum_{\varhexagon \subset \Lambda_L} J_{\varhexagon} X_{\varhexagon} - \sum_{\varhexagon \subset \Lambda_L} \tilde{J}_{\varhexagon} Z_{\varhexagon},
\]
where $J_{\varhexagon}, \tilde J_{\varhexagon} \in \RR$ for every $\varhexagon \subset \Lambda_L$, and
\[
X_{\varhexagon} := \bigotimes_{i \in \varhexagon} \sigma_x^i, \quad Z_{\varhexagon} := \bigotimes_{i \in \varhexagon} \sigma_z^i,
\]
are both plaquette interactions (see \cref{fig:colorcodeinter}).

\begin{figure}[]
    \centering
    \input{Figures/latticehoneycomb.tikz}
    \caption{Visual representation of a $4 \times 4$ honeycomb lattice. We use single and double arrows to show the periodic boundary conditions. Note that the dotted lines are auxiliary and do not correspond to the edges of the graph.}
    \label{fig:colorcodelattice}
\end{figure}

 \begin{figure}[]
    \centering
    \input{Figures/color.tikz}
    \caption{Visual representation of the $X_{\varhexagon}$ and $Z_{\varhexagon}$ interactions of the color code on a 2D honeycomb lattice.}
    \label{fig:colorcodeinter}
\end{figure}

In this case, the final $Z$ matrix obtained after the pseudo-Gaussian elimination depends on the size of the system. Indeed, if $L \textup{ mod } 3 = 0$ or $1$ then the final $Z$ matrix contains exactly one non-zero element in every row and every column, which implies that the final Hamiltonian is non-interacting and each of its terms acts non-trivially on exactly one spin. Note that this result is even stronger than the duality shown in \cite{weinstein2019universality}. 

Lastly, when $L \textup{ mod } 3 = 2$, the rows and columns of the $Z$ matrix can be rearranged after performing the pseudo-Gaussian elimination so that $Z$ is of the form
\[
\begin{pmatrix}
  \bigI & \rvline & \bigzero & \rvline & 0000\\
  \cline{1-3}
  1 \cdots 1 1 \cdots 1 0 \cdots 0 & \rvline & 0 \quad \ \quad \cdots \quad  \quad \ 0 & \rvline & \vdots \\
  0 \cdots 0 1 \cdots 1 1 \cdots 1 & \rvline & 0 \quad \ \quad \cdots \quad  \quad \ 0 & \rvline & \vdots \\
  \cline{1-3}
  \bigzero & \rvline & \bigI & \rvline & \vdots\\
  \cline{1-3}
  0 \quad \ \quad \cdots \quad  \quad \ 0 & \rvline & 1 \cdots 1 1 \cdots 1 0 \cdots 0 & \rvline & \vdots \\
  0 \quad \ \quad \cdots \quad  \quad \ 0 & \rvline & 0 \cdots 0 1 \cdots 1 1 \cdots 1 & \rvline & 0000\\
\end{pmatrix},
\]
where we note that the first $4(L+1)$ rows correspond to the $X_{\varhexagon}$ operators and the last $4(L+1)$ rows correspond to the $Z_{\varhexagon}$ operators. In this case, there are always exactly four all-zero columns. 

This structure allows us to immediately identify two decoupled models, each corresponding to a different type of interaction. 

\begin{figure}[]
    \centering
    \input{Figures/lassoising.tikz}
    \caption{Visual representation of a $7$-spin lasso Ising chain. The ellipses denote $\sigma_z \otimes \sigma_z$ interactions. We omit any possible magnetic fields for readability.}
    \label{fig:lassoising}
\end{figure}

In fact, each model is a \textit{lasso Ising chain}, i.e. an Ising chain where the last spin of the chain has an extra interaction with a spin situated in the bulk of the chain (see \cref{fig:lassoising}). Indeed, let us consider the columns and the rows of one of the two decoupled systems:
\[
\begin{pmatrix}
 & & & & I & & & & \\
\hline
1 & \cdots &  1 &  1 & \cdots & 1 & 0 & \cdots & 0 \\
0 & \cdots &  0 &  1 & \cdots & 1 & 1 & \cdots & 1 \\
\end{pmatrix}.
\]

Once again, by applying $\CX$ gates as discussed in the 2D toric code case, we arrive at
\[
\begin{pmatrix}
& & & & A & & & & & \\
 \hline
0 & \cdots &  0 &  0 & \cdots & 1 & 0 & \cdots & 0  & 0 \\
0 & \cdots &  1 &  0 & \cdots & 0 & 0 & \cdots & 0 & 1 \\
\end{pmatrix},
\]
where
\[
A := 
\begin{pmatrix}
1 & 0 & 0 & 0 & & & \cdots & & & 0\\
 1 & 1 & 0 & 0 & & & \cdots & & & 0\\
 0 & 1 & 1 & 0 & & & \cdots & & & 0\\
 \vdots & & & & & & \ddots& & & \vdots\\
 0 & 0 & 0 & 0 & & & \cdots & & 1& 1    
\end{pmatrix}.
\]
This implies that the final Hamiltonian corresponds to an Ising chain with a magnetic field in one spin, and an extra 2-body interaction, which creates the lasso shape of the model. The Davies generator associated to this Hamiltonian is shown to have a positive MLSI at any positive temperature in \cref{sec:lasso}, thus yielding efficient sampling. 

In this case, the number of ground states varies depending on the system size; when $L \text{ mod } 3 \neq 2$ our findings show that the number of ground states is $2^a$, where $a$ is the number of zero couplings $J_{\varhexagon}, \tilde{J}_{\varhexagon}$. In the other case the four all-zero columns should be related to the existence of $2^4$ ground states \cite{mehdicolorcode}. 

\subsection{Rotated surface code}

Next, we consider the rotated surface code, which is a slight generalization of the defected toric code studied in \cite{hwang2024gibbsstatepreparationcommuting}. 

This model is defined on a regular square lattice $(V_L, \mathcal{E}_L)$ on $[0, L]^2$, which corresponds to considering open boundary conditions. One spin is situated at each vertex $v \in V_L$. Furthermore, the lattice is ``colored'' in a chessboard manner, using blue and red (see \cref{fig:rotatedsurfacelattice}). The interactions of the model will depend on the color of each square. 

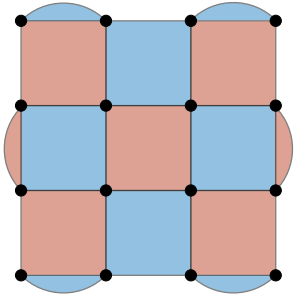
\begin{figure}[]
    \centering
    \input{Figures/rotatedsurfacelattice.tikz}
    \caption{Visual representation of a $3 \times 3$ rotated surface code lattice. Some circular segments are added to represent the interactions considered at the boundary of the lattice.}
    \label{fig:rotatedsurfacelattice}
\end{figure}

In order to define the Hamiltonian associated to the model, let us denote by $p \subset \Lambda_L$ any four-spin set corresponding to the corners of a square in the lattice. Let us also denote by $s_h$ (resp. $s_v$) any set of two adjacent spins situated at the top or bottom (resp. left or right) boundaries of the lattice. We say that $p$ is red (resp. blue) if its corresponding square is red (resp. blue). Similarly, we say that $s_v$ or $s_h$ is red (resp. blue) if the edge connecting them belongs to a red (resp. blue) square.

We define $H$ as
\begin{align*}
H &= -\sum_{\substack{p \subset \Lambda_L\\p \text{ red}}} J_p A_p - \sum_{\substack{p \subset \Lambda_L\\p \text{ blue}}} \tilde{J}_p B_p 
\\&\quad - \sum_{\substack{s_v \subset \Lambda_L\\s_v \text{ red}}} J_{s_v} C_{s_v} - \sum_{\substack{s_h \subset \Lambda_L\\s_v \text{ blue}}} \tilde{J}_{s_h} D_{s_h},
\end{align*}
where $J_p, \tilde J_p, J_{s_v}, \tilde J_{s_h} \in \RR$ and 
\begin{align*}
A_p := \bigotimes_{i \in p} \sigma^i_x, \quad B_p := \bigotimes_{i \in p} \sigma^i_z\\
C_{s_v} := \bigotimes_{i \in s_v} \sigma^i_x,\quad D_{s_h} := \bigotimes_{i \in s_v} \sigma^i_z.
\end{align*}
See \cref{fig:rotatedsurfaceinter} for a visual representation of the interactions. 

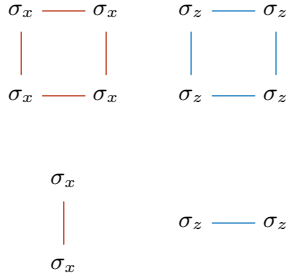
\begin{figure}[]
    \centering
    \input{Figures/rotatedsurfaceinter.tikz}
    \caption{Visual representation of the $A_p$, $B_p$ (top row), $C_{s_v}$ and $D_{s_h}$ (bottom row) operators.}
    \label{fig:rotatedsurfaceinter}
\end{figure}

After performing the Gaussian elimination on the final $Z$ matrix of the tableau, we encounter an analogous situation to the one reached when considering the color code with $L \text{ mod } 3 \neq 2$. In particular, the final $Z$ matrix contains exactly one all-zero column, whilst the rest of the columns contains exactly one non-zero element. Furthermore, there is exactly one non-zero element in every row of $Z$, which in particular implies that the final Hamiltonian is non-interacting and single-site. 

Moreover, the above structure implies that the number of ground states of the model is $2^{1 + a}$, where $a$ is the number of zero couplings $J_p, \tilde J_p, J_{s_v}, \tilde J_{s_h}$. 

\subsection{Haah's code}

Next, we study Haah's code \cite{haah2011localstab,haah2013quantumself}. This model is defined on the same cubic lattice as the 3D toric code. This time, we consider two spins in each vertex $v \in V_L$ of the lattice (see \cref{fig:haahlattice}). Since we consider a lattice in $\mathbb{S}_L \times \mathbb{S}_L \times \mathbb{S}_L$, the system will contain $2 L^3$ spins. 

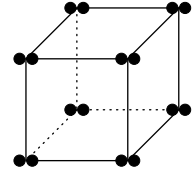
\begin{figure}
    \centering
    \input{Figures/latticehaah.tikz}
    \caption{Visual representation of a single cube in Haah's code model.}
    \label{fig:haahlattice}
\end{figure}

The Hamiltonian of Haah's code is defined as
\[
H = - \sum_{v \in V_L} J_v A_v - \sum_{v \in V_L} \tilde J_v B_v,
\]
where $J_v, \tilde J_v \in \RR$ and both $A_v$ and $B_v$ act on every spin of a single block of the lattice, whose bottom down leftmost vertex is $v$ (see \cref{fig:haahinter}).

As done in \cite{weinstein2019universality}, we only consider those values of $L$ for which the ground state degeneracy of the Hamiltonian is $4$ when $J_v = \tilde J_v = 1$ for every $v \in V_L$, which correspond to odd values of $L$ such that $L \text{ mod } 4^p - 1 \neq 0$ for every $p \geq 2$. 

\begin{figure}
    \centering
    \input{Figures/haah.tikz}
    \caption{Visual representation of the $A_v$ and $B_v$ cube interactions of Haah's code}
    \label{fig:haahinter}
\end{figure}
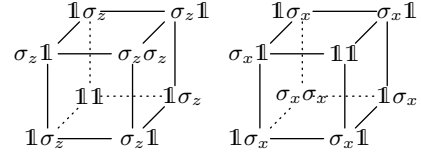

After performing the pseudo-Gaussian elimination, the rows and columns of the $Z$ matrix can be arranged so that it is of the same form as in \cref{eq:matrixising}, where the first $L^3$ rows of the tableau correspond to $B_v$ operators, and the last $L^3$ rows correspond to $A_v$ operators. 

This allows us to conclude that the $B_v$ interactions get mapped to one Ising chain chain and the $A_v$ interactions get mapped to the other. 

Once again, since the sign vector remains unchanged, we know that when $J_v = \tilde J_v = 1$ for every $v$, this corresponds to two decoupled ferromagnetic Ising chains with a ferromagnetic magnetic field at both ends. Thus, since both chains have a unique ground state, the final two all-zero columns of $Z$ correspond to the four ground states of the model. 

\subsection{X-cube}

Another model considered is the X-cube \cite{weinstein2020Xcube,vijay2016fracton}. This time, motivated by \cite{weinstein2019universality}, we consider cylindrical boundary conditions, which correspond to defining the model on a cubic lattice $(V_L, \mathcal{E}_L)$ in $\mathbb{S}_L \times \mathbb{S}_L \times [0, L]$, where the vertices have integer coordinates. As in the 3D toric code, one spin is placed at the midpoint of every edge in $\mathcal{E}_L$. This time, due to the different boundary conditions, the system contains $3L^3 + 2L^2$ spins. 

Let us denote by $b \subset \mathcal{E}_L ( \equiv \Lambda_L)$ the set of $12$ spins in any given cube of the lattice. Furthermore, given any vertex of the lattice $v \in V_L$, let us denote by $\partial^{\mathit{xy}} v$ the set four spins situated in the edges adjacent to $v$ that lie in the $\mathit{xy}$ plane. Similarly, we define $\partial^{\mathit{yz}} v$ and $\partial^{\mathit{xz}} v$. We define the X-cube Hamiltonian as
\[
H = - \sum_{b \subset \mathcal{E}_L} J_b A_b - \sum_{\substack{v \in V_L\\v \in \mathbb{S}_L \times \mathbb{S}_L \times (0,L)}} \big( J_v B_v + \tilde{J}_v C_v + \hat{J}_v D_v \big),
\]
where $J_b, J_b, \tilde J_v, \hat J_v \in \RR$, and
\begin{align*}
A_b &:= \bigotimes_{i \in b} \sigma_x^i,\quad & B_v &:= \bigotimes_{i \in \partial^{\mathit{yz}} v} \sigma^i_z,\\
C_v &:= \bigotimes_{i \in \partial^{\mathit{xy}} v} \sigma^i_z, \quad & D_v &:= \bigotimes_{i \in \partial^{\mathit{xz}} v} \sigma^i_z.
\end{align*}
Note that $\mathbb{S}_L \times \mathbb{S}_L \times (0,L)$ contains the whole lattice except for the vertices at the boundary in the direction in which we consider open boundary conditions. See \cref{fig:xcubeinter} for a visual representation of the operators. 

\begin{figure}
    \centering
    \input{Figures/Xcube.tikz}
    \caption{Visual representation of the cube operator $A_b$---left---and the three cross operators---$B_v$, $C_v$ and $D_v$, respectively---of the X-cube model. }
    \label{fig:xcubeinter}
\end{figure}
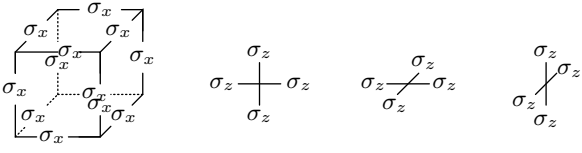

After finishing the pseudo-Gaussian elimination we obtain several all-zero columns. In fact, the number of all-zero columns grows with the dimension of the lattice as $4L^2 + 2L - 1$, which is consistent with the fact that the X-cube has a number of ground states that grows superpolynomially with the system size \cite[Appendix A]{weinstein2020Xcube}.

After rearranging the rows and columns of the $Z$ matrix associated to the final tableau, we find the following block-diagonal structure,
\[
\left(
  \begin{array}{c|c|c|c|c|c}
    \mathcal{D}_1 &   &   &  &  & \\
    \hline
      & \ddots &   &   & & \\
    \hline
      &   & \mathcal{D}_L &   & & \\
    \hline
      &   &   & \mathcal{D}'_1 & &\\
      \hline
      &   &   & & \ddots &\\
      \hline
      &   &   & &  & \mathcal{D}'_{L-1}
  \end{array}
  \right),
\]
where the rows in which each $\mathcal{D}_i$ block is situated correspond to those operators $A_b$ for which $b$ is such that its corner vertices have coordinates $(x,y,i-1)$ and $(x,y,i)$ for some $x, y \in \{0, \dotsc, L-1\}$. Every matrix $\mathcal{D}_i$,  $i \in \{1, \dotsc, L\}$ has $L^2$ rows and is of the same form as in \cref{eq:identitywithones}. This allows us to conclude that there are $L$ decoupled Ising chains, each corresponding to a block $\mathcal{D}_{i}$, which in turn corresponds to a ``slice" parallel to the $xy$ plane of the cube operators.

Let us now study the remaining $L-1$ decoupled systems, which correspond to the blocks $\mathcal{D}'_i$, which again correspond to a ``slice" of the cross operators. 

First of all, the rows in which each $\mathcal{D}'_i$ block is situated correspond to those cross operators $B_v, C_v, D_v$ for which $v$ has coordinates $(x,y, i)$ for some $x, y \in \{0, \dotsc, L-1\}$. Every matrix $\mathcal{D}'_i$, $i \in \{1, \dotsc, L-1\}$ has $3L^2$ rows, $2(L^2-1)+1$ columns and is of the form:
\begin{equation}
\label{eq:expressionZXblock}
\mathcal{D}'_i :=  \begin{pmatrix}
\mathfrak{D} & \rvline &  & \rvline &  & \rvline &  & \rvline & 0\\
\cline{1-7}
 & \rvline & \mathfrak{D} & \rvline & & \rvline &  & \rvline & \vdots \\
\cline{1-7}
 & \rvline &  & \rvline & \ddots & \rvline &  & \rvline & \vdots\\
\cline{1-7}
 & \rvline &  & \rvline &  & \rvline & \mathfrak{D}  & \rvline & 0 \\
\hline
\begin{matrix}0 & 0 \\0 & 1 \\0 & 1 \end{matrix} & \rvline & \begin{matrix}0 & 0 \\0 & 1 \\0 & 1 \end{matrix} & \rvline & \cdots & \rvline & \begin{matrix}0 & 0 \\0 & 1 \\0 & 1 \end{matrix} & \rvline & \begin{matrix}1\\0\\1\end{matrix}
\end{pmatrix},
\end{equation}
where
\[
\mathfrak{D} := \begin{pmatrix}1 & 0 \\0 & 1 \\1 & 1 \end{pmatrix}.
\]

Relabeling the columns of \cref{eq:expressionZXblock} as $1, 2, \dotsc, 2(L^2-1)+1$, for simplicity and conjugating its associated Hamiltonian by 
\begin{align*}
U &:= \CX(2(L^2-1)-2, 2(L^2-1))
\\&\quad \times \CX(2(L^2-1)-3, 2(L^2-1)-1)\dotsm\CX(1, 3)
\\&\quad \times \CX(2, 4)
\end{align*}
we obtain 
\[
\begin{pmatrix}
\mathfrak{D} & \rvline &  & \rvline &  & \rvline &  & \rvline & 0\\
\cline{1-7}
\mathfrak{D} & \rvline & \mathfrak{D} & \rvline & & \rvline &  & \rvline & \vdots \\
\cline{1-7}
 & \rvline &  & \rvline & \ddots & \rvline &  & \rvline & \vdots\\
\cline{1-7}
 & \rvline &  & \rvline & \mathfrak{D} & \rvline & \mathfrak{D}  & \rvline & 0\\
\hline 
0 0 & & \cdots & & 0 0 & & 0 0 & \rvline & 1\\
0 0 & & \cdots & & 0 0 & & 0 1 & \rvline & 0\\
0 0 & & \cdots & & 0 0 & & 0 1 & \rvline & 1
\end{pmatrix},
\]
which, after blocking each spin situated in an odd column to the spin situated in the next column, results in a nearest-neighbor 1D system. 

\subsection{Stabilizer subsystem toric code}

The final two models that we consider correspond to simplified versions of the subsystem toric code model \cite{kubica2022,kubica2023}, which are again defined on a cubic lattice $(V_L, \mathcal{E}_L)$ on the three-dimensional torus $\mathbb{S}_L \times \mathbb{S}_L \times \mathbb{S}_L$. In these models, the lattice is colored with two alternating colors in a chessboard manner (see \cref{fig:subsystemtoriclattice}). This way, the interactions defined in each cube will depend on its color.

As in the 3D toric code and the X-cube models, one spin is situated at the midpoint of each edge (see \cref{fig:3Dtoriclattice}). Thus, the system contains $3L^3$ spins. We only consider system sizes where $L$ is even; otherwise, the checkerboard pattern cannot be accomplished, due to the periodic boundary conditions. 

\begin{figure}
    \centering
    \input{Figures/latticesubsystemtoric.tikz}
    \caption{Visual representation of a $2 \times 2 \times 2$ subsystem toric lattice model. We have omitted the spins for simplicity. }
    \label{fig:subsystemtoriclattice}
\end{figure}
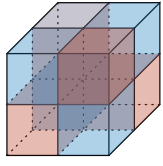

The first simplification of the subsystem toric code arises from only considering the stabilizer operators, i.e, defining $b \subset \mathcal{E}_L$ as in the X-cube model, consider the set of $8$ vertices of the cube defined by $b$. We say that $b$ is red if the sum of the coordinates of its bottom down leftmost vertex is even, and we say that it is blue otherwise.

The Hamiltonian is defined as
\[
H = - \sum_{\substack{b \subset \mathcal{E}_L \\ b \text{ red}}} J_b A_b - \sum_{\substack{b \subset \mathcal{E}_L \\ b \text{ blue}}} \tilde J_b B_b,
\]
where $J_b, \tilde J_b \in \RR$ and $A_b$ and $B_b$ are defined as 
\[
A_b = \bigotimes_{i \in b} \sigma_x^i,\quad  B_b = \bigotimes_{i \in b} \sigma^i_z.
\]
(see \cref{fig:stabsubsysteminter}). 

\begin{figure}
    \centering
    \input{Figures/stabsubsystoric.tikz}
    \caption{Visual representation of the stabilizer operators $A_b$ and $B_b$, respectively.}
    \label{fig:stabsubsysteminter}
\end{figure}
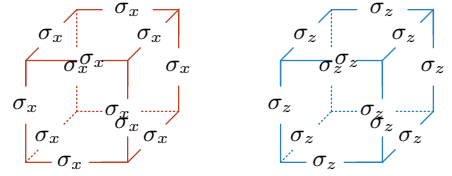

In this case, the final $Z$ matrix, after performing the pseudo-Gaussian elimination, contains $2(L^3+1)$ all-zero columns, which again is related to the existence of an exponential number of ground states. 

After rearranging the rows and columns of $Z$, it is of the same form as the one obtained when studying the two-dimensional toric code. Therefore, this model is again dual to two decoupled Ising chains. 

Similarly to the previous models, the first $L^3/2$ rows of $Z$---those corresponding to the $A_b$ operators---correspond to one of the chains, whilst the last $L^3/2$ rows of $Z$, corresponding to $B_b$ operators, describe the second chain. 

\subsection{Commuting checks subsystem toric code}
\label{sec:subsystemtoric}

The last model considered corresponds to the second simplification of the subsystem toric code model.

In this case, the simplification only considers some of the check operators defined originally in \cite{kubica2022,kubica2023}. Each check operator acts on three spins situated in the edges adjacent to a given vertex $v \in V_L$ (see \cref{fig:trianglesubsystoricinter}).

In order to formally describe the Hamiltonian associated to this model, let $b \subset \mathcal{E}_L$ be a $12$-spin subset as defined in the X-cube model. As we mentioned previously, there are eight vertices that lie in the cube associated to $b$. Let us denote by $V_1^b, V_2^b$ the two four-vertex subsets in $b$ such that if $v \in V_1^b$ (resp. $v \in V_2^b$), then $v' \notin V_1^b$ (resp. $v' \notin V_2^b$) for every $v'$ connected to $v$ by an edge of $b$ (recall that $b$ can be regarded both as a subset of $\Lambda_L$ or $\mathcal{E}_L$). We denote by $V_1^b$ the subset containing the bottom down leftmost vertex of the cube defined by $b$. 

Let $b \subset \mathcal{E}_L$. For every $v \in V_i^b$, $i \in \{1, 2\}$, let $\partial v$ be the set of three spins in $b$ such that they are situated in the edges adjacent to $v$. 

With the above notation in mind, we define $H$ as
\[
H = - \sum_{\substack{b \subset \mathcal{E}_L\\ b \text{ red}}} \sum_{v \in V_2^b} \tilde J_v A_v - \sum_{\substack{b \subset \mathcal{E}_L\\ b \text{ blue}}} \sum_{v \in V_1^b} J_v B_v,
\]
where $J_v, \tilde J_v \in \RR$ and 
\[
A_v := \bigotimes_{i \in \partial v} \sigma_x^i, \quad B_v := \bigotimes_{i \in \partial v} \sigma_z^i,
\]
(see \cref{fig:trianglesubsystoricinter}). In this case, the total number of interactions is $4L^3$.

\begin{figure}
    \centering
    \input{Figures/triangsubsystoric.tikz}
    \caption{Visual representation of the four check operators considered in each cube, depending on the color. Each three-spin operator in the cubes is of the form shown on their right-hand side, which correspond to $A_v$ and $B_v$, respectively.}
    \label{fig:trianglesubsystoricinter}
\end{figure}
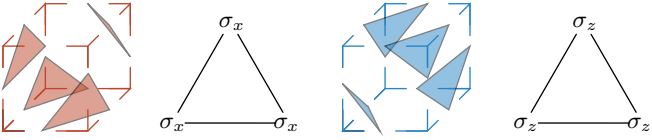

After performing the pseudo-Gaussian elimination on $Z$, there are no all-zero columns and each row will either have one or three non-zero elements. Furthermore, let $i$ be a row $i \in \{1, \dotsc, 4L^3\}$ containing three non-zero elements, which we denote by $z_{ij}$, $z_{ik}$ and $z_{il}$. Then, the $j$-th, $k$-th and $l$-th column of $Z$ will contain exactly two non-zero elements. One of which will be $z_{ij}$, $z_{ik}$ or $z_{il}$, respectively, and the other will be in certain different rows $p, q, s$. There will not be any more non-zero elements in either of those rows. 

The above structure implies that the final Hamiltonian corresponds to $L^3$ decoupled classical Ising chains with exactly three spins, with a magnetic field on both ends. This accounts for all the $4L^3$ original terms. In this case, we are not able to identify which interactions get mapped to which Ising spin chains, apart from noting that every Ising chain is associated to $A_v$-only rows, or $B_v$-only rows. 

\section{Lindbladians}

\subsection{Uniform family of Hamiltonians}

Let us present now the kind of Hamiltonians considered in this work. Let us define a \textit{potential} $\Phi$ as  a family $\Phi = \{ \Phi_\Lambda \}_{\Lambda \Subset V},$  where each $\Phi_\Lambda \in \mathcal{A}_\Lambda$.  We can then consider families of many-body Hamiltonians $\{H_\Lambda\}_{\Lambda\Subset V}$ such that
\begin{align} \label{def:UnfiformHamiltonian}
    H_\Lambda=\sum_{X\subset \Lambda}\Phi_X.
\end{align} 
By the definition of the potential, note that for each $X\subset\Lambda$, $\Phi_X$ is a self-adjoint operator acting only non-trivially on the sub-region $X$. 

Let us fix hereafter $\Lambda \Subset V$. The potential is called \textit{commuting} if for each $X,Y\subset\Lambda$, $[\Phi_X,\Phi_Y]=0$.
We denote by  $J:=\max_{X\subset\Lambda}\{\|\Phi_X\|\}$ the \textit{strength} of the interaction, and by $r:=\max\{\text{diam}(X) \, | \, X\subset\Lambda, \Phi_X\neq 0\}$ the \textit{range} or \textit{locality} of the interaction, where $\text{diam}(X)$ stands for the diameter of region $X$ with respect to the graph distance introduced above. We say that a potential has range $r$ if, and only if, it is $r$-local, and when $r=2$ we say that the potential has \textit{nearest-neighbor} interactions.

Then, the family $\{H_\Lambda \}_{\Lambda\Subset V}$  is called the  associated {\textit{family of Hamiltonians} to $(V,\Phi)$. It is called \textit{uniformly bounded} if the strength $J$ of all Hamiltonians is uniformly bounded. In this work, we will only consider uniformly bounded families of $k$-local Hamiltonians, which will frequently be commuting as well.

For any $\Lambda \Subset V$, the associated Gibbs state of the local Hamiltonian at inverse temperature $0 \leq \beta < \infty$ is denoted by
\begin{align}
    \sigma^\Lambda(\beta):=\frac{e^{-\beta H_\Lambda}}{\Tr[e^{-\beta H_\Lambda}]}\in\mathcal{S}(\mathcal{H}_\Lambda) \, ,
\end{align} and any marginal or reduced state onto any subregion $X\subset\Lambda$ is given by
\begin{align}
    \sigma_X:=\Tr_{\Lambda\setminus X}[\sigma^\Lambda]\in\mathcal{S}(\mathcal{H}_X) \, .
\end{align} 
Similarly to the description of the associated family of Hamiltonians to the potential $(V, \Phi)$ and a $\beta$, we call the family of their corresponding Gibbs states $\{\sigma^\Lambda\}_{\Lambda\Subset V}$ the \textit{family of Gibbs states} associated to $(V,\Phi,\beta)$.

\subsection{Uniform family of Lindbladians}

The dynamics of an open quantum system weakly coupled to a thermal bath can be described by a quantum Markov semigroup. Given a finite-dimensional Hilbert space $\mathcal{H}$, a \textit{quantum Markov semigroup} (QMS) is a 1-parameter continuous semigroup of quantum channels $\{ \mathcal{T}_t \}_{t\geq 0} $, with $\mathcal{T}_t : \mathcal{S} (\mathcal{H}) \rightarrow \mathcal{S} (\mathcal{H})$ for every $t\geq 0$ satisfying:
\begin{itemize}
    \item $\mathcal{T}_0=\operatorname{id}$. 
    \item $\mathcal{T}_t \circ \mathcal{T}_s=\mathcal{T}_{t+s} \; \; \forall t,s\geq 0$. 
\end{itemize}
Associated to any QMS, there is an infinitesimal generator $\mathcal{L}: \mathcal{S} (\mathcal{H}) \rightarrow \mathcal{S} (\mathcal{H})$ such that $\mathcal{T}_t = e^{t \mathcal{L}}$ for every $t \geq 0$. This also leads to the following differential equation:
\begin{equation}
    \frac{d}{dt} \mathcal{T}_t = \mathcal{L} \circ \mathcal{T}_t \, , \quad \forall t \geq 0 \, .
\end{equation}
This is frequently known as \textit{master equation} or \textit{Liouville's equation}. The conditions above precisely dictate the structure of the generator, which is shown to satisfy the GKLS form  \cite{Gorini1976CompletelyPositiveDynamicalSemigroups,Lindblad1976GeneratorsOfQuantumDynamicalSemigroups}:
\begin{equation} \label{eq:GKSL_Lindblad_form}
    \cL(\rho) = -i[H, \rho] + \sum_{k=1}^{(d^{|\Lambda|})^2-1} \gamma_k \left[L_k\rho L_k^\dagger - \frac{1}{2}\{L_k^\dagger L_k, \rho\} \right] \, ,
\end{equation}
    for $ \rho \in \mathcal{S} (\mathcal{H})$, with nonnegative constants $\gamma_k \geq 0$ termed \textit{relaxation rates}, a \textit{Hamiltonian} $H=H^\dagger$, and bounded operators $\{L_k\} \subseteq \cB(\cH_S)$ known as \textit{Lindblad} or \textit{jump operators}.

A state $\eta \in \mathcal{S} (\mathcal{H})$ is said to be \textit{invariant}, \textit{stationary} or simply a \textit{fixed point} of $\mathcal{L}$ if $\mathcal{L}(\eta)=0$. We say that $\mathcal{L}$ and its associated QMS are \textit{faithful} if there is a full-rank invariant state. Moreover,  $\mathcal{L}$ and the QMS are said to be \textit{primitive} if this full-rank invariant state is the unique fixed point, which we denote by $\sigma$. Primitivity thus ensures convergence for any initial state to $\sigma$, namely
\begin{equation}
e^{t \mathcal{L}}(\rho) \overset{t \rightarrow \infty }{\longrightarrow} \sigma \, .
\end{equation} 
This is crucial in the context of Gibbs sampling.

Given a graph $(V,E)$ and a finite subset $\Lambda \Subset V$, we consider a family of Lindbladians $\mathcal{L}_V=\{\mathcal{L}_\Lambda\}_{\Lambda\Subset V}$, such that
\begin{align}\label{eq:locallind}
    \mathcal{L}_\Lambda = \sum_{k\in \Lambda}\mathcal{L}_k \, ,
\end{align}
where $\{ \mathcal{L}_k\}_{k\in \Lambda}$ are local and finite-range Lindbladians. We assume that the \textit{strength} of the Lindbladians is uniformly bounded, namely:
\begin{equation}
    \overline{J}:= \underset{k \in V}{\text{sup}} \norm{\mathcal{L}_k}_{\diamond} < \infty \, .
\end{equation}
The set $\mathcal{L}_V=\{\mathcal{L}_\Lambda\}_{\Lambda\Subset V}$ thus constitutes a \textit{uniform family of Lindbladians} which are local, bounded and of finite range. 

We say that the family is: 
\begin{itemize}
\item \textit{Locally primitive} if $\mathcal{L}_\Lambda$ is primitive for every $\Lambda \Subset V$. 
    \item \textit{Locally reversible} if each $\mathcal{L}_\Lambda$ is locally reversible, or KMS-detailed balance, namely
    \begin{equation*}
        \langle \rho, \mathcal{L}_\Lambda(\eta) \rangle^{\text{KMS}}_{\sigma^\Gamma} = \langle  \mathcal{L}_\Lambda(\rho), \eta \rangle^{\text{KMS}}_{\sigma^\Gamma}
    \end{equation*}
    for any $\rho,\eta \in \mathcal{S}_{\Lambda}$ and any finite $\Gamma \supset \Lambda$. 
    \item \textit{Frustration free} if for any two $\Lambda \subset \Gamma \Subset V$, $\mathcal{L}_\Gamma (\rho) = 0$ implies   $\mathcal{L}_\Lambda (\rho) = 0$. 
\end{itemize}
The uniform families of Lindbladians considered in this work will be locally primitive, locally reversible and frustration-free. Note that the properties of local primitivity and frustration freeness are compatible because the former concerns for any $\mathcal{L}_\Lambda$ the existence of a unique full-rank invariant state $\sigma_\Lambda \in \mathcal{S}(\mathcal{H}_\Lambda)$, whereas the latter refers to multiple invariant states for $\mathcal{L}_\Lambda$ supported in larger spaces, i.e. in $\mathcal{S}_\Gamma$ for any finite $\Gamma \supset \Lambda$.

\subsection{Mixing times and functional inequalities}\label{sec:mixing_times}

Hereafter we assume that $\mathcal{L}_V= \{ \mathcal{L}_\Lambda \}_{\Lambda \Subset V}$ is a uniform family of local, bounded, finite-range, locally primitive, locally reversible and frustration-free Lindbladians.

The canonical notion employed to study the convergence time of quantum Markov semigroups to their fixed points is the \textit{mixing time}. Given a Lindbladian $\mathcal{L}$ with associated QMS $\{ e^{t \mathcal{L}}\}_{t \geq 0}$, and unique full-rank fixed point $\sigma$, for $\varepsilon>0 $ the mixing time of the QMS is given by 
\begin{equation}\label{eq:mixing_time}
\tau_{\text{mix}}(\varepsilon) := \text{inf} \left\{ t> 0 \, : \, \underset{\rho \in \mathcal{S}(\mathcal{H})}{\text{sup}} \left\| e^{t \mathcal{L}}(\rho) - \sigma \right\|_1 < \varepsilon \right\} \, .
\end{equation}
In this work, we are interested in the cases in which the mixing time is short enough. We say that the family of Lindbladians $ \{ \mathcal{L}_\Lambda \}_{\Lambda \Subset V}$ satisfies:
\begin{itemize}
\item \textit{Rapid mixing} if 
\begin{equation}
\tau_{\text{mix}} = \mathcal{O}(\text{poly} \text{log} |\Lambda|) \, ,
\end{equation}
\item \textit{Fast mixing}, or \textit{poly-time mixing} if 
\begin{equation}
\tau_{\text{mix}} = \mathcal{O}(\text{poly} |\Lambda|) \, ,
\end{equation}
\end{itemize} 
for every $\Lambda \Subset V$. Note that, from Eq. \eqref{eq:mixing_time}, rapid mixing can be thus equivalently written as 
\begin{equation}\label{eq:bound_mixingtime_rapidmixing}
\left\| e^{t \mathcal{L}}(\rho) - \sigma \right\|_1 \leq \text{poly}(|\Lambda|) e^{-\gamma t} \, ,
\end{equation}
and fast mixing as
\begin{equation}\label{eq:bound_mixingtime_fastmixing}
\left\| e^{t \mathcal{L}}(\rho) - \sigma \right\|_1 \leq \text{exp}(|\Lambda|) e^{-\kappa t} \, .
\end{equation}
To facilitate the practical estimation of mixing times, there are in general two central quantities which impose strict upper limits for $\tau_{\text{mix}}$: The spectral gap $\lambda$ and the modified logarithmic Sobolev constant $\alpha$ of a Lindbladian $\mathcal{L}$.

Given any primitive Lindbladian $\mathcal{L}$ with fixed point $\sigma$, its \textit{spectral gap} $\lambda(\mathcal{L})$ is defined as
    \begin{equation}\label{eq:def_spectral_gap}
        \lambda(\mathcal{L}) := \underset{\mu \in \spec(\cL)\backslash\{0\}}{\text{min}} \{|\text{Re}(\mu)|\} \; ,
    \end{equation}
    where the spectrum of $\cL$ given as
\begin{equation}
     \spec(\cL) = \{\nu \in \C \, | \; \exists \; \rho  \in \cS(\cH): \cL(\rho) = \nu \rho \} \;.
\end{equation}
When the Lindbladian is additionally locally reversible, its spectral gap is the optimal constant appearing in the \textit{Poincaré inequality}, given by 
\begin{equation}
    \lambda (\mathcal{L}) \text{Var}_\sigma ( X) \leq - \frac{d}{dt} \Big|_{t=0} \text{Var}_\sigma ( X_t) = - \langle X, \mathcal{L}^*(X) \rangle_\sigma^{\text{KMS}} \, , 
\end{equation}
for any $X \in \cB (\cH)$, where $X_t=\sigma^{-1/2} e^{t \mathcal{L}}(\rho) \sigma^{-1/2}$ for every $t\geq 0$, and  $\text{Var}_\sigma ( X)= \langle X - \Tr[\sigma X] \identity ,  X - \Tr[\sigma X] \identity \rangle_\sigma^\text{KMS} $. The existence of a positive spectral gap automatically yields an exponential decay of the variance with time by Grönwall's lemma, namely
\begin{equation}
    \text{Var}_\sigma ( X_t) \leq \text{Var}_\sigma ( X) e^{- \lambda (\cL) t} \, . 
\end{equation}
A stronger condition can be defined from the exponential decay rate of a stronger measure of distinguishability between states, namely the relative entropy. Given a primitive Lindbladian $\mathcal{L}$ with fixed point $\sigma$, it satisfies a \textit{modified logarithmic Sobolev inequality} (MLSI) if there exists a constant $\alpha$ such that
\begin{align}
    2\alpha D(\rho \| \sigma) &\leq  -\frac{d}{dt}\Big|_{t=0} D(e^{t\mathcal{L}}(\rho)\| \sigma)  \\
    & = - \Tr[ \cL(\rho) (\log \rho - \log \sigma)] \, \label{eq:entropy_production},
\end{align}
for any $\rho \in \cS(\cH)$, where the relative entropy is given by $D(\rho \| \sigma) = \Tr[\rho(\log \rho - \log \sigma)]$ and the term in \eqref{eq:entropy_production} is called \textit{entropy production}. The optimal constant for which this inequality holds is called \textit{modified logarithmic Sobolev constant}, and it is given by
\begin{equation}
    \alpha (\cL) := \underset{\rho \in \cS(\cH)}{\text{inf}} \frac{- \Tr[ \cL(\rho) (\log \rho - \log \sigma)] }{2 D(\rho \| \sigma)} \, .
\end{equation}
Similarly to the case of the spectral gap, the previous inequality is equivalent to 
\begin{equation}
    D(e^{t \cL}(\rho) \| \sigma) \leq D(\rho \| \sigma) e^{-2 \alpha(\cL) t} \, .
\end{equation}
From both the spectral gap and the MLSI of a Lindbladian, we can derive the following estimates on the mixing time of the evolution, respectively:
\begin{itemize}
\item From the spectral gap,
        \begin{equation}\label{eq:bound_mixingtime_gap}
            \|e^{t \cL}(\rho)-\sigma \|_1 \leq \sqrt{1/\sigma_{\text{min}}} \; e^{-\lambda (\cL) t} \; ,
        \end{equation}

        \item and from the MLSI,
        
        \begin{equation}\label{eq:bound_mixingtime_MLSI}
            \|e^{t \cL}(\rho)-\sigma\|_1 \leq \sqrt{2\log(1/\sigma_{\text{min}})} \; e^{-\alpha (\cL) t}\; ,
        \end{equation}
    \end{itemize}
    where $\sigma_{\text{min}}$ represents the smallest eigenvalue of $\sigma$ (note that $\sigma$ has full-rank and therefore $\sigma_{\text{min}} > 0$). Since $\sigma_{\text{min}} = \Omega (1/\exp(|\Lambda|))$ due to $\sigma$ being a Gibbs state of a local Hamiltonian, then \eqref{eq:bound_mixingtime_gap} implies \eqref{eq:bound_mixingtime_fastmixing} with an $\sqrt{\exp(|\Lambda|)}$ prefactor, whereas \eqref{eq:bound_mixingtime_MLSI} implies \eqref{eq:bound_mixingtime_rapidmixing} with a $|\Lambda|$ prefactor.

\subsection{Quantum Gibbs samplers}

In this section, given a Hamiltonian $H$ in a finite-dimensional Hilbert space, we are concerned with the sampling of the quantum Gibbs state 
\begin{equation}
  \sigma_\beta = \frac{e^{-\beta H}}{Z}  \; \text{ with } \; Z=\operatorname{Tr}[e^{-\beta H}] \;.
\end{equation}
There are different approaches to prepare Gibbs states such as quantum Metropolis algorithms, which try to mirror their classical counterpart \cite{Temme2010Metropolis}, quantum imaginary time evolutions (QITE) \cite{Motta2020DeterminingEigenstates} and, most relevant here, dissipative simulations based on Lindbladian dynamics. The discussion below concentrates on this last path. A dissipative sampler succeeds only if the evolution drives every initial state towards the desired Gibbs state. Consequently, the most important requirement for a Lindbladian is that $\sigma$ is its unique fixed point, and thus primitivity and detailed balance are crucial in our setting. 

Let us consider a finite lattice $\Lambda \Subset V$, a Hamiltonian $H_\Lambda$ on it, and an inverse temperature $0\leq \beta < \infty$. In the remainder of this section, we will introduce several Lindbladians with the desired properties that guarantee that $\sigma^\Lambda(\beta) = e^{-\beta H_\Lambda}/\Tr[e^{-\beta H_\Lambda}]$ is its unique fixed point.

\subsubsection{Davies generators}

Given a local and commuting Hamiltonian $H_\Lambda$, the \textit{Davies generator} constitutes the canonical model to describe quantum spin systems weakly coupled to a thermal bath and is given by
\begin{equation}
    \cL_\Lambda^{\text{D}}(\rho) := -i[H_\Lambda, \rho] + \sum_{k\in\Lambda} \cL_k^{\text{D}}(\rho) \;,
\end{equation}
where the local dissipators are defined as

\begin{align}
    \cL_k^{\text{D}}(\rho):=\sum_{\alpha(k), \omega} &\chi_{\alpha(k)}(\omega)\left(S_{\alpha(k)}^{\dagger}(\omega)  \rho S_{\alpha(k)}(\omega)\right.\\
    &\left.-\frac{1}{2}\left\{S_{\alpha(k)}^{\dagger}(\omega) S_{\alpha(k)}(\omega), \rho\right\}\right) \;.
\end{align}
In the upper expression, $\omega$ are the Bohr-frequencies, defined by all possible differences between the energies in $\spec(H_\Lambda)$, and the operators $S_{\alpha(k)}(\omega)$ can be understood as the Fourier components of $S_{\alpha(k)}$, decomposed into the energy basis
\begin{equation}
    S_{\alpha(k)}(\omega) = \sum_{E_i-E_j = \omega} \Pi_{E_i} S_{\alpha(k)} \Pi_{E_j} \;,
\end{equation}
where $\Pi_{E_i}$ is the projection onto the eigenspace of $H_\Lambda$ with energy $E_i$
\begin{equation}
    H_\Lambda = \sum_i E_i \; \Pi_{E_i} \;.
\end{equation}
The $S_{\alpha(k)}(\omega)$ are local, self-adjoint operators, and the $\chi_{\alpha(k)}(\omega)$  depend on $\beta$, characterize the influence of the bath onto the system, and fulfill the so-called KMS-condition $\chi_{\alpha(k)}(-\omega) = e^{-\beta \omega}\chi_{\alpha(k)}(\omega)$. 

Given a family of systems $\{\Lambda\}_{\Lambda \Subset V}$, a family of Davies Lindbladians $\{ \cL_\Lambda^D\}_{\Lambda \Subset V}$ with associated Hamiltonians $\{ H_\Lambda\}_{\Lambda \Subset V}$ is a family of locally primitive, locally reversible and frustration-free Lindbladians \cite{Kastoryano2016GibbsSamplersCommuting}, for which the locality of the Lindbladians is contained between once and twice that of the associated Hamiltonians \cite{Capel2024RapidThermalizationDissipative}.

For any $A \subseteq \Lambda$, we can consider
\begin{equation}
    \cL^{\text{D}}_{\Lambda,A}(\rho) := -i[H_A, \rho] + \sum_{k\in A} \cL_k^D(\rho) \, ,
\end{equation}
and define the \textit{Davies conditional expectation} on $A$ by
\begin{equation}\label{eq:cond_exp_Davies}
    {E}^{\text{D}}_A (\rho ) := \underset{t \rightarrow \infty}{\text{lim}} e^{t  \cL^{\text{D}}_{\Lambda,A} } (\rho) \, .
\end{equation}

Other Lindbladians that are frequently employed in the literature in the context of commuting Hamiltonians are the \textit{heat-bath generator} \cite{Kastoryano2016GibbsSamplersCommuting,BardetCapelLuciaPerezGarciaRouze-HeatBath1DMLSI-2019} and the \textit{Schmidt generator} \cite{art:2localPaper,art:bravyi2004commutative}, which will be mentioned in more detail later in this text.

\subsubsection{Chen-Kastoryano-Gilyén (CKG) generator}
The first exactly detailed‐balanced Lindbladian for arbitrary non-commuting Hamiltonians was constructed by Chen, Kastoryano and Gilyén in \cite{Chen2023EfficientExact}. We provide here a brief overview of its construction, and refer the interested reader to \cite{Chen2023EfficientExact} and \cite{Chen2023StatePreparation}. 

In contrast to the Davies generator, which derives jump operators and rates from a physical bath in the weak‐coupling limit, the CKG sampler is a purely algorithmic construction. It is designed for an arbitrary finite-dimensional Hamiltonian $H$ and---as the general GKSL form---decomposes into a coherent and a dissipative part 
\begin{equation}
     \cL^{\text{CKG}}(\rho) = -\,\frac{i}{\hbar}\bigl[B,\rho\bigr] + T(\rho) \;,
\end{equation}
where the latter is given by
\begin{align}
    &T(\rho) \\
    &\hspace{-0.08cm} =\hspace{-0.08cm} \sum_a\int_{-\infty}^\infty \hspace{-0.12cm}\gamma(\omega)\Bigl[\widehat A_a(\omega)\,\rho\,\widehat A_a(\omega)^\dagger
      \hspace{-0.1cm}-\hspace{-0.08cm}\tfrac12\{\widehat A_a(\omega)^\dagger\widehat A_a(\omega),\rho\}\Bigr] d\omega \,.
\end{align}
A central innovation of the sampler lies now in the definition of the frequency-dependent jump operators as
\begin{equation}
  \widehat A_a(\omega)\;=\;\frac1{\sqrt{2\pi}}\int_{-\infty}^\infty e^{iHt}\,A_a\,e^{-iHt}\,e^{-i\omega t}f(t)\,dt,
\end{equation}
with
\begin{equation}
  f(t)=\Bigl(\tfrac{\sigma_E}{\sqrt{2\pi}}\Bigr)^{\!1/2}e^{-\sigma_E^2t^2/2} \, ,
\end{equation}
derived from the Hamiltonian $H$ and an initial set of self-adjoint jump proposals $\{A_a\}$ that can be chosen relatively freely. The applied Gaussian energy filter blurs every transition in energy by a gentle Gaussian window $\sigma_E$ instead of demanding exact Bohr frequencies $\omega$. This leads to an advantage over other Lindbladians such as the Davies generator, which for an accurate implementation requires the precise resolution of the Bohr frequencies, which is often exponentially expensive \cite{Ding2024EfficientQuantumGibbs}. The soft blurring incurs a tiny violation of detailed‑balance, which is then canceled by a suitable chosen coherent part $B$.

As a consequence, the CKG construction yields a Lindbladian which satisfies exact detailed balance with respect to the Gibbs state of $H$ and this state is its unique steady state. Finally, shortly after the appearance of \cite{Chen2023EfficientExact}, another construction with similar properties appeared in \cite{Ding2024EfficientQuantumGibbs}, with the advantage that it uses a finite set of jump operators instead of a continuously parameterized set as in the latter work. These new Gibbs samplers have a comparable quantum simulation cost, but present greater
design flexibility, a simpler implementation and error analysis, and generalize the CKG Lindbladian.

\section{Results on Gibbs sampling}

Given a graph $(V,E)$, a uniform family of Hamiltonians on it $\{ H_\Lambda \}_{\Lambda \Subset V}$ and a fixed inverse temperature $0 \leq \beta < \infty$, consider the associated family of Gibbs states 
\begin{equation}
    \left\{ \sigma^\Lambda(\beta) = \frac{e^{-\beta H_\Lambda}}{\Tr[e^{-\beta H_\Lambda}]}  \right\}_{\Lambda \Subset V} \, .
\end{equation}
A family of quantum circuits $\{\cC_{\Lambda,\varepsilon}\}$ is called an \emph{efficient Gibbs sampler} (for $\{ \sigma^\Lambda(\beta) \}_{\Lambda \Subset V} $) if for every finite $\Lambda$ and every precision $\varepsilon>0$ the circuit $\cC_{\Lambda,\varepsilon}$ produces an output state 
$\widetilde\rho_\Lambda$ satisfying
\[
  \bigl\|\widetilde\rho_\Lambda-\sigma^\Lambda(\beta)\bigr\|_1 \;\le\;\varepsilon,
\]
with total depth bounded by
\[
  \text{depth}(\cC_{\Lambda,\varepsilon})
  = \cO\left(\operatorname{poly}\!\bigl(|\Lambda|,\;\log\tfrac{1}{\varepsilon}\bigr)\right) \;.
\]
In this manuscript, we focus on the construction of quantum Gibbs samplers with Lindbladians, i.e. the circuits above describe the implementation of certain Lindbladians with fixed points the corresponding Gibbs states. In this context, two ingredients are crucial to derive efficient Gibbs samplers:
\begin{itemize}
    \item \textit{Efficient implementation of Lindbladians.} Each channel $e^{t\cL_\Lambda}$ can be approximated to error $\varepsilon$ by a circuit of size 
  $\operatorname{poly}(|\Lambda|,t,\log\tfrac{1}{\varepsilon})$. 
    \item \textit{Quick mixing of the Lindbladians towards their fixed points.} The semigroup $e^{t\cL_\Lambda}$ converges in $\varepsilon$-trace‐distance to $\sigma^\Lambda(\beta)$ in time 
  $\tau_{\operatorname{mix}}(\varepsilon)=\operatorname{poly}(|\Lambda|,\log\tfrac{1}{\varepsilon})$. 
\end{itemize}
The latter is studied through the estimation of mixing times via spectral gaps and MLSIs as described above. We present an overview on prior results regarding fast/rapid mixing of Lindbladians in \Cref{sec:mixing_times_Lindbladians}. For the former, we need to study the \textit{implementation time} or \textit{implementation complexity}, which represents the quantum computational cost required to approximate a given Lindblad evolution to within a specified error tolerance. This metric depends sensitively on the particular simulation algorithm employed. This is explored in the next sections. 

\subsection{Prior results on Gibbs sampling}\label{sec:previous_Gibbs_sampling}

Sampling from Gibbs distributions lies at the heart of many problems in statistical physics, machine learning and probabilistic inference. Equilibrium states of physical systems are represented by Gibbs distributions, and they provide a framework for modeling complex probability spaces of high dimensions. In the context of Gibbs sampling, multiple classical algorithms such as Markov Chain Monte Carlo (MCMC) have been widely used in the literature \cite{Levin.2008}. These methods are typically efficient at high temperatures \cite{Martinelli.1999}, though more generally believed to be efficient in practice \cite{Brooks.2011}. The purpose of this subsection is to review the extension of this and other classical works such as \cite{Alaoui.2022,Chen.2024a,Anari.2021,Anari.2021a} to the study of efficiency of the quantum version of this problem.

The most natural extensions of the aforementioned works to prepare quantum Gibbs states are those based on quantum algorithms inspired by the classical Monte Carlo \cite{Temme2010Metropolis,Rall.2023,Wocjan.2023,Jiang.2024,Gilyen.2024}, but they generally lack provable guarantees unless further theoretical assumptions are made. Multiple other algorithms for quantum Gibbs sampling are based on dissipation \cite{Kastoryano2016GibbsSamplersCommuting,Ding2024EfficientQuantumGibbs,Gilyen.2024,Bardet.2023}, i.e. the existence of a Lindbladian that drives any initial state to the desired Gibbs state. The efficiency of these algorithms is based on a good implementation of the Lindbladians and quick convergence of the dissipative process governed by the Lindbladian. For general non-commuting Hamiltonians, an algorithm based on dissipation was proposed in \cite{Chen2023StatePreparation,Chen2023EfficientExact}
 and subsequently shown to be efficient in \cite{Rouze.2024b,Bakshi.2024,Rouze.2024c}. 

 In the case of local, commuting Hamiltonians, we can restrict to quantum Gibbs samplers employing a Davies generator. \cite{Kastoryano2016GibbsSamplersCommuting} showed an equivalence between the spectral gap of the Lindbladian and a form of decay of correlations in the Gibbs state, yielding thus fast mixing for the Gibbs sampler. An exponentially stronger result follows in the presence of an MLSI, which was shown to exist (depending logarithmically on the system size) for translation-invariant 1D systems in \cite{Bardet.2023,Bardet.2024} at any positive temperature, to be subsequently improved to constant MLSI in \cite{Capel2024RapidThermalizationDissipative}, where this was also extended at high temperature to any dimensional square lattices and to binary trees. Other works regarding MLSIs for commuting Hamiltonians for various systems and under different constraints involve, among others, \cite{BardetCapelLuciaPerezGarciaRouze-HeatBath1DMLSI-2019,art:2localPaper,Capel2024Gibbs,hwang2024gibbsstatepreparationcommuting}. At low temperature and in high dimensions, efficient Gibbs samplers based on positive spectral gaps are only known to exist for Kitaev’s quantum double models in 2D \cite{Alicki_2009,ding2025polynomialtimepreparationlowtemperaturegibbs,Lucia_Perez-Garcia_Perez-Hernandez_2023,Lucia_Pérez-García_Pérez-Hernández_2025}.

Many other algorithms for Gibbs sampling based on other methods than dissipation are \cite{Poulin.2009,Chowdhury.2017} based on Grover approaches, or \cite{Motta.2019} based on quantum imaginary time evolution, among others. Other works in this direction worth mentioning are \cite{Fang.2024,Harrow.2020,Scalet.2025, Fawzi.2024}, and a recent review on the subject can be found in \cite{Lin.2025}.

\subsection{Implementation of Lindbladians}\label{sec:efficient_implementation}

The CKG approach is unusual in the sense that the framework is already defined with circuit‐depth and gate‐count bounds for implementing the entire semigroup. Other Lindbladian samplers as the Davies generator require additional work to implement their dynamics on a quantum device. In this section, we give a short overview of the methods available to implement the Davies dynamics, and thereafter briefly refer back to the CKG implementation protocol.

In general, there are many procedures in the literature to simulate a given Lindblad evolution on a quantum computer; three  particularly influential approaches are the following:

\begin{enumerate} 
  \item \textbf{Trotter–Suzuki Decomposition:}  
    In this method the semigroup \(e^{t(\cL_H+\cL_D)}\) is approximated by alternating small Trotter steps of the coherent \(\exp(\Delta t\,\cL_H)\) and the dissipator part \(\exp(\Delta t\,\cL_D)\), which has to be implemented with small ancilla registers. The gate depth scales linearly in the number of segments and the Lindblad‐operator complexity \cite{Kliesch2011Dissipative}.

  \item \textbf{Block-Encoding \& Linear Combination of Unitaries (LCU):}  
    With exponential improvement of the gate count at same precision in respect to the Trotter method \cite{Lubetzki2025EfficientQuantumGibbsSampling}, here each local part of the Lindbladian is embedded into a larger unitary (block-encoding). Then amplitude amplification is used to implement their weighted sum \cite{CleveWang2016Efficient}.

  \item \textbf{Collision-Model Simulation}  
    In this case, the dissipative dynamics are implemented by repeating a fixed, low-depth gate layer that couples the system to fresh environment ancillas (reset after each interaction). In the stroboscopic limit, these repeated gate layers approximate the target dissipative dynamics with minimal circuit overhead \cite{Ciccarello2022QuantumCollision}.
\end{enumerate}

\subsubsection{Implementation of the Davies Lindbladian}\label{sec:implementation_davies}

An efficient way to implement the Davies generator for local, commuting Hamiltonians was recently proposed in \cite{Lubetzki2025EfficientQuantumGibbsSampling}, based on a framework introduced in \cite{Li2023SimulatingMarkovianOQS}. In the constructed circuit, several tools including a series expansion of the semigroup and block encodings are employed, yielding a Gibbs sampler with circuit depth scaling nearly linearly in the system size.

\begin{theorem}[Efficient Gibbs sampler for a local, commuting Hamiltonian, \cite{Lubetzki2025EfficientQuantumGibbsSampling, Li2023SimulatingMarkovianOQS}] \label{thm efficient gibbs sampler}
    Let $\{H_\Lambda\}_{\Lambda \Subset \Z^d}$ be a $k$-local, commuting family of bounded Hamiltonians with strength $J$ on finite subsets $\Lambda \Subset \Z^d$. Let $\{\cL_\Lambda\}_{\Lambda \Subset \Z^d}$ be the uniform family of Davies Lindbladians associated to $\{H_\Lambda\}_{\Lambda \Subset \Z^d}$. If they they satisfy MLSIs with respective MLSI constants $\alpha_\Lambda$, then, one can construct a quantum circuit that prepares the Gibbs states $\sigma^\Lambda(\beta) = \frac{e^{-\beta H_\Lambda}}{\Tr[e^{-\beta H_\Lambda}]}$ up to precision $\varepsilon$ with
    \begin{equation*}
    \begin{aligned}
        \cO\left(\frac{|\Lambda|^2}{\alpha_\Lambda} \operatorname{polylog}\left(\frac{|\Lambda|}{\alpha_\Lambda\varepsilon}\right)\right) &\quad \text{1- and 2- qubit gates} \\
        \cO\left(\frac{|\Lambda|}{\alpha_\Lambda} \operatorname{polylog}\left(\frac{|\Lambda|}{\alpha_\Lambda\varepsilon}\right)\right) &\quad \text{circuit depth} \;,
    \end{aligned}
    \end{equation*}
    where the dependencies on $k$, $J$, and the inverse temperature $\beta$ are absorbed in the Hamiltonian.

    If the Lindbladians $\Lambda \Subset \Z^d$ are gapped with spectral gap $\lambda_{\Lambda}$, then, the circuit that prepares the Gibbs states up to precision $\varepsilon$ has 
    \begin{equation*}
    \begin{aligned}
        \cO\left(\frac{|\Lambda|^3}{\lambda_\Lambda} \operatorname{polylog}\left(\frac{|\Lambda|}{\lambda_\Lambda\varepsilon}\right)\right) &\quad \text{1- and 2- qubit gates} \\
        \cO\left(\frac{|\Lambda|^2}{\lambda_\Lambda} \operatorname{polylog}\left(\frac{|\Lambda|}{\lambda_\Lambda\varepsilon}\right)\right) &\quad \text{circuit depth} \;,
    \end{aligned}
    \end{equation*}
    where the dependencies on $k$, $J$, and the inverse temperature $\beta$ are absorbed in the Hamiltonian.
    
\end{theorem}

Our main result for the 2D toric code achieves a circuit complexity, i.e. number of $1$- and $2$- qubit gates, of $\mathcal{O}(N^{3/2})$ for $N$ qubits, at any positive temperature, independently of $\beta$. Furthermore, as mentioned in the main text, by the Gottesman-Knill theorem, the circuit can be simulated classically in $\mathcal{O}(N^2 N^{3/2}) = \mathcal{O}(N^{7/2})$ time. This can be compared to other results that have previously appeared in the literature. 

In \cite{hwang2024gibbsstatepreparationcommuting}, the authors derive a quantum circuit of depth $\widetilde{\mathcal{O}}(N^2)$ to implement the Gibbs state of the defected toric code at any positive temperature, independently of $\beta$. 

In contrast, there are some results in dissipative quantum Gibbs sampling of the 2D toric code based on the efficient implementation of the Davies generator and a short mixing time of the dynamics. In \cite{Alicki_2009}, the spectral gap of the Davies generator was proven to be constant with the system size at any positive temperature, with an exponential scaling in $\beta$. By \cref{thm efficient gibbs sampler}, this yields an algorithm to implement the Gibbs state of the 2D toric code with circuit depth $\widetilde{\mathcal{O}}(N^{2}\exp(\beta))$ and circuit complexity of $\widetilde{\mathcal{O}}(N^{3}\exp(\beta) )$, at any positive temperature. The result of \cite{Ding2024EfficientQuantumGibbs} proves that the gap of the Davies generator can be also considered to be quadratic with respect to the system size, thus yielding an algorithm for the same object with circuit depth $\widetilde{\mathcal{O}}(N^{4}\beta)$ and circuit complexity of $\widetilde{\mathcal{O}}(N^{5}\beta )$, at any positive temperature. The recent \cite{stengele2025modifiedlogarithmicsobolevinequalities}, which shows an MLSI for the Davies generator associated to the 2D toric code at any positive temperature, improves upon \cite{Alicki_2009} and yields by \cref{thm efficient gibbs sampler} a circuit depth of $\widetilde{\mathcal{O}}(N\exp(\beta))$ and circuit complexity of $\widetilde{\mathcal{O}}(N^{2}\exp(\beta) )$.

Lastly, the recent work presented in \cite{schmidhuber2025hamiltoniandecodedquantuminterferometry} provides a classical algorithm running in $\mathcal{O}(N^4)$ time which prepares the Gibbs state of the 2D toric code.

\subsubsection{Implementation of the CKG Lindbladian}

As noted earlier, the CKG sampler comes with an explicit implementation protocol. The construction first dilates the target Lindbladian into a block‑encodable Hamiltonian (leveraging the block-encoding and LCU framework mentioned above) and then applies qubitization to simulate its continuous-time evolution. The resulting quantum circuit has depth \cite{Chen2023EfficientExact}
\begin{equation}
    \widetilde\cO(t \cdot \beta)\;,
\end{equation}
where the $\widetilde\cO$ notation absorbs several logarithmic dependencies.
To build a full Gibbs sampler from this framework, one needs to set $t=\tau_{\operatorname{mix}}(\varepsilon)$, which can cause the circuit depth to vary greatly. A key advantage is that this depth bound remains valid even for non-commuting local Hamiltonians.

\subsection{Efficient sampling under unitary conjugation}\label{sec:appendix_technical}

Here we show the proof of \cref{lem:efficientgibbssampling}. For any $\Lambda\Subset V$, the Gibbs state of $\widetilde{H}_\Lambda$ is given by 
    \begin{align*}
        \widetilde{\sigma}^\Lambda(\beta) & := \frac{e^{-\beta \,U_\Lambda H_\Lambda U^\dagger_\Lambda}}{\Tr[e^{-\beta \,U_\Lambda H_\Lambda U^\dagger_\Lambda}]} \\
        &= \frac{U_\Lambda \, e^{-\beta \,H_\Lambda}\,U^\dagger_\Lambda}{\Tr[U_\Lambda \, e^{-\beta \,H_\Lambda}\,U^\dagger_\Lambda]} = U_\Lambda \, \sigma^\Lambda (\beta) \, U^\dagger_\Lambda \;.
    \end{align*}
    Assume that there exist efficient Gibbs samplers $\{\cC_{\Lambda, \varepsilon}\}_{\Lambda \Subset V}$ for $\{H_\Lambda\}_{\Lambda \Subset V}$ (up to $\varepsilon$-trace-distance). Then one can construct efficient Gibbs samplers $\{\widetilde{\cC}_{\Lambda, \varepsilon}\}_{\Lambda \Subset V}$ for $\{\widetilde{H}_\Lambda\}_{\Lambda \Subset V}$ as follows: To prepare any $\widetilde\sigma^\Lambda(\beta)$ within $\varepsilon$-trace-distance, one has to
    \begin{enumerate}
        \item Prepare a state $\rho$ within $\varepsilon$ distance of $\sigma^\Lambda(\beta)$ with the corresponding circuit $\cC_{\Lambda, \varepsilon}$, which exists by assumption.
        \item Apply the circuit $U_\Lambda$ to the prepared state $\rho$ to reach the state $\widetilde{\rho}=U_\Lambda\,\rho \, U_\Lambda^\dagger $.
    \end{enumerate}
    The final state $\widetilde{\rho}$ is then $\varepsilon$-close to $\widetilde{\sigma}_\beta^\Lambda$ since
    \begin{equation}
        \|\widetilde{\rho}-\widetilde{\sigma}^\Lambda(\beta)\|_1 = \|U_\Lambda(\rho-{\sigma}^\Lambda(\beta))U_\Lambda^\dagger \|_1 = \|\rho-{\sigma}^\Lambda(\beta) \|_1 \leq \varepsilon \;,
    \end{equation}
    and the circuit depth is given by
    \begin{equation}
    \begin{aligned}
        \text{depth}(\widetilde{\cC}_{\Lambda, \varepsilon}) &= \text{depth}(U_\Lambda) + \text{depth}(\cC_{\Lambda, \varepsilon})\\
        &= \cO\left(\operatorname{poly}(|\Lambda|)\right) + \cO\left(\operatorname{poly}(|\Lambda|, \log\tfrac1\varepsilon)\right) \\
        &=\cO\left(\operatorname{poly}(|\Lambda|, \log\tfrac1\varepsilon)\right)\;,
    \end{aligned}
    \end{equation}
    proving that $\{\widetilde{\cC}_{\Lambda, \varepsilon}\}$ is an efficient Gibbs sampler and one may simply set
    \begin{equation}
        \widetilde{\cC}_{\Lambda, \varepsilon} := U_\Lambda \circ \cC_{\Lambda, \varepsilon}  \;\;\; \forall \Lambda \Subset V\;.
    \end{equation}

\subsection{Mixing times of Lindbladians}\label{sec:mixing_times_Lindbladians}

This section provides a general overview of the current state of the art regarding estimates on the mixing times of some Lindbladians, such as the Davies and CKG.

Let us first consider the simplest possible case. Given a finite lattice $\Lambda \Subset V$, a $1$-local (i.e. non-interacting) Hamiltonian on it, $H_\Lambda$, and given its Gibbs state at a fixed temperature $\beta >0$, $\sigma \equiv \sigma^\Lambda(\beta)$, for any $A \subseteq \Lambda$ we define the \textit{heat-bath generator} \cite{Kastoryano2016GibbsSamplersCommuting,BardetCapelLuciaPerezGarciaRouze-HeatBath1DMLSI-2019} as 
\begin{equation}\label{eq:heat_bath_generator}
    \cL^{\text{HB}}_\Lambda (\rho_\Lambda) := \underset{x \in \Lambda}{\sum} \left( \rho_{x^c} \otimes \sigma_x - \rho_\Lambda \right) \, ,
\end{equation}
for any $\rho_\Lambda \in \cS_\Lambda$. The semigroup generated by this Lindbladian reduces to the generalized depolarizing semigroup \cite{art:QuantumConditionalEntropyCapel_2018}.

\begin{theorem}[See \cite{muller2016relative,muller2016entropy} for $\sigma=\identity/d$ and \cite{art:QuantumConditionalEntropyCapel_2018,BeigiDattaRouze-ReverseHypercontractivity-2018} for general product $\sigma$] 
    Consider $\cL^{\text{HB}}_\Lambda$ defined as above, over a finite-dimensional Hilbert space in a finite graph $\Lambda$, with unique fixed point $\sigma= \underset{x\in \Lambda}{\otimes} \sigma_x$. Then, $\cL^{\text{HB}}_\Lambda$ has a positive MLSI, with optimal constant given by $1/2$.  
\end{theorem}

This generator is not unrelated to the Davies Lindbladian. Each summand of \eqref{eq:heat_bath_generator} can be written in terms of a so-called conditional expectation onto a subgraph, which when applied infinitely many times coincides with the corresponding conditional expectation associated to the Davies generator \cite{art:ApproxTensorizationBardet_2021}. 

Considering now more general commuting Hamiltonians and their associated Davies Lindbladians, we can summarize the current knowledge about positive MLSIs as follows.

\begin{theorem}[\cite{Bardet.2023,Bardet.2024,Capel2024RapidThermalizationDissipative}] \label{thm:MLSI_commuting}
    Let $V$ be a graph and consider a family of finite subgraphs $ \{ \Lambda\}_{\Lambda \Subset V} $. Consider $\{\cL^{\text{D}}_\Lambda\}_{\Lambda \Subset V}$ defined as above, with corresponding unique fixed points $\{\sigma_\Lambda\}_{\Lambda \Subset V}$ for $\{ H_\Lambda\}_{\Lambda \Subset V}$ and inverse temperature $\beta < \infty$. Then, we have the following results:
    \begin{itemize}
        \item If $V=\mathbb{Z}$ and $\{ H_\Lambda\}_{\Lambda \Subset V}$ are commuting and finite-range Hamiltonians, then $\{\cL^{\text{D}}_\Lambda\}_{\Lambda \Subset V}$ have a positive MLSI constant, independent of the system size, at any $0 < \beta < \infty$. 
        \item If $V=\mathbb{Z}^D$ and $\{ H_\Lambda\}_{\Lambda \Subset V}$ are commuting and nearest-neighbor Hamiltonians,  then $\{\cL^{\text{D}}_\Lambda\}_{\Lambda \Subset V}$ have a positive MLSI constant, independent of the system size, for $0 < \beta < \beta_*$.
        \item If $V=\mathbb{T}_b$ is a binary tree and $\{ H_\Lambda\}_{\Lambda \Subset V}$ are commuting and nearest-neighbor Hamiltonians,  then $\{\cL^{\text{D}}_\Lambda\}_{\Lambda \Subset V}$ have a log-decreasing MLSI constant for $0 < \beta < \beta_\#$.
    \end{itemize}
\end{theorem}
For the specific value on each critical temperature, we refer the reader to \cite{Capel2024RapidThermalizationDissipative}. The extension from nearest-neighbor to $k$-local Hamiltonians is unfortunately highly non-trivial and we can only guarantee that an MLSI follows for a $k$-local commuting Hamiltonian from the existence of a positive spectral gap in the Lindbladian and an exponential decay of a condition named \textit{Matrix-valued Conditional Mutual Information}, which in particular implies the decay of the conditional mutual information and the mutual information \cite{Capel2024Gibbs}.  

For specific models such as the so-called \textit{Kitaev's quantum double models} in 2D, a positive spectral gap for the associated Davies generator is known to hold at every positive temperature. The canonical representative of this family of models is the 2D toric code. 

\begin{theorem}[\cite{Alicki_2009,ding2025polynomialtimepreparationlowtemperaturegibbs,Komar.2016,Lucia_Perez-Garcia_Perez-Hernandez_2023}]\label{thm:double_models}
Let $\{ H_\Lambda\}$ be a family of Hamiltonians associated to a possibly non-Abelian quantum double model in 2D, $0 < \beta < \infty$ an inverse temperature, and $\{\cL^{\text{D}}_\Lambda\}$ a family of Davies Lindbladians with unique fixed point the Gibbs states $\{ \sigma^\Lambda(\beta)\}$. Then, $\{\cL^{\text{D}}_\Lambda\}$ has a positive spectral gap independently of the system size.
\end{theorem}

We conclude this section with the more general case of non-commuting Hamiltonians, even possibly with long-range interactions. In that setting, rapid mixing is known to hold for the CKG generators at high temperature.

\begin{theorem}[\cite{Rouze.2024b,Bakshi.2024,Rouze.2024c}]\label{thm:non_commuting}
Let $\{ H_\Lambda\}_{\Lambda \Subset \mathbb{Z}^D}$ be a family of non-commuting Hamiltonians and $\{\cL^{\text{CKG}}_\Lambda\}_{\Lambda \Subset \mathbb{Z}^D}$ a family of CKG Lindbladians with unique fixed point the Gibbs states $\{ \sigma^\Lambda(\beta)\}_{\Lambda \Subset \mathbb{Z}^D}$ at inverse temperature $0 < \beta < \infty$. Then, there exists $\beta_c >0$ such that for all $\beta < \beta_c $, $\{\cL^{\text{CKG}}_\Lambda\}_{\Lambda \Subset \mathbb{Z}^D}$ have rapid mixing.
\end{theorem}

\section{Preservation of mixing times under conjugation by unitary}
\label{sec:preservation_mixing_proof}

The following constitutes a formal version of \Cref{thm:preservation_mixing_times_informal}. 

\begin{theorem}\label{thm:preservation_mixing_times}
     Let $\cH$ be a finite-dimensional Hilbert space, $\cL: \; \cB(\cH) \to \cB(\cH)$ a primitive Lindbladian with unique fixed point $\sigma$, $U \in \cU(\cH)$ a unitary and $\widetilde{\cL}$ defined as in \eqref{eq rotated Lindbladian}. Then the following statements hold:
    
    \begin{enumerate}
    
    \item \emph{Primitive Lindbladian.}
    $\widetilde{\cL}$ is a primitive Lindbladian and its unique fixed point is $\widetilde{\sigma} = U \sigma U^\dagger$.
    
    \item \emph{Spectral gap.}
    The spectral gaps of $\cL$ and $\widetilde{\cL}$ coincide: $\lambda = \widetilde{\lambda}$.
    
    \item \emph{MLSI constant.}
    If $\cL$ satisfies an MLSI with constant $\alpha > 0$, then so does $\widetilde{\cL}$ with the same constant $\alpha$.
    
    \item \emph{Mixing time.}
    Given the mixing time $\tau_{\operatorname{mix}}(\varepsilon)$ of $\cL$ towards its fixed state $\sigma$, $\widetilde{\cL}$ has an identical mixing time $\widetilde{\tau}_{\operatorname{mix}}(\varepsilon)$ towards its respective fixed state $\widetilde{\sigma}$, i.e. $\forall\, \eps > 0: \;\;\widetilde{\tau}_{\operatorname{mix}}(\varepsilon) = \tau_{\operatorname{mix}}(\varepsilon)$.
\end{enumerate}
\end{theorem}

\begin{proof}

\textbf{(i) Primitive Lindbladian.}
Firstly, since unitary transformations conserve the rank, $\widetilde{\sigma} = U\sigma U^\dagger$ has full-rank. By definition of $\widetilde{\cL}$, it satisfies the form of \eqref{eq GKSL Lindblad form (rep in unitary invariance)} replacing $H$ with $\widetilde{H} = UHU^\dagger$ and $L_k$ by $\widetilde{L}_k = UL_kU^\dagger$ for every $k$. Since
\begin{equation}
    \widetilde{\cL}(\widetilde{\sigma}) = (\Ad_{U}\circ\cL\circ\Ad_{U^\dagger})(\widetilde{\sigma}) =(\Ad_{U}\circ\cL) (\sigma) = 0 \; ,
\end{equation}
$\widetilde{\sigma}$ is a fixed state of $\widetilde{\cL}$. This fixed point is unique: If there exists another fixed state $\widetilde{\sigma}' \neq \widetilde{\sigma}$ of $\widetilde{\cL}$, then 
\begin{align}
    \cL(U^\dagger \widetilde{\sigma}' U)& = (\Ad_{U^\dagger}\circ\,\widetilde{\cL}\circ\Ad_{U})(U^\dagger \widetilde{\sigma}' U) \\
    &= (\Ad_{U^\dagger}\circ \, \widetilde{\cL})(\widetilde{\sigma}')=0 \;,
\end{align}
so $U^\dagger\widetilde{\sigma}'U\neq\sigma$ would be a second fixed point of $\cL$, which contradicts its primitivity of $\cL$.

\textbf{(ii) Spectral gap.}
By the definition of the spectral gap (see \eqref{eq:def_spectral_gap}),
then for any eigenstate $\rho$ with $\cL(\rho) = \nu \rho$, the state $U\rho U^\dagger$ is an eigenstate of $\widetilde{\cL}$ with the same eigenvalue $\nu$. On the other hand, any eigenstate $\varrho$ of $\widetilde{\cL}$ automatically corresponds to the eigenstate $U^\dagger \varrho U$ of $\cL$, again both with the same eigenvalue. This ultimately leads to the conclusion of identical spectra for $\cL$ and $\widetilde{\cL}$, and thus according to its definition to an identical spectral gap $\lambda=\widetilde{\lambda}$.

\textbf{(iii) MLSI constant.}
Assuming that $\cL$ satisfies an MLSI with constant $\alpha \geq 0$, we have
\begin{equation} \label{eq theorem MLSI equation}
     \alpha=  \underset{\rho \in \cS(\cH)}{\operatorname{inf}}\frac{- \Tr[\cL(\rho)(\log\rho-\log\sigma)]}{2D(\rho\|\sigma)}   \;.
\end{equation}
Note that 
\begin{equation}
\begin{aligned}
    &-\Tr[\widetilde{\cL}(\rho)(\log(\rho)-\log(\widetilde{\sigma}))] \\
    &= -\Tr[\cL(U^\dagger\rho U)(\text{log}(U^\dagger\rho U)-\log(\sigma))] \;
\end{aligned}
\end{equation}
by unitary invariance of the trace. Since for any $U \in \cU(\cH)$, we have $\{ \rho \in \cS(\cH)\}= \{U^\dagger\rho U: \rho \in \cS(\cH) \}$, then
\begin{equation} \label{eq:equiv_MLSI}
\begin{aligned}
        \alpha&=  \underset{U^\dagger\rho U \in \cS(\cH)}{\operatorname{inf}}\frac{- \Tr[\cL(U^\dagger\rho U)(\log(U^\dagger\rho U)-\log\sigma)]}{2D(U^\dagger\rho U\|\sigma)}  \\
        & =  \underset{\rho \in \cS(\cH)}{\operatorname{inf}}\frac{- \Tr[\widetilde{\cL}(\rho)(\log(\rho)-\log\widetilde{\sigma})]}{2D(\rho \|\widetilde{\sigma})} \, ,
\end{aligned}
\end{equation}
and hence, $\widetilde{\cL}$ satisfies an MLSI with the same constant $\alpha$.

\textbf{(iv) Mixing time.} It follows trivially from the definition of mixing time, the unitary invariance of the 1-norm, and the fact that $\{ \rho \in \cS(\cH)\}= \{U^\dagger\rho U:\rho \in \cS(\cH) \}$.

\end{proof}

We can immediately extend the previous result to families of Lindbladians associated to families of Hamiltonians. 

\begin{cor}
\label{cor:unitary_invariance_families}
Let $\{\cL_\Lambda\}:= \{\cL_\Lambda\}_{\Lambda \Subset V}$ be a primitive uniform family of Lindbladians with corresponding fixed states $\{\sigma^\Lambda\}_{\Lambda \Subset V}$. For a family of unitaries $\{U_\Lambda\}_{\Lambda \Subset V}$ with $U_\Lambda \in \cU(\cH_\Lambda)$, the transformed family $\{\widetilde{\cL}_\Lambda\}$ given by 
\begin{equation}
    \{\widetilde{\cL}_\Lambda\} :=  \left\{\Ad_{U_\Lambda} \circ \cL_\Lambda \circ \Ad_{U^\dagger_\Lambda} \right\}_{\Lambda \Subset V}\;,
\end{equation}
satisfies the same mixing bounds as $\{\cL_\Lambda\}$. In particular $\{\widetilde{\cL}_\Lambda\}$ achieves fast/rapid mixing if, and only if, $\{\cL_\Lambda\}$ does so.
\end{cor}

In the next result, we explore the preservation of efficient sampling under conjugation by circuits of polylog-depth, provided that an initial family of Lindbladians satisfies an MLSI.

\begin{cor}\label{cor:examples_MLSI_gap}
Consider $\{ H_\Lambda \}_{\Lambda\Subset V}$ a family of Hamiltonians with corresponding Gibbs states at $\beta \geq 0$ $\{\sigma^\Lambda(\beta)\}_{\Lambda\Subset V}$ and $\{ \cL_\Lambda \}_{\Lambda\Subset V}$ a family of Lindbladians that will be taken to be the Davies or the CKG. Given any family of unitaries $\{ U_\Lambda \}_{\Lambda\Subset V}$, consider the dual families $\{ \widetilde{H}_\Lambda \}_{\Lambda\Subset V}$ and $\{ \widetilde{\cL}_\Lambda \}_{\Lambda\Subset V}$ as above. Then, we have:
\begin{itemize}
    \item For $\{ \cL_\Lambda \}_{\Lambda\Subset V}$ Davies generators and $\{ U_\Lambda \}_{\Lambda\Subset V}$ of polylog-depth, $\widetilde{\sigma}$ can be sampled with a circuit of depth $\mathcal{O}(|\Lambda| \operatorname{polylog}(|\Lambda|))$ in the following cases:
    \begin{itemize}
        \item If $\{ H_\Lambda \}_{\Lambda\Subset V}$ are non-interacting, $\forall \beta < \infty$ .
        \item If $V=\mathbb{Z}$ and $\{ H_\Lambda\}_{\Lambda \Subset V}$ are commuting and finite-range Hamiltonians at any $0 < \beta < \infty$.
         \item If $V=\mathbb{Z}^D$ and $\{ H_\Lambda\}_{\Lambda \Subset V}$ are commuting and nearest-neighbor Hamiltonians for $0 < \beta < \beta_*$.
        \item If $V=\mathbb{T}_b$ is a binary tree and $\{ H_\Lambda\}_{\Lambda \Subset V}$ are commuting and nearest-neighbor Hamiltonians for $0 < \beta < \beta_\#$.
    \end{itemize}
    \item For $\{ \cL_\Lambda \}_{\Lambda\Subset V}$ CKG generators and $\{ U_\Lambda \}_{\Lambda\Subset V}$ of polylog-depth,  $\widetilde{\sigma}$ can be sampled with a circuit of depth $\mathcal{O}(|\Lambda| \operatorname{polylog}(|\Lambda|))$  if $V=\mathbb{Z}^D$ and $\{ H_\Lambda\}_{\Lambda \Subset V}$ are non-commuting Hamiltonians for $0 < \beta < \beta_c$.
\end{itemize}

\end{cor}

\begin{proof}
   These are straightforward consequences of \cref{thm:preservation_mixing_times} jointly with the existence of a positive MLSI and/or rapid mixing for these models at the corresponding temperatures mentioned in the statement, as recalled in \cref{thm:MLSI_commuting} and \cref{thm:non_commuting}, respectively, and the efficiency on the implementation of the sampler of the original family of Gibbs states, as recalled in \cref{thm efficient gibbs sampler}. 
\end{proof}

This result can be translated to circuits of depth $\mathcal{O}(\operatorname{poly}(|\Lambda|))$ in the case in which the initial Lindbladian satisfies a positive spectral gap (possibly depending polynomially on the system size), and the unitaries considered are of poly-depth. This applies for all the models of  \cref{thm:double_models}.

\section{Lasso 1D Ising model}\label{sec:lasso}

This section is devoted to a more thorough understanding of the lasso 1D Ising model. In particular, here we show that the lasso 1D Ising model can be efficiently sampled. Even more, we show that the Davies generator with unique fixed point its Gibbs state at any positive temperature satisfies a positive MLSI constant.

Beforehand, we need to introduce another mathematically tractable conditional expectation (and therefore local Lindbladian) that has been employed in the recent literature of quantum Markov semigroups. This construction requires commuting Hamiltonians which are additionally nearest-neighbor, and can therefore be applied in our setting. 

The \textit{Schmidt conditional expectation} was introduced in \cite{art:2localPaper} following a construction of \cite{art:bravyi2004commutative}. Given a graph $(V,E)$ and a $2$-local commuting interaction on it, for $\Lambda \Subset V$ the corresponding Hamiltonian
\begin{equation}
    H_\Lambda  = \underset{(i,j)\in E_\Lambda}{\sum} H_{i,j}
\end{equation}
is composed of terms that act non-trivially only on vertices $i$ and $j$. Given any two terms of the interaction, they overlap in, at most, one vertex. Let us consider the simplified case of vertices $A,B,C$, and interactions $H_{AB}$ and $H_{BC}$, for a fixed $\beta >0$. Then, we can Schmidt decompose $e^{-\beta H_{AB}}$ and $e^{-\beta H_{BC}}$, simultaneously, with a compatible decomposition of $\cB_B$ for both interactions, in the sense that the algebras generated by the ``$B$-part'' of the decompositions of both $e^{-\beta H_{AB}}$ and $e^{-\beta H_{BC}}$ commute. This is extended to any subregion of $V$ in the same way, by considering any 3 consecutive sites of that region and performing this Schmidt decomposition, verifying that the decompositions of the algebras in the intermediate sites are always compatible. 

Then, for $A$, this allows us to define a conditional expectation ${E}_A^{\text{S}}$ onto a *-subalgebra that acts trivially on $A$, decomposes its boundary so that the action of two conditional expectations with two overlapping boundaries commute, and acts non-trivially in the complement of the lattice. For any subregion of $V$, the local Lindbladian onto such a subregion is defined as the difference between such a conditional expectation and the identity. This construction in particular yields the property that, for $A_1,A_2 \subset \Lambda \Subset V$ with $d(A_1,A_2)\geq 2$, the following holds:
\begin{equation}\label{eq:tensorization_Schmidt}
    {E}_{A_1}^{\text{S}} \circ {E}_{A_2}^{\text{S}}  = {E}_{A_2}^{\text{S}}\circ {E}_{A_1}^{\text{S}} = {E}_{A_1\cup A_2}^{\text{S}} \, ,
\end{equation}
which is crucial for the derivation of estimates of mixing times. For a general description of these dynamics, see \cite[Section 3.2]{Capel2024RapidThermalizationDissipative}. 

Now we are in position to state and prove the main result of this section.

\begin{theorem}
    For any $0 <\beta < \infty$, let $\{ H_\Lambda \}$ be a family of Lasso 1D Ising models (see \cref{fig:lassoising}) and let $\{ \cL^{\operatorname{D}}_\Lambda \}$ be the associated Davies generators with unique fixed points $\{ \sigma^\Lambda(\beta)\}$. Then, $\{ \cL^{\operatorname{D}}_\Lambda \}$ have a positive MLSI independent of the system size. 
\end{theorem}

\begin{proof}
    Here we will only follow the steps of the proof of \cite[Theorem 5.3]{Capel2024RapidThermalizationDissipative}, where the analogous result for 1D spin chains was shown. As a 1D lasso is a slight modification of the geometry of 1D chains, we will only present here a sketch of the proof focusing particularly on those steps which differ from the aforementioned result. 

    \begin{figure}
    \centering
    \input{Figures/lasso_plain.tikz}
    \caption{2-coloring of the lasso 1D Ising, where we are coloring $\Lambda_0$ in red and $\Lambda_1$ in green.}
    \label{fig:fig1}
\end{figure}
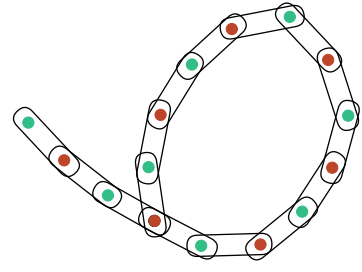

    Let us first fix $\Lambda$, an inverse temperature $0 < \beta < \infty$, and consider a 2-coloring of the lasso $\{ \Lambda_0 , \Lambda_1\}$ as in \cref{fig:fig1}, where we enforce that the site of connection with the lasso belongs to $\Lambda_0$. In the case of an odd number of sites in the cycle we block two sites into one, so that a 2-coloring exists without loss of generality. Let us recall that we want to prove 
    \begin{equation}
      2 \alpha (\cL^{\operatorname{D}}_\Lambda)   D(\rho \| \sigma^\Lambda) \leq - \tr[ \cL^{\operatorname{D}}_\Lambda( \rho ) (\log \rho - \log \sigma^\Lambda)]
    \end{equation}
    for every $\rho\in \cD(\cH)$, where we drop $\beta$ from $\sigma^\Lambda$ for simplicity. The first step of the proof requires conditioning onto the set of points $\Lambda_0$ and the chain rule for the relative entropy \cite{OhyaPetz1993Book}, \cite[Lemma 3.4]{art:ChainRuleForRelEntropyConExpJunge_2022}:
    \begin{equation}\label{eq:chain_rule_RelEnt}
        D(\rho \| \sigma^\Lambda) = D(\rho \| E_{\Lambda_0}^{\text{S}} (\rho) ) + D( E_{\Lambda_0}^{\text{S}} (\rho) \| \sigma^\Lambda ) \, . 
    \end{equation}
    Let us denote $\omega := E_{\Lambda_0}^{\text{S}} (\rho) $. The first summand of \eqref{eq:chain_rule_RelEnt} satisfies exact tensorization due to \eqref{eq:tensorization_Schmidt}, as all points of $\Lambda_0$ are at least at distance $2$. Therefore,
    \begin{equation*}
        D(\rho \| \omega) \leq \underset{x_k \in \Lambda_0}{\sum} D(\rho \| E_{\{ x_k\}}^{\text{S}} (\rho)) \, .
    \end{equation*}
    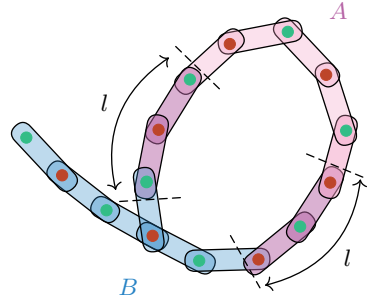
\begin{figure}
    \centering
    \input{Figures/lasso_superposition.tikz}
    \caption{Splitting of $\Lambda$ into two overlapping regions $A$ and $B$, where the interactions fully contained in $A$ are represented in pink, those in $B$ in blue, and those in the overlap in purple. The overlap has two connected components, each of size $\ell$, and in the figure $\ell=3$, as the size is measured in number of sites.}
    \label{fig:fig2}
\end{figure}
    Moreover, since we are interested in obtaining an MLSI for the Davies generator, not the Schmidt one, we can further upper bound this term as in \cite[Lemma 3.10]{Capel2024RapidThermalizationDissipative} by 
    \begin{equation}\label{eq:bound_ChR_1}
        D(\rho \| \omega) \leq \underset{x_k \in \Lambda_0}{\sum} D(\rho \| E_{\{ x_k\}\partial}^{\text{D}} (\rho)) \, ,
    \end{equation}
    where $E_{\{ x_k\}\partial}^{\text{D}}$ is the Davies conditional expectation from \eqref{eq:cond_exp_Davies}. The second summand in the RHS of \eqref{eq:chain_rule_RelEnt} is more complicated. It requires a standard argument of ``divide and conquer'' in order to upper bound $D(\omega \| \sigma^\Lambda ) $ by a sum of relative entropies involving Davies conditional expectations. This argument typically splits the graph considered at each step into two overlapping pieces, and obtains an upper bound for a conditional relative entropy in that graph in terms of two conditional relative entropies in each of the pieces, and a multiplicative error term depending on how correlations decay on the Gibbs state between the regions in the complement of the overlap. In particular, it completely follows the steps of  \cite[Lemma 5.11]{Capel2024RapidThermalizationDissipative}, except for the first step, since now, in order to split the lasso into two pieces, we have to ``cut'' from two parts, and thus consider two overlaps (see \cref{fig:fig2}). By \cite[Theorem 8]{art:2localPaper} and its modified version \cite[Theorem 5.9]{Capel2024RapidThermalizationDissipative}, we have
    \begin{equation}\label{eq:approximate_tensorization}
        D(\omega \| \sigma^\Lambda) \leq \frac{1}{1 - 2 \eta_{A,B}} \left[ D(\omega \| E_A^{\operatorname{S}}(\omega)) + D(\omega \| E_B^{\operatorname{S}}(\omega)) \right] \, ,
    \end{equation}
    where $\eta_{A,B}$ is a condition of decay of correlations between $A^c$ and $B^c$ called q$\mathbb{L}_1\rightarrow\mathbb{L}_\infty$-clustering. Note that, up to blocking of a few sites and their corresponding interactions as shown in \cref{fig:fig3},  $A^c$ and $B^c$ constitute two 1D segments, which we denote by $X$ and $Y$, respectively (see \cref{fig:fig4}). Thus, we can use \cite[Theorem 4.22 and Theorem 4.23]{Capel2024RapidThermalizationDissipative} to conclude that the decay of the standard covariance 
\begin{figure}
    \centering
    \input{Figures/lasso_blocking.tikz}
    \caption{Blocking performed in one of the sites of the $B$ region, where 4 sites are blocked into 1, and all their 2-local interactions are transformed into a larger 1-local interaction. }
    \label{fig:fig3}
\end{figure}
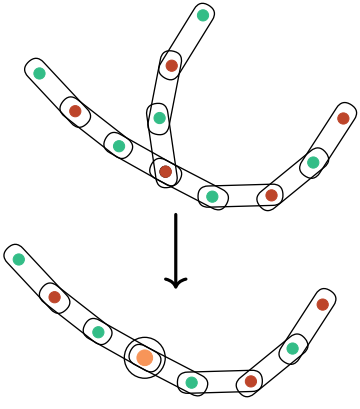
    \begin{figure}
    \centering
    \input{Figures/lasso_complementary.tikz}
    \caption{Complements of $A$ and $B$, denoted by $X$ and $Y$, respectively. The distance between $X$ and $Y$ is $\ell+1$, i.e. 4 in the figure. }
    \label{fig:fig4}
\end{figure}
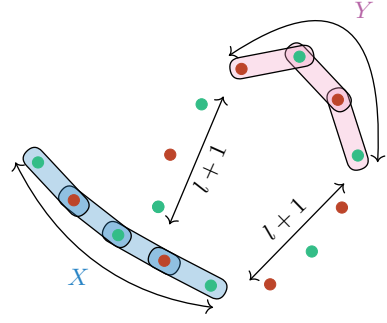
    \begin{align*}
        &\operatorname{Cov}_{\sigma^\Lambda}(A^c:B^c) \\
        &=   \hspace{-0.3cm} \underset{\norm{O_{A^c}}=\norm{O_{B^c}}=1}{\operatorname{sup}} \hspace{-0.3cm} \big| \tr[\sigma^\Lambda O_{A^c} O_{B^c}]-\tr[\sigma^\Lambda O_{A^c}]\tr[\sigma^\Lambda O_{B^c}]\big|
    \end{align*}
    is equivalent to the q$\mathbb{L}_1\rightarrow\mathbb{L}_\infty$-clustering, and thus, by \cite{Kimura_2025}, it decays exponentially with the distance between $A^c$ and $B^c$. In more detail, note that even though $X$ and $Y$ are 1D segments, it is not immediate to interpret the decay of correlations between them with $\ell$ as contained in a 1D chain as in \cite{Araki1969} and \cite{art:Bluhm2022exponentialdecayof}. However, we could pairwise identify and block all the sites in  \cref{fig:fig4} except for the last two on the left, and ``bend'' both segments $X$ and $Y$ so that the construction is transformed into a fully 1D model. Using thus that correlations decay exponentially in 1D even for commuting, non-translation invariant interactions at any positive temperature, as proven in \cite{Kimura_2025}, we conclude the exponential decay of $\eta_{A,B}$ with $\ell +1$. 
    
    This is the first step for the recursion in \cite[Lemma 5.11]{Capel2024RapidThermalizationDissipative}. In the next step, each of the segments $X$ and $Y$ is split into two overlapping parts, using a similar result of approximate tensorization to \eqref{eq:approximate_tensorization}. This yields 4 conditional relative entropies (in each of the ``halves'' of each of the 2 segments above) and 2 multiplicative error terms of the form of $\eta$, which will decay exponentially with the size of the overlaps. By carefully choosing the size of the ``halves'' and the size of the overlaps, it can be shown as in \cite[Lemma 5.11]{Capel2024RapidThermalizationDissipative} that, after a multiple step recursion, we have
    \begin{align*}
        D(\omega \| \sigma^\Lambda) &= D(\omega \| E_{\Lambda}(\omega))  \leq C \underset{R_k \subset \Lambda}{\sum} D(\omega \| E_{R_k}(\omega)) \, ,
    \end{align*}
    for $C$ a positive constant, where all the $R_k$ are given by segments of size $\ell_0$ intersected with $\Lambda_0$. By data-processing inequality and again the use of \eqref{eq:tensorization_Schmidt} and \cite[Lemma 3.10]{Capel2024RapidThermalizationDissipative}, we have
    \begin{align}\label{eq:bound_ChR_2}
         D(\omega \| \sigma^\Lambda) & \leq C \underset{R_k \subset \Lambda}{\sum} D(E_{\Lambda_0}^{\operatorname{S}}(\rho) \| E_{R_k}^{\operatorname{S}}\circ E_{\Lambda_0}^{\operatorname{S}}(\rho)) \nonumber \\
         & \leq C \underset{R_k \subset \Lambda}{\sum} D( \rho \| E_{R_k}^{\operatorname{S}} (\rho)) \nonumber \\
         & \leq C \underset{R_k \subset \Lambda}{\sum} D( \rho \| E_{R_k \partial}^{\operatorname{D}} (\rho)) \, .
    \end{align}
    Now, combining \eqref{eq:chain_rule_RelEnt} with \eqref{eq:bound_ChR_1} and \eqref{eq:bound_ChR_2}, and using that by \cite[Theorem 1.4]{art:CompleteEntropicInequalities_GaoRouze_2022} there exist positive cMLSI constants $\alpha_1$ and $\alpha_2$ such that for any $j,k$,
    \begin{equation}
        \alpha_1 \hspace{-1pt} D(\rho \| E_{\{ x_j\}\partial}^{\text{D}} (\rho)\hspace{-1pt})\hspace{-2.5pt} \leq \hspace{-2.5pt} \underbrace{- \hspace{-1pt} \tr [ \hspace{-0.5pt} \cL_{\{ x_j\}\partial}^{\text{D}}(\rho) \hspace{-1pt} \big(\hspace{-1.5pt} \log \rho \hspace{-1pt} - \hspace{-1pt} \log \hspace{-0.5pt} E_{\{ x_j\}\partial}^{\text{D}} (\rho) \big)\hspace{-1pt}]}_{\operatorname{EP}_{\cL_{\{ x_j\}\partial}^{\text{D}}}(\rho)},
    \end{equation}
    and 
    \begin{equation}
        \alpha_2 D(\rho \| E_{R_k\partial}^{\text{D}} (\rho)) \leq \underbrace{- \tr[ \cL_{R_k\partial}^{\text{D}}(\rho) \big( \log \rho - \log E_{R_k\partial}^{\text{D}} (\rho) \big)]}_{\operatorname{EP}_{\cL_{R_k\partial}^{\text{D}}}(\rho)} ,
    \end{equation}
    where EP stands for \textit{entropy production}, we conclude
    \begin{align*}
        D(\rho \| \sigma^\Lambda) \leq & \frac{1}{\operatorname{min} \{ \alpha_1 , \alpha_2\}} \\
        & \cdot \left( \underset{x_j \in \Lambda_0}{\sum}  \operatorname{EP}_{\cL_{\{ x_j\}\partial}^{\text{D}}}(\rho) + C\underset{R_k \subset \Lambda}{\sum}  \operatorname{EP}_{\cL_{R_k\partial}^{\text{D}}}(\rho) \right) \\
       & \leq \frac{m C}{\operatorname{min} \{ \alpha_1 , \alpha_2\}} \operatorname{EP}_{\cL_{\Lambda}^{\text{D}}}(\rho) \, ,
    \end{align*}
    where we are using the positivity and additivity of the entropy production, as well as the fact that, by construction, each site is contained in at most a constant number of $R_k$, say $m$. Since all multiplicative factors in the last term are positive constants, this concludes the proof.  
\end{proof}

In an analogous way, we could have proven a positive MLSI for classical generators associated to the Lasso 1D Ising, such as the Glauber dynamics, as the Hamiltonian is purely classical. We have preferred to include the proof for the Davies dynamics, as it concerns quantum systems evolving towards the Gibbs state of the Lasso, and thus is more general in principle. 
\end{document}

%% file: Figures/gibbssampler.tikz
\begin{tikzpicture}[scale=1]
    \filldraw[shift={(1.094, 4.163)}, scale=0.5, fill=RoyalBlue!80, opacity=0.3] (0, 0) rectangle (0.564, -1.411);
    \draw[shift={(1.094, 4.022)}, scale=0.5] (0, 0) -- (-1.693, 0);
    \draw[shift={(1.094, 3.598)}, scale=0.5] (0, 0) -- (-1.693, 0);
    \draw[shift={(1.376, 3.598)}, scale=0.5] (0, 0) -- (0.282, 0);
    \draw[shift={(1.799, 3.598)}, scale=0.5] (0, 0) -- (0.282, 0);
    \draw[shift={(1.376, 4.022)}, scale=0.5] (0, 0) -- (1.129, 0);
    \draw[shift={(2.223, 4.022)}, scale=0.5] (0, 0) -- (0.564, 0);
    \draw[shift={(2.223, 3.598)}, scale=0.5] (0, 0) -- (0.564, 0);
    \draw[shift={(1.094, 3.104)}, scale=0.5] (0, 0) -- (-1.693, 0);
    \draw[shift={(1.517, 3.104)}, scale=0.5] (0, 0) -- (-0.282, 0);
    \draw[shift={(1.799, 3.104)}, scale=0.5] (0, 0) -- (0.282, 0);
    \filldraw[shift={(1.94, 4.163)}, scale=0.5, fill=RoyalBlue!80, opacity=0.3] (0, 0) rectangle (0.564, -1.411);
    \filldraw[shift={(1.517, 3.739)}, scale=0.5, fill=RoyalBlue!80, opacity=0.3] (0, 0) rectangle (0.564, -1.411);
    \filldraw[shift={(1.094, 2.893)}, scale=0.5, fill=RoyalBlue!80, opacity=0.3] (0, 0) -- (0, 0.706) -- (0.564, 0.706) -- (0.564, 0);
    \filldraw[shift={(1.94, 2.893)}, scale=0.5, fill=RoyalBlue!80, opacity=0.3] (0, 0) -- (0, 0.706) -- (0.564, 0.706) -- (0.564, 0);
    \draw[shift={(2.223, 3.104)}, scale=0.5] (0, 0) -- (0.564, 0);
    \filldraw[shift={(1.094, 0.847)}, xscale=0.5, yscale=-0.5, fill=RoyalBlue!80, opacity=0.3] (0, 0) rectangle (0.564, -1.411);
    \draw[shift={(1.094, 0.988)}, xscale=0.5, yscale=-0.5] (0, 0) -- (-1.693, 0);
    \draw[shift={(1.094, 1.411)}, xscale=0.5, yscale=-0.5] (0, 0) -- (-1.693, 0);
    \draw[shift={(1.376, 1.411)}, xscale=0.5, yscale=-0.5] (0, 0) -- (0.282, 0);
    \draw[shift={(1.799, 1.411)}, xscale=0.5, yscale=-0.5] (0, 0) -- (0.282, 0);
    \draw[shift={(1.376, 0.988)}, xscale=0.5, yscale=-0.5] (0, 0) -- (1.129, 0);
    \draw[shift={(2.223, 0.988)}, xscale=0.5, yscale=-0.5] (0, 0) -- (0.564, 0);
    \draw[shift={(2.223, 1.411)}, xscale=0.5, yscale=-0.5] (0, 0) -- (0.564, 0);
    \filldraw[shift={(1.94, 0.847)}, xscale=0.5, yscale=-0.5, fill=RoyalBlue!80, opacity=0.3] (0, 0) rectangle (0.564, -1.411);
    \filldraw[shift={(1.517, 1.623)}, xscale=0.5, yscale=-0.5, fill=RoyalBlue!80, opacity=0.3] (0, 0) -- (0, 0.706) -- (0.564, 0.706) -- (0.564, 0);
    \draw[shift={(4.198, 4.022)}, scale=0.5] (0, 0) -- (0.564, 0);
    \draw[shift={(4.198, 3.598)}, scale=0.5] (0, 0) -- (0.564, 0);
    \filldraw[shift={(3.916, 4.163)}, scale=0.5, fill=RoyalBlue!80, opacity=0.3] (0, 0) rectangle (0.564, -1.411);
    \filldraw[shift={(3.916, 2.893)}, scale=0.5, fill=RoyalBlue!80, opacity=0.3] (0, 0) -- (0, 0.706) -- (0.564, 0.706) -- (0.564, 0);
    \draw[shift={(4.198, 3.104)}, scale=0.5] (0, 0) -- (0.564, 0);
    \draw[shift={(3.634, 4.022)}, scale=0.5] (0, 0) -- (0.564, 0);
    \draw[shift={(3.634, 3.598)}, scale=0.5] (0, 0) -- (0.564, 0);
    \draw[shift={(3.634, 3.104)}, scale=0.5] (0, 0) -- (0.564, 0);
    \draw[shift={(4.198, 0.988)}, xscale=0.5, yscale=-0.5] (0, 0) -- (0.564, 0);
    \draw[shift={(4.198, 1.411)}, xscale=0.5, yscale=-0.5] (0, 0) -- (0.564, 0);
    \filldraw[shift={(3.916, 0.847)}, xscale=0.5, yscale=-0.5, fill=RoyalBlue!80, opacity=0.3] (0, 0) rectangle (0.564, -1.411);
    \draw[shift={(3.634, 0.988)}, xscale=0.5, yscale=-0.5] (0, 0) -- (0.564, 0);
    \draw[shift={(3.634, 1.411)}, xscale=0.5, yscale=-0.5] (0, 0) -- (0.564, 0);
    \draw[shift={(1.658, 2.752)}, scale=0.5, dotted] (0, 0) -- (0, -1.693);
    \draw[shift={(4.057, 2.681)}, scale=0.5, dotted] (0, 0) -- (0, -1.693);
    \draw[shift={(2.646, 3.528)}, scale=0.5, dotted] (0, 0) -- (1.693, 0);
    \draw[shift={(2.646, 1.199)}, scale=0.5, dotted] (0, 0) -- (1.693, 0);
    \draw[shift={(6.668, 4.022)}, xscale=-0.5, yscale=0.5] (0, 0) -- (0.564, 0);
    \draw[shift={(6.668, 0.988)}, scale=-0.5] (0, 0) -- (0.564, 0);
    \filldraw[fill=Lavender, opacity=0.3] (4.48, 5.644) rectangle (6.385, 0.494);
    \node[anchor=center, font=\LARGE] at (3.069, 2.328) {$U$};
    \node[anchor=center, font=\LARGE] at (5.433, 2.54) {$\mathcal{C}$};
    \node[anchor=center] at (6.979, 4.022) {$|0\rangle$};
    \node[anchor=center] at (6.985, 0.988) {$|0\rangle$};
    \draw[dotted] (6.773, 3.387) -- (6.773, 1.693);
    \draw[shift={(6.668, 5.434)}, xscale=-0.5, yscale=0.5] (0, 0) -- (0.564, 0);
    \draw[shift={(6.668, 4.587)}, xscale=-0.5, yscale=0.5] (0, 0) -- (0.564, 0);
    \draw[shift={(6.668, 4.587)}, xscale=-0.5, yscale=0.5] (0, 0) -- (0.564, 0);
    \node[anchor=center] at (7.267, 5.01) {$|0\rangle^{\otimes m}$};
    \draw[shift={(6.668, 4.728)}, xscale=-0.5, yscale=0.5] (0, 0) -- (0.564, 0);
    \draw[shift={(6.668, 5.257)}, xscale=-0.5, yscale=0.5] (0, 0) -- (0.564, 0);
    \draw[dotted] (6.526, 5.157) -- (6.526, 4.842);
    \draw[->] (6.773, 6.209) -- (1.976, 6.209);
    \node[anchor=center] at (4.516, 6.491) {Time};
    \draw[shift={(4.48, 5.434)}, xscale=-0.5, yscale=0.5] (0, 0) -- (0.564, 0);
    \draw[shift={(4.48, 4.587)}, xscale=-0.5, yscale=0.5] (0, 0) -- (0.564, 0);
    \draw[shift={(4.48, 4.728)}, xscale=-0.5, yscale=0.5] (0, 0) -- (0.564, 0);
    \draw[shift={(4.48, 5.257)}, xscale=-0.5, yscale=0.5] (0, 0) -- (0.564, 0);
    \draw[dotted] (4.339, 5.157) -- (4.339, 4.842);
    \node[anchor=center] at (3.387, 5.009) {discard};
    \draw [thick,decorate,decoration={calligraphic brace,amplitude=10pt}] (4.233, 0.423) -- (0.988, 0.423) node[midway,below,yshift=-10pt,font=\small\sffamily]{Poly($n$)};
    \draw [thick,decorate,decoration={calligraphic brace,amplitude=10pt}] (6.35, 0) -- (4.516, 0) node[midway,below,yshift=-10pt,font=\small\sffamily]{Poly($n$)};
    \draw [thick,decorate,decoration={calligraphic brace,amplitude=10pt}] (7.62, 4.092) -- (7.62, 0.847) node[midway,right,font=\small\sffamily]{\ \ \ $n$};
    \draw [thick,decorate,decoration={calligraphic brace,amplitude=10pt}] (0, 0.847) -- (0, 4.092) node[midway,left,font=\small\sffamily]{$n$\ \ \ \ };
\end{tikzpicture}

%% file: Figures/toric3x3.tikz
\begin{tikzpicture}[scale=1]
    \node[Mark, Mark_disk] at (1.049, 1.138) {};
    \node[Mark, Mark_disk] at (2.178, 1.138) {};
    \node[Mark, Mark_disk] at (2.742, 1.703) {};
    \node[Mark, Mark_disk] at (1.613, 1.703) {};
    \node[Mark, Mark_disk] at (0.484, 1.703) {};
    \node[Mark, Mark_disk] at (1.049, 2.267) {};
    \node[Mark, Mark_disk] at (0.484, 2.832) {};
    \node[Mark, Mark_disk] at (1.049, 3.396) {};
    \node[Mark, Mark_disk] at (2.178, 3.396) {};
    \node[Mark, Mark_disk] at (1.613, 2.832) {};
    \node[Mark, Mark_disk] at (2.178, 2.267) {};
    \node[Mark, Mark_disk] at (2.742, 2.832) {};
    \draw[shift={(0.484, 1.138)}, yscale=-1] (0, 0) rectangle (2.258, -2.258);
    \draw[shift={(0.484, 2.267)}, yscale=-1] (0, 0) -- (2.258, 0);
    \draw[shift={(1.613, 1.138)}, yscale=-1] (0, 0) -- (0, -2.258);
    \draw[shift={(0.484, 1.138)}, yscale=-1] (0, 0) -- (0, 1.129) -- (1.129, 1.129) -- (1.129, 0);
    \draw[shift={(1.613, 0.009)}, yscale=-1] (0, 0) -- (1.129, 0) -- (1.129, -1.129);
    \draw[shift={(2.742, 0.009)}, yscale=-1] (0, 0) -- (1.129, 0) -- (1.129, -1.129) -- (0, -1.129);
    \draw[shift={(3.871, 1.138)}, yscale=-1] (0, 0) -- (0, -1.129) -- (-1.129, -1.129);
    \draw[shift={(3.871, 2.267)}, yscale=-1] (0, 0) -- (0, -1.129) -- (-1.129, -1.129);
    \node[Mark, Mark_disk] at (0.484, 0.574) {};
    \node[Mark, Mark_disk, black!50] at (1.049, 0.009) {};
    \node[Mark, Mark_disk] at (1.613, 0.574) {};
    \node[Mark, Mark_disk, black!50] at (2.178, 0.009) {};
    \node[Mark, Mark_disk] at (2.742, 0.574) {};
    \node[Mark, Mark_disk, black!50] at (3.307, 0.009) {};
    \node[Mark, Mark_disk, black!50] at (3.871, 0.574) {};
    \node[Mark, Mark_disk, black!50] at (3.871, 1.703) {};
    \node[Mark, Mark_disk] at (3.307, 2.267) {};
    \node[Mark, Mark_disk, black!50] at (3.871, 2.832) {};
    \node[Mark, Mark_disk] at (3.307, 3.396) {};
    \node[Mark, Mark_disk] at (3.307, 1.138) {};
    \node[anchor=center] at (1.049, 3.678) {$1$};
    \node[anchor=center] at (2.178, 3.678) {$2$};
    \node[anchor=center] at (3.307, 3.678) {$3$};
    \node[anchor=center] at (0.202, 2.832) {$4$};
    \node[anchor=center] at (1.331, 2.832) {$5$};
    \node[anchor=center] at (2.46, 2.832) {$6$};
    \node[anchor=center] at (1.049, 2.055) {$7$};
    \node[anchor=center] at (2.178, 2.055) {$8$};
    \node[anchor=center] at (3.307, 2.055) {$9$};
    \node[anchor=center] at (0.202, 1.703) {$10$};
    \node[anchor=center] at (1.331, 1.703) {$11$};
    \node[anchor=center] at (2.46, 1.703) {$12$};
    \node[anchor=center] at (1.049, 0.927) {$13$};
    \node[anchor=center] at (2.178, 0.927) {$14$};
    \node[anchor=center] at (3.307, 0.927) {$15$};
    \node[anchor=center] at (0.202, 0.574) {$16$};
    \node[anchor=center] at (1.331, 0.574) {$17$};
    \node[anchor=center] at (2.46, 0.574) {$18$};
    \fill[shift={(0.847, 0.049)}, xscale=0.831, yscale=-0.883, RoyalBlue!80, opacity=0.5] (0, 0) -- (0, -0.073) -- (0.763, -0.073) -- (0.763, -1.127) -- (-0.291, -1.127) -- (-0.291, -0.073) -- (0, -0.073) -- (0, 0) -- (-0.364, 0) -- (-0.364, -1.199) -- (0.835, -1.199) -- (0.835, 0) -- (0, 0);
    \fill[shift={(1.976, 0.049)}, xscale=0.831, yscale=-0.883, RoyalBlue!80, opacity=0.5] (0, 0) -- (0, -0.073) -- (0.763, -0.073) -- (0.763, -1.127) -- (-0.291, -1.127) -- (-0.291, -0.073) -- (0, -0.073) -- (0, 0) -- (-0.364, 0) -- (-0.364, -1.199) -- (0.835, -1.199) -- (0.835, 0) -- (0, 0);
    \fill[shift={(3.105, 0.049)}, xscale=0.831, yscale=-0.883, RoyalBlue!80, opacity=0.5] (0, 0) -- (0, -0.073) -- (0.763, -0.073) -- (0.763, -1.127) -- (-0.291, -1.127) -- (-0.291, -0.073) -- (0, -0.073) -- (0, 0) -- (-0.364, 0) -- (-0.364, -1.199) -- (0.835, -1.199) -- (0.835, 0) -- (0, 0);
    \fill[shift={(0.847, 1.178)}, xscale=0.831, yscale=-0.883, RoyalBlue!80, opacity=0.5] (0, 0) -- (0, -0.073) -- (0.763, -0.073) -- (0.763, -1.127) -- (-0.291, -1.127) -- (-0.291, -0.073) -- (0, -0.073) -- (0, 0) -- (-0.364, 0) -- (-0.364, -1.199) -- (0.835, -1.199) -- (0.835, 0) -- (0, 0);
    \fill[shift={(1.976, 1.178)}, xscale=0.831, yscale=-0.883, RoyalBlue!80, opacity=0.5] (0, 0) -- (0, -0.073) -- (0.763, -0.073) -- (0.763, -1.127) -- (-0.291, -1.127) -- (-0.291, -0.073) -- (0, -0.073) -- (0, 0) -- (-0.364, 0) -- (-0.364, -1.199) -- (0.835, -1.199) -- (0.835, 0) -- (0, 0);
    \fill[shift={(3.105, 1.178)}, xscale=0.831, yscale=-0.883, RoyalBlue!80, opacity=0.5] (0, 0) -- (0, -0.073) -- (0.763, -0.073) -- (0.763, -1.127) -- (-0.291, -1.127) -- (-0.291, -0.073) -- (0, -0.073) -- (0, 0) -- (-0.364, 0) -- (-0.364, -1.199) -- (0.835, -1.199) -- (0.835, 0) -- (0, 0);
    \fill[shift={(0.847, 2.307)}, xscale=0.831, yscale=-0.883, RoyalBlue!80, opacity=0.5] (0, 0) -- (0, -0.073) -- (0.763, -0.073) -- (0.763, -1.127) -- (-0.291, -1.127) -- (-0.291, -0.073) -- (0, -0.073) -- (0, 0) -- (-0.364, 0) -- (-0.364, -1.199) -- (0.835, -1.199) -- (0.835, 0) -- (0, 0);
    \fill[shift={(1.976, 2.307)}, xscale=0.831, yscale=-0.883, RoyalBlue!80, opacity=0.5] (0, 0) -- (0, -0.073) -- (0.763, -0.073) -- (0.763, -1.127) -- (-0.291, -1.127) -- (-0.291, -0.073) -- (0, -0.073) -- (0, 0) -- (-0.364, 0) -- (-0.364, -1.199) -- (0.835, -1.199) -- (0.835, 0) -- (0, 0);
    \fill[shift={(3.105, 2.307)}, xscale=0.831, yscale=-0.883, RoyalBlue!80, opacity=0.5] (0, 0) -- (0, -0.073) -- (0.763, -0.073) -- (0.763, -1.127) -- (-0.291, -1.127) -- (-0.291, -0.073) -- (0, -0.073) -- (0, 0) -- (-0.364, 0) -- (-0.364, -1.199) -- (0.835, -1.199) -- (0.835, 0) -- (0, 0);
    \fill[shift={(1.543, 1.076)}, yscale=-1, Lavender, opacity=0.5] (0, 0) -- (-0.414, 0) -- (-0.414, -0.13) -- (0, -0.13) -- (0, -0.543) -- (0.13, -0.543) -- (0.13, -0.13) -- (0.543, -0.13) -- (0.543, 0) -- (0.13, 0) -- (0.13, 0.414) -- (0, 0.414) -- cycle;
    \fill[shift={(0.414, 1.076)}, yscale=-1, Lavender, opacity=0.5] (0, 0) -- (-0.414, 0) -- (-0.414, -0.13) -- (0, -0.13) -- (0, -0.543) -- (0.13, -0.543) -- (0.13, -0.13) -- (0.543, -0.13) -- (0.543, 0) -- (0.13, 0) -- (0.13, 0.414) -- (0, 0.414) -- cycle;
    \fill[shift={(2.671, 1.076)}, yscale=-1, Lavender, opacity=0.5] (0, 0) -- (-0.414, 0) -- (-0.414, -0.13) -- (0, -0.13) -- (0, -0.543) -- (0.13, -0.543) -- (0.13, -0.13) -- (0.543, -0.13) -- (0.543, 0) -- (0.13, 0) -- (0.13, 0.414) -- (0, 0.414) -- cycle;
    \fill[shift={(1.543, 2.205)}, yscale=-1, Lavender, opacity=0.5] (0, 0) -- (-0.414, 0) -- (-0.414, -0.13) -- (0, -0.13) -- (0, -0.543) -- (0.13, -0.543) -- (0.13, -0.13) -- (0.543, -0.13) -- (0.543, 0) -- (0.13, 0) -- (0.13, 0.414) -- (0, 0.414) -- cycle;
    \fill[shift={(0.414, 2.205)}, yscale=-1, Lavender, opacity=0.5] (0, 0) -- (-0.414, 0) -- (-0.414, -0.13) -- (0, -0.13) -- (0, -0.543) -- (0.13, -0.543) -- (0.13, -0.13) -- (0.543, -0.13) -- (0.543, 0) -- (0.13, 0) -- (0.13, 0.414) -- (0, 0.414) -- cycle;
    \fill[shift={(2.671, 2.205)}, yscale=-1, Lavender, opacity=0.5] (0, 0) -- (-0.414, 0) -- (-0.414, -0.13) -- (0, -0.13) -- (0, -0.543) -- (0.13, -0.543) -- (0.13, -0.13) -- (0.543, -0.13) -- (0.543, 0) -- (0.13, 0) -- (0.13, 0.414) -- (0, 0.414) -- cycle;
    \fill[shift={(1.543, 3.334)}, yscale=-1, Lavender, opacity=0.5] (0, 0) -- (-0.414, 0) -- (-0.414, -0.13) -- (0, -0.13) -- (0, -0.543) -- (0.13, -0.543) -- (0.13, -0.13) -- (0.543, -0.13) -- (0.543, 0) -- (0.13, 0) -- (0.13, 0.414) -- (0, 0.414) -- cycle;
    \fill[shift={(0.414, 3.334)}, yscale=-1, Lavender, opacity=0.5] (0, 0) -- (-0.414, 0) -- (-0.414, -0.13) -- (0, -0.13) -- (0, -0.543) -- (0.13, -0.543) -- (0.13, -0.13) -- (0.543, -0.13) -- (0.543, 0) -- (0.13, 0) -- (0.13, 0.414) -- (0, 0.414) -- cycle;
    \fill[shift={(2.671, 3.334)}, yscale=-1, Lavender, opacity=0.5] (0, 0) -- (-0.414, 0) -- (-0.414, -0.13) -- (0, -0.13) -- (0, -0.543) -- (0.13, -0.543) -- (0.13, -0.13) -- (0.543, -0.13) -- (0.543, 0) -- (0.13, 0) -- (0.13, 0.414) -- (0, 0.414) -- cycle;
    \fill[shift={(5.353, 2.593)}, yscale=-1, Lavender, opacity=0.5] (0, 0) -- (-0.414, 0) -- (-0.414, -0.13) -- (0, -0.13) -- (0, -0.543) -- (0.13, -0.543) -- (0.13, -0.13) -- (0.543, -0.13) -- (0.543, 0) -- (0.13, 0) -- (0.13, 0.414) -- (0, 0.414) -- cycle;
    \fill[shift={(5.257, 0.261)}, xscale=0.831, yscale=-0.883, RoyalBlue!80, opacity=0.5] (0, 0) -- (0, -0.073) -- (0.763, -0.073) -- (0.763, -1.127) -- (-0.291, -1.127) -- (-0.291, -0.073) -- (0, -0.073) -- (0, 0) -- (-0.364, 0) -- (-0.364, -1.199) -- (0.835, -1.199) -- (0.835, 0) -- (0, 0);
    \node[anchor=center] at (5.459, 1.443) {$\sigma_z$};
    \node[anchor=center] at (6.164, 0.808) {$\sigma_z$};
    \node[anchor=center] at (5.459, 0.102) {$\sigma_z$};
    \node[anchor=center] at (4.753, 0.808) {$\sigma_z$};
    \node[anchor=center] at (5.423, 3.277) {$\sigma_x$};
    \node[anchor=center] at (6.129, 2.642) {$\sigma_x$};
    \node[anchor=center] at (4.753, 2.642) {$\sigma_x$};
    \node[anchor=center] at (5.423, 2.007) {$\sigma_x$};
\end{tikzpicture}

%% file: Figures/toricIsing.tikz
\begin{tikzpicture}[scale=1]
    \node[Mark, Mark_disk] at (1.128, 1.129) {};
    \node[Mark, Mark_disk] at (2.257, 1.129) {};
    \node[Mark, Mark_disk] at (2.822, 1.693) {};
    \node[Mark, Mark_disk] at (1.693, 1.693) {};
    \node[Mark, Mark_disk] at (0.564, 1.693) {};
    \node[Mark, Mark_disk] at (1.128, 2.258) {};
    \node[Mark, Mark_disk] at (0.564, 2.822) {};
    \node[Mark, Mark_disk] at (1.128, 3.387) {};
    \node[Mark, Mark_disk] at (2.257, 3.387) {};
    \node[Mark, Mark_disk] at (1.693, 2.822) {};
    \node[Mark, Mark_disk] at (2.257, 2.258) {};
    \node[Mark, Mark_disk] at (2.822, 2.822) {};
    \draw[shift={(0.564, 1.129)}, yscale=-1] (0, 0) rectangle (2.258, -2.258);
    \draw[shift={(0.564, 2.258)}, yscale=-1] (0, 0) -- (2.258, 0);
    \draw[shift={(1.693, 1.129)}, yscale=-1] (0, 0) -- (0, -2.258);
    \draw[shift={(0.564, 1.129)}, yscale=-1] (0, 0) -- (0, 1.129) -- (1.129, 1.129) -- (1.129, 0);
    \draw[shift={(1.693, 0)}, yscale=-1] (0, 0) -- (1.129, 0) -- (1.129, -1.129);
    \draw[shift={(2.822, 0)}, yscale=-1] (0, 0) -- (1.129, 0) -- (1.129, -1.129) -- (0, -1.129);
    \draw[shift={(3.95, 1.129)}, yscale=-1] (0, 0) -- (0, -1.129) -- (-1.129, -1.129);
    \draw[shift={(3.95, 2.258)}, yscale=-1] (0, 0) -- (0, -1.129) -- (-1.129, -1.129);
    \node[Mark, Mark_disk] at (0.564, 0.564) {};
    \node[Mark, Mark_disk] at (1.693, 0.564) {};
    \node[Mark, Mark_disk] at (2.822, 0.564) {};
    \node[Mark, Mark_disk] at (3.386, 2.258) {};
    \node[Mark, Mark_disk] at (3.386, 3.387) {};
    \node[Mark, Mark_disk] at (3.386, 1.129) {};
    \node[anchor=center] at (1.128, 3.739) {$1$};
    \node[anchor=center] at (2.257, 3.739) {$2$};
    \node[anchor=center] at (3.386, 3.739) {$3$};
    \node[anchor=center] at (0.246, 2.822) {$4$};
    \node[anchor=center] at (1.41, 2.822) {$5$};
    \node[anchor=center] at (2.539, 2.822) {$6$};
    \node[anchor=center] at (1.128, 2.046) {$7$};
    \node[anchor=center] at (2.257, 2.046) {$8$};
    \node[anchor=center] at (3.386, 2.011) {$9$};
    \node[anchor=center] at (0.211, 1.693) {$10$};
    \node[anchor=center] at (1.41, 1.693) {$11$};
    \node[anchor=center] at (2.539, 1.693) {$12$};
    \node[anchor=center] at (0.846, 0.988) {$13$};
    \node[anchor=center] at (2.61, 0.953) {$14$};
    \node[anchor=center] at (3.386, 0.917) {$15$};
    \node[anchor=center] at (0.176, 0.564) {$16$};
    \node[anchor=center] at (1.481, 0.212) {$17$};
    \node[anchor=center] at (2.539, 0.564) {$18$};
    \filldraw[fill=RoyalBlue!80, fill opacity=0.3] 
    (1.693, 0.564) ellipse[x radius=0.2, y radius=-0.2];
    \filldraw[fill=RoyalBlue!80, fill opacity=0.3] 
    (3.386, 2.258) ellipse[x radius=0.2, y radius=-0.2];
    \filldraw[fill=Lavender, fill opacity=0.3] 
    (1.128, 3.387) ellipse[x radius=0.2, y radius=-0.2];
    \filldraw[fill=Lavender, fill opacity=0.3] 
    (2.822, 1.693) ellipse[x radius=0.2, y radius=-0.2];
    \node[Mark, Mark_disk, black!50] at (1.128, 0) {};
    \node[Mark, Mark_disk, black!50] at (2.257, 0) {};
    \node[Mark, Mark_disk, black!50] at (3.386, 0) {};
    \node[Mark, Mark_disk, black!50] at (3.95, 0.564) {};
    \node[Mark, Mark_disk, black!50] at (3.95, 1.693) {};
    \node[Mark, Mark_disk, black!50] at (3.95, 2.822) {};
    \filldraw[shift={(1.693, 0.431)}, rotate=90, fill=RoyalBlue!80, fill opacity=0.3] (0, 0) .. controls (0, 0.094) and (0.047, 0.141) .. (0.141, 0.141) .. controls (0.235, 0.141) and (0.282, 0.094) .. (0.282, 0) -- (0.282, -1.129) .. controls (0.282, -1.223) and (0.235, -1.27) .. (0.141, -1.27) .. controls (0.047, -1.27) and (0, -1.223) .. (0, -1.129) -- cycle;
    \filldraw[shift={(1.128, 0.995)}, rotate=90, fill=RoyalBlue!80, fill opacity=0.3] (0, 0) .. controls (0, 0.094) and (0.047, 0.141) .. (0.141, 0.141) .. controls (0.235, 0.141) and (0.282, 0.094) .. (0.282, 0) -- (0.282, -1.129) .. controls (0.282, -1.223) and (0.235, -1.27) .. (0.141, -1.27) .. controls (0.047, -1.27) and (0, -1.223) .. (0, -1.129) -- cycle;
    \filldraw[shift={(1.128, 2.124)}, rotate=90, fill=RoyalBlue!80, fill opacity=0.3] (0, 0) .. controls (0, 0.094) and (0.047, 0.141) .. (0.141, 0.141) .. controls (0.235, 0.141) and (0.282, 0.094) .. (0.282, 0) -- (0.282, -1.129) .. controls (0.282, -1.223) and (0.235, -1.27) .. (0.141, -1.27) .. controls (0.047, -1.27) and (0, -1.223) .. (0, -1.129) -- cycle;
    \filldraw[shift={(2.257, 2.124)}, rotate=90, fill=RoyalBlue!80, fill opacity=0.3] (0, 0) .. controls (0, 0.094) and (0.047, 0.141) .. (0.141, 0.141) .. controls (0.235, 0.141) and (0.282, 0.094) .. (0.282, 0) -- (0.282, -1.129) .. controls (0.282, -1.223) and (0.235, -1.27) .. (0.141, -1.27) .. controls (0.047, -1.27) and (0, -1.223) .. (0, -1.129) -- cycle;
    \filldraw[shift={(1.27, 1.136)}, scale=-1, fill=RoyalBlue!80, fill opacity=0.3] (0, 0) .. controls (0, 0.094) and (0.047, 0.141) .. (0.141, 0.141) .. controls (0.235, 0.141) and (0.282, 0.094) .. (0.282, 0) -- (0.282, -1.129) .. controls (0.282, -1.223) and (0.235, -1.27) .. (0.141, -1.27) .. controls (0.047, -1.27) and (0, -1.223) .. (0, -1.129) -- cycle;
    \filldraw[shift={(1.128, 3.253)}, rotate=90, fill=Lavender, fill opacity=0.3] (0, 0) .. controls (0, 0.094) and (0.047, 0.141) .. (0.141, 0.141) .. controls (0.235, 0.141) and (0.282, 0.094) .. (0.282, 0) -- (0.282, -1.129) .. controls (0.282, -1.223) and (0.235, -1.27) .. (0.141, -1.27) .. controls (0.047, -1.27) and (0, -1.223) .. (0, -1.129) -- cycle;
    \filldraw[shift={(0.564, 2.688)}, rotate=90, fill=Lavender, fill opacity=0.3] (0, 0) .. controls (0, 0.094) and (0.047, 0.141) .. (0.141, 0.141) .. controls (0.235, 0.141) and (0.282, 0.094) .. (0.282, 0) -- (0.282, -1.129) .. controls (0.282, -1.223) and (0.235, -1.27) .. (0.141, -1.27) .. controls (0.047, -1.27) and (0, -1.223) .. (0, -1.129) -- cycle;
    \filldraw[shift={(1.693, 2.688)}, rotate=90, fill=Lavender, fill opacity=0.3] (0, 0) .. controls (0, 0.094) and (0.047, 0.141) .. (0.141, 0.141) .. controls (0.235, 0.141) and (0.282, 0.094) .. (0.282, 0) -- (0.282, -1.129) .. controls (0.282, -1.223) and (0.235, -1.27) .. (0.141, -1.27) .. controls (0.047, -1.27) and (0, -1.223) .. (0, -1.129) -- cycle;
    \filldraw[shift={(0.564, 1.559)}, rotate=90, fill=Lavender, fill opacity=0.3] (0, 0) .. controls (0, 0.094) and (0.047, 0.141) .. (0.141, 0.141) .. controls (0.235, 0.141) and (0.282, 0.094) .. (0.282, 0) -- (0.282, -1.129) .. controls (0.282, -1.223) and (0.235, -1.27) .. (0.141, -1.27) .. controls (0.047, -1.27) and (0, -1.223) .. (0, -1.129) -- cycle;
    \filldraw[shift={(1.693, 1.559)}, rotate=90, fill=Lavender, fill opacity=0.3] (0, 0) .. controls (0, 0.094) and (0.047, 0.141) .. (0.141, 0.141) .. controls (0.235, 0.141) and (0.282, 0.094) .. (0.282, 0) -- (0.282, -1.129) .. controls (0.282, -1.223) and (0.235, -1.27) .. (0.141, -1.27) .. controls (0.047, -1.27) and (0, -1.223) .. (0, -1.129) -- cycle;
    \filldraw[shift={(0.705, 1.701)}, scale=-1, fill=Lavender, fill opacity=0.3] (0, 0) .. controls (0, 0.094) and (0.047, 0.141) .. (0.141, 0.141) .. controls (0.235, 0.141) and (0.282, 0.094) .. (0.282, 0) -- (0.282, -1.129) .. controls (0.282, -1.223) and (0.235, -1.27) .. (0.141, -1.27) .. controls (0.047, -1.27) and (0, -1.223) .. (0, -1.129) -- cycle;
    \filldraw[shift={(2.078, 3.391)}, rotate=-135, fill=Lavender, fill opacity=0.3] (0, 0) .. controls (0, -0.094) and (-0.047, -0.141) .. (-0.141, -0.141) .. controls (-0.235, -0.141) and (-0.282, -0.094) .. (-0.282, 0) -- (-0.282, 1.023) .. controls (-0.282, 1.117) and (-0.235, 1.165) .. (-0.141, 1.165) .. controls (-0.047, 1.165) and (0, 1.117) .. (0, 1.023) -- cycle;
    \filldraw[shift={(2.257, 0.995)}, rotate=90, fill=RoyalBlue!80, fill opacity=0.3] (0, 0) .. controls (0, 0.094) and (0.047, 0.141) .. (0.141, 0.141) .. controls (0.235, 0.141) and (0.282, 0.094) .. (0.282, 0) -- (0.282, -1.129) .. controls (0.282, -1.223) and (0.235, -1.27) .. (0.141, -1.27) .. controls (0.047, -1.27) and (0, -1.223) .. (0, -1.129) -- cycle;
    \filldraw[shift={(2.642, 0.574)}, rotate=135, yscale=-1, fill=RoyalBlue!80, fill opacity=0.3] (0, 0) .. controls (0, -0.094) and (-0.047, -0.141) .. (-0.141, -0.141) .. controls (-0.235, -0.141) and (-0.282, -0.094) .. (-0.282, 0) -- (-0.282, 1.023) .. controls (-0.282, 1.117) and (-0.235, 1.165) .. (-0.141, 1.165) .. controls (-0.047, 1.165) and (0, 1.117) .. (0, 1.023) -- cycle;
\end{tikzpicture}

%% file: Figures/1Dising.tikz
\begin{tikzpicture}[scale=1]
    \node[Mark, Mark_disk] at (0.141, 4.657) {};
    \node[Mark, Mark_disk] at (1.27, 4.657) {};
    \node[Mark, Mark_disk] at (2.399, 4.657) {};
    \node[Mark, Mark_disk] at (3.528, 4.657) {};
    \node[Mark, Mark_disk] at (4.656, 4.657) {};
    \node[Mark, Mark_disk] at (0.141, 0.141) {};
    \node[Mark, Mark_disk] at (1.27, 0.141) {};
    \node[Mark, Mark_disk] at (2.399, 0.141) {};
    \node[Mark, Mark_disk] at (3.528, 0.141) {};
    \node[Mark, Mark_disk] at (4.656, 0.141) {};
    \filldraw[shift={(0.141, 4.798)}, rotate=-90, yscale=-1, fill=RoyalBlue!80, fill opacity=0.3] (0, 0) .. controls (0, 0.094) and (0.047, 0.141) .. (0.141, 0.141) .. controls (0.235, 0.141) and (0.282, 0.094) .. (0.282, 0) -- (0.282, -1.129) .. controls (0.282, -1.223) and (0.235, -1.27) .. (0.141, -1.27) .. controls (0.047, -1.27) and (0, -1.223) .. (0, -1.129) -- cycle;
    \filldraw[shift={(1.27, 4.798)}, rotate=-90, yscale=-1, fill=RoyalBlue!80, fill opacity=0.3] (0, 0) .. controls (0, 0.094) and (0.047, 0.141) .. (0.141, 0.141) .. controls (0.235, 0.141) and (0.282, 0.094) .. (0.282, 0) -- (0.282, -1.129) .. controls (0.282, -1.223) and (0.235, -1.27) .. (0.141, -1.27) .. controls (0.047, -1.27) and (0, -1.223) .. (0, -1.129) -- cycle;
    \filldraw[shift={(2.399, 4.798)}, rotate=-90, yscale=-1, fill=RoyalBlue!80, fill opacity=0.3] (0, 0) .. controls (0, 0.094) and (0.047, 0.141) .. (0.141, 0.141) .. controls (0.235, 0.141) and (0.282, 0.094) .. (0.282, 0) -- (0.282, -1.129) .. controls (0.282, -1.223) and (0.235, -1.27) .. (0.141, -1.27) .. controls (0.047, -1.27) and (0, -1.223) .. (0, -1.129) -- cycle;
    \filldraw[shift={(3.528, 4.798)}, rotate=-90, yscale=-1, fill=RoyalBlue!80, fill opacity=0.3] (0, 0) .. controls (0, 0.094) and (0.047, 0.141) .. (0.141, 0.141) .. controls (0.235, 0.141) and (0.282, 0.094) .. (0.282, 0) -- (0.282, -1.129) .. controls (0.282, -1.223) and (0.235, -1.27) .. (0.141, -1.27) .. controls (0.047, -1.27) and (0, -1.223) .. (0, -1.129) -- cycle;
    \filldraw[fill=RoyalBlue!80, fill opacity=0.3] 
    (1.27, 0.141) ellipse[x radius=0.141, y radius=-0.141];
    \filldraw[fill=RoyalBlue!80, fill opacity=0.3] 
    (2.399, 0.141) ellipse[x radius=0.141, y radius=-0.141];
    \filldraw[fill=RoyalBlue!80, fill opacity=0.3] 
    (3.528, 0.141) ellipse[x radius=0.141, y radius=-0.141];
    \filldraw[fill=RoyalBlue!80, fill opacity=0.3] 
    (4.656, 0.141) ellipse[x radius=0.141, y radius=-0.141];
    \node[Mark, Mark_disk] at (0.141, 1.27) {};
    \node[Mark, Mark_disk] at (1.27, 1.27) {};
    \node[Mark, Mark_disk] at (2.399, 1.27) {};
    \node[Mark, Mark_disk] at (3.528, 1.27) {};
    \node[Mark, Mark_disk] at (4.656, 1.27) {};
    \filldraw[fill=RoyalBlue!80, fill opacity=0.3] 
    (2.399, 1.27) ellipse[x radius=0.141, y radius=-0.141];
    \filldraw[fill=RoyalBlue!80, fill opacity=0.3] 
    (3.528, 1.27) ellipse[x radius=0.141, y radius=-0.141];
    \filldraw[fill=RoyalBlue!80, fill opacity=0.3] 
    (4.656, 1.27) ellipse[x radius=0.141, y radius=-0.141];
    \filldraw[shift={(0.141, 1.411)}, rotate=-90, yscale=-1, fill=RoyalBlue!80, fill opacity=0.3] (0, 0) .. controls (0, 0.094) and (0.047, 0.141) .. (0.141, 0.141) .. controls (0.235, 0.141) and (0.282, 0.094) .. (0.282, 0) -- (0.282, -1.129) .. controls (0.282, -1.223) and (0.235, -1.27) .. (0.141, -1.27) .. controls (0.047, -1.27) and (0, -1.223) .. (0, -1.129) -- cycle;
    \node[Mark, Mark_disk] at (0.141, 2.399) {};
    \node[Mark, Mark_disk] at (1.27, 2.399) {};
    \node[Mark, Mark_disk] at (2.399, 2.399) {};
    \node[Mark, Mark_disk] at (3.528, 2.399) {};
    \node[Mark, Mark_disk] at (4.656, 2.399) {};
    \filldraw[fill=RoyalBlue!80, fill opacity=0.3] 
    (3.528, 2.399) ellipse[x radius=0.141, y radius=-0.141];
    \filldraw[fill=RoyalBlue!80, fill opacity=0.3] 
    (4.656, 2.399) ellipse[x radius=0.141, y radius=-0.141];
    \filldraw[shift={(0.141, 2.54)}, rotate=-90, yscale=-1, fill=RoyalBlue!80, fill opacity=0.3] (0, 0) .. controls (0, 0.094) and (0.047, 0.141) .. (0.141, 0.141) .. controls (0.235, 0.141) and (0.282, 0.094) .. (0.282, 0) -- (0.282, -1.129) .. controls (0.282, -1.223) and (0.235, -1.27) .. (0.141, -1.27) .. controls (0.047, -1.27) and (0, -1.223) .. (0, -1.129) -- cycle;
    \filldraw[shift={(1.27, 2.54)}, rotate=-90, yscale=-1, fill=RoyalBlue!80, fill opacity=0.3] (0, 0) .. controls (0, 0.094) and (0.047, 0.141) .. (0.141, 0.141) .. controls (0.235, 0.141) and (0.282, 0.094) .. (0.282, 0) -- (0.282, -1.129) .. controls (0.282, -1.223) and (0.235, -1.27) .. (0.141, -1.27) .. controls (0.047, -1.27) and (0, -1.223) .. (0, -1.129) -- cycle;
    \node[Mark, Mark_disk] at (0.141, 3.528) {};
    \node[Mark, Mark_disk] at (1.27, 3.528) {};
    \node[Mark, Mark_disk] at (2.399, 3.528) {};
    \node[Mark, Mark_disk] at (3.528, 3.528) {};
    \node[Mark, Mark_disk] at (4.656, 3.528) {};
    \filldraw[fill=RoyalBlue!80, fill opacity=0.3] 
    (4.656, 3.528) ellipse[x radius=0.141, y radius=-0.141];
    \filldraw[shift={(0.141, 3.669)}, rotate=-90, yscale=-1, fill=RoyalBlue!80, fill opacity=0.3] (0, 0) .. controls (0, 0.094) and (0.047, 0.141) .. (0.141, 0.141) .. controls (0.235, 0.141) and (0.282, 0.094) .. (0.282, 0) -- (0.282, -1.129) .. controls (0.282, -1.223) and (0.235, -1.27) .. (0.141, -1.27) .. controls (0.047, -1.27) and (0, -1.223) .. (0, -1.129) -- cycle;
    \filldraw[shift={(1.27, 3.669)}, rotate=-90, yscale=-1, fill=RoyalBlue!80, fill opacity=0.3] (0, 0) .. controls (0, 0.094) and (0.047, 0.141) .. (0.141, 0.141) .. controls (0.235, 0.141) and (0.282, 0.094) .. (0.282, 0) -- (0.282, -1.129) .. controls (0.282, -1.223) and (0.235, -1.27) .. (0.141, -1.27) .. controls (0.047, -1.27) and (0, -1.223) .. (0, -1.129) -- cycle;
    \filldraw[shift={(2.399, 3.669)}, rotate=-90, yscale=-1, fill=RoyalBlue!80, fill opacity=0.3] (0, 0) .. controls (0, 0.094) and (0.047, 0.141) .. (0.141, 0.141) .. controls (0.235, 0.141) and (0.282, 0.094) .. (0.282, 0) -- (0.282, -1.129) .. controls (0.282, -1.223) and (0.235, -1.27) .. (0.141, -1.27) .. controls (0.047, -1.27) and (0, -1.223) .. (0, -1.129) -- cycle;
    \node[anchor=center] at (0.141, 5.045) {$1$};
    \node[anchor=center] at (1.27, 5.045) {$2$};
    \node[anchor=center] at (2.399, 5.045) {$3$};
    \node[anchor=center] at (3.528, 5.045) {$4$};
    \node[anchor=center] at (4.656, 5.045) {$5$};
    \draw[->] (2.399, 4.375) -- (2.399, 3.81);
    \draw[->] (2.399, 3.246) -- (2.399, 2.681);
    \draw[->] (2.399, 2.117) -- (2.399, 1.552);
    \draw[->] (2.399, 0.988) -- (2.399, 0.424);
    \node[anchor=center] at (3.245, 4.092) {$CX(4, 5)$};
    \node[anchor=center] at (3.245, 2.964) {$CX(3, 4)$};
    \node[anchor=center] at (3.245, 1.835) {$CX(2, 3)$};
    \node[anchor=center] at (3.245, 0.706) {$CX(1, 2)$};
\end{tikzpicture}

%% file: Figures/latticetoric2D.tikz
\begin{tikzpicture}[scale=1.2]
    \node[Mark, Mark_disk] at (1.199, 1.764) {};
    \node[Mark, Mark_disk] at (0.071, 1.764) {};
    \node[Mark, Mark_disk] at (0.635, 1.199) {};
    \node[Mark, Mark_disk] at (0.071, 0.635) {};
    \node[Mark, Mark_disk] at (0.635, 0.071) {};
    \node[Mark, Mark_disk] at (1.764, 0.071) {};
    \node[Mark, Mark_disk] at (1.199, 0.635) {};
    \node[Mark, Mark_disk] at (1.764, 1.199) {};
    \draw (0.071, 2.328) rectangle (2.328, 0.071);
    \draw (0.071, 1.199) -- (2.328, 1.199);
    \draw (1.199, 2.328) -- (1.199, 0.071);
    \draw[shift={(0.071, 1.543)}, rotate=60] (0, 0) -- (-0.141, 0);
    \draw[shift={(0.141, 1.421)}, rotate=-60, yscale=-1] (0, 0) -- (-0.141, 0);
    \draw[shift={(2.328, 1.543)}, rotate=60] (0, 0) -- (-0.141, 0);
    \draw[shift={(2.399, 1.421)}, rotate=-60, yscale=-1] (0, 0) -- (-0.141, 0);
    \draw[shift={(0.071, 0.414)}, rotate=60] (0, 0) -- (-0.141, 0);
    \draw[shift={(0.141, 0.292)}, rotate=-60, yscale=-1] (0, 0) -- (-0.141, 0);
    \draw[shift={(2.328, 0.414)}, rotate=60] (0, 0) -- (-0.141, 0);
    \draw[shift={(2.399, 0.292)}, rotate=-60, yscale=-1] (0, 0) -- (-0.141, 0);
    \draw[shift={(1.543, 0.071)}, rotate=-30] (0, 0) -- (-0.141, 0);
    \draw[shift={(1.421, 0)}, rotate=-150, yscale=-1] (0, 0) -- (-0.141, 0);
    \draw[shift={(1.472, 0.071)}, rotate=-30] (0, 0) -- (-0.141, 0);
    \draw[shift={(1.35, 0)}, rotate=-150, yscale=-1] (0, 0) -- (-0.141, 0);
    \draw[shift={(0.414, 0.071)}, rotate=-30] (0, 0) -- (-0.141, 0);
    \draw[shift={(0.292, 0)}, rotate=-150, yscale=-1] (0, 0) -- (-0.141, 0);
    \draw[shift={(0.343, 0.071)}, rotate=-30] (0, 0) -- (-0.141, 0);
    \draw[shift={(0.221, 0)}, rotate=-150, yscale=-1] (0, 0) -- (-0.141, 0);
    \draw[shift={(0.414, 2.328)}, rotate=-30] (0, 0) -- (-0.141, 0);
    \draw[shift={(0.292, 2.258)}, rotate=-150, yscale=-1] (0, 0) -- (-0.141, 0);
    \draw[shift={(0.343, 2.328)}, rotate=-30] (0, 0) -- (-0.141, 0);
    \draw[shift={(0.221, 2.258)}, rotate=-150, yscale=-1] (0, 0) -- (-0.141, 0);
    \draw[shift={(1.543, 2.328)}, rotate=-30] (0, 0) -- (-0.141, 0);
    \draw[shift={(1.421, 2.258)}, rotate=-150, yscale=-1] (0, 0) -- (-0.141, 0);
    \draw[shift={(1.472, 2.328)}, rotate=-30] (0, 0) -- (-0.141, 0);
    \draw[shift={(1.35, 2.258)}, rotate=-150, yscale=-1] (0, 0) -- (-0.141, 0);
    \node[Mark, Mark_disk, black!50] at (0.635, 2.328) {};
    \node[Mark, Mark_disk, black!50] at (1.764, 2.328) {};
    \node[Mark, Mark_disk, black!50] at (2.328, 1.764) {};
    \node[Mark, Mark_disk, black!50] at (2.328, 0.635) {};
\end{tikzpicture}

%% file: Figures/toriccoords.tikz
\begin{tikzpicture}[scale=1]
    \draw (1.082, 0.007) -- (1.082, 1.129) -- (2.211, 1.136) -- (2.211, 0.007) -- cycle;
    \node[Mark, Mark_disk] at (1.646, 0) {};
    \node[Mark, Mark_disk] at (1.082, 0.564) {};
    \node[Mark, Mark_disk] at (1.646, 1.129) {};
    \node[Mark, Mark_disk] at (2.211, 0.564) {};
    \node[anchor=center, font=\small] at (0.623, 1.348) {$(i,j)$};
    \node[anchor=center, font=\small] at (1.646, 1.419) {$(i,j,h)$};
    \node[anchor=center] at (0.518, 0.572) {$(i,j,v)$};
    \node[anchor=center, font=\small] at (3.1, 0.572) {$(i,j+1,v)$};
\end{tikzpicture}

%% file: Figures/2Dtoric.tikz
\begin{tikzpicture}[scale=1]
    \draw (0.464, 0.949) -- (1.592, 0.949);
    \draw (1.028, 1.513) -- (1.028, 0.384);
    \draw (2.721, 1.231) -- (2.721, 1.513) -- (3.004, 1.513);
    \draw (3.568, 1.513) -- (3.85, 1.513) -- (3.85, 1.231);
    \draw (2.721, 0.666) -- (2.721, 0.384) -- (3.004, 0.384);
    \draw (3.568, 0.384) -- (3.85, 0.384) -- (3.85, 0.666);
    \node[anchor=center] at (1.028, 1.795) {$\sigma_x$};
    \node[anchor=center] at (1.875, 0.949) {$\sigma_x$};
    \node[anchor=center] at (1.028, 0.102) {$\sigma_x$};
    \node[anchor=center] at (0.181, 0.949) {$\sigma_x$};
    \node[anchor=center] at (2.721, 0.949) {$\sigma_z$};
    \node[anchor=center] at (3.286, 1.513) {$\sigma_z$};
    \node[anchor=center] at (3.85, 0.949) {$\sigma_z$};
    \node[anchor=center] at (3.286, 0.384) {$\sigma_z$};
\end{tikzpicture}

%% file: Figures/gates1.tikz
\begin{tikzpicture}[scale=1]
    \draw (3.528, 4.649) -- (4.657, 4.65) -- (4.657, 0.141) -- (0.141, 0.141) -- (0.141, 1.27);
    \node[Mark, Mark_disk] at (0.706, 2.392) {};
    \node[Mark, Mark_disk] at (1.834, 2.392) {};
    \node[Mark, Mark_disk] at (0.706, 3.521) {};
    \node[Mark, Mark_disk] at (1.834, 3.521) {};
    \draw[shift={(0.141, 2.392)}, yscale=-1] (0, 0) rectangle (2.258, -2.258);
    \draw[shift={(0.141, 3.521)}, yscale=-1] (0, 0) -- (2.258, 0);
    \draw[shift={(1.27, 2.392)}, yscale=-1] (0, 0) -- (0, -2.258);
    \draw[shift={(0.141, 2.392)}, yscale=-1] (0, 0) -- (0, 1.129) -- (1.129, 1.129) -- (1.129, 0);
    \draw[shift={(1.27, 1.263)}, yscale=-1] (0, 0) -- (1.129, 0) -- (1.129, -1.129);
    \draw[shift={(2.399, 1.263)}, yscale=-1] (0, 0) -- (1.129, 0) -- (1.129, -1.129) -- (0, -1.129);
    \draw[shift={(3.528, 2.392)}, yscale=-1] (0, 0) -- (0, -1.129) -- (-1.129, -1.129);
    \draw[shift={(3.528, 3.521)}, yscale=-1] (0, 0) -- (0, -1.129) -- (-1.129, -1.129);
    \node[Mark, Mark_disk] at (2.963, 3.521) {};
    \node[Mark, Mark_disk] at (2.963, 2.392) {};
    \node[Mark, Mark_disk] at (0.706, 1.263) {};
    \node[Mark, Mark_disk] at (1.834, 1.263) {};
    \node[Mark, Mark_disk] at (2.963, 1.263) {};
    \draw (1.27, 1.27) -- (1.27, 0.141);
    \draw (2.399, 1.27) -- (2.399, 0.141);
    \draw (3.528, 1.27) -- (3.528, 0.141);
    \node[Mark, Mark_disk] at (4.092, 1.263) {};
    \node[Mark, Mark_disk] at (4.092, 2.392) {};
    \node[Mark, Mark_disk] at (4.092, 3.521) {};
    \filldraw[fill=RoyalBlue!80, fill opacity=0.3] (0, 4.092) .. controls (0, 4.186) and (0.047, 4.233) .. (0.141, 4.233) .. controls (0.235, 4.233) and (0.282, 4.186) .. (0.282, 4.092) -- (0.282, 0.706) .. controls (0.282, 0.612) and (0.235, 0.565) .. (0.141, 0.565) .. controls (0.047, 0.565) and (0, 0.612) .. (0, 0.706) -- cycle;
    \filldraw[fill=RoyalBlue!80, fill opacity=0.3] (0, 2.964) .. controls (0, 3.057) and (0.047, 3.104) .. (0.141, 3.104) .. controls (0.235, 3.104) and (0.282, 3.057) .. (0.282, 2.964) -- (0.282, 0.706) .. controls (0.282, 0.612) and (0.235, 0.565) .. (0.141, 0.565) .. controls (0.047, 0.565) and (0, 0.612) .. (0, 0.706) -- cycle;
    \filldraw[fill=RoyalBlue!80, fill opacity=0.3] (0, 1.835) .. controls (0, 1.929) and (0.047, 1.976) .. (0.141, 1.976) .. controls (0.235, 1.976) and (0.282, 1.929) .. (0.282, 1.835) -- (0.282, 0.706) .. controls (0.282, 0.612) and (0.235, 0.565) .. (0.141, 0.565) .. controls (0.047, 0.565) and (0, 0.612) .. (0, 0.706) -- cycle;
    \filldraw[fill=RoyalBlue!80, fill opacity=0.3] (1.129, 4.092) .. controls (1.129, 4.186) and (1.176, 4.233) .. (1.27, 4.233) .. controls (1.364, 4.233) and (1.411, 4.186) .. (1.411, 4.092) -- (1.411, 0.706) .. controls (1.411, 0.612) and (1.364, 0.565) .. (1.27, 0.565) .. controls (1.176, 0.565) and (1.129, 0.612) .. (1.129, 0.706) -- cycle;
    \filldraw[fill=RoyalBlue!80, fill opacity=0.3] (1.129, 2.964) .. controls (1.129, 3.057) and (1.176, 3.104) .. (1.27, 3.104) .. controls (1.364, 3.104) and (1.411, 3.057) .. (1.411, 2.964) -- (1.411, 0.706) .. controls (1.411, 0.612) and (1.364, 0.565) .. (1.27, 0.565) .. controls (1.176, 0.565) and (1.129, 0.612) .. (1.129, 0.706) -- cycle;
    \filldraw[fill=RoyalBlue!80, fill opacity=0.3] (1.129, 1.835) .. controls (1.129, 1.929) and (1.176, 1.976) .. (1.27, 1.976) .. controls (1.364, 1.976) and (1.411, 1.929) .. (1.411, 1.835) -- (1.411, 0.706) .. controls (1.411, 0.612) and (1.364, 0.565) .. (1.27, 0.565) .. controls (1.176, 0.565) and (1.129, 0.612) .. (1.129, 0.706) -- cycle;
    \filldraw[fill=RoyalBlue!80, fill opacity=0.3] (2.258, 4.092) .. controls (2.258, 4.186) and (2.305, 4.233) .. (2.399, 4.233) .. controls (2.493, 4.233) and (2.54, 4.186) .. (2.54, 4.092) -- (2.54, 0.706) .. controls (2.54, 0.612) and (2.493, 0.565) .. (2.399, 0.565) .. controls (2.305, 0.565) and (2.258, 0.612) .. (2.258, 0.706) -- cycle;
    \filldraw[fill=RoyalBlue!80, fill opacity=0.3] (2.258, 2.964) .. controls (2.258, 3.057) and (2.305, 3.104) .. (2.399, 3.104) .. controls (2.493, 3.104) and (2.54, 3.057) .. (2.54, 2.964) -- (2.54, 0.706) .. controls (2.54, 0.612) and (2.493, 0.565) .. (2.399, 0.565) .. controls (2.305, 0.565) and (2.258, 0.612) .. (2.258, 0.706) -- cycle;
    \filldraw[fill=RoyalBlue!80, fill opacity=0.3] (2.258, 1.835) .. controls (2.258, 1.929) and (2.305, 1.976) .. (2.399, 1.976) .. controls (2.493, 1.976) and (2.54, 1.929) .. (2.54, 1.835) -- (2.54, 0.706) .. controls (2.54, 0.612) and (2.493, 0.565) .. (2.399, 0.565) .. controls (2.305, 0.565) and (2.258, 0.612) .. (2.258, 0.706) -- cycle;
    \filldraw[shift={(0.706, 4.516)}, rotate=90, fill=RoyalBlue!80, fill opacity=0.3] (0, 0) .. controls (0, 0.094) and (0.047, 0.141) .. (0.141, 0.141) .. controls (0.235, 0.141) and (0.282, 0.094) .. (0.282, 0) -- (0.282, -3.387) .. controls (0.282, -3.481) and (0.235, -3.528) .. (0.141, -3.528) .. controls (0.047, -3.528) and (0, -3.481) .. (0, -3.387) -- cycle;
    \filldraw[shift={(1.834, 4.516)}, rotate=90, fill=RoyalBlue!80, fill opacity=0.3] (0, 0) .. controls (0, 0.094) and (0.047, 0.141) .. (0.141, 0.141) .. controls (0.235, 0.141) and (0.282, 0.094) .. (0.282, 0) -- (0.282, -2.258) .. controls (0.282, -2.352) and (0.235, -2.399) .. (0.141, -2.399) .. controls (0.047, -2.399) and (0, -2.352) .. (0, -2.258) -- cycle;
    \filldraw[shift={(2.963, 4.516)}, rotate=90, fill=RoyalBlue!80, fill opacity=0.3] (0, 0) .. controls (0, 0.094) and (0.047, 0.141) .. (0.141, 0.141) .. controls (0.235, 0.141) and (0.282, 0.094) .. (0.282, 0) -- (0.282, -1.129) .. controls (0.282, -1.223) and (0.235, -1.27) .. (0.141, -1.27) .. controls (0.047, -1.27) and (0, -1.223) .. (0, -1.129) -- cycle;
    \node[Mark, Mark_disk, BrickRed!90] at (2.399, 2.956) {};
    \node[Mark, Mark_disk, BrickRed!90] at (1.27, 2.956) {};
    \node[Mark, Mark_disk, BrickRed!90] at (0.141, 2.956) {};
    \node[Mark, Mark_disk, BrickRed!90] at (0.141, 4.085) {};
    \node[Mark, Mark_disk, BrickRed!90] at (0.706, 4.649) {};
    \node[Mark, Mark_disk, BrickRed!90] at (1.834, 4.649) {};
    \node[Mark, Mark_disk, BrickRed!90] at (1.27, 4.085) {};
    \node[Mark, Mark_disk, BrickRed!90] at (2.399, 4.085) {};
    \node[Mark, Mark_disk, BrickRed!90] at (0.141, 1.827) {};
    \node[Mark, Mark_disk, BrickRed!90] at (1.27, 1.827) {};
    \node[Mark, Mark_disk, BrickRed!90] at (2.399, 1.827) {};
    \node[Mark, Mark_disk, BrickRed!90] at (2.963, 4.649) {};
    \node[Mark, Mark_disk] at (4.092, 4.649) {};
    \node[Mark, Mark_disk] at (0.141, 0.698) {};
    \node[Mark, Mark_disk] at (1.27, 0.698) {};
    \node[Mark, Mark_disk] at (2.399, 0.698) {};
    \draw (3.528, 1.263) -- (4.657, 1.262);
    \draw (3.528, 2.392) -- (4.657, 2.39);
    \draw (3.528, 3.521) -- (4.657, 3.52);
    \filldraw[draw=black!25, fill=RoyalBlue!80, fill opacity=0.1] (4.516, 4.092) .. controls (4.516, 4.186) and (4.563, 4.233) .. (4.657, 4.233) .. controls (4.751, 4.233) and (4.798, 4.186) .. (4.798, 4.092) -- (4.798, 0.706) .. controls (4.798, 0.612) and (4.751, 0.565) .. (4.657, 0.565) .. controls (4.563, 0.565) and (4.516, 0.612) .. (4.516, 0.706) -- cycle;
    \filldraw[draw=black!25, fill=RoyalBlue!80, fill opacity=0.1] (4.516, 2.964) .. controls (4.516, 3.057) and (4.563, 3.104) .. (4.657, 3.104) .. controls (4.751, 3.104) and (4.798, 3.057) .. (4.798, 2.964) -- (4.798, 0.706) .. controls (4.798, 0.612) and (4.751, 0.565) .. (4.657, 0.565) .. controls (4.563, 0.565) and (4.516, 0.612) .. (4.516, 0.706) -- cycle;
    \filldraw[draw=black!25, fill=RoyalBlue!80, fill opacity=0.1] (4.516, 1.835) .. controls (4.516, 1.929) and (4.563, 1.976) .. (4.657, 1.976) .. controls (4.751, 1.976) and (4.798, 1.929) .. (4.798, 1.835) -- (4.798, 0.706) .. controls (4.798, 0.612) and (4.751, 0.565) .. (4.657, 0.565) .. controls (4.563, 0.565) and (4.516, 0.612) .. (4.516, 0.706) -- cycle;
    \node[Mark, Mark_disk, BrickRed!40] at (4.657, 1.827) {};
    \node[Mark, Mark_disk, BrickRed!40] at (4.657, 2.956) {};
    \node[Mark, Mark_disk, BrickRed!40] at (4.657, 4.085) {};
    \filldraw[shift={(0.706, 0)}, rotate=90, draw=black!25, fill=RoyalBlue!80, fill opacity=0.1] (0, 0) .. controls (0, 0.094) and (0.047, 0.141) .. (0.141, 0.141) .. controls (0.235, 0.141) and (0.282, 0.094) .. (0.282, 0) -- (0.282, -3.387) .. controls (0.282, -3.481) and (0.235, -3.528) .. (0.141, -3.528) .. controls (0.047, -3.528) and (0, -3.481) .. (0, -3.387) -- cycle;
    \filldraw[shift={(1.834, 0)}, rotate=90, draw=black!25, fill=RoyalBlue!80, fill opacity=0.1] (0, 0) .. controls (0, 0.094) and (0.047, 0.141) .. (0.141, 0.141) .. controls (0.235, 0.141) and (0.282, 0.094) .. (0.282, 0) -- (0.282, -2.258) .. controls (0.282, -2.352) and (0.235, -2.399) .. (0.141, -2.399) .. controls (0.047, -2.399) and (0, -2.352) .. (0, -2.258) -- cycle;
    \filldraw[shift={(2.963, 0)}, rotate=90, draw=black!25, fill=RoyalBlue!80, fill opacity=0.1] (0, 0) .. controls (0, 0.094) and (0.047, 0.141) .. (0.141, 0.141) .. controls (0.235, 0.141) and (0.282, 0.094) .. (0.282, 0) -- (0.282, -1.129) .. controls (0.282, -1.223) and (0.235, -1.27) .. (0.141, -1.27) .. controls (0.047, -1.27) and (0, -1.223) .. (0, -1.129) -- cycle;
    \node[Mark, Mark_disk, BrickRed!40] at (0.706, 0.134) {};
    \node[Mark, Mark_disk, BrickRed!40] at (1.834, 0.134) {};
    \node[Mark, Mark_disk, BrickRed!40] at (2.963, 0.134) {};
    \node[Mark, Mark_disk, black!50] at (4.092, 0.134) {};
    \node[Mark, Mark_disk, black!50] at (4.657, 0.698) {};
    \filldraw[fill=RoyalBlue!80, fill opacity=0.3] (3.387, 4.092) .. controls (3.387, 4.186) and (3.434, 4.233) .. (3.528, 4.233) .. controls (3.622, 4.233) and (3.669, 4.186) .. (3.669, 4.092) -- (3.669, 0.706) .. controls (3.669, 0.612) and (3.622, 0.565) .. (3.528, 0.565) .. controls (3.434, 0.565) and (3.387, 0.612) .. (3.387, 0.706) -- cycle;
    \filldraw[fill=RoyalBlue!80, fill opacity=0.3] (3.387, 2.964) .. controls (3.387, 3.057) and (3.434, 3.104) .. (3.528, 3.104) .. controls (3.622, 3.104) and (3.669, 3.057) .. (3.669, 2.964) -- (3.669, 0.706) .. controls (3.669, 0.612) and (3.622, 0.565) .. (3.528, 0.565) .. controls (3.434, 0.565) and (3.387, 0.612) .. (3.387, 0.706) -- cycle;
    \filldraw[fill=RoyalBlue!80, fill opacity=0.3] (3.387, 1.835) .. controls (3.387, 1.929) and (3.434, 1.976) .. (3.528, 1.976) .. controls (3.622, 1.976) and (3.669, 1.929) .. (3.669, 1.835) -- (3.669, 0.706) .. controls (3.669, 0.612) and (3.622, 0.565) .. (3.528, 0.565) .. controls (3.434, 0.565) and (3.387, 0.612) .. (3.387, 0.706) -- cycle;
    \node[Mark, Mark_disk, BrickRed!90] at (3.528, 1.827) {};
    \node[Mark, Mark_disk, BrickRed!90] at (3.528, 2.956) {};
    \node[Mark, Mark_disk, BrickRed!90] at (3.528, 4.085) {};
    \node[Mark, Mark_disk] at (3.528, 0.698) {};
\end{tikzpicture}

%% file: Figures/gates2.tikz
\begin{tikzpicture}[scale=1]
    \node[Mark, Mark_disk] at (0.781, 2.674) {};
    \node[Mark, Mark_disk] at (1.91, 2.674) {};
    \node[Mark, Mark_disk] at (2.474, 3.238) {};
    \node[Mark, Mark_disk] at (1.345, 3.238) {};
    \node[Mark, Mark_disk] at (0.216, 3.238) {};
    \draw[shift={(0.216, 2.674)}, yscale=-1] (0, 0) rectangle (2.258, -2.258);
    \draw[shift={(0.216, 3.803)}, yscale=-1] (0, 0) -- (2.258, 0);
    \draw[shift={(1.345, 2.674)}, yscale=-1] (0, 0) -- (0, -2.258);
    \draw[shift={(0.216, 2.674)}, yscale=-1] (0, 0) -- (0, 1.129) -- (1.129, 1.129) -- (1.129, 0);
    \draw[shift={(1.345, 1.545)}, yscale=-1] (0, 0) -- (1.129, 0) -- (1.129, -1.129);
    \draw[shift={(2.474, 1.545)}, yscale=-1] (0, 0) -- (1.129, 0) -- (1.129, -1.129) -- (0, -1.129);
    \draw[shift={(3.603, 2.674)}, yscale=-1] (0, 0) -- (0, -1.129) -- (-1.129, -1.129);
    \draw[shift={(3.603, 3.803)}, yscale=-1] (0, 0) -- (0, -1.129) -- (-1.129, -1.129);
    \node[Mark, Mark_disk] at (0.216, 2.109) {};
    \node[Mark, Mark_disk] at (1.345, 2.109) {};
    \node[Mark, Mark_disk] at (2.474, 2.109) {};
    \node[Mark, Mark_disk] at (3.039, 2.674) {};
    \node[Mark, Mark_disk] at (0.781, 1.545) {};
    \node[Mark, Mark_disk] at (1.91, 1.545) {};
    \node[Mark, Mark_disk] at (3.039, 1.545) {};
    \node[Mark, Mark_disk] at (3.603, 2.109) {};
    \node[Mark, Mark_disk] at (3.603, 3.238) {};
    \draw (1.345, 1.552) -- (1.345, 0.423);
    \draw (2.474, 1.552) -- (2.474, 0.423);
    \draw (3.603, 1.552) -- (3.603, 0.423);
    \node[Mark, Mark_disk, black!50] at (4.732, 2.109) {};
    \node[Mark, Mark_disk, black!50] at (4.732, 3.238) {};
    \node[Mark, Mark_disk, BrickRed!40] at (4.732, 4.367) {};
    \node[Mark, Mark_disk, BrickRed!40] at (0.781, 0.416) {};
    \node[Mark, Mark_disk, BrickRed!40] at (1.91, 0.416) {};
    \node[Mark, Mark_disk, BrickRed!40] at (3.039, 0.416) {};
    \node[Mark, Mark_disk, black!50] at (4.168, 0.416) {};
    \node[Mark, Mark_disk, black!50] at (4.732, 0.98) {};
    \node[Mark, Mark_disk] at (4.168, 1.545) {};
    \node[Mark, Mark_disk] at (4.168, 2.674) {};
    \node[Mark, Mark_disk] at (4.168, 3.803) {};
    \node[Mark, Mark_disk] at (4.168, 4.931) {};
    \node[Mark, Mark_disk] at (0.216, 0.98) {};
    \node[Mark, Mark_disk] at (1.345, 0.98) {};
    \node[Mark, Mark_disk] at (2.474, 0.98) {};
    \node[Mark, Mark_disk] at (3.603, 0.98) {};
    \filldraw[fill=RoyalBlue!80, fill opacity=0.3] (0.64, 4.939) .. controls (0.64, 5.033) and (0.687, 5.08) .. (0.781, 5.08) .. controls (0.875, 5.08) and (0.922, 5.033) .. (0.922, 4.939) -- (0.922, 3.81) .. controls (0.922, 3.716) and (0.875, 3.669) .. (0.781, 3.669) .. controls (0.687, 3.669) and (0.64, 3.716) .. (0.64, 3.81) -- cycle;
    \filldraw[fill=RoyalBlue!80, fill opacity=0.3] (1.769, 4.939) .. controls (1.769, 5.033) and (1.816, 5.08) .. (1.91, 5.08) .. controls (2.004, 5.08) and (2.051, 5.033) .. (2.051, 4.939) -- (2.051, 3.81) .. controls (2.051, 3.716) and (2.004, 3.669) .. (1.91, 3.669) .. controls (1.816, 3.669) and (1.769, 3.716) .. (1.769, 3.81) -- cycle;
    \filldraw[fill=RoyalBlue!80, fill opacity=0.3] (2.898, 4.939) .. controls (2.898, 5.033) and (2.945, 5.08) .. (3.039, 5.08) .. controls (3.133, 5.08) and (3.18, 5.033) .. (3.18, 4.939) -- (3.18, 3.81) .. controls (3.18, 3.716) and (3.133, 3.669) .. (3.039, 3.669) .. controls (2.945, 3.669) and (2.898, 3.716) .. (2.898, 3.81) -- cycle;
    \filldraw[shift={(0.099, 4.292)}, rotate=45, fill=RoyalBlue!80, fill opacity=0.3] (0, 0) .. controls (0, 0.094) and (0.047, 0.141) .. (0.141, 0.141) .. controls (0.235, 0.141) and (0.282, 0.094) .. (0.282, 0) -- (0.282, -0.847) .. controls (0.282, -0.941) and (0.235, -0.988) .. (0.141, -0.988) .. controls (0.047, -0.988) and (0, -0.941) .. (0, -0.847) -- cycle;
    \filldraw[shift={(1.228, 4.292)}, rotate=45, fill=RoyalBlue!80, fill opacity=0.3] (0, 0) .. controls (0, 0.094) and (0.047, 0.141) .. (0.141, 0.141) .. controls (0.235, 0.141) and (0.282, 0.094) .. (0.282, 0) -- (0.282, -0.847) .. controls (0.282, -0.941) and (0.235, -0.988) .. (0.141, -0.988) .. controls (0.047, -0.988) and (0, -0.941) .. (0, -0.847) -- cycle;
    \filldraw[shift={(2.357, 4.292)}, rotate=45, fill=RoyalBlue!80, fill opacity=0.3] (0, 0) .. controls (0, 0.094) and (0.047, 0.141) .. (0.141, 0.141) .. controls (0.235, 0.141) and (0.282, 0.094) .. (0.282, 0) -- (0.282, -0.847) .. controls (0.282, -0.941) and (0.235, -0.988) .. (0.141, -0.988) .. controls (0.047, -0.988) and (0, -0.941) .. (0, -0.847) -- cycle;
    \filldraw[shift={(1.993, 3.693)}, rotate=135, fill=RoyalBlue!80, fill opacity=0.3] (0, 0) .. controls (0, 0.094) and (0.047, 0.141) .. (0.141, 0.141) .. controls (0.235, 0.141) and (0.282, 0.094) .. (0.282, 0) -- (0.282, -0.847) .. controls (0.282, -0.941) and (0.235, -0.988) .. (0.141, -0.988) .. controls (0.047, -0.988) and (0, -0.941) .. (0, -0.847) -- cycle;
    \filldraw[shift={(3.122, 3.693)}, rotate=135, fill=RoyalBlue!80, fill opacity=0.3] (0, 0) .. controls (0, 0.094) and (0.047, 0.141) .. (0.141, 0.141) .. controls (0.235, 0.141) and (0.282, 0.094) .. (0.282, 0) -- (0.282, -0.847) .. controls (0.282, -0.941) and (0.235, -0.988) .. (0.141, -0.988) .. controls (0.047, -0.988) and (0, -0.941) .. (0, -0.847) -- cycle;
    \filldraw[shift={(0.864, 3.693)}, rotate=135, fill=RoyalBlue!80, fill opacity=0.3] (0, 0) .. controls (0, 0.094) and (0.047, 0.141) .. (0.141, 0.141) .. controls (0.235, 0.141) and (0.282, 0.094) .. (0.282, 0) -- (0.282, -0.847) .. controls (0.282, -0.941) and (0.235, -0.988) .. (0.141, -0.988) .. controls (0.047, -0.988) and (0, -0.941) .. (0, -0.847) -- cycle;
    \node[Mark, Mark_disk] at (0.781, 3.803) {};
    \node[Mark, Mark_disk, BrickRed!90] at (0.216, 4.367) {};
    \node[Mark, Mark_disk, BrickRed!90] at (0.781, 4.931) {};
    \node[Mark, Mark_disk, BrickRed!90] at (1.91, 4.931) {};
    \node[Mark, Mark_disk, BrickRed!90] at (1.345, 4.367) {};
    \node[Mark, Mark_disk] at (1.91, 3.803) {};
    \node[Mark, Mark_disk, BrickRed!90] at (2.474, 4.367) {};
    \node[Mark, Mark_disk] at (3.039, 3.803) {};
    \node[Mark, Mark_disk, BrickRed!90] at (3.039, 4.931) {};
    \node[Mark, Mark_disk, BrickRed!90] at (3.603, 4.367) {};
    \draw (3.603, 4.931) -- (4.732, 4.932) -- (4.732, 0.423) -- (0.216, 0.423) -- (0.216, 1.552);
    \draw (3.603, 1.545) -- (4.732, 1.544);
    \draw (3.603, 2.674) -- (4.732, 2.672);
    \draw (3.603, 3.803) -- (4.732, 3.802);
    \filldraw[draw=black!25, fill=RoyalBlue!80, fill opacity=0.1] (0.64, 0) -- (0.64, 0.423) .. controls (0.64, 0.517) and (0.687, 0.564) .. (0.781, 0.564) .. controls (0.875, 0.564) and (0.922, 0.517) .. (0.922, 0.423) -- (0.922, 0);
    \filldraw[draw=black!25, fill=RoyalBlue!80, fill opacity=0.1] (4.92, 3.986) -- (4.621, 4.286) .. controls (4.555, 4.352) and (4.555, 4.419) .. (4.621, 4.485) .. controls (4.688, 4.552) and (4.754, 4.552) .. (4.821, 4.485) -- (5.12, 4.186);
    \filldraw[draw=black!25, fill=RoyalBlue!80, fill opacity=0.1] (1.769, 0) -- (1.769, 0.423) .. controls (1.769, 0.517) and (1.816, 0.564) .. (1.91, 0.564) .. controls (2.004, 0.564) and (2.051, 0.517) .. (2.051, 0.423) -- (2.051, 0);
    \filldraw[draw=black!25, fill=RoyalBlue!80, fill opacity=0.1] (2.898, 0) -- (2.898, 0.423) .. controls (2.898, 0.517) and (2.945, 0.564) .. (3.039, 0.564) .. controls (3.133, 0.564) and (3.18, 0.517) .. (3.18, 0.423) -- (3.18, 0);
\end{tikzpicture}

%% file: Figures/gates3.tikz
\begin{tikzpicture}[scale=1]
    \draw[shift={(0.216, 2.674)}, yscale=-1] (0, 0) rectangle (2.258, -2.258);
    \draw[shift={(0.216, 3.803)}, yscale=-1] (0, 0) -- (2.258, 0);
    \draw[shift={(1.345, 2.674)}, yscale=-1] (0, 0) -- (0, -2.258);
    \draw[shift={(0.216, 2.674)}, yscale=-1] (0, 0) -- (0, 1.129) -- (1.129, 1.129) -- (1.129, 0);
    \draw[shift={(1.345, 1.545)}, yscale=-1] (0, 0) -- (1.129, 0) -- (1.129, -1.129);
    \draw[shift={(2.474, 1.545)}, yscale=-1] (0, 0) -- (1.129, 0) -- (1.129, -1.129) -- (0, -1.129);
    \draw[shift={(3.603, 2.674)}, yscale=-1] (0, 0) -- (0, -1.129) -- (-1.129, -1.129);
    \draw[shift={(3.603, 3.803)}, yscale=-1] (0, 0) -- (0, -1.129) -- (-1.129, -1.129);
    \node[Mark, Mark_disk] at (0.216, 2.109) {};
    \node[Mark, Mark_disk] at (1.345, 2.109) {};
    \node[Mark, Mark_disk] at (2.474, 2.109) {};
    \node[Mark, Mark_disk] at (0.781, 1.545) {};
    \node[Mark, Mark_disk] at (1.91, 1.545) {};
    \node[Mark, Mark_disk] at (3.039, 1.545) {};
    \node[Mark, Mark_disk] at (3.603, 2.109) {};
    \draw (1.345, 1.552) -- (1.345, 0.423);
    \draw (2.474, 1.552) -- (2.474, 0.423);
    \draw (3.603, 1.552) -- (3.603, 0.423);
    \node[Mark, Mark_disk, black!50] at (4.732, 2.109) {};
    \node[Mark, Mark_disk, black!50] at (4.168, 0.416) {};
    \node[Mark, Mark_disk, black!50] at (4.732, 0.98) {};
    \node[Mark, Mark_disk] at (4.168, 1.545) {};
    \node[Mark, Mark_disk] at (4.168, 2.674) {};
    \node[Mark, Mark_disk] at (4.168, 3.803) {};
    \node[Mark, Mark_disk] at (4.168, 4.931) {};
    \node[Mark, Mark_disk] at (0.216, 0.98) {};
    \node[Mark, Mark_disk] at (1.345, 0.98) {};
    \node[Mark, Mark_disk] at (2.474, 0.98) {};
    \node[Mark, Mark_disk] at (3.603, 0.98) {};
    \filldraw[fill=RoyalBlue!80, fill opacity=0.3] (0.64, 4.939) .. controls (0.64, 5.033) and (0.687, 5.08) .. (0.781, 5.08) .. controls (0.875, 5.08) and (0.922, 5.033) .. (0.922, 4.939) -- (0.922, 2.681) .. controls (0.922, 2.587) and (0.875, 2.54) .. (0.781, 2.54) .. controls (0.687, 2.54) and (0.64, 2.587) .. (0.64, 2.681) -- cycle;
    \filldraw[shift={(0.099, 3.163)}, rotate=45, fill=RoyalBlue!80, fill opacity=0.3] (0, 0) .. controls (0, 0.094) and (0.047, 0.141) .. (0.141, 0.141) .. controls (0.235, 0.141) and (0.282, 0.094) .. (0.282, 0) -- (0.282, -0.847) .. controls (0.282, -0.941) and (0.235, -0.988) .. (0.141, -0.988) .. controls (0.047, -0.988) and (0, -0.941) .. (0, -0.847) -- cycle;
    \filldraw[shift={(0.864, 2.564)}, rotate=135, fill=RoyalBlue!80, fill opacity=0.3] (0, 0) .. controls (0, 0.094) and (0.047, 0.141) .. (0.141, 0.141) .. controls (0.235, 0.141) and (0.282, 0.094) .. (0.282, 0) -- (0.282, -0.847) .. controls (0.282, -0.941) and (0.235, -0.988) .. (0.141, -0.988) .. controls (0.047, -0.988) and (0, -0.941) .. (0, -0.847) -- cycle;
    \filldraw[shift={(0.069, 4.334)}, rotate=18.936, fill=RoyalBlue!80, fill opacity=0.3] (0, 0) .. controls (0, 0.094) and (0.047, 0.141) .. (0.141, 0.141) .. controls (0.235, 0.141) and (0.282, 0.094) .. (0.282, 0) -- (0.282, -1.834) .. controls (0.282, -1.929) and (0.235, -1.976) .. (0.141, -1.976) .. controls (0.047, -1.976) and (0, -1.929) .. (0, -1.834) -- cycle;
    \filldraw[shift={(0.626, 2.704)}, rotate=-18.936, yscale=-1, fill=RoyalBlue!80, fill opacity=0.3] (0, 0) .. controls (0, 0.094) and (0.047, 0.141) .. (0.141, 0.141) .. controls (0.235, 0.141) and (0.282, 0.094) .. (0.282, 0) -- (0.282, -1.834) .. controls (0.282, -1.929) and (0.235, -1.976) .. (0.141, -1.976) .. controls (0.047, -1.976) and (0, -1.929) .. (0, -1.834) -- cycle;
    \filldraw[fill=RoyalBlue!80, fill opacity=0.3] (1.768, 4.952) .. controls (1.768, 5.046) and (1.815, 5.093) .. (1.909, 5.093) .. controls (2.003, 5.093) and (2.05, 5.046) .. (2.05, 4.952) -- (2.05, 2.694) .. controls (2.05, 2.6) and (2.003, 2.553) .. (1.909, 2.553) .. controls (1.815, 2.553) and (1.768, 2.6) .. (1.768, 2.694) -- cycle;
    \filldraw[shift={(1.227, 3.176)}, rotate=45, fill=RoyalBlue!80, fill opacity=0.3] (0, 0) .. controls (0, 0.094) and (0.047, 0.141) .. (0.141, 0.141) .. controls (0.235, 0.141) and (0.282, 0.094) .. (0.282, 0) -- (0.282, -0.847) .. controls (0.282, -0.941) and (0.235, -0.988) .. (0.141, -0.988) .. controls (0.047, -0.988) and (0, -0.941) .. (0, -0.847) -- cycle;
    \filldraw[shift={(1.992, 2.577)}, rotate=135, fill=RoyalBlue!80, fill opacity=0.3] (0, 0) .. controls (0, 0.094) and (0.047, 0.141) .. (0.141, 0.141) .. controls (0.235, 0.141) and (0.282, 0.094) .. (0.282, 0) -- (0.282, -0.847) .. controls (0.282, -0.941) and (0.235, -0.988) .. (0.141, -0.988) .. controls (0.047, -0.988) and (0, -0.941) .. (0, -0.847) -- cycle;
    \filldraw[shift={(1.197, 4.347)}, rotate=18.936, fill=RoyalBlue!80, fill opacity=0.3] (0, 0) .. controls (0, 0.094) and (0.047, 0.141) .. (0.141, 0.141) .. controls (0.235, 0.141) and (0.282, 0.094) .. (0.282, 0) -- (0.282, -1.834) .. controls (0.282, -1.929) and (0.235, -1.976) .. (0.141, -1.976) .. controls (0.047, -1.976) and (0, -1.929) .. (0, -1.834) -- cycle;
    \filldraw[shift={(1.754, 2.717)}, rotate=-18.936, yscale=-1, fill=RoyalBlue!80, fill opacity=0.3] (0, 0) .. controls (0, 0.094) and (0.047, 0.141) .. (0.141, 0.141) .. controls (0.235, 0.141) and (0.282, 0.094) .. (0.282, 0) -- (0.282, -1.834) .. controls (0.282, -1.929) and (0.235, -1.976) .. (0.141, -1.976) .. controls (0.047, -1.976) and (0, -1.929) .. (0, -1.834) -- cycle;
    \filldraw[fill=RoyalBlue!80, fill opacity=0.3] (2.892, 4.956) .. controls (2.892, 5.05) and (2.939, 5.097) .. (3.033, 5.097) .. controls (3.128, 5.097) and (3.174, 5.05) .. (3.174, 4.956) -- (3.174, 2.698) .. controls (3.174, 2.604) and (3.127, 2.557) .. (3.033, 2.557) .. controls (2.939, 2.557) and (2.892, 2.604) .. (2.892, 2.698) -- cycle;
    \filldraw[shift={(2.352, 3.18)}, rotate=45, fill=RoyalBlue!80, fill opacity=0.3] (0, 0) .. controls (0, 0.094) and (0.047, 0.141) .. (0.141, 0.141) .. controls (0.235, 0.141) and (0.282, 0.094) .. (0.282, 0) -- (0.282, -0.847) .. controls (0.282, -0.941) and (0.235, -0.988) .. (0.141, -0.988) .. controls (0.047, -0.988) and (0, -0.941) .. (0, -0.847) -- cycle;
    \filldraw[shift={(3.116, 2.581)}, rotate=135, fill=RoyalBlue!80, fill opacity=0.3] (0, 0) .. controls (0, 0.094) and (0.047, 0.141) .. (0.141, 0.141) .. controls (0.235, 0.141) and (0.282, 0.094) .. (0.282, 0) -- (0.282, -0.847) .. controls (0.282, -0.941) and (0.235, -0.988) .. (0.141, -0.988) .. controls (0.047, -0.988) and (0, -0.941) .. (0, -0.847) -- cycle;
    \filldraw[shift={(2.322, 4.351)}, rotate=18.936, fill=RoyalBlue!80, fill opacity=0.3] (0, 0) .. controls (0, 0.094) and (0.047, 0.141) .. (0.141, 0.141) .. controls (0.235, 0.141) and (0.282, 0.094) .. (0.282, 0) -- (0.282, -1.834) .. controls (0.282, -1.929) and (0.235, -1.976) .. (0.141, -1.976) .. controls (0.047, -1.976) and (0, -1.929) .. (0, -1.834) -- cycle;
    \filldraw[shift={(2.878, 2.721)}, rotate=-18.936, yscale=-1, fill=RoyalBlue!80, fill opacity=0.3] (0, 0) .. controls (0, 0.094) and (0.047, 0.141) .. (0.141, 0.141) .. controls (0.235, 0.141) and (0.282, 0.094) .. (0.282, 0) -- (0.282, -1.834) .. controls (0.282, -1.929) and (0.235, -1.976) .. (0.141, -1.976) .. controls (0.047, -1.976) and (0, -1.929) .. (0, -1.834) -- cycle;
    \node[Mark, Mark_disk] at (0.781, 2.674) {};
    \node[Mark, Mark_disk] at (1.91, 2.674) {};
    \node[Mark, Mark_disk, BrickRed!90] at (2.474, 3.238) {};
    \node[Mark, Mark_disk, BrickRed!90] at (1.345, 3.238) {};
    \node[Mark, Mark_disk, BrickRed!90] at (0.216, 3.238) {};
    \node[Mark, Mark_disk] at (0.781, 3.803) {};
    \node[Mark, Mark_disk, BrickRed!90] at (0.216, 4.367) {};
    \node[Mark, Mark_disk, BrickRed!90] at (0.781, 4.931) {};
    \node[Mark, Mark_disk, BrickRed!90] at (1.91, 4.931) {};
    \node[Mark, Mark_disk, BrickRed!90] at (1.345, 4.367) {};
    \node[Mark, Mark_disk] at (1.91, 3.803) {};
    \node[Mark, Mark_disk, BrickRed!90] at (2.474, 4.367) {};
    \node[Mark, Mark_disk] at (3.039, 3.803) {};
    \node[Mark, Mark_disk, BrickRed!90] at (3.039, 4.931) {};
    \node[Mark, Mark_disk] at (3.039, 2.674) {};
    \node[Mark, Mark_disk, BrickRed!90] at (3.603, 3.238) {};
    \node[Mark, Mark_disk, BrickRed!90] at (3.603, 4.367) {};
    \draw (3.603, 4.931) -- (4.732, 4.932) -- (4.732, 0.423) -- (0.216, 0.423) -- (0.216, 1.552);
    \draw (3.603, 1.545) -- (4.732, 1.544);
    \draw (3.603, 2.674) -- (4.732, 2.672);
    \draw (3.603, 3.803) -- (4.732, 3.802);
    \filldraw[draw=black!25, fill=RoyalBlue!80, fill opacity=0.1] (0.64, 0) -- (0.64, 0.423) .. controls (0.64, 0.517) and (0.687, 0.564) .. (0.781, 0.564) .. controls (0.875, 0.564) and (0.922, 0.517) .. (0.922, 0.423) -- (0.922, 0);
    \filldraw[draw=black!25, fill=RoyalBlue!80, fill opacity=0.1] (1.769, 0) -- (1.769, 0.423) .. controls (1.769, 0.517) and (1.816, 0.564) .. (1.91, 0.564) .. controls (2.004, 0.564) and (2.051, 0.517) .. (2.051, 0.423) -- (2.051, 0);
    \filldraw[draw=black!25, fill=RoyalBlue!80, fill opacity=0.1] (2.898, 0) -- (2.898, 0.423) .. controls (2.898, 0.517) and (2.945, 0.564) .. (3.039, 0.564) .. controls (3.133, 0.564) and (3.18, 0.517) .. (3.18, 0.423) -- (3.18, 0);
    \node[Mark, Mark_disk, BrickRed!40] at (0.781, 0.416) {};
    \node[Mark, Mark_disk, BrickRed!40] at (1.91, 0.416) {};
    \node[Mark, Mark_disk, BrickRed!40] at (3.039, 0.416) {};
    \filldraw[draw=black!25, fill=RoyalBlue!80, fill opacity=0.1] (4.876, 2.901) -- (4.615, 3.163) .. controls (4.549, 3.229) and (4.549, 3.296) .. (4.615, 3.362) .. controls (4.682, 3.429) and (4.748, 3.429) .. (4.815, 3.362) -- (5.073, 3.104);
    \filldraw[draw=black!25, fill=RoyalBlue!80, fill opacity=0.1] (4.769, 3.796) -- (4.585, 4.334) .. controls (4.554, 4.423) and (4.584, 4.483) .. (4.672, 4.513) .. controls (4.761, 4.544) and (4.821, 4.515) .. (4.852, 4.426) -- (5.036, 3.89);
    \node[Mark, Mark_disk, BrickRed!40] at (4.732, 4.367) {};
    \node[Mark, Mark_disk, BrickRed!40] at (4.732, 3.238) {};
\end{tikzpicture}

%% file: Figures/gates4.tikz
\begin{tikzpicture}[scale=1]
    \draw (3.603, 3.803) -- (4.732, 3.802);
    \draw[shift={(0.216, 2.674)}, yscale=-1] (0, 0) rectangle (2.258, -2.258);
    \draw[shift={(0.216, 3.803)}, yscale=-1] (0, 0) -- (2.258, 0);
    \draw[shift={(1.345, 2.674)}, yscale=-1] (0, 0) -- (0, -2.258);
    \draw[shift={(0.216, 2.674)}, yscale=-1] (0, 0) -- (0, 1.129) -- (1.129, 1.129) -- (1.129, 0);
    \draw[shift={(1.345, 1.545)}, yscale=-1] (0, 0) -- (1.129, 0) -- (1.129, -1.129);
    \draw[shift={(2.474, 1.545)}, yscale=-1] (0, 0) -- (1.129, 0) -- (1.129, -1.129) -- (0, -1.129);
    \draw[shift={(3.603, 2.674)}, yscale=-1] (0, 0) -- (0, -1.129) -- (-1.129, -1.129);
    \draw[shift={(3.603, 3.803)}, yscale=-1] (0, 0) -- (0, -1.129) -- (-1.129, -1.129);
    \draw (1.345, 1.552) -- (1.345, 0.423);
    \draw (2.474, 1.552) -- (2.474, 0.423);
    \draw (3.603, 1.552) -- (3.603, 0.423);
    \node[Mark, Mark_disk, black!50] at (4.168, 0.416) {};
    \node[Mark, Mark_disk, black!50] at (4.732, 0.98) {};
    \node[Mark, Mark_disk] at (4.168, 1.545) {};
    \node[Mark, Mark_disk] at (4.168, 2.674) {};
    \node[Mark, Mark_disk] at (4.168, 3.803) {};
    \node[Mark, Mark_disk] at (4.168, 4.931) {};
    \node[Mark, Mark_disk] at (0.216, 0.98) {};
    \node[Mark, Mark_disk] at (1.345, 0.98) {};
    \node[Mark, Mark_disk] at (2.474, 0.98) {};
    \node[Mark, Mark_disk] at (3.603, 0.98) {};
    \filldraw[fill=RoyalBlue!80, fill opacity=0.3] (0.64, 4.939) .. controls (0.64, 5.033) and (0.687, 5.08) .. (0.781, 5.08) .. controls (0.875, 5.08) and (0.922, 5.033) .. (0.922, 4.939) -- (0.922, 1.552) .. controls (0.922, 1.458) and (0.875, 1.411) .. (0.781, 1.411) .. controls (0.687, 1.411) and (0.64, 1.458) .. (0.64, 1.552) -- cycle;
    \filldraw[fill=RoyalBlue!80, fill opacity=0.3] (1.769, 4.939) .. controls (1.769, 5.033) and (1.816, 5.08) .. (1.91, 5.08) .. controls (2.004, 5.08) and (2.051, 5.033) .. (2.051, 4.939) -- (2.051, 1.552) .. controls (2.051, 1.458) and (2.004, 1.411) .. (1.91, 1.411) .. controls (1.816, 1.411) and (1.769, 1.458) .. (1.769, 1.552) -- cycle;
    \filldraw[fill=RoyalBlue!80, fill opacity=0.3] (2.898, 4.939) .. controls (2.898, 5.033) and (2.945, 5.08) .. (3.039, 5.08) .. controls (3.133, 5.08) and (3.18, 5.033) .. (3.18, 4.939) -- (3.18, 1.552) .. controls (3.18, 1.458) and (3.133, 1.411) .. (3.039, 1.411) .. controls (2.945, 1.411) and (2.898, 1.458) .. (2.898, 1.552) -- cycle;
    \filldraw[shift={(0.099, 2.034)}, rotate=45, fill=RoyalBlue!80, fill opacity=0.3] (0, 0) .. controls (0, 0.094) and (0.047, 0.141) .. (0.141, 0.141) .. controls (0.235, 0.141) and (0.282, 0.094) .. (0.282, 0) -- (0.282, -0.847) .. controls (0.282, -0.941) and (0.235, -0.988) .. (0.141, -0.988) .. controls (0.047, -0.988) and (0, -0.941) .. (0, -0.847) -- cycle;
    \filldraw[shift={(0.864, 1.435)}, rotate=135, fill=RoyalBlue!80, fill opacity=0.3] (0, 0) .. controls (0, 0.094) and (0.047, 0.141) .. (0.141, 0.141) .. controls (0.235, 0.141) and (0.282, 0.094) .. (0.282, 0) -- (0.282, -0.847) .. controls (0.282, -0.941) and (0.235, -0.988) .. (0.141, -0.988) .. controls (0.047, -0.988) and (0, -0.941) .. (0, -0.847) -- cycle;
    \filldraw[shift={(1.228, 2.034)}, rotate=45, fill=RoyalBlue!80, fill opacity=0.3] (0, 0) .. controls (0, 0.094) and (0.047, 0.141) .. (0.141, 0.141) .. controls (0.235, 0.141) and (0.282, 0.094) .. (0.282, 0) -- (0.282, -0.847) .. controls (0.282, -0.941) and (0.235, -0.988) .. (0.141, -0.988) .. controls (0.047, -0.988) and (0, -0.941) .. (0, -0.847) -- cycle;
    \filldraw[shift={(1.993, 1.435)}, rotate=135, fill=RoyalBlue!80, fill opacity=0.3] (0, 0) .. controls (0, 0.094) and (0.047, 0.141) .. (0.141, 0.141) .. controls (0.235, 0.141) and (0.282, 0.094) .. (0.282, 0) -- (0.282, -0.847) .. controls (0.282, -0.941) and (0.235, -0.988) .. (0.141, -0.988) .. controls (0.047, -0.988) and (0, -0.941) .. (0, -0.847) -- cycle;
    \filldraw[shift={(2.357, 2.034)}, rotate=45, fill=RoyalBlue!80, fill opacity=0.3] (0, 0) .. controls (0, 0.094) and (0.047, 0.141) .. (0.141, 0.141) .. controls (0.235, 0.141) and (0.282, 0.094) .. (0.282, 0) -- (0.282, -0.847) .. controls (0.282, -0.941) and (0.235, -0.988) .. (0.141, -0.988) .. controls (0.047, -0.988) and (0, -0.941) .. (0, -0.847) -- cycle;
    \filldraw[shift={(3.121, 1.435)}, rotate=135, fill=RoyalBlue!80, fill opacity=0.3] (0, 0) .. controls (0, 0.094) and (0.047, 0.141) .. (0.141, 0.141) .. controls (0.235, 0.141) and (0.282, 0.094) .. (0.282, 0) -- (0.282, -0.847) .. controls (0.282, -0.941) and (0.235, -0.988) .. (0.141, -0.988) .. controls (0.047, -0.988) and (0, -0.941) .. (0, -0.847) -- cycle;
    \filldraw[shift={(0.069, 3.205)}, rotate=18.936, fill=RoyalBlue!80, fill opacity=0.3] (0, 0) .. controls (0, 0.094) and (0.047, 0.141) .. (0.141, 0.141) .. controls (0.235, 0.141) and (0.282, 0.094) .. (0.282, 0) -- (0.282, -1.834) .. controls (0.282, -1.929) and (0.235, -1.976) .. (0.141, -1.976) .. controls (0.047, -1.976) and (0, -1.929) .. (0, -1.834) -- cycle;
    \filldraw[shift={(0.626, 1.575)}, rotate=-18.936, yscale=-1, fill=RoyalBlue!80, fill opacity=0.3] (0, 0) .. controls (0, 0.094) and (0.047, 0.141) .. (0.141, 0.141) .. controls (0.235, 0.141) and (0.282, 0.094) .. (0.282, 0) -- (0.282, -1.834) .. controls (0.282, -1.929) and (0.235, -1.976) .. (0.141, -1.976) .. controls (0.047, -1.976) and (0, -1.929) .. (0, -1.834) -- cycle;
    \filldraw[shift={(1.198, 3.205)}, rotate=18.936, fill=RoyalBlue!80, fill opacity=0.3] (0, 0) .. controls (0, 0.094) and (0.047, 0.141) .. (0.141, 0.141) .. controls (0.235, 0.141) and (0.282, 0.094) .. (0.282, 0) -- (0.282, -1.834) .. controls (0.282, -1.929) and (0.235, -1.976) .. (0.141, -1.976) .. controls (0.047, -1.976) and (0, -1.929) .. (0, -1.834) -- cycle;
    \filldraw[shift={(1.755, 1.575)}, rotate=-18.936, yscale=-1, fill=RoyalBlue!80, fill opacity=0.3] (0, 0) .. controls (0, 0.094) and (0.047, 0.141) .. (0.141, 0.141) .. controls (0.235, 0.141) and (0.282, 0.094) .. (0.282, 0) -- (0.282, -1.834) .. controls (0.282, -1.929) and (0.235, -1.976) .. (0.141, -1.976) .. controls (0.047, -1.976) and (0, -1.929) .. (0, -1.834) -- cycle;
    \filldraw[shift={(2.327, 3.205)}, rotate=18.936, fill=RoyalBlue!80, fill opacity=0.3] (0, 0) .. controls (0, 0.094) and (0.047, 0.141) .. (0.141, 0.141) .. controls (0.235, 0.141) and (0.282, 0.094) .. (0.282, 0) -- (0.282, -1.834) .. controls (0.282, -1.929) and (0.235, -1.976) .. (0.141, -1.976) .. controls (0.047, -1.976) and (0, -1.929) .. (0, -1.834) -- cycle;
    \filldraw[shift={(2.884, 1.575)}, rotate=-18.936, yscale=-1, fill=RoyalBlue!80, fill opacity=0.3] (0, 0) .. controls (0, 0.094) and (0.047, 0.141) .. (0.141, 0.141) .. controls (0.235, 0.141) and (0.282, 0.094) .. (0.282, 0) -- (0.282, -1.834) .. controls (0.282, -1.929) and (0.235, -1.976) .. (0.141, -1.976) .. controls (0.047, -1.976) and (0, -1.929) .. (0, -1.834) -- cycle;
    \filldraw[shift={(2.323, 4.374)}, rotate=11.288, fill=RoyalBlue!80, fill opacity=0.3] (0, 0) .. controls (0, 0.094) and (0.047, 0.141) .. (0.141, 0.141) .. controls (0.235, 0.141) and (0.282, 0.094) .. (0.282, 0) -- (0.282, -2.963) .. controls (0.282, -3.057) and (0.235, -3.104) .. (0.141, -3.104) .. controls (0.047, -3.104) and (0, -3.057) .. (0, -2.963) -- cycle;
    \filldraw[shift={(2.887, 1.523)}, rotate=-11.288, yscale=-1, fill=RoyalBlue!80, fill opacity=0.3] (0, 0) .. controls (0, 0.094) and (0.047, 0.141) .. (0.141, 0.141) .. controls (0.235, 0.141) and (0.282, 0.094) .. (0.282, 0) -- (0.282, -2.963) .. controls (0.282, -3.057) and (0.235, -3.104) .. (0.141, -3.104) .. controls (0.047, -3.104) and (0, -3.057) .. (0, -2.963) -- cycle;
    \filldraw[shift={(1.194, 4.374)}, rotate=11.288, fill=RoyalBlue!80, fill opacity=0.3] (0, 0) .. controls (0, 0.094) and (0.047, 0.141) .. (0.141, 0.141) .. controls (0.235, 0.141) and (0.282, 0.094) .. (0.282, 0) -- (0.282, -2.963) .. controls (0.282, -3.057) and (0.235, -3.104) .. (0.141, -3.104) .. controls (0.047, -3.104) and (0, -3.057) .. (0, -2.963) -- cycle;
    \filldraw[shift={(1.759, 1.523)}, rotate=-11.288, yscale=-1, fill=RoyalBlue!80, fill opacity=0.3] (0, 0) .. controls (0, 0.094) and (0.047, 0.141) .. (0.141, 0.141) .. controls (0.235, 0.141) and (0.282, 0.094) .. (0.282, 0) -- (0.282, -2.963) .. controls (0.282, -3.057) and (0.235, -3.104) .. (0.141, -3.104) .. controls (0.047, -3.104) and (0, -3.057) .. (0, -2.963) -- cycle;
    \filldraw[shift={(0.065, 4.374)}, rotate=11.288, fill=RoyalBlue!80, fill opacity=0.3] (0, 0) .. controls (0, 0.094) and (0.047, 0.141) .. (0.141, 0.141) .. controls (0.235, 0.141) and (0.282, 0.094) .. (0.282, 0) -- (0.282, -2.963) .. controls (0.282, -3.057) and (0.235, -3.104) .. (0.141, -3.104) .. controls (0.047, -3.104) and (0, -3.057) .. (0, -2.963) -- cycle;
    \filldraw[shift={(0.63, 1.523)}, rotate=-11.288, yscale=-1, fill=RoyalBlue!80, fill opacity=0.3] (0, 0) .. controls (0, 0.094) and (0.047, 0.141) .. (0.141, 0.141) .. controls (0.235, 0.141) and (0.282, 0.094) .. (0.282, 0) -- (0.282, -2.963) .. controls (0.282, -3.057) and (0.235, -3.104) .. (0.141, -3.104) .. controls (0.047, -3.104) and (0, -3.057) .. (0, -2.963) -- cycle;
    \node[Mark, Mark_disk] at (0.781, 2.674) {};
    \node[Mark, Mark_disk] at (1.91, 2.674) {};
    \node[Mark, Mark_disk, BrickRed!90] at (2.474, 3.238) {};
    \node[Mark, Mark_disk, BrickRed!90] at (1.345, 3.238) {};
    \node[Mark, Mark_disk, BrickRed!90] at (0.216, 3.238) {};
    \node[Mark, Mark_disk] at (0.781, 3.803) {};
    \node[Mark, Mark_disk, BrickRed!90] at (0.216, 4.367) {};
    \node[Mark, Mark_disk, BrickRed!90] at (0.781, 4.931) {};
    \node[Mark, Mark_disk, BrickRed!90] at (1.91, 4.931) {};
    \node[Mark, Mark_disk, BrickRed!90] at (1.345, 4.367) {};
    \node[Mark, Mark_disk] at (1.91, 3.803) {};
    \node[Mark, Mark_disk, BrickRed!90] at (2.474, 4.367) {};
    \node[Mark, Mark_disk, BrickRed!90] at (0.216, 2.109) {};
    \node[Mark, Mark_disk, BrickRed!90] at (1.345, 2.109) {};
    \node[Mark, Mark_disk, BrickRed!90] at (2.474, 2.109) {};
    \node[Mark, Mark_disk] at (3.039, 3.803) {};
    \node[Mark, Mark_disk, BrickRed!90] at (3.039, 4.931) {};
    \node[Mark, Mark_disk] at (3.039, 2.674) {};
    \node[Mark, Mark_disk] at (0.781, 1.545) {};
    \node[Mark, Mark_disk] at (1.91, 1.545) {};
    \node[Mark, Mark_disk] at (3.039, 1.545) {};
    \node[Mark, Mark_disk, BrickRed!90] at (3.603, 2.109) {};
    \node[Mark, Mark_disk, BrickRed!90] at (3.603, 3.238) {};
    \node[Mark, Mark_disk, BrickRed!90] at (3.603, 4.367) {};
    \draw (3.603, 4.931) -- (4.732, 4.932) -- (4.732, 0.423) -- (0.216, 0.423) -- (0.216, 1.552);
    \draw (3.603, 1.545) -- (4.732, 1.544);
    \draw (3.603, 2.674) -- (4.732, 2.672);
    \filldraw[draw=black!25, fill=RoyalBlue!80, fill opacity=0.1] (0.64, 0) -- (0.64, 0.423) .. controls (0.64, 0.517) and (0.687, 0.564) .. (0.781, 0.564) .. controls (0.875, 0.564) and (0.922, 0.517) .. (0.922, 0.423) -- (0.922, 0);
    \filldraw[draw=black!25, fill=RoyalBlue!80, fill opacity=0.1] (1.769, 0) -- (1.769, 0.423) .. controls (1.769, 0.517) and (1.816, 0.564) .. (1.91, 0.564) .. controls (2.004, 0.564) and (2.051, 0.517) .. (2.051, 0.423) -- (2.051, 0);
    \filldraw[draw=black!25, fill=RoyalBlue!80, fill opacity=0.1] (2.898, 0) -- (2.898, 0.423) .. controls (2.898, 0.517) and (2.945, 0.564) .. (3.039, 0.564) .. controls (3.133, 0.564) and (3.18, 0.517) .. (3.18, 0.423) -- (3.18, 0);
    \filldraw[draw=black!25, fill=RoyalBlue!80, fill opacity=0.1] (4.827, 1.822) -- (4.615, 2.034) .. controls (4.549, 2.1) and (4.549, 2.167) .. (4.615, 2.233) .. controls (4.682, 2.3) and (4.748, 2.3) .. (4.815, 2.233) -- (5.029, 2.019);
    \filldraw[draw=black!25, fill=RoyalBlue!80, fill opacity=0.1] (4.794, 2.594) -- (4.585, 3.205) .. controls (4.554, 3.294) and (4.584, 3.354) .. (4.672, 3.384) .. controls (4.761, 3.415) and (4.821, 3.386) .. (4.852, 3.297) -- (5.062, 2.684);
    \filldraw[draw=black!25, fill=RoyalBlue!80, fill opacity=0.1] (4.703, 3.761) -- (4.581, 4.374) .. controls (4.562, 4.466) and (4.599, 4.522) .. (4.692, 4.54) .. controls (4.784, 4.558) and (4.839, 4.521) .. (4.858, 4.429) -- (4.981, 3.811);
    \node[Mark, Mark_disk, BrickRed!40] at (0.781, 0.416) {};
    \node[Mark, Mark_disk, BrickRed!40] at (1.91, 0.416) {};
    \node[Mark, Mark_disk, BrickRed!40] at (3.039, 0.416) {};
    \node[Mark, Mark_disk, BrickRed!40] at (4.732, 2.109) {};
    \node[Mark, Mark_disk, BrickRed!40] at (4.732, 3.238) {};
    \node[Mark, Mark_disk, BrickRed!40] at (4.732, 4.367) {};
\end{tikzpicture}

%% file: Figures/gates5.tikz
\begin{tikzpicture}[scale=1]
    \draw (3.606, 3.784) -- (4.735, 3.783);
    \node[Mark, Mark_disk] at (0.784, 2.655) {};
    \node[Mark, Mark_disk] at (1.913, 2.655) {};
    \node[Mark, Mark_disk] at (2.477, 3.22) {};
    \node[Mark, Mark_disk] at (1.348, 3.22) {};
    \node[Mark, Mark_disk] at (0.22, 3.22) {};
    \node[Mark, Mark_disk] at (0.784, 3.784) {};
    \node[Mark, Mark_disk] at (1.348, 4.349) {};
    \node[Mark, Mark_disk] at (1.913, 3.784) {};
    \draw[shift={(0.22, 2.655)}, yscale=-1] (0, 0) rectangle (2.258, -2.258);
    \draw[shift={(0.22, 3.784)}, yscale=-1] (0, 0) -- (2.258, 0);
    \draw[shift={(1.348, 2.655)}, yscale=-1] (0, 0) -- (0, -2.258);
    \draw[shift={(0.22, 2.655)}, yscale=-1] (0, 0) -- (0, 1.129) -- (1.129, 1.129) -- (1.129, 0);
    \draw[shift={(1.348, 1.527)}, yscale=-1] (0, 0) -- (1.129, 0) -- (1.129, -1.129);
    \draw[shift={(2.477, 1.527)}, yscale=-1] (0, 0) -- (1.129, 0) -- (1.129, -1.129) -- (0, -1.129);
    \draw[shift={(3.606, 2.655)}, yscale=-1] (0, 0) -- (0, -1.129) -- (-1.129, -1.129);
    \draw[shift={(3.606, 3.784)}, yscale=-1] (0, 0) -- (0, -1.129) -- (-1.129, -1.129);
    \node[Mark, Mark_disk] at (0.22, 2.091) {};
    \node[Mark, Mark_disk] at (1.348, 2.091) {};
    \node[Mark, Mark_disk] at (2.477, 2.091) {};
    \node[Mark, Mark_disk] at (3.042, 2.655) {};
    \node[Mark, Mark_disk] at (0.784, 1.527) {};
    \node[Mark, Mark_disk] at (1.913, 1.527) {};
    \node[Mark, Mark_disk] at (3.042, 1.527) {};
    \node[Mark, Mark_disk] at (3.606, 2.091) {};
    \node[Mark, Mark_disk] at (3.606, 3.22) {};
    \draw (1.348, 1.534) -- (1.348, 0.405);
    \draw (2.477, 1.534) -- (2.477, 0.405);
    \draw (3.606, 1.534) -- (3.606, 0.405);
    \node[Mark, Mark_disk, black!50] at (4.735, 2.091) {};
    \node[Mark, Mark_disk, black!50] at (4.735, 3.22) {};
    \node[Mark, Mark_disk, black!50] at (4.171, 0.398) {};
    \node[Mark, Mark_disk, black!50] at (4.735, 0.962) {};
    \node[Mark, Mark_disk] at (4.171, 1.527) {};
    \node[Mark, Mark_disk] at (4.171, 2.655) {};
    \node[Mark, Mark_disk] at (4.171, 4.913) {};
    \node[Mark, Mark_disk] at (0.22, 0.962) {};
    \node[Mark, Mark_disk] at (1.348, 0.962) {};
    \node[Mark, Mark_disk] at (2.477, 0.962) {};
    \node[Mark, Mark_disk] at (3.606, 0.962) {};
    \filldraw[shift={(3.489, 4.273)}, rotate=45, fill=RoyalBlue!80, fill opacity=0.3] (0, 0) .. controls (0, 0.094) and (0.047, 0.141) .. (0.141, 0.141) .. controls (0.235, 0.141) and (0.282, 0.094) .. (0.282, 0) -- (0.282, -0.847) .. controls (0.282, -0.941) and (0.235, -0.988) .. (0.141, -0.988) .. controls (0.047, -0.988) and (0, -0.941) .. (0, -0.847) -- cycle;
    \filldraw[shift={(0.692, 4.794)}, rotate=71.477, fill=RoyalBlue!80, fill opacity=0.3] (0, 0) .. controls (0, 0.094) and (0.047, 0.141) .. (0.141, 0.141) .. controls (0.235, 0.141) and (0.282, 0.094) .. (0.282, 0) -- (0.282, -3.669) .. controls (0.282, -3.763) and (0.235, -3.81) .. (0.141, -3.81) .. controls (0.047, -3.81) and (0, -3.763) .. (0, -3.669) -- cycle;
    \filldraw[shift={(0.14, 4.201)}, rotate=82.191, fill=RoyalBlue!80, fill opacity=0.3] (0, 0) .. controls (0, 0.094) and (0.047, 0.141) .. (0.141, 0.141) .. controls (0.235, 0.141) and (0.282, 0.094) .. (0.282, 0) -- (0.282, -4.092) .. controls (0.282, -4.186) and (0.235, -4.233) .. (0.141, -4.233) .. controls (0.047, -4.233) and (0, -4.186) .. (0, -4.092) -- cycle;
    \filldraw[shift={(1.837, 4.783)}, rotate=64.125, fill=RoyalBlue!80, fill opacity=0.3] (0, 0) .. controls (0, 0.094) and (0.047, 0.141) .. (0.141, 0.141) .. controls (0.235, 0.141) and (0.282, 0.094) .. (0.282, 0) -- (0.282, -2.54) .. controls (0.282, -2.634) and (0.235, -2.681) .. (0.141, -2.681) .. controls (0.047, -2.681) and (0, -2.634) .. (0, -2.54) -- cycle;
    \filldraw[shift={(2.898, 4.841)}, rotate=46.146, fill=RoyalBlue!80, fill opacity=0.3] (0, 0) .. controls (0, 0.094) and (0.047, 0.141) .. (0.141, 0.141) .. controls (0.235, 0.141) and (0.282, 0.094) .. (0.282, 0) -- (0.282, -1.693) .. controls (0.282, -1.787) and (0.235, -1.834) .. (0.141, -1.834) .. controls (0.047, -1.834) and (0, -1.787) .. (0, -1.693) -- cycle;
    \node[Mark, Mark_disk, BrickRed!90] at (0.22, 4.349) {};
    \node[Mark, Mark_disk, BrickRed!90] at (0.784, 4.913) {};
    \node[Mark, Mark_disk, BrickRed!90] at (1.913, 4.913) {};
    \node[Mark, Mark_disk] at (2.477, 4.349) {};
    \node[Mark, Mark_disk] at (3.042, 3.784) {};
    \node[Mark, Mark_disk, BrickRed!90] at (3.042, 4.913) {};
    \node[Mark, Mark_disk, BrickRed!90] at (3.606, 4.349) {};
    \node[Mark, Mark_disk] at (4.171, 3.784) {};
    \draw (3.606, 4.913) -- (4.735, 4.913) -- (4.735, 0.405) -- (0.22, 0.405) -- (0.22, 1.534);
    \draw (3.606, 1.527) -- (4.735, 1.526);
    \draw (3.606, 2.655) -- (4.735, 2.654);
    \filldraw[draw=black!25, fill=RoyalBlue!80, fill opacity=0.1] (1.513, 0.004) -- (0.692, 0.279) .. controls (0.603, 0.309) and (0.573, 0.368) .. (0.603, 0.457) .. controls (0.633, 0.547) and (0.692, 0.576) .. (0.781, 0.546) -- (1.61, 0.269);
    \filldraw[draw=black!25, fill=RoyalBlue!80, fill opacity=0.1] (2.39, 0) -- (1.837, 0.268) .. controls (1.753, 0.309) and (1.731, 0.372) .. (1.772, 0.456) .. controls (1.813, 0.541) and (1.876, 0.563) .. (1.961, 0.522) -- (2.519, 0.251);
    \filldraw[draw=black!25, fill=RoyalBlue!80, fill opacity=0.1] (3.222, 0.014) -- (2.898, 0.325) .. controls (2.83, 0.391) and (2.829, 0.457) .. (2.894, 0.525) .. controls (2.959, 0.593) and (3.026, 0.594) .. (3.093, 0.529) -- (3.423, 0.212);
    \node[Mark, Mark_disk, BrickRed!40] at (0.784, 0.398) {};
    \node[Mark, Mark_disk, BrickRed!40] at (1.913, 0.398) {};
    \node[Mark, Mark_disk, BrickRed!40] at (3.042, 0.398) {};
    \filldraw[draw=black!25, fill=RoyalBlue!80, fill opacity=0.1] (5.014, 4.132) -- (4.612, 4.188) .. controls (4.519, 4.2) and (4.478, 4.253) .. (4.491, 4.347) .. controls (4.504, 4.44) and (4.557, 4.48) .. (4.65, 4.467) -- (5.059, 4.411);
    \node[Mark, Mark_disk, BrickRed!40] at (4.735, 4.349) {};
\end{tikzpicture}

%% file: Figures/gates6.tikz
\begin{tikzpicture}[scale=1]
    \draw (3.606, 2.682) -- (4.735, 2.681);
    \draw (3.606, 3.811) -- (4.735, 3.81);
    \node[Mark, Mark_disk] at (0.784, 2.682) {};
    \node[Mark, Mark_disk] at (1.913, 2.682) {};
    \node[Mark, Mark_disk] at (0.784, 3.811) {};
    \draw[shift={(0.22, 2.682)}, yscale=-1] (0, 0) rectangle (2.258, -2.258);
    \draw[shift={(0.22, 3.811)}, yscale=-1] (0, 0) -- (2.258, 0);
    \draw[shift={(1.348, 2.682)}, yscale=-1] (0, 0) -- (0, -2.258);
    \draw[shift={(0.22, 2.682)}, yscale=-1] (0, 0) -- (0, 1.129) -- (1.129, 1.129) -- (1.129, 0);
    \draw[shift={(1.348, 1.553)}, yscale=-1] (0, 0) -- (1.129, 0) -- (1.129, -1.129);
    \draw[shift={(2.477, 1.553)}, yscale=-1] (0, 0) -- (1.129, 0) -- (1.129, -1.129) -- (0, -1.129);
    \draw[shift={(3.606, 2.682)}, yscale=-1] (0, 0) -- (0, -1.129) -- (-1.129, -1.129);
    \draw[shift={(3.606, 3.811)}, yscale=-1] (0, 0) -- (0, -1.129) -- (-1.129, -1.129);
    \node[Mark, Mark_disk] at (0.22, 2.118) {};
    \node[Mark, Mark_disk] at (1.348, 2.118) {};
    \node[Mark, Mark_disk] at (2.477, 2.118) {};
    \node[Mark, Mark_disk] at (0.784, 1.553) {};
    \node[Mark, Mark_disk] at (1.913, 1.553) {};
    \node[Mark, Mark_disk] at (3.042, 1.553) {};
    \node[Mark, Mark_disk] at (3.606, 2.118) {};
    \draw (1.348, 1.561) -- (1.348, 0.432);
    \draw (2.477, 1.561) -- (2.477, 0.432);
    \draw (3.606, 1.561) -- (3.606, 0.432);
    \draw (3.606, 1.561) -- (4.735, 1.561);
    \node[Mark, Mark_disk, black!50] at (4.735, 2.118) {};
    \node[Mark, Mark_disk, black!50] at (4.171, 0.424) {};
    \node[Mark, Mark_disk, black!50] at (4.735, 0.989) {};
    \node[Mark, Mark_disk] at (4.171, 1.553) {};
    \node[Mark, Mark_disk] at (4.171, 3.811) {};
    \node[Mark, Mark_disk] at (4.171, 4.94) {};
    \node[Mark, Mark_disk] at (0.22, 0.989) {};
    \node[Mark, Mark_disk] at (1.348, 0.989) {};
    \node[Mark, Mark_disk] at (2.477, 0.989) {};
    \node[Mark, Mark_disk] at (3.606, 0.989) {};
    \filldraw[shift={(3.489, 3.171)}, rotate=45, fill=RoyalBlue!80, fill opacity=0.3] (0, 0) .. controls (0, 0.094) and (0.047, 0.141) .. (0.141, 0.141) .. controls (0.235, 0.141) and (0.282, 0.094) .. (0.282, 0) -- (0.282, -0.847) .. controls (0.282, -0.941) and (0.235, -0.988) .. (0.141, -0.988) .. controls (0.047, -0.988) and (0, -0.941) .. (0, -0.847) -- cycle;
    \filldraw[shift={(0.14, 3.098)}, rotate=82.191, fill=RoyalBlue!80, fill opacity=0.3] (0, 0) .. controls (0, 0.094) and (0.047, 0.141) .. (0.141, 0.141) .. controls (0.235, 0.141) and (0.282, 0.094) .. (0.282, 0) -- (0.282, -4.092) .. controls (0.282, -4.186) and (0.235, -4.233) .. (0.141, -4.233) .. controls (0.047, -4.233) and (0, -4.186) .. (0, -4.092) -- cycle;
    \filldraw[shift={(3.459, 4.359)}, rotate=18.786, fill=RoyalBlue!80, fill opacity=0.3] (0, 0) .. controls (0, 0.094) and (0.047, 0.141) .. (0.141, 0.141) .. controls (0.235, 0.141) and (0.282, 0.094) .. (0.282, 0) -- (0.282, -1.834) .. controls (0.282, -1.929) and (0.235, -1.976) .. (0.141, -1.976) .. controls (0.047, -1.976) and (0, -1.929) .. (0, -1.834) -- cycle;
    \filldraw[shift={(2.909, 4.889)}, rotate=26.72, fill=RoyalBlue!80, fill opacity=0.3] (0, 0) .. controls (0, 0.094) and (0.047, 0.141) .. (0.141, 0.141) .. controls (0.235, 0.141) and (0.282, 0.094) .. (0.282, 0) -- (0.282, -2.54) .. controls (0.282, -2.634) and (0.235, -2.681) .. (0.141, -2.681) .. controls (0.047, -2.681) and (0, -2.634) .. (0, -2.54) -- cycle;
    \filldraw[shift={(1.786, 4.895)}, rotate=44.596, fill=RoyalBlue!80, fill opacity=0.3] (0, 0) .. controls (0, 0.094) and (0.047, 0.141) .. (0.141, 0.141) .. controls (0.235, 0.141) and (0.282, 0.094) .. (0.282, 0) -- (0.282, -3.246) .. controls (0.282, -3.34) and (0.235, -3.387) .. (0.141, -3.387) .. controls (0.047, -3.387) and (0, -3.34) .. (0, -3.246) -- cycle;
    \filldraw[shift={(0.599, 4.895)}, rotate=56.41, fill=RoyalBlue!80, fill opacity=0.3] (0, 0) .. controls (0, 0.094) and (0.047, 0.141) .. (0.141, 0.141) .. controls (0.235, 0.141) and (0.282, 0.094) .. (0.282, 0) -- (0.282, -4.233) .. controls (0.282, -4.327) and (0.235, -4.374) .. (0.141, -4.374) .. controls (0.047, -4.374) and (0, -4.327) .. (0, -4.233) -- cycle;
    \filldraw[shift={(0.14, 4.26)}, rotate=66.503, fill=RoyalBlue!80, fill opacity=0.3] (0, 0) .. controls (0, 0.094) and (0.047, 0.141) .. (0.141, 0.141) .. controls (0.235, 0.141) and (0.282, 0.094) .. (0.282, 0) -- (0.282, -4.374) .. controls (0.282, -4.469) and (0.235, -4.516) .. (0.141, -4.516) .. controls (0.047, -4.516) and (0, -4.469) .. (0, -4.374) -- cycle;
    \node[Mark, Mark_disk] at (2.477, 3.247) {};
    \node[Mark, Mark_disk, BrickRed!90] at (0.22, 3.247) {};
    \node[Mark, Mark_disk, BrickRed!90] at (0.22, 4.375) {};
    \node[Mark, Mark_disk, BrickRed!90] at (0.784, 4.94) {};
    \node[Mark, Mark_disk, BrickRed!90] at (1.913, 4.94) {};
    \node[Mark, Mark_disk] at (1.913, 3.811) {};
    \node[Mark, Mark_disk] at (2.477, 4.375) {};
    \node[Mark, Mark_disk] at (3.042, 3.811) {};
    \node[Mark, Mark_disk, BrickRed!90] at (3.042, 4.94) {};
    \node[Mark, Mark_disk, BrickRed!90] at (3.606, 3.247) {};
    \node[Mark, Mark_disk, BrickRed!90] at (3.606, 4.375) {};
    \node[Mark, Mark_disk] at (4.171, 2.682) {};
    \node[Mark, Mark_disk] at (3.042, 2.682) {};
    \node[Mark, Mark_disk] at (1.348, 3.247) {};
    \node[Mark, Mark_disk] at (1.348, 4.375) {};
    \draw (3.606, 4.94) -- (4.735, 4.94) -- (4.735, 0.432) -- (0.22, 0.432) -- (0.22, 1.561);
    \draw (3.606, 1.553) -- (4.735, 1.553);
    \filldraw[draw=black!25, fill=RoyalBlue!80, fill opacity=0.1] (3.108, 0.015) -- (2.919, 0.389) .. controls (2.877, 0.473) and (2.898, 0.536) .. (2.982, 0.578) .. controls (3.066, 0.621) and (3.129, 0.6) .. (3.171, 0.516) -- (3.371, 0.119);
    \filldraw[draw=black!25, fill=RoyalBlue!80, fill opacity=0.1] (2.175, 0.011) -- (1.796, 0.395) .. controls (1.73, 0.462) and (1.731, 0.528) .. (1.798, 0.594) .. controls (1.865, 0.66) and (1.931, 0.66) .. (1.997, 0.593) -- (2.378, 0.206);
    \filldraw[draw=black!25, fill=RoyalBlue!80, fill opacity=0.1] (1.202, 0) -- (0.609, 0.394) .. controls (0.531, 0.446) and (0.518, 0.511) .. (0.57, 0.59) .. controls (0.622, 0.668) and (0.687, 0.681) .. (0.765, 0.629) -- (1.375, 0.224);
    \filldraw[draw=black!25, fill=RoyalBlue!80, fill opacity=0.1] (5.056, 3.079) -- (4.65, 3.134) .. controls (4.557, 3.147) and (4.516, 3.2) .. (4.529, 3.293) .. controls (4.542, 3.386) and (4.595, 3.427) .. (4.688, 3.414) -- (5.085, 3.359);
    \filldraw[draw=black!25, fill=RoyalBlue!80, fill opacity=0.1] (5.109, 4.061) -- (4.65, 4.261) .. controls (4.564, 4.298) and (4.539, 4.36) .. (4.577, 4.447) .. controls (4.614, 4.533) and (4.676, 4.557) .. (4.762, 4.52) -- (5.214, 4.323);
    \node[Mark, Mark_disk, BrickRed!40] at (0.784, 0.424) {};
    \node[Mark, Mark_disk, BrickRed!40] at (1.913, 0.424) {};
    \node[Mark, Mark_disk, BrickRed!40] at (3.042, 0.424) {};
    \node[Mark, Mark_disk, BrickRed!40] at (4.735, 3.247) {};
    \node[Mark, Mark_disk, BrickRed!40] at (4.735, 4.375) {};
\end{tikzpicture}

%% file: Figures/gates7.tikz
\begin{tikzpicture}[scale=1]
    \draw (3.647, 3.803) -- (4.776, 3.802);
    \node[Mark, Mark_disk] at (0.824, 2.674) {};
    \node[Mark, Mark_disk] at (1.389, 3.239) {};
    \draw[shift={(0.26, 2.674)}, yscale=-1] (0, 0) rectangle (2.258, -2.258);
    \draw[shift={(0.26, 3.803)}, yscale=-1] (0, 0) -- (2.258, 0);
    \draw[shift={(1.389, 2.674)}, yscale=-1] (0, 0) -- (0, -2.258);
    \draw[shift={(0.26, 2.674)}, yscale=-1] (0, 0) -- (0, 1.129) -- (1.129, 1.129) -- (1.129, 0);
    \draw[shift={(1.389, 1.546)}, yscale=-1] (0, 0) -- (1.129, 0) -- (1.129, -1.129);
    \draw[shift={(2.518, 1.546)}, yscale=-1] (0, 0) -- (1.129, 0) -- (1.129, -1.129) -- (0, -1.129);
    \draw[shift={(3.647, 2.674)}, yscale=-1] (0, 0) -- (0, -1.129) -- (-1.129, -1.129);
    \draw[shift={(3.647, 3.803)}, yscale=-1] (0, 0) -- (0, -1.129) -- (-1.129, -1.129);
    \node[Mark, Mark_disk] at (1.389, 2.11) {};
    \node[Mark, Mark_disk] at (3.082, 3.803) {};
    \node[Mark, Mark_disk] at (0.824, 1.546) {};
    \node[Mark, Mark_disk] at (1.953, 1.546) {};
    \draw (3.647, 4.932) -- (4.776, 4.932) -- (4.776, 0.424) -- (0.26, 0.424) -- (0.26, 1.553);
    \draw (1.389, 1.553) -- (1.389, 0.424);
    \draw (2.518, 1.553) -- (2.518, 0.424);
    \draw (3.647, 1.553) -- (3.647, 0.424);
    \draw (3.647, 1.546) -- (4.776, 1.545);
    \draw (3.647, 2.674) -- (4.776, 2.673);
    \node[Mark, Mark_disk, black!50] at (4.211, 0.417) {};
    \node[Mark, Mark_disk, black!50] at (4.776, 0.981) {};
    \node[Mark, Mark_disk] at (4.211, 2.674) {};
    \node[Mark, Mark_disk] at (4.211, 3.803) {};
    \node[Mark, Mark_disk] at (4.211, 4.932) {};
    \node[Mark, Mark_disk] at (0.26, 0.981) {};
    \node[Mark, Mark_disk] at (1.389, 0.981) {};
    \node[Mark, Mark_disk] at (2.518, 0.981) {};
    \node[Mark, Mark_disk] at (3.647, 0.981) {};
    \filldraw[shift={(3.53, 2.035)}, rotate=45, fill=RoyalBlue!80, fill opacity=0.3] (0, 0) .. controls (0, 0.094) and (0.047, 0.141) .. (0.141, 0.141) .. controls (0.235, 0.141) and (0.282, 0.094) .. (0.282, 0) -- (0.282, -0.847) .. controls (0.282, -0.941) and (0.235, -0.988) .. (0.141, -0.988) .. controls (0.047, -0.988) and (0, -0.941) .. (0, -0.847) -- cycle;
    \filldraw[shift={(0.18, 1.962)}, rotate=82.191, fill=RoyalBlue!80, fill opacity=0.3] (0, 0) .. controls (0, 0.094) and (0.047, 0.141) .. (0.141, 0.141) .. controls (0.235, 0.141) and (0.282, 0.094) .. (0.282, 0) -- (0.282, -4.092) .. controls (0.282, -4.186) and (0.235, -4.233) .. (0.141, -4.233) .. controls (0.047, -4.233) and (0, -4.186) .. (0, -4.092) -- cycle;
    \filldraw[shift={(3.5, 3.223)}, rotate=18.786, fill=RoyalBlue!80, fill opacity=0.3] (0, 0) .. controls (0, 0.094) and (0.047, 0.141) .. (0.141, 0.141) .. controls (0.235, 0.141) and (0.282, 0.094) .. (0.282, 0) -- (0.282, -1.834) .. controls (0.282, -1.929) and (0.235, -1.976) .. (0.141, -1.976) .. controls (0.047, -1.976) and (0, -1.929) .. (0, -1.834) -- cycle;
    \filldraw[shift={(0.18, 3.124)}, rotate=66.503, fill=RoyalBlue!80, fill opacity=0.3] (0, 0) .. controls (0, 0.094) and (0.047, 0.141) .. (0.141, 0.141) .. controls (0.235, 0.141) and (0.282, 0.094) .. (0.282, 0) -- (0.282, -4.374) .. controls (0.282, -4.469) and (0.235, -4.516) .. (0.141, -4.516) .. controls (0.047, -4.516) and (0, -4.469) .. (0, -4.374) -- cycle;
    \filldraw[shift={(3.51, 4.394)}, rotate=11.446, fill=RoyalBlue!80, fill opacity=0.3] (0, 0) .. controls (0, 0.094) and (0.047, 0.141) .. (0.141, 0.141) .. controls (0.235, 0.141) and (0.282, 0.094) .. (0.282, 0) -- (0.282, -2.963) .. controls (0.282, -3.057) and (0.235, -3.104) .. (0.141, -3.104) .. controls (0.047, -3.104) and (0, -3.057) .. (0, -2.963) -- cycle;
    \filldraw[shift={(2.924, 4.938)}, rotate=18.73, fill=RoyalBlue!80, fill opacity=0.3] (0, 0) .. controls (0, 0.094) and (0.047, 0.141) .. (0.141, 0.141) .. controls (0.235, 0.141) and (0.282, 0.094) .. (0.282, 0) -- (0.282, -3.669) .. controls (0.282, -3.763) and (0.235, -3.81) .. (0.141, -3.81) .. controls (0.047, -3.81) and (0, -3.763) .. (0, -3.669) -- cycle;
    \filldraw[shift={(0.665, 4.889)}, rotate=45.146, fill=RoyalBlue!80, fill opacity=0.3] (0, 0) .. controls (0, 0.094) and (0.047, 0.141) .. (0.141, 0.141) .. controls (0.235, 0.141) and (0.282, 0.094) .. (0.282, 0) -- (0.28, -4.945) .. controls (0.28, -5.039) and (0.233, -5.086) .. (0.139, -5.086) .. controls (0.044, -5.086) and (-0.003, -5.039) .. (-0.003, -4.945) -- cycle;
    \filldraw[shift={(1.82, 4.897)}, rotate=33.288, fill=RoyalBlue!80, fill opacity=0.3] (0, 0) .. controls (0, 0.094) and (0.047, 0.141) .. (0.141, 0.141) .. controls (0.235, 0.141) and (0.282, 0.094) .. (0.282, 0) -- (0.282, -4.092) .. controls (0.282, -4.186) and (0.235, -4.233) .. (0.141, -4.233) .. controls (0.047, -4.233) and (0, -4.186) .. (0, -4.092) -- cycle;
    \filldraw[shift={(0.115, 4.296)}, rotate=54.763, fill=RoyalBlue!80, fill opacity=0.3] (0, 0) .. controls (0, 0.094) and (0.047, 0.141) .. (0.141, 0.141) .. controls (0.235, 0.141) and (0.282, 0.094) .. (0.282, 0) -- (0.282, -4.939) .. controls (0.282, -5.033) and (0.235, -5.08) .. (0.141, -5.08) .. controls (0.047, -5.08) and (0, -5.033) .. (0, -4.939) -- cycle;
    \node[Mark, Mark_disk] at (1.953, 2.674) {};
    \node[Mark, Mark_disk] at (2.518, 3.239) {};
    \node[Mark, Mark_disk, BrickRed!90] at (0.26, 3.239) {};
    \node[Mark, Mark_disk] at (0.824, 3.803) {};
    \node[Mark, Mark_disk, BrickRed!90] at (0.26, 4.368) {};
    \node[Mark, Mark_disk, BrickRed!90] at (0.824, 4.932) {};
    \node[Mark, Mark_disk, BrickRed!90] at (1.953, 4.932) {};
    \node[Mark, Mark_disk] at (1.389, 4.368) {};
    \node[Mark, Mark_disk] at (1.953, 3.803) {};
    \node[Mark, Mark_disk] at (2.518, 4.368) {};
    \node[Mark, Mark_disk, BrickRed!90] at (0.26, 2.11) {};
    \node[Mark, Mark_disk] at (2.518, 2.11) {};
    \node[Mark, Mark_disk, BrickRed!90] at (3.082, 4.932) {};
    \node[Mark, Mark_disk] at (3.082, 2.674) {};
    \node[Mark, Mark_disk] at (3.082, 1.546) {};
    \node[Mark, Mark_disk, BrickRed!90] at (3.647, 2.11) {};
    \node[Mark, Mark_disk, BrickRed!90] at (3.647, 3.239) {};
    \node[Mark, Mark_disk, BrickRed!90] at (3.647, 4.368) {};
    \node[Mark, Mark_disk] at (4.211, 1.546) {};
    \filldraw[draw=black!25, fill=RoyalBlue!80, fill opacity=0.1] (5.171, 1.897) -- (4.696, 1.962) .. controls (4.603, 1.975) and (4.562, 2.028) .. (4.575, 2.121) .. controls (4.588, 2.214) and (4.641, 2.254) .. (4.734, 2.241) -- (5.21, 2.176);
    \filldraw[draw=black!25, fill=RoyalBlue!80, fill opacity=0.1] (5.152, 2.925) -- (4.696, 3.124) .. controls (4.61, 3.161) and (4.585, 3.223) .. (4.623, 3.309) .. controls (4.66, 3.396) and (4.722, 3.42) .. (4.808, 3.383) -- (5.251, 3.19);
    \filldraw[draw=black!25, fill=RoyalBlue!80, fill opacity=0.1] (5.125, 3.947) -- (4.631, 4.296) .. controls (4.554, 4.35) and (4.543, 4.416) .. (4.597, 4.493) .. controls (4.651, 4.569) and (4.717, 4.581) .. (4.794, 4.526) -- (5.276, 4.186);
    \filldraw[draw=black!25, fill=RoyalBlue!80, fill opacity=0.1] (3.063, 0.013) -- (2.924, 0.423) .. controls (2.894, 0.512) and (2.923, 0.572) .. (3.012, 0.602) .. controls (3.101, 0.632) and (3.161, 0.602) .. (3.191, 0.513) -- (3.334, 0.091);
    \filldraw[draw=black!25, fill=RoyalBlue!80, fill opacity=0.1] (1.04, 0) -- (0.665, 0.374) .. controls (0.598, 0.44) and (0.598, 0.506) .. (0.664, 0.573) .. controls (0.731, 0.64) and (0.797, 0.64) .. (0.864, 0.574) -- (1.237, 0.202);
    \filldraw[draw=black!25, fill=RoyalBlue!80, fill opacity=0.1] (2.065, 0.009) -- (1.82, 0.382) .. controls (1.769, 0.46) and (1.782, 0.525) .. (1.861, 0.577) .. controls (1.939, 0.628) and (2.005, 0.615) .. (2.056, 0.536) -- (2.313, 0.144);
    \node[Mark, Mark_disk, BrickRed!40] at (4.776, 2.11) {};
    \node[Mark, Mark_disk, BrickRed!40] at (4.776, 3.239) {};
    \node[Mark, Mark_disk, BrickRed!40] at (4.776, 4.368) {};
    \node[Mark, Mark_disk, BrickRed!40] at (0.824, 0.417) {};
    \node[Mark, Mark_disk, BrickRed!40] at (1.953, 0.417) {};
    \node[Mark, Mark_disk, BrickRed!40] at (3.082, 0.417) {};
\end{tikzpicture}

%% file: Figures/combstoric.tikz
\begin{tikzpicture}[scale=1]
    \draw (3.528, 1.689) -- (4.657, 1.689);
    \draw (3.528, 2.818) -- (4.657, 2.817);
    \node[Mark, Mark_disk] at (0.706, 2.818) {};
    \draw[shift={(0.141, 2.818)}, yscale=-1] (0, 0) rectangle (2.258, -2.258);
    \draw[shift={(0.141, 3.947)}, yscale=-1] (0, 0) -- (2.258, 0);
    \draw[shift={(1.27, 2.818)}, yscale=-1] (0, 0) -- (0, -2.258);
    \draw[shift={(0.141, 2.818)}, yscale=-1] (0, 0) -- (0, 1.129) -- (1.129, 1.129) -- (1.129, 0);
    \draw[shift={(1.27, 1.689)}, yscale=-1] (0, 0) -- (1.129, 0) -- (1.129, -1.129);
    \draw[shift={(2.399, 1.689)}, yscale=-1] (0, 0) -- (1.129, 0) -- (1.129, -1.129) -- (0, -1.129);
    \draw[shift={(3.528, 2.818)}, yscale=-1] (0, 0) -- (0, -1.129) -- (-1.129, -1.129);
    \draw[shift={(3.528, 3.947)}, yscale=-1] (0, 0) -- (0, -1.129) -- (-1.129, -1.129);
    \node[Mark, Mark_disk, BrickRed!90] at (1.27, 2.254) {};
    \node[Mark, Mark_disk] at (2.963, 3.947) {};
    \node[Mark, Mark_disk] at (0.706, 1.689) {};
    \node[Mark, Mark_disk] at (1.834, 1.689) {};
    \draw (1.27, 1.697) -- (1.27, 0.568);
    \draw (2.399, 1.697) -- (2.399, 0.568);
    \draw (3.528, 1.697) -- (3.528, 0.568);
    \node[Mark, Mark_disk, black!50] at (4.092, 0.56) {};
    \node[Mark, Mark_disk, black!50] at (4.657, 1.125) {};
    \node[Mark, Mark_disk] at (4.092, 2.818) {};
    \node[Mark, Mark_disk] at (4.092, 3.947) {};
    \node[Mark, Mark_disk] at (4.092, 5.076) {};
    \node[Mark, Mark_disk] at (0.141, 1.125) {};
    \node[Mark, Mark_disk] at (1.27, 1.125) {};
    \node[Mark, Mark_disk] at (2.399, 1.125) {};
    \node[Mark, Mark_disk] at (3.528, 1.125) {};
    \node[Mark, Mark_disk] at (1.834, 2.818) {};
    \node[Mark, Mark_disk, BrickRed!90] at (2.399, 3.383) {};
    \node[Mark, Mark_disk, BrickRed!90] at (0.141, 3.383) {};
    \node[Mark, Mark_disk] at (0.706, 3.947) {};
    \node[Mark, Mark_disk, BrickRed!90] at (0.141, 4.512) {};
    \node[Mark, Mark_disk, BrickRed!90] at (0.706, 5.076) {};
    \node[Mark, Mark_disk, BrickRed!90] at (1.834, 5.076) {};
    \node[Mark, Mark_disk, BrickRed!90] at (1.27, 4.512) {};
    \node[Mark, Mark_disk] at (1.834, 3.947) {};
    \node[Mark, Mark_disk, BrickRed!90] at (2.399, 4.512) {};
    \node[Mark, Mark_disk, BrickRed!90] at (0.141, 2.254) {};
    \node[Mark, Mark_disk, BrickRed!90] at (2.399, 2.254) {};
    \node[Mark, Mark_disk, BrickRed!90] at (2.963, 5.076) {};
    \node[Mark, Mark_disk] at (2.963, 2.818) {};
    \node[Mark, Mark_disk] at (2.963, 1.689) {};
    \node[Mark, Mark_disk, BrickRed!90] at (3.528, 2.254) {};
    \node[Mark, Mark_disk, BrickRed!90] at (3.528, 3.383) {};
    \node[Mark, Mark_disk, BrickRed!90] at (3.528, 4.512) {};
    \node[Mark, Mark_disk] at (4.092, 1.689) {};
    \fill[RoyalBlue!80, opacity=0.3] (0, 0.991) -- (0.282, 0.991) -- (0.282, 1.273) -- (0.706, 1.556) -- (1.129, 1.273) -- (1.129, 0.991) -- (1.411, 0.991) -- (1.411, 1.273) -- (1.834, 1.556) -- (2.258, 1.273) -- (2.258, 0.991) -- (2.54, 0.991) -- (2.54, 1.273) -- (2.963, 1.556) -- (3.387, 1.273) -- (3.387, 0.991) -- (3.951, 0.427) -- (4.233, 0.427) -- (4.798, 0.991) -- (4.798, 1.273) -- (4.233, 1.697) -- (4.092, 1.556) -- (4.516, 1.273) -- (4.516, 0.991) -- (4.233, 0.709) -- (3.951, 0.709) -- (3.669, 0.991) -- (3.669, 1.273) -- (4.092, 1.556) -- (4.233, 1.697) -- (4.233, 5.224) -- (3.951, 5.224) -- (3.951, 1.697) -- (3.951, 1.697) -- (3.528, 1.414) -- (3.104, 1.697) -- (3.104, 4.096) -- (2.822, 4.096) -- (2.822, 1.697) -- (2.399, 1.414) -- (1.976, 1.697) -- (1.976, 4.096) -- (1.693, 4.096) -- (1.693, 1.697) -- (1.27, 1.414) -- (0.847, 1.697) -- (0.847, 4.096) -- (0.564, 4.096) -- (0.564, 1.697) -- (0, 1.273) -- (0, 0.991);
    \fill[BrickRed!90, opacity=0.3] (0, 2.12) -- (0.282, 2.12) -- (0.282, 4.66) -- (0.706, 4.942) -- (1.129, 4.66) -- (1.129, 2.12) -- (1.411, 2.12) -- (1.411, 4.66) -- (1.834, 4.942) -- (2.258, 4.66) -- (2.258, 2.12) -- (2.54, 2.12) -- (2.54, 4.66) -- (2.963, 4.942) -- (3.387, 4.66) -- (3.387, 2.12) -- (3.669, 2.12) -- (3.669, 4.66) -- (3.669, 4.801) -- (2.963, 5.224) -- (2.399, 4.942) -- (1.834, 5.224) -- (1.27, 4.942) -- (0.706, 5.224) -- (0, 4.801) -- (0, 4.519) -- (0, 2.12);
    \draw (3.528, 5.076) -- (4.657, 5.076) -- (4.657, 0.568) -- (0.141, 0.568) -- (0.141, 1.697);
    \draw (3.528, 3.947) -- (4.657, 3.946);
    \fill[BrickRed!90, opacity=0.1] (4.516, 2.12) -- (4.798, 2.12) -- (4.798, 4.66) -- (5.077, 4.846) -- (5.077, 5.138) -- (4.516, 4.801) -- (4.516, 4.519) -- (4.516, 2.12);
    \fill[BrickRed!90, opacity=0.1] (0.282, 0) -- (0.282, 0.144) -- (0.706, 0.427) -- (1.129, 0.144) -- (1.129, 0.001) -- (1.411, 0.004) -- (1.411, 0.144) -- (1.834, 0.427) -- (2.258, 0.144) -- (2.258, 0.003) -- (2.54, 0.008) -- (2.54, 0.144) -- (2.963, 0.427) -- (3.387, 0.144) -- (3.387, 0.005) -- (3.669, 0.005) -- (3.669, 0.144) -- (3.669, 0.286) -- (2.963, 0.709) -- (2.399, 0.427) -- (1.834, 0.709) -- (1.27, 0.427) -- (0.706, 0.709) -- (0, 0.286) -- (0, 0.003);
    \node[Mark, Mark_disk, BrickRed!40] at (4.657, 2.254) {};
    \node[Mark, Mark_disk, BrickRed!40] at (4.657, 3.383) {};
    \node[Mark, Mark_disk, BrickRed!40] at (4.657, 4.512) {};
    \node[Mark, Mark_disk, BrickRed!40] at (0.706, 0.56) {};
    \node[Mark, Mark_disk, BrickRed!40] at (1.834, 0.56) {};
    \node[Mark, Mark_disk, BrickRed!40] at (2.963, 0.56) {};
    \node[Mark, Mark_disk, BrickRed!90] at (1.27, 3.383) {};
\end{tikzpicture}

%% file: Figures/finalplaquettes.tikz
\begin{tikzpicture}[scale=1]
    \node[Mark, Mark_disk] at (0.564, 2.258) {};
    \draw[shift={(0, 2.258)}, yscale=-1] (0, 0) rectangle (2.258, -2.258);
    \draw[shift={(0, 3.387)}, yscale=-1] (0, 0) -- (2.258, 0);
    \draw[shift={(1.129, 2.258)}, yscale=-1] (0, 0) -- (0, -2.258);
    \draw[shift={(0, 2.258)}, yscale=-1] (0, 0) -- (0, 1.129) -- (1.129, 1.129) -- (1.129, 0);
    \draw[shift={(1.129, 1.129)}, yscale=-1] (0, 0) -- (1.129, 0) -- (1.129, -1.129);
    \draw[shift={(2.258, 1.129)}, yscale=-1] (0, 0) -- (1.129, 0) -- (1.129, -1.129) -- (0, -1.129);
    \draw[shift={(3.387, 2.258)}, yscale=-1] (0, 0) -- (0, -1.129) -- (-1.129, -1.129);
    \draw[shift={(3.387, 3.387)}, yscale=-1] (0, 0) -- (0, -1.129) -- (-1.129, -1.129);
    \node[Mark, Mark_disk, BrickRed!90] at (1.129, 1.693) {};
    \node[Mark, Mark_disk] at (2.822, 3.387) {};
    \node[Mark, Mark_disk] at (0.564, 1.129) {};
    \node[Mark, Mark_disk] at (1.693, 1.129) {};
    \draw (1.129, 1.136) -- (1.129, 0.007);
    \draw (2.258, 1.136) -- (2.258, 0.007);
    \draw (3.387, 1.136) -- (3.387, 0.007);
    \node[Mark, Mark_disk, black!50] at (3.951, 0) {};
    \node[Mark, Mark_disk] at (3.951, 2.258) {};
    \node[Mark, Mark_disk] at (3.951, 3.387) {};
    \node[Mark, Mark_disk] at (3.951, 4.516) {};
    \node[Mark, Mark_disk] at (0, 0.564) {};
    \node[Mark, Mark_disk] at (1.129, 0.564) {};
    \node[Mark, Mark_disk] at (2.258, 0.564) {};
    \node[Mark, Mark_disk] at (3.387, 0.564) {};
    \node[Mark, Mark_disk] at (1.693, 2.258) {};
    \node[Mark, Mark_disk, BrickRed!90] at (2.258, 2.822) {};
    \node[Mark, Mark_disk, BrickRed!90] at (0, 2.822) {};
    \node[Mark, Mark_disk] at (0.564, 3.387) {};
    \node[Mark, Mark_disk, BrickRed!90] at (0, 3.951) {};
    \node[Mark, Mark_disk, BrickRed!90] at (0.564, 4.516) {};
    \node[Mark, Mark_disk, BrickRed!90] at (1.693, 4.516) {};
    \node[Mark, Mark_disk, BrickRed!90] at (1.129, 3.951) {};
    \node[Mark, Mark_disk] at (1.693, 3.387) {};
    \node[Mark, Mark_disk, BrickRed!90] at (2.258, 3.951) {};
    \node[Mark, Mark_disk, BrickRed!90] at (0, 1.693) {};
    \node[Mark, Mark_disk, BrickRed!90] at (2.258, 1.693) {};
    \node[Mark, Mark_disk, BrickRed!90] at (2.822, 4.516) {};
    \node[Mark, Mark_disk] at (2.822, 2.258) {};
    \node[Mark, Mark_disk] at (2.822, 1.129) {};
    \node[Mark, Mark_disk, BrickRed!90] at (3.387, 1.693) {};
    \node[Mark, Mark_disk, BrickRed!90] at (3.387, 2.822) {};
    \node[Mark, Mark_disk, BrickRed!90] at (3.387, 3.951) {};
    \node[Mark, Mark_disk] at (3.951, 1.129) {};
    \fill[RoyalBlue!80, opacity=0.5] (0.363, 1.096) -- (0.363, 1.032) -- (0.997, 1.032) -- (0.997, 0.101) -- (0.121, 0.101) -- (0.121, 1.032) -- (0.363, 1.032) -- (0.363, 1.096) -- (0.06, 1.096) -- (0.06, 0.101) -- (1.057, 0.101) -- (1.057, 1.096) -- (0.363, 1.096);
    \fill[RoyalBlue!80, opacity=0.5] (0.06, 3.424) rectangle (1.057, 3.488);
    \fill[RoyalBlue!80, opacity=0.5] (1.189, 3.424) rectangle (2.186, 3.488);
    \fill[RoyalBlue!80, opacity=0.5] (2.318, 3.424) rectangle (3.315, 3.488);
    \fill[RoyalBlue!80, opacity=0.5] (3.447, 3.424) rectangle (4.444, 3.488);
    \fill[RoyalBlue!80, opacity=0.5] (3.447, 4.411) rectangle (4.444, 4.476);
    \fill[shift={(3.977, 3.451)}, rotate=90, RoyalBlue!80, opacity=0.5] (0, 0) rectangle (0.997, 0.064);
    \fill[RoyalBlue!80, opacity=0.5] (3.447, 2.295) rectangle (4.444, 2.359);
    \fill[RoyalBlue!80, opacity=0.5] (3.447, 3.283) rectangle (4.444, 3.347);
    \fill[shift={(3.978, 2.322)}, rotate=90, RoyalBlue!80, opacity=0.5] (0, 0) rectangle (0.997, 0.064);
    \fill[RoyalBlue!80, opacity=0.5] (2.318, 2.295) rectangle (3.315, 2.359);
    \fill[RoyalBlue!80, opacity=0.5] (2.318, 3.283) rectangle (3.315, 3.347);
    \fill[shift={(2.849, 2.322)}, rotate=90, RoyalBlue!80, opacity=0.5] (0, 0) rectangle (0.997, 0.064);
    \fill[RoyalBlue!80, opacity=0.5] (1.189, 2.295) rectangle (2.186, 2.359);
    \fill[RoyalBlue!80, opacity=0.5] (1.189, 3.283) rectangle (2.186, 3.347);
    \fill[shift={(1.72, 2.322)}, rotate=90, RoyalBlue!80, opacity=0.5] (0, 0) rectangle (0.997, 0.064);
    \fill[RoyalBlue!80, opacity=0.5] (0.06, 2.295) rectangle (1.057, 2.359);
    \fill[RoyalBlue!80, opacity=0.5] (0.06, 3.283) rectangle (1.057, 3.347);
    \fill[shift={(0.591, 2.322)}, rotate=90, RoyalBlue!80, opacity=0.5] (0, 0) rectangle (0.997, 0.064);
    \fill[RoyalBlue!80, opacity=0.5] (0.06, 1.166) rectangle (1.057, 1.23);
    \fill[RoyalBlue!80, opacity=0.5] (0.06, 2.154) rectangle (1.057, 2.218);
    \fill[shift={(0.591, 1.193)}, rotate=90, RoyalBlue!80, opacity=0.5] (0, 0) rectangle (0.997, 0.064);
    \fill[RoyalBlue!80, opacity=0.5] (1.189, 1.166) rectangle (2.186, 1.23);
    \fill[RoyalBlue!80, opacity=0.5] (1.189, 2.154) rectangle (2.186, 2.218);
    \fill[shift={(1.72, 1.193)}, rotate=90, RoyalBlue!80, opacity=0.5] (0, 0) rectangle (0.997, 0.064);
    \fill[RoyalBlue!80, opacity=0.5] (2.318, 1.166) rectangle (3.315, 1.23);
    \fill[RoyalBlue!80, opacity=0.5] (2.318, 2.154) rectangle (3.315, 2.218);
    \fill[shift={(2.849, 1.193)}, rotate=90, RoyalBlue!80, opacity=0.5] (0, 0) rectangle (0.997, 0.064);
    \fill[RoyalBlue!80, opacity=0.5] (3.447, 1.166) rectangle (4.444, 1.23);
    \fill[RoyalBlue!80, opacity=0.5] (3.447, 2.154) rectangle (4.444, 2.218);
    \fill[shift={(3.978, 1.193)}, rotate=90, RoyalBlue!80, opacity=0.5] (0, 0) rectangle (0.997, 0.064);
    \fill[RoyalBlue!80, opacity=0.5] (1.492, 1.096) -- (1.492, 1.032) -- (2.126, 1.032) -- (2.126, 0.101) -- (1.25, 0.101) -- (1.25, 1.032) -- (1.492, 1.032) -- (1.492, 1.096) -- (1.189, 1.096) -- (1.189, 0.101) -- (2.186, 0.101) -- (2.186, 1.096) -- (1.492, 1.096);
    \fill[RoyalBlue!80, opacity=0.5] (2.621, 1.096) -- (2.621, 1.032) -- (3.255, 1.032) -- (3.255, 0.101) -- (2.378, 0.101) -- (2.378, 1.032) -- (2.62, 1.032) -- (2.62, 1.096) -- (2.318, 1.096) -- (2.318, 0.101) -- (3.315, 0.101) -- (3.315, 1.096) -- (2.621, 1.096);
    \fill[shift={(3.75, 1.096)}, xscale=0.831, yscale=0.883, RoyalBlue!80, opacity=0.5] (0, 0) -- (0, -0.073) -- (0.763, -0.073) -- (0.763, -1.127) -- (-0.291, -1.127) -- (-0.291, -0.073) -- (0, -0.073) -- (0, 0) -- (-0.364, 0) -- (-0.364, -1.199) -- (0.835, -1.199) -- (0.835, 0) -- (0, 0);
    \draw (3.387, 4.516) -- (4.516, 4.516) -- (4.516, 0.007) -- (0, 0.007) -- (0, 1.136);
    \draw (3.387, 1.129) -- (4.516, 1.128);
    \draw (3.387, 2.258) -- (4.516, 2.256);
    \draw (3.387, 3.387) -- (4.516, 3.386);
    \begin{scope}[shift={(-1.693, -5.637)}]
        \node[Mark, Mark_disk, BrickRed!40] at (2.258, 5.637) {};
        \node[Mark, Mark_disk, BrickRed!40] at (3.387, 5.637) {};
        \node[Mark, Mark_disk, BrickRed!40] at (4.516, 5.637) {};
    \end{scope}
    \node[Mark, Mark_disk, BrickRed!40] at (4.516, 1.693) {};
    \node[Mark, Mark_disk, BrickRed!40] at (4.516, 2.822) {};
    \node[Mark, Mark_disk, BrickRed!40] at (4.516, 3.951) {};
    \node[Mark, Mark_disk, black!50] at (4.516, 0.564) {};
    \node[Mark, Mark_disk, BrickRed!90] at (1.129, 2.822) {};
\end{tikzpicture}

%% file: Figures/finalstars.tikz
\begin{tikzpicture}[scale=1]
    \draw (3.457, 1.603) -- (4.586, 1.602);
    \draw (3.457, 2.732) -- (4.586, 2.73);
    \draw (3.457, 3.861) -- (4.586, 3.86);
    \draw (3.457, 4.99) -- (4.586, 4.99) -- (4.586, 0.481) -- (0.071, 0.481) -- (0.071, 1.61);
    \node[Mark, Mark_disk] at (0.635, 2.732) {};
    \draw[shift={(0.071, 2.732)}, yscale=-1] (0, 0) rectangle (2.258, -2.258);
    \draw[shift={(0.071, 3.861)}, yscale=-1] (0, 0) -- (2.258, 0);
    \draw[shift={(1.2, 2.732)}, yscale=-1] (0, 0) -- (0, -2.258);
    \draw[shift={(0.071, 2.732)}, yscale=-1] (0, 0) -- (0, 1.129) -- (1.129, 1.129) -- (1.129, 0);
    \draw[shift={(1.2, 1.603)}, yscale=-1] (0, 0) -- (1.129, 0) -- (1.129, -1.129);
    \draw[shift={(2.329, 1.603)}, yscale=-1] (0, 0) -- (1.129, 0) -- (1.129, -1.129) -- (0, -1.129);
    \draw[shift={(3.457, 2.732)}, yscale=-1] (0, 0) -- (0, -1.129) -- (-1.129, -1.129);
    \draw[shift={(3.457, 3.861)}, yscale=-1] (0, 0) -- (0, -1.129) -- (-1.129, -1.129);
    \node[Mark, Mark_disk, BrickRed!90] at (1.2, 2.167) {};
    \node[Mark, Mark_disk] at (2.893, 3.861) {};
    \node[Mark, Mark_disk] at (0.635, 1.603) {};
    \node[Mark, Mark_disk] at (1.764, 1.603) {};
    \draw (1.2, 1.61) -- (1.2, 0.481);
    \draw (2.329, 1.61) -- (2.329, 0.481);
    \draw (3.457, 1.61) -- (3.457, 0.481);
    \node[Mark, Mark_disk, BrickRed!40] at (4.586, 2.167) {};
    \node[Mark, Mark_disk, BrickRed!40] at (4.586, 3.296) {};
    \node[Mark, Mark_disk, BrickRed!40] at (4.586, 4.425) {};
    \node[Mark, Mark_disk, BrickRed!40] at (0.635, 0.474) {};
    \node[Mark, Mark_disk, BrickRed!40] at (1.764, 0.474) {};
    \node[Mark, Mark_disk, BrickRed!40] at (2.893, 0.474) {};
    \node[Mark, Mark_disk, black!50] at (4.022, 0.474) {};
    \node[Mark, Mark_disk, black!50] at (4.586, 1.038) {};
    \node[Mark, Mark_disk] at (4.022, 2.732) {};
    \node[Mark, Mark_disk] at (4.022, 3.861) {};
    \node[Mark, Mark_disk] at (4.022, 4.99) {};
    \node[Mark, Mark_disk] at (0.071, 1.038) {};
    \node[Mark, Mark_disk] at (1.2, 1.038) {};
    \node[Mark, Mark_disk] at (2.329, 1.038) {};
    \node[Mark, Mark_disk] at (3.457, 1.038) {};
    \node[Mark, Mark_disk] at (1.764, 2.732) {};
    \node[Mark, Mark_disk, BrickRed!90] at (2.329, 3.296) {};
    \node[Mark, Mark_disk, BrickRed!90] at (0.071, 3.296) {};
    \node[Mark, Mark_disk] at (0.635, 3.861) {};
    \node[Mark, Mark_disk, BrickRed!90] at (0.071, 4.425) {};
    \node[Mark, Mark_disk, BrickRed!90] at (0.635, 4.99) {};
    \node[Mark, Mark_disk, BrickRed!90] at (1.764, 4.99) {};
    \node[Mark, Mark_disk, BrickRed!90] at (1.2, 4.425) {};
    \node[Mark, Mark_disk] at (1.764, 3.861) {};
    \node[Mark, Mark_disk, BrickRed!90] at (2.329, 4.425) {};
    \node[Mark, Mark_disk, BrickRed!90] at (0.071, 2.167) {};
    \node[Mark, Mark_disk, BrickRed!90] at (2.329, 2.167) {};
    \node[Mark, Mark_disk, BrickRed!90] at (2.893, 4.99) {};
    \node[Mark, Mark_disk] at (2.893, 2.732) {};
    \node[Mark, Mark_disk] at (2.893, 1.603) {};
    \node[Mark, Mark_disk, BrickRed!90] at (3.457, 2.167) {};
    \node[Mark, Mark_disk, BrickRed!90] at (3.457, 3.296) {};
    \node[Mark, Mark_disk, BrickRed!90] at (3.457, 4.425) {};
    \node[Mark, Mark_disk] at (4.022, 1.603) {};
    \filldraw[draw=Mulberry!80, fill=Lavender, opacity=0.5] (0, 5.059) -- (0, 5.059) -- (0, 4.929) -- (0, 4.929) -- (0, 4.516) -- (0.13, 4.516) -- (0.13, 4.929) -- (0.543, 4.929) -- (0.543, 5.059) -- (0.13, 5.059) -- (0.13, 5.059) -- (0, 5.059) -- cycle;
    \filldraw[draw=Mulberry!80, fill=Lavender, opacity=0.5] (1.129, 5.059) -- (0.715, 5.059) -- (0.715, 4.929) -- (1.129, 4.929) -- (1.129, 4.516) -- (1.259, 4.516) -- (1.259, 4.929) -- (1.672, 4.929) -- (1.672, 5.059) -- (1.259, 5.059) -- (1.259, 5.059) -- (1.129, 5.059) -- cycle;
    \filldraw[draw=Mulberry!80, fill=Lavender, opacity=0.5] (1.129, 3.93) -- (1.129, 3.93) -- (1.129, 3.8) -- (1.129, 3.8) -- (1.129, 3.387) -- (1.259, 3.387) -- (1.259, 3.8) -- (1.259, 3.8) -- (1.259, 3.939) -- (1.259, 3.93) -- (1.259, 4.344) -- (1.129, 4.344) -- cycle;
    \filldraw[draw=Mulberry!80, fill=Lavender, opacity=0.5] (0, 1.672) -- (0, 1.543) -- (0.13, 1.543) -- (0.13, 1.672) -- (0.13, 2.086) -- (0, 2.086) -- cycle;
    \filldraw[shift={(3.523, 5.059)}, xscale=-1, draw=Mulberry!80, fill=Lavender, opacity=0.5] (0, 0) -- (0, 0) -- (0, -0.13) -- (0, -0.13) -- (0, -0.543) -- (0.13, -0.543) -- (0.13, -0.13) -- (0.543, -0.13) -- (0.543, 0) -- (0.13, 0) -- (0.13, 0) -- (0, 0) -- cycle;
    \filldraw[draw=Mulberry!80, fill=Lavender, opacity=0.5] (2.262, 5.061) -- (1.848, 5.061) -- (1.848, 4.932) -- (2.262, 4.932) -- (2.262, 4.518) -- (2.391, 4.518) -- (2.391, 4.932) -- (2.805, 4.932) -- (2.805, 5.061) -- (2.391, 5.061) -- (2.391, 5.061) -- (2.262, 5.061) -- cycle;
    \filldraw[draw=Mulberry!80, fill=Lavender, opacity=0.5] (0, 3.93) -- (0, 3.93) -- (0, 3.8) -- (0, 3.8) -- (0, 3.387) -- (0.13, 3.387) -- (0.13, 3.8) -- (0.13, 3.8) -- (0.13, 3.939) -- (0.13, 3.93) -- (0.13, 4.344) -- (0, 4.344) -- cycle;
    \filldraw[draw=Mulberry!80, fill=Lavender, opacity=0.5] (3.387, 3.93) -- (3.387, 3.93) -- (3.387, 3.8) -- (3.387, 3.8) -- (3.387, 3.387) -- (3.516, 3.387) -- (3.516, 3.8) -- (3.516, 3.8) -- (3.516, 3.939) -- (3.516, 3.93) -- (3.516, 4.344) -- (3.387, 4.344) -- cycle;
    \filldraw[draw=Mulberry!80, fill=Lavender, opacity=0.5] (2.258, 3.93) -- (2.258, 3.93) -- (2.258, 3.8) -- (2.258, 3.8) -- (2.258, 3.387) -- (2.387, 3.387) -- (2.387, 3.8) -- (2.387, 3.8) -- (2.387, 3.939) -- (2.387, 3.93) -- (2.387, 4.344) -- (2.258, 4.344) -- cycle;
    \filldraw[draw=Mulberry!80, fill=Lavender, opacity=0.5] (1.129, 2.801) -- (1.129, 2.801) -- (1.129, 2.671) -- (1.129, 2.671) -- (1.129, 2.258) -- (1.259, 2.258) -- (1.259, 2.671) -- (1.259, 2.671) -- (1.259, 2.81) -- (1.259, 2.801) -- (1.259, 3.215) -- (1.129, 3.215) -- cycle;
    \filldraw[draw=Mulberry!80, fill=Lavender, opacity=0.5] (0, 2.801) -- (0, 2.801) -- (0, 2.671) -- (0, 2.671) -- (0, 2.258) -- (0.13, 2.258) -- (0.13, 2.671) -- (0.13, 2.671) -- (0.13, 2.81) -- (0.13, 2.801) -- (0.13, 3.215) -- (0, 3.215) -- cycle;
    \filldraw[draw=Mulberry!80, fill=Lavender, opacity=0.5] (3.387, 2.801) -- (3.387, 2.801) -- (3.387, 2.671) -- (3.387, 2.671) -- (3.387, 2.258) -- (3.516, 2.258) -- (3.516, 2.671) -- (3.516, 2.671) -- (3.516, 2.81) -- (3.516, 2.801) -- (3.516, 3.215) -- (3.387, 3.215) -- cycle;
    \filldraw[draw=Mulberry!80, fill=Lavender, opacity=0.5] (2.258, 2.801) -- (2.258, 2.801) -- (2.258, 2.671) -- (2.258, 2.671) -- (2.258, 2.258) -- (2.387, 2.258) -- (2.387, 2.671) -- (2.387, 2.671) -- (2.387, 2.81) -- (2.387, 2.801) -- (2.387, 3.215) -- (2.258, 3.215) -- cycle;
    \filldraw[draw=Mulberry!80, fill=Lavender, opacity=0.5] (1.129, 1.672) -- (1.129, 1.543) -- (1.259, 1.543) -- (1.259, 1.672) -- (1.259, 2.086) -- (1.129, 2.086) -- cycle;
    \filldraw[draw=Mulberry!80, fill=Lavender, opacity=0.5] (2.258, 1.675) -- (2.258, 1.546) -- (2.387, 1.546) -- (2.387, 1.675) -- (2.387, 2.089) -- (2.258, 2.089) -- cycle;
    \filldraw[draw=Mulberry!80, fill=Lavender, opacity=0.5] (3.387, 1.675) -- (3.387, 1.546) -- (3.516, 1.546) -- (3.516, 1.675) -- (3.516, 2.089) -- (3.387, 2.089) -- cycle;
    \filldraw[draw=RoyalBlue!80, fill=Lavender, opacity=0.3] (4.516, 5.059) -- (4.516, 5.059) -- (4.516, 4.929) -- (4.516, 4.929) -- (4.516, 4.516) -- (4.645, 4.516) -- (4.645, 4.929) -- (5.059, 4.929) -- (5.059, 5.059) -- (4.645, 5.059) -- (4.645, 5.059) -- (4.516, 5.059) -- cycle;
    \filldraw[draw=RoyalBlue!80, fill=Lavender, opacity=0.3] (4.516, 3.93) -- (4.516, 3.93) -- (4.516, 3.8) -- (4.516, 3.8) -- (4.516, 3.387) -- (4.645, 3.387) -- (4.645, 3.8) -- (4.645, 3.8) -- (4.645, 3.939) -- (4.645, 3.93) -- (4.645, 4.344) -- (4.516, 4.344) -- cycle;
    \filldraw[draw=RoyalBlue!80, fill=Lavender, opacity=0.3] (4.516, 2.801) -- (4.516, 2.801) -- (4.516, 2.671) -- (4.516, 2.671) -- (4.516, 2.258) -- (4.645, 2.258) -- (4.645, 2.671) -- (4.645, 2.671) -- (4.645, 2.81) -- (4.645, 2.801) -- (4.645, 3.215) -- (4.516, 3.215) -- cycle;
    \filldraw[draw=RoyalBlue!80, fill=Lavender, opacity=0.3] (4.516, 1.672) -- (4.516, 1.543) -- (4.645, 1.543) -- (4.645, 1.672) -- (4.645, 2.086) -- (4.516, 2.086) -- cycle;
    \filldraw[draw=RoyalBlue!80, fill=Lavender, opacity=0.3] (0, 0.543) -- (0, 0.543) -- (0, 0.414) -- (0, 0.414) -- (0, 0) -- (0.13, 0) -- (0.13, 0.414) -- (0.543, 0.414) -- (0.543, 0.543) -- (0.13, 0.543) -- (0.13, 0.543) -- (0, 0.543) -- cycle;
    \filldraw[draw=RoyalBlue!80, fill=Lavender, opacity=0.3] (1.129, 0.543) -- (0.715, 0.543) -- (0.715, 0.414) -- (1.129, 0.414) -- (1.129, 0) -- (1.259, 0) -- (1.259, 0.414) -- (1.672, 0.414) -- (1.672, 0.543) -- (1.259, 0.543) -- (1.259, 0.543) -- (1.129, 0.543) -- cycle;
    \filldraw[shift={(3.523, 0.544)}, xscale=-1, draw=RoyalBlue!80, fill=Lavender, opacity=0.3] (0, 0) -- (0, 0) -- (0, -0.13) -- (0, -0.13) -- (0, -0.543) -- (0.13, -0.543) -- (0.13, -0.13) -- (0.543, -0.13) -- (0.543, 0) -- (0.13, 0) -- (0.13, 0) -- (0, 0) -- cycle;
    \filldraw[draw=RoyalBlue!80, fill=Lavender, opacity=0.3] (2.262, 0.546) -- (1.848, 0.546) -- (1.848, 0.416) -- (2.262, 0.416) -- (2.262, 0.003) -- (2.391, 0.003) -- (2.391, 0.416) -- (2.805, 0.416) -- (2.805, 0.546) -- (2.391, 0.546) -- (2.391, 0.546) -- (2.262, 0.546) -- cycle;
    \node[Mark, Mark_disk, BrickRed!90] at (1.2, 3.296) {};
\end{tikzpicture}

%% file: Figures/finalhamiltonian2D.tikz
\begin{tikzpicture}[scale=1]
    \filldraw[fill=RoyalBlue!80, fill opacity=0.3] (1.363, 3.221) .. controls (1.363, 3.315) and (1.41, 3.362) .. (1.504, 3.362) .. controls (1.598, 3.362) and (1.645, 3.315) .. (1.645, 3.221) -- (1.645, 2.092) .. controls (1.645, 1.998) and (1.598, 1.951) .. (1.504, 1.951) .. controls (1.41, 1.951) and (1.363, 1.998) .. (1.363, 2.092) -- cycle;
    \filldraw[fill=RoyalBlue!80, fill opacity=0.3] (2.492, 3.221) .. controls (2.492, 3.315) and (2.539, 3.362) .. (2.633, 3.362) .. controls (2.727, 3.362) and (2.774, 3.315) .. (2.774, 3.221) -- (2.774, 2.092) .. controls (2.774, 1.998) and (2.727, 1.951) .. (2.633, 1.951) .. controls (2.539, 1.951) and (2.492, 1.998) .. (2.492, 2.092) -- cycle;
    \filldraw[fill=RoyalBlue!80, fill opacity=0.3] (1.927, 3.785) .. controls (1.928, 3.879) and (1.975, 3.926) .. (2.069, 3.926) .. controls (2.163, 3.926) and (2.21, 3.879) .. (2.21, 3.785) -- (2.21, 2.656) .. controls (2.21, 2.562) and (2.163, 2.515) .. (2.068, 2.515) .. controls (1.974, 2.515) and (1.927, 2.562) .. (1.927, 2.656) -- cycle;
    \filldraw[fill=RoyalBlue!80, fill opacity=0.3] (3.056, 3.785) .. controls (3.056, 3.879) and (3.104, 3.926) .. (3.198, 3.926) .. controls (3.292, 3.926) and (3.339, 3.879) .. (3.338, 3.785) -- (3.338, 2.656) .. controls (3.338, 2.562) and (3.291, 2.515) .. (3.197, 2.515) .. controls (3.103, 2.515) and (3.056, 2.562) .. (3.056, 2.656) -- cycle;
    \filldraw[fill=RoyalBlue!80, fill opacity=0.3] (3.621, 4.35) .. controls (3.621, 4.443) and (3.668, 4.49) .. (3.762, 4.49) .. controls (3.856, 4.49) and (3.903, 4.443) .. (3.903, 4.35) -- (3.903, 3.221) .. controls (3.903, 3.127) and (3.856, 3.08) .. (3.762, 3.08) .. controls (3.668, 3.08) and (3.621, 3.127) .. (3.621, 3.221) -- cycle;
    \filldraw[fill=RoyalBlue!80, fill opacity=0.3] (3.621, 3.221) .. controls (3.621, 3.315) and (3.668, 3.362) .. (3.762, 3.362) .. controls (3.856, 3.362) and (3.903, 3.315) .. (3.903, 3.221) -- (3.903, 2.092) .. controls (3.903, 1.998) and (3.856, 1.951) .. (3.762, 1.951) .. controls (3.668, 1.951) and (3.621, 1.998) .. (3.621, 2.092) -- cycle;
    \filldraw[fill=RoyalBlue!80, fill opacity=0.3] 
    (1.508, 2.082) circle[radius=0.211];
    \filldraw[fill=RoyalBlue!80, fill opacity=0.3] 
    (2.637, 2.082) circle[radius=0.211];
    \filldraw[fill=RoyalBlue!80, fill opacity=0.3] 
    (3.766, 2.082) circle[radius=0.211];
    \filldraw[fill=RoyalBlue!80, fill opacity=0.3] 
    (3.202, 3.776) circle[radius=0.211];
    \filldraw[draw=black!25, fill=RoyalBlue!80, fill opacity=0.1] (1.248, 0.242) -- (1.036, 0.454) .. controls (0.97, 0.52) and (0.903, 0.52) .. (0.837, 0.454) .. controls (0.77, 0.387) and (0.77, 0.321) .. (0.837, 0.254) -- (1.034, 0.058)(1.937, 0.025) -- (2.183, 0.271) .. controls (2.249, 0.338) and (2.249, 0.404) .. (2.183, 0.471) .. controls (2.116, 0.537) and (2.049, 0.537) .. (1.983, 0.471) -- (1.726, 0.214);
    \filldraw[fill=RoyalBlue!80, fill opacity=0.3] 
    (0.38, 2.082) circle[radius=0.211];
    \filldraw[fill=RoyalBlue!80, fill opacity=0.3] (0.234, 4.35) .. controls (0.234, 4.443) and (0.281, 4.49) .. (0.375, 4.49) .. controls (0.469, 4.49) and (0.516, 4.443) .. (0.516, 4.35) -- (0.516, 3.221) .. controls (0.516, 3.127) and (0.469, 3.08) .. (0.375, 3.08) .. controls (0.281, 3.08) and (0.234, 3.127) .. (0.234, 3.221) -- cycle;
    \filldraw[fill=RoyalBlue!80, fill opacity=0.3] (0.234, 3.221) .. controls (0.234, 3.315) and (0.281, 3.362) .. (0.375, 3.362) .. controls (0.469, 3.362) and (0.516, 3.315) .. (0.516, 3.221) -- (0.516, 2.092) .. controls (0.516, 1.998) and (0.469, 1.951) .. (0.375, 1.951) .. controls (0.281, 1.951) and (0.234, 1.998) .. (0.234, 2.092) -- cycle;
    \filldraw[fill=RoyalBlue!80, fill opacity=0.3] (1.363, 4.35) .. controls (1.363, 4.443) and (1.41, 4.49) .. (1.504, 4.49) .. controls (1.598, 4.49) and (1.645, 4.443) .. (1.645, 4.35) -- (1.645, 3.221) .. controls (1.645, 3.127) and (1.598, 3.08) .. (1.504, 3.08) .. controls (1.41, 3.08) and (1.363, 3.127) .. (1.363, 3.221) -- cycle;
    \filldraw[fill=RoyalBlue!80, fill opacity=0.3] (2.492, 4.35) .. controls (2.492, 4.443) and (2.539, 4.49) .. (2.633, 4.49) .. controls (2.727, 4.49) and (2.774, 4.443) .. (2.774, 4.35) -- (2.774, 3.221) .. controls (2.774, 3.127) and (2.727, 3.08) .. (2.633, 3.08) .. controls (2.539, 3.08) and (2.492, 3.127) .. (2.492, 3.221) -- cycle;
    \filldraw[fill=RoyalBlue!80, fill opacity=0.3] (0.798, 3.785) .. controls (0.799, 3.879) and (0.846, 3.926) .. (0.94, 3.926) .. controls (1.034, 3.926) and (1.081, 3.879) .. (1.081, 3.785) -- (1.081, 2.656) .. controls (1.081, 2.562) and (1.034, 2.515) .. (0.94, 2.515) .. controls (0.846, 2.515) and (0.798, 2.562) .. (0.798, 2.656) -- cycle;
    \filldraw[fill=RoyalBlue!80, fill opacity=0.3] (0.798, 2.656) .. controls (0.799, 2.75) and (0.846, 2.797) .. (0.94, 2.797) .. controls (1.034, 2.797) and (1.081, 2.75) .. (1.081, 2.656) -- (1.081, 1.527) .. controls (1.081, 1.433) and (1.034, 1.386) .. (0.94, 1.386) .. controls (0.846, 1.386) and (0.798, 1.433) .. (0.798, 1.527) -- cycle;
    \filldraw[fill=RoyalBlue!80, fill opacity=0.3] (1.927, 2.656) .. controls (1.928, 2.75) and (1.975, 2.797) .. (2.069, 2.797) .. controls (2.163, 2.797) and (2.21, 2.75) .. (2.21, 2.656) -- (2.21, 1.527) .. controls (2.21, 1.433) and (2.163, 1.386) .. (2.068, 1.386) .. controls (1.974, 1.386) and (1.927, 1.433) .. (1.927, 1.527) -- cycle;
    \filldraw[fill=RoyalBlue!80, fill opacity=0.3] (3.056, 2.656) .. controls (3.056, 2.75) and (3.104, 2.797) .. (3.198, 2.797) .. controls (3.292, 2.797) and (3.339, 2.75) .. (3.338, 2.656) -- (3.338, 1.527) .. controls (3.338, 1.433) and (3.291, 1.386) .. (3.197, 1.386) .. controls (3.103, 1.386) and (3.056, 1.433) .. (3.056, 1.527) -- cycle;
    \filldraw[fill=RoyalBlue!80, fill opacity=0.3] (4.185, 4.914) .. controls (4.185, 5.008) and (4.232, 5.055) .. (4.326, 5.055) .. controls (4.421, 5.055) and (4.468, 5.008) .. (4.467, 4.914) -- (4.467, 3.785) .. controls (4.467, 3.691) and (4.42, 3.644) .. (4.326, 3.644) .. controls (4.232, 3.644) and (4.185, 3.691) .. (4.185, 3.785) -- cycle;
    \filldraw[fill=RoyalBlue!80, fill opacity=0.3] (4.185, 3.785) .. controls (4.185, 3.879) and (4.232, 3.926) .. (4.326, 3.926) .. controls (4.421, 3.926) and (4.468, 3.879) .. (4.467, 3.785) -- (4.467, 2.656) .. controls (4.467, 2.562) and (4.42, 2.515) .. (4.326, 2.515) .. controls (4.232, 2.515) and (4.185, 2.562) .. (4.185, 2.656) -- cycle;
    \filldraw[fill=RoyalBlue!80, fill opacity=0.3] (4.185, 2.656) .. controls (4.185, 2.75) and (4.232, 2.797) .. (4.326, 2.797) .. controls (4.421, 2.797) and (4.468, 2.75) .. (4.467, 2.656) -- (4.467, 1.527) .. controls (4.467, 1.433) and (4.42, 1.386) .. (4.326, 1.386) .. controls (4.232, 1.386) and (4.185, 1.433) .. (4.185, 1.527) -- cycle;
    \filldraw[shift={(3.08, 4.831)}, rotate=45, fill=RoyalBlue!80, fill opacity=0.3] (0, 0) .. controls (0, 0.094) and (0.047, 0.141) .. (0.141, 0.141) .. controls (0.235, 0.141) and (0.282, 0.094) .. (0.282, 0) -- (0.282, -0.847) .. controls (0.282, -0.941) and (0.235, -0.988) .. (0.141, -0.988) .. controls (0.047, -0.988) and (0, -0.941) .. (0, -0.847) -- cycle;
    \filldraw[shift={(0.258, 4.432)}, rotate=-45, yscale=-1, fill=RoyalBlue!80, fill opacity=0.3] (0, 0) .. controls (0, 0.094) and (0.047, 0.141) .. (0.141, 0.141) .. controls (0.235, 0.141) and (0.282, 0.094) .. (0.282, 0) -- (0.282, -0.847) .. controls (0.282, -0.941) and (0.235, -0.988) .. (0.141, -0.988) .. controls (0.047, -0.988) and (0, -0.941) .. (0, -0.847) -- cycle;
    \filldraw[fill=RoyalBlue!80, fill opacity=0.3] (1.621, 1.046) .. controls (1.687, 0.979) and (1.687, 0.913) .. (1.621, 0.846) .. controls (1.554, 0.779) and (1.488, 0.779) .. (1.421, 0.846) -- (1.023, 1.245) .. controls (0.966, 1.294) and (0.91, 1.294) .. (0.856, 1.244) -- (0.475, 0.863) .. controls (0.408, 0.796) and (0.342, 0.796) .. (0.275, 0.863) .. controls (0.209, 0.929) and (0.209, 0.996) .. (0.275, 1.062) -- (0.774, 1.562) .. controls (0.885, 1.633) and (0.995, 1.633) .. (1.105, 1.562) -- cycle;
    \filldraw[fill=RoyalBlue!80, fill opacity=0.3] (2.75, 1.046) .. controls (2.816, 0.979) and (2.816, 0.913) .. (2.75, 0.846) .. controls (2.683, 0.779) and (2.617, 0.779) .. (2.55, 0.846) -- (2.152, 1.245) .. controls (2.095, 1.294) and (2.039, 1.294) .. (1.984, 1.244) -- (1.604, 0.863) .. controls (1.537, 0.796) and (1.471, 0.796) .. (1.404, 0.863) .. controls (1.338, 0.929) and (1.338, 0.996) .. (1.404, 1.062) -- (1.903, 1.562) .. controls (2.013, 1.633) and (2.124, 1.633) .. (2.234, 1.562) -- cycle;
    \filldraw[fill=RoyalBlue!80, fill opacity=0.3] (3.879, 1.046) .. controls (3.945, 0.979) and (3.945, 0.913) .. (3.879, 0.846) .. controls (3.812, 0.779) and (3.746, 0.779) .. (3.679, 0.846) -- (3.28, 1.245) .. controls (3.223, 1.294) and (3.168, 1.294) .. (3.113, 1.244) -- (2.733, 0.863) .. controls (2.666, 0.796) and (2.599, 0.796) .. (2.533, 0.863) .. controls (2.466, 0.929) and (2.466, 0.996) .. (2.533, 1.062) -- (3.032, 1.562) .. controls (3.142, 1.633) and (3.253, 1.633) .. (3.363, 1.562) -- cycle;
    \filldraw[shift={(2.185, 4.791)}, yscale=-1, fill=RoyalBlue!80, fill opacity=0.3] (0, 0) .. controls (0.067, -0.067) and (0.067, -0.133) .. (0, -0.2) .. controls (-0.067, -0.266) and (-0.133, -0.266) .. (-0.2, -0.2) -- (-0.598, 0.199) .. controls (-0.655, 0.249) and (-0.711, 0.248) .. (-0.765, 0.198) -- (-1.146, -0.183) .. controls (-1.213, -0.249) and (-1.279, -0.249) .. (-1.346, -0.183) .. controls (-1.412, -0.116) and (-1.412, -0.05) .. (-1.346, 0.017) -- (-0.847, 0.516) .. controls (-0.736, 0.587) and (-0.626, 0.587) .. (-0.516, 0.516) -- cycle;
    \filldraw[shift={(3.314, 4.791)}, yscale=-1, fill=RoyalBlue!80, fill opacity=0.3] (0, 0) .. controls (0.067, -0.067) and (0.067, -0.133) .. (0, -0.2) .. controls (-0.067, -0.266) and (-0.133, -0.266) .. (-0.2, -0.2) -- (-0.598, 0.199) .. controls (-0.655, 0.249) and (-0.711, 0.248) .. (-0.765, 0.198) -- (-1.146, -0.183) .. controls (-1.213, -0.249) and (-1.279, -0.249) .. (-1.346, -0.183) .. controls (-1.412, -0.116) and (-1.412, -0.05) .. (-1.346, 0.017) -- (-0.847, 0.516) .. controls (-0.736, 0.587) and (-0.626, 0.587) .. (-0.516, 0.516) -- cycle;
    \node[Mark, Mark_disk] at (0.94, 2.649) {};
    \draw[shift={(0.375, 2.649)}, yscale=-1] (0, 0) rectangle (2.258, -2.258);
    \draw[shift={(0.375, 3.778)}, yscale=-1] (0, 0) -- (2.258, 0);
    \draw[shift={(1.504, 2.649)}, yscale=-1] (0, 0) -- (0, -2.258);
    \draw[shift={(0.375, 2.649)}, yscale=-1] (0, 0) -- (0, 1.129) -- (1.129, 1.129) -- (1.129, 0);
    \draw[shift={(1.504, 1.52)}, yscale=-1] (0, 0) -- (1.129, 0) -- (1.129, -1.129);
    \draw[shift={(2.633, 1.52)}, yscale=-1] (0, 0) -- (1.129, 0) -- (1.129, -1.129) -- (0, -1.129);
    \draw[shift={(3.762, 2.649)}, yscale=-1] (0, 0) -- (0, -1.129) -- (-1.129, -1.129);
    \draw[shift={(3.762, 3.778)}, yscale=-1] (0, 0) -- (0, -1.129) -- (-1.129, -1.129);
    \node[Mark, Mark_disk, BrickRed!90] at (1.504, 2.084) {};
    \node[Mark, Mark_disk] at (3.197, 3.778) {};
    \node[Mark, Mark_disk] at (0.94, 1.52) {};
    \node[Mark, Mark_disk] at (2.068, 1.52) {};
    \draw (1.504, 1.527) -- (1.504, 0.398);
    \draw (2.633, 1.527) -- (2.633, 0.398);
    \draw (3.762, 1.527) -- (3.762, 0.398);
    \node[Mark, Mark_disk, black!50] at (4.326, 0.391) {};
    \node[Mark, Mark_disk] at (4.326, 2.649) {};
    \node[Mark, Mark_disk] at (4.326, 3.778) {};
    \node[Mark, Mark_disk] at (4.326, 4.907) {};
    \node[Mark, Mark_disk] at (0.375, 0.955) {};
    \node[Mark, Mark_disk] at (1.504, 0.955) {};
    \node[Mark, Mark_disk] at (2.633, 0.955) {};
    \node[Mark, Mark_disk] at (3.762, 0.955) {};
    \node[Mark, Mark_disk] at (2.068, 2.649) {};
    \node[Mark, Mark_disk, BrickRed!90] at (2.633, 3.213) {};
    \node[Mark, Mark_disk, BrickRed!90] at (0.375, 3.213) {};
    \node[Mark, Mark_disk] at (0.94, 3.778) {};
    \node[Mark, Mark_disk, BrickRed!90] at (0.375, 4.342) {};
    \node[Mark, Mark_disk, BrickRed!90] at (0.94, 4.907) {};
    \node[Mark, Mark_disk, BrickRed!90] at (2.068, 4.907) {};
    \node[Mark, Mark_disk, BrickRed!90] at (1.504, 4.342) {};
    \node[Mark, Mark_disk] at (2.068, 3.778) {};
    \node[Mark, Mark_disk, BrickRed!90] at (2.633, 4.342) {};
    \node[Mark, Mark_disk, BrickRed!90] at (0.375, 2.084) {};
    \node[Mark, Mark_disk, BrickRed!90] at (2.633, 2.084) {};
    \node[Mark, Mark_disk, BrickRed!90] at (3.197, 4.907) {};
    \node[Mark, Mark_disk] at (3.197, 2.649) {};
    \node[Mark, Mark_disk] at (3.197, 1.52) {};
    \node[Mark, Mark_disk, BrickRed!90] at (3.762, 2.084) {};
    \node[Mark, Mark_disk, BrickRed!90] at (3.762, 3.213) {};
    \node[Mark, Mark_disk, BrickRed!90] at (3.762, 4.342) {};
    \node[Mark, Mark_disk] at (4.326, 1.52) {};
    \draw (3.762, 4.907) -- (4.891, 4.907) -- (4.891, 0.398) -- (0.375, 0.398) -- (0.375, 1.527);
    \draw (3.762, 1.52) -- (4.891, 1.519);
    \draw (3.762, 2.649) -- (4.891, 2.647);
    \draw (3.762, 3.778) -- (4.891, 3.777);
    \begin{scope}[shift={(-1.318, -5.246)}]
        \node[Mark, Mark_disk, BrickRed!40] at (2.258, 5.637) {};
        \node[Mark, Mark_disk, BrickRed!40] at (3.387, 5.637) {};
        \node[Mark, Mark_disk, BrickRed!40] at (4.516, 5.637) {};
    \end{scope}
    \node[Mark, Mark_disk, BrickRed!40] at (4.891, 2.084) {};
    \node[Mark, Mark_disk, BrickRed!40] at (4.891, 3.213) {};
    \node[Mark, Mark_disk, BrickRed!40] at (4.891, 4.342) {};
    \node[Mark, Mark_disk, black!50] at (4.891, 0.955) {};
    \node[Mark, Mark_disk, BrickRed!90] at (1.504, 3.213) {};
    \filldraw[fill=RoyalBlue!80, fill opacity=0.3] 
    (0.944, 3.776) circle[radius=0.211];
    \filldraw[fill=RoyalBlue!80, fill opacity=0.3] 
    (2.073, 3.776) circle[radius=0.211];
    \filldraw[draw=black!25, fill=RoyalBlue!80, fill opacity=0.1] 
    (4.895, 2.082) circle[radius=0.211];
    \filldraw[draw=black!25, fill=RoyalBlue!80, fill opacity=0.1] (4.75, 4.35) .. controls (4.75, 4.443) and (4.797, 4.49) .. (4.891, 4.49) .. controls (4.985, 4.49) and (5.032, 4.443) .. (5.032, 4.35) -- (5.032, 3.221) .. controls (5.032, 3.127) and (4.985, 3.08) .. (4.891, 3.08) .. controls (4.797, 3.08) and (4.75, 3.127) .. (4.75, 3.221) -- cycle;
    \filldraw[draw=black!25, fill=RoyalBlue!80, fill opacity=0.1] (4.75, 3.221) .. controls (4.75, 3.315) and (4.797, 3.362) .. (4.891, 3.362) .. controls (4.985, 3.362) and (5.032, 3.315) .. (5.032, 3.221) -- (5.032, 2.092) .. controls (5.032, 1.998) and (4.985, 1.951) .. (4.891, 1.951) .. controls (4.797, 1.951) and (4.75, 1.998) .. (4.75, 2.092) -- cycle;
    \filldraw[draw=black!25, fill=RoyalBlue!80, fill opacity=0.1] (5.031, 4.69) -- (4.774, 4.432) .. controls (4.707, 4.366) and (4.707, 4.299) .. (4.774, 4.233) .. controls (4.84, 4.166) and (4.907, 4.166) .. (4.973, 4.233) -- (5.219, 4.478);
    \filldraw[draw=black!25, fill=RoyalBlue!80, fill opacity=0.1] (2.36, 0.249) -- (2.148, 0.461) .. controls (2.081, 0.527) and (2.015, 0.527) .. (1.948, 0.461) .. controls (1.881, 0.394) and (1.881, 0.327) .. (1.948, 0.261) -- (2.145, 0.064)(3.048, 0.032) -- (3.294, 0.278) .. controls (3.36, 0.344) and (3.36, 0.411) .. (3.294, 0.477) .. controls (3.227, 0.544) and (3.161, 0.544) .. (3.094, 0.477) -- (2.838, 0.221);
    \filldraw[shift={(3.338, 0.02)}, yscale=-1, draw=black!25, fill=RoyalBlue!80, fill opacity=0.1] (0, 0) -- (-0.258, -0.258) .. controls (-0.324, -0.324) and (-0.324, -0.391) .. (-0.258, -0.457) .. controls (-0.191, -0.524) and (-0.125, -0.524) .. (-0.058, -0.457) -- (0.188, -0.212);
    \filldraw[shift={(0.79, 0)}, scale=-1, draw=black!25, fill=RoyalBlue!80, fill opacity=0.1] (0, 0) -- (-0.258, -0.258) .. controls (-0.324, -0.324) and (-0.324, -0.391) .. (-0.258, -0.457) .. controls (-0.191, -0.524) and (-0.125, -0.524) .. (-0.058, -0.457) -- (0.188, -0.212);
    \filldraw[fill=RoyalBlue!80, fill opacity=0.3] (3.728, 1.128) -- (4.161, 1.562) -- (4.161, 1.562) .. controls (4.271, 1.633) and (4.382, 1.633) .. (4.492, 1.562) -- (4.925, 1.128) .. controls (4.996, 1.018) and (4.995, 0.908) .. (4.924, 0.797) -- (4.492, 0.364) .. controls (4.381, 0.295) and (4.271, 0.295) .. (4.161, 0.366) -- (3.728, 0.798) .. controls (3.657, 0.908) and (3.658, 1.018) .. (3.729, 1.128) -- cycle;
    \filldraw[draw=black!50, fill=white] (3.961, 0.963) -- (4.326, 1.328) -- (4.691, 0.963) -- (4.326, 0.598) -- (3.961, 0.963);
    \filldraw[draw=black!25, fill=RoyalBlue!80, fill opacity=0.1] (4.924, 5.313) -- (4.492, 4.88) .. controls (4.381, 4.81) and (4.271, 4.811) .. (4.161, 4.882) -- (3.728, 5.313);
    \filldraw[draw=black!10, fill=white] (4.525, 5.313) -- (4.326, 5.114) -- (4.127, 5.313);
    \filldraw[draw=black!25, fill=RoyalBlue!80, fill opacity=0.1] (0.001, 1.55) -- (0.433, 1.118) .. controls (0.503, 1.007) and (0.502, 0.897) .. (0.431, 0.787) -- (0, 0.354);
    \filldraw[draw=black!10, fill=white] (0.001, 1.151) -- (0.2, 0.952) -- (0, 0.753);
\end{tikzpicture}

%% file: Figures/intermediatestars.tikz
\begin{tikzpicture}[scale=1]
    \filldraw[fill=RoyalBlue!80, fill opacity=0.3] 
    (0.776, 4.933) circle[radius=0.211];
    \filldraw[fill=RoyalBlue!80, fill opacity=0.3] 
    (1.34, 4.369) circle[radius=0.211];
    \filldraw[fill=RoyalBlue!80, fill opacity=0.3] 
    (1.34, 2.111) circle[radius=0.211];
    \filldraw[fill=RoyalBlue!80, fill opacity=0.3] 
    (0.211, 2.111) circle[radius=0.211];
    \filldraw[fill=RoyalBlue!80, fill opacity=0.3] 
    (0.211, 3.24) circle[radius=0.211];
    \filldraw[fill=RoyalBlue!80, fill opacity=0.3] 
    (0.211, 4.369) circle[radius=0.211];
    \filldraw[fill=RoyalBlue!80, fill opacity=0.3] (2.324, 4.378) .. controls (2.324, 4.472) and (2.371, 4.519) .. (2.465, 4.519) .. controls (2.559, 4.519) and (2.606, 4.472) .. (2.606, 4.378) -- (2.606, 3.25) .. controls (2.606, 3.155) and (2.559, 3.108) .. (2.465, 3.108) .. controls (2.371, 3.108) and (2.324, 3.155) .. (2.324, 3.25) -- cycle;
    \filldraw[fill=RoyalBlue!80, fill opacity=0.3] (2.324, 3.25) .. controls (2.324, 3.343) and (2.371, 3.39) .. (2.465, 3.39) .. controls (2.559, 3.39) and (2.606, 3.343) .. (2.606, 3.25) -- (2.606, 2.121) .. controls (2.606, 2.027) and (2.559, 1.98) .. (2.465, 1.98) .. controls (2.371, 1.98) and (2.324, 2.027) .. (2.324, 2.121) -- cycle;
    \filldraw[fill=RoyalBlue!80, fill opacity=0.3] (3.452, 4.378) .. controls (3.453, 4.472) and (3.5, 4.519) .. (3.594, 4.519) .. controls (3.688, 4.519) and (3.735, 4.472) .. (3.735, 4.378) -- (3.735, 3.25) .. controls (3.735, 3.155) and (3.688, 3.108) .. (3.594, 3.108) .. controls (3.499, 3.108) and (3.452, 3.155) .. (3.452, 3.25) -- cycle;
    \filldraw[fill=RoyalBlue!80, fill opacity=0.3] (3.452, 3.25) .. controls (3.453, 3.343) and (3.5, 3.39) .. (3.594, 3.39) .. controls (3.688, 3.39) and (3.735, 3.343) .. (3.735, 3.25) -- (3.735, 2.121) .. controls (3.735, 2.027) and (3.688, 1.98) .. (3.594, 1.98) .. controls (3.499, 1.98) and (3.452, 2.027) .. (3.452, 2.121) -- cycle;
    \filldraw[shift={(2.912, 4.86)}, rotate=45, fill=RoyalBlue!80, fill opacity=0.3] (0, 0) .. controls (0, 0.094) and (0.047, 0.141) .. (0.141, 0.141) .. controls (0.235, 0.141) and (0.282, 0.094) .. (0.282, 0) -- (0.282, -0.847) .. controls (0.282, -0.941) and (0.235, -0.988) .. (0.141, -0.988) .. controls (0.047, -0.988) and (0, -0.941) .. (0, -0.847) -- cycle;
    \filldraw[fill=RoyalBlue!80, fill opacity=0.3] 
    (2.469, 2.111) circle[radius=0.211];
    \filldraw[fill=RoyalBlue!80, fill opacity=0.3] 
    (3.598, 2.111) circle[radius=0.211];
    \filldraw[shift={(1.9, 4.802)}, rotate=90, fill=RoyalBlue!80, fill opacity=0.3] (0, 0) .. controls (0, 0.094) and (0.047, 0.141) .. (0.141, 0.141) .. controls (0.235, 0.141) and (0.282, 0.094) .. (0.282, 0) -- (0.282, -1.129) .. controls (0.282, -1.223) and (0.235, -1.27) .. (0.141, -1.27) .. controls (0.047, -1.27) and (0, -1.223) .. (0, -1.129) -- cycle;
    \filldraw[fill=RoyalBlue!80, fill opacity=0.3] (1.92, 4.723) -- (1.477, 4.281) -- (1.477, 2.121) .. controls (1.477, 2.027) and (1.43, 1.98) .. (1.336, 1.98) .. controls (1.242, 1.98) and (1.195, 2.027) .. (1.195, 2.121) -- (1.195, 4.28) -- (0.751, 4.723) -- (0.348, 4.32) -- (0.348, 2.121) .. controls (0.348, 2.027) and (0.301, 1.98) .. (0.207, 1.98) .. controls (0.113, 1.98) and (0.066, 2.027) .. (0.066, 2.121) -- (0.073, 4.442) -- (0.564, 4.935) .. controls (0.702, 5.034) and (0.827, 5.034) .. (0.938, 4.935) -- (1.252, 4.622) .. controls (1.306, 4.572) and (1.362, 4.571) .. (1.419, 4.621) -- (1.734, 4.935) .. controls (1.845, 5.034) and (1.969, 5.034) .. (2.107, 4.935) -- (2.582, 4.461) .. controls (2.648, 4.395) and (2.648, 4.328) .. (2.582, 4.262) .. controls (2.515, 4.195) and (2.449, 4.195) .. (2.382, 4.262) -- cycle;
    \filldraw[draw=black!25, fill=RoyalBlue!80, fill opacity=0.1] 
    (0.776, 0.418) circle[radius=0.211];
    \filldraw[draw=black!25, fill=RoyalBlue!80, fill opacity=0.1] (3.257, 0) -- (2.912, 0.345) .. controls (2.846, 0.411) and (2.846, 0.478) .. (2.912, 0.544) .. controls (2.979, 0.611) and (3.045, 0.611) .. (3.112, 0.544) -- (3.479, 0.177);
    \filldraw[shift={(1.9, 0.286)}, rotate=90, draw=black!25, fill=RoyalBlue!80, fill opacity=0.1] (0, 0) .. controls (0, 0.094) and (0.047, 0.141) .. (0.141, 0.141) .. controls (0.235, 0.141) and (0.282, 0.094) .. (0.282, 0) -- (0.282, -1.129) .. controls (0.282, -1.223) and (0.235, -1.27) .. (0.141, -1.27) .. controls (0.047, -1.27) and (0, -1.223) .. (0, -1.129) -- cycle;
    \filldraw[draw=black!25, fill=RoyalBlue!80, fill opacity=0.1] (2.123, 0.005) -- (1.92, 0.208) -- (1.717, 0.005)(0.955, 0.004) -- (0.751, 0.208) -- (0.552, 0.008)(0.152, 0.006) -- (0.564, 0.42) .. controls (0.702, 0.519) and (0.827, 0.519) .. (0.938, 0.42) -- (1.252, 0.106) .. controls (1.306, 0.056) and (1.362, 0.056) .. (1.419, 0.105) -- (1.734, 0.42) .. controls (1.845, 0.519) and (1.969, 0.519) .. (2.107, 0.42) -- (2.518, 0.004);
    \filldraw[draw=black!25, fill=RoyalBlue!80, fill opacity=0.1] 
    (4.727, 2.111) circle[radius=0.211];
    \filldraw[draw=black!25, fill=RoyalBlue!80, fill opacity=0.1] 
    (4.727, 3.24) circle[radius=0.211];
    \filldraw[draw=black!25, fill=RoyalBlue!80, fill opacity=0.1] 
    (4.727, 4.369) circle[radius=0.211];
    \filldraw[draw=black!25, fill=RoyalBlue!80, fill opacity=0.1] (5.267, 4.723) -- (4.864, 4.32) -- (4.864, 2.121) .. controls (4.864, 2.027) and (4.817, 1.98) .. (4.722, 1.98) .. controls (4.628, 1.98) and (4.581, 2.027) .. (4.581, 2.121) -- (4.588, 4.442) -- (5.08, 4.935);
    \filldraw[fill=RoyalBlue!80, fill opacity=0.3] 
    (1.34, 3.24) circle[radius=0.211];
    \node[Mark, Mark_disk] at (0.771, 2.678) {};
    \draw[shift={(0.207, 2.678)}, yscale=-1] (0, 0) rectangle (2.258, -2.258);
    \draw[shift={(0.207, 3.807)}, yscale=-1] (0, 0) -- (2.258, 0);
    \draw[shift={(1.336, 2.678)}, yscale=-1] (0, 0) -- (0, -2.258);
    \draw[shift={(0.207, 2.678)}, yscale=-1] (0, 0) -- (0, 1.129) -- (1.129, 1.129) -- (1.129, 0);
    \draw[shift={(1.336, 1.549)}, yscale=-1] (0, 0) -- (1.129, 0) -- (1.129, -1.129);
    \draw[shift={(2.465, 1.549)}, yscale=-1] (0, 0) -- (1.129, 0) -- (1.129, -1.129) -- (0, -1.129);
    \draw[shift={(3.594, 2.678)}, yscale=-1] (0, 0) -- (0, -1.129) -- (-1.129, -1.129);
    \draw[shift={(3.594, 3.807)}, yscale=-1] (0, 0) -- (0, -1.129) -- (-1.129, -1.129);
    \node[Mark, Mark_disk, BrickRed!90] at (1.336, 2.113) {};
    \node[Mark, Mark_disk] at (3.029, 3.807) {};
    \node[Mark, Mark_disk] at (0.771, 1.549) {};
    \node[Mark, Mark_disk] at (1.9, 1.549) {};
    \draw (1.336, 1.556) -- (1.336, 0.427);
    \draw (2.465, 1.556) -- (2.465, 0.427);
    \draw (3.594, 1.556) -- (3.594, 0.427);
    \node[Mark, Mark_disk, black!50] at (4.158, 0.42) {};
    \node[Mark, Mark_disk] at (4.158, 2.678) {};
    \node[Mark, Mark_disk] at (4.158, 3.807) {};
    \node[Mark, Mark_disk] at (4.158, 4.935) {};
    \node[Mark, Mark_disk] at (0.207, 0.984) {};
    \node[Mark, Mark_disk] at (1.336, 0.984) {};
    \node[Mark, Mark_disk] at (2.465, 0.984) {};
    \node[Mark, Mark_disk] at (3.594, 0.984) {};
    \node[Mark, Mark_disk] at (1.9, 2.678) {};
    \node[Mark, Mark_disk, BrickRed!90] at (2.465, 3.242) {};
    \node[Mark, Mark_disk, BrickRed!90] at (0.207, 3.242) {};
    \node[Mark, Mark_disk] at (0.771, 3.807) {};
    \node[Mark, Mark_disk, BrickRed!90] at (0.207, 4.371) {};
    \node[Mark, Mark_disk, BrickRed!90] at (0.771, 4.935) {};
    \node[Mark, Mark_disk, BrickRed!90] at (1.9, 4.935) {};
    \node[Mark, Mark_disk, BrickRed!90] at (1.336, 4.371) {};
    \node[Mark, Mark_disk] at (1.9, 3.807) {};
    \node[Mark, Mark_disk, BrickRed!90] at (2.465, 4.371) {};
    \node[Mark, Mark_disk, BrickRed!90] at (0.207, 2.113) {};
    \node[Mark, Mark_disk, BrickRed!90] at (2.465, 2.113) {};
    \node[Mark, Mark_disk, BrickRed!90] at (3.029, 4.935) {};
    \node[Mark, Mark_disk] at (3.029, 2.678) {};
    \node[Mark, Mark_disk] at (3.029, 1.549) {};
    \node[Mark, Mark_disk, BrickRed!90] at (3.594, 2.113) {};
    \node[Mark, Mark_disk, BrickRed!90] at (3.594, 3.242) {};
    \node[Mark, Mark_disk, BrickRed!90] at (3.594, 4.371) {};
    \node[Mark, Mark_disk] at (4.158, 1.549) {};
    \draw (3.594, 4.935) -- (4.722, 4.936) -- (4.722, 0.427) -- (0.207, 0.427) -- (0.207, 1.556);
    \draw (3.594, 1.549) -- (4.722, 1.548);
    \draw (3.594, 2.678) -- (4.722, 2.676);
    \draw (3.594, 3.807) -- (4.722, 3.806);
    \begin{scope}[shift={(-1.486, -5.217)}]
        \node[Mark, Mark_disk, BrickRed!40] at (2.258, 5.637) {};
        \node[Mark, Mark_disk, BrickRed!40] at (3.387, 5.637) {};
        \node[Mark, Mark_disk, BrickRed!40] at (4.516, 5.637) {};
    \end{scope}
    \node[Mark, Mark_disk, BrickRed!40] at (4.722, 2.113) {};
    \node[Mark, Mark_disk, BrickRed!40] at (4.722, 3.242) {};
    \node[Mark, Mark_disk, BrickRed!40] at (4.722, 4.371) {};
    \node[Mark, Mark_disk, black!50] at (4.722, 0.984) {};
    \node[Mark, Mark_disk, BrickRed!90] at (1.336, 3.242) {};
\end{tikzpicture}

%% file: Figures/intermediateplaquettes.tikz
\begin{tikzpicture}[scale=1]
    \filldraw[fill=RoyalBlue!80, fill opacity=0.3] (0.423, 3.394) .. controls (0.423, 3.488) and (0.471, 3.535) .. (0.565, 3.535) .. controls (0.659, 3.535) and (0.706, 3.488) .. (0.706, 3.394) -- (0.706, 2.265) .. controls (0.706, 2.171) and (0.659, 2.124) .. (0.564, 2.124) .. controls (0.47, 2.124) and (0.423, 2.171) .. (0.423, 2.265) -- cycle;
    \filldraw[fill=RoyalBlue!80, fill opacity=0.3] (0.423, 2.265) .. controls (0.423, 2.359) and (0.471, 2.406) .. (0.565, 2.406) .. controls (0.659, 2.406) and (0.706, 2.359) .. (0.706, 2.265) -- (0.706, 1.136) .. controls (0.706, 1.042) and (0.659, 0.995) .. (0.564, 0.995) .. controls (0.47, 0.995) and (0.423, 1.042) .. (0.423, 1.136) -- cycle;
    \filldraw[fill=RoyalBlue!80, fill opacity=0.3] (1.552, 3.394) .. controls (1.552, 3.488) and (1.599, 3.535) .. (1.694, 3.535) .. controls (1.788, 3.535) and (1.835, 3.488) .. (1.834, 3.394) -- (1.834, 2.265) .. controls (1.834, 2.171) and (1.787, 2.124) .. (1.693, 2.124) .. controls (1.599, 2.124) and (1.552, 2.171) .. (1.552, 2.265) -- cycle;
    \filldraw[fill=RoyalBlue!80, fill opacity=0.3] (1.552, 2.265) .. controls (1.552, 2.359) and (1.599, 2.406) .. (1.694, 2.406) .. controls (1.788, 2.406) and (1.835, 2.359) .. (1.834, 2.265) -- (1.834, 1.136) .. controls (1.834, 1.042) and (1.787, 0.995) .. (1.693, 0.995) .. controls (1.599, 0.995) and (1.552, 1.042) .. (1.552, 1.136) -- cycle;
    \filldraw[fill=RoyalBlue!80, fill opacity=0.3] (2.681, 3.394) .. controls (2.681, 3.488) and (2.728, 3.535) .. (2.822, 3.535) .. controls (2.917, 3.535) and (2.963, 3.488) .. (2.963, 3.394) -- (2.963, 2.265) .. controls (2.963, 2.171) and (2.916, 2.124) .. (2.822, 2.124) .. controls (2.728, 2.124) and (2.681, 2.171) .. (2.681, 2.265) -- cycle;
    \filldraw[fill=RoyalBlue!80, fill opacity=0.3] (2.681, 2.265) .. controls (2.681, 2.359) and (2.728, 2.406) .. (2.822, 2.406) .. controls (2.917, 2.406) and (2.963, 2.359) .. (2.963, 2.265) -- (2.963, 1.136) .. controls (2.963, 1.042) and (2.916, 0.995) .. (2.822, 0.995) .. controls (2.728, 0.995) and (2.681, 1.042) .. (2.681, 1.136) -- cycle;
    \filldraw[fill=RoyalBlue!80, fill opacity=0.3] (3.81, 3.394) .. controls (3.81, 3.488) and (3.857, 3.535) .. (3.951, 3.535) .. controls (4.045, 3.535) and (4.092, 3.488) .. (4.092, 3.394) -- (4.092, 2.265) .. controls (4.092, 2.171) and (4.045, 2.124) .. (3.951, 2.124) .. controls (3.857, 2.124) and (3.81, 2.171) .. (3.81, 2.265) -- cycle;
    \filldraw[fill=RoyalBlue!80, fill opacity=0.3] (3.81, 2.265) .. controls (3.81, 2.359) and (3.857, 2.406) .. (3.951, 2.406) .. controls (4.045, 2.406) and (4.092, 2.359) .. (4.092, 2.265) -- (4.092, 1.136) .. controls (4.092, 1.042) and (4.045, 0.995) .. (3.951, 0.995) .. controls (3.857, 0.995) and (3.81, 1.042) .. (3.81, 1.136) -- cycle;
    \filldraw[fill=RoyalBlue!80, fill opacity=0.3] (2.375, 0.655) .. controls (2.441, 0.588) and (2.441, 0.521) .. (2.375, 0.455) .. controls (2.308, 0.388) and (2.242, 0.388) .. (2.175, 0.455) -- (1.776, 0.854) .. controls (1.719, 0.903) and (1.664, 0.903) .. (1.609, 0.853) -- (1.228, 0.472) .. controls (1.162, 0.405) and (1.095, 0.405) .. (1.029, 0.472) .. controls (0.962, 0.538) and (0.962, 0.605) .. (1.029, 0.671) -- (1.528, 1.171) .. controls (1.638, 1.242) and (1.748, 1.242) .. (1.859, 1.171) -- cycle;
    \filldraw[fill=RoyalBlue!80, fill opacity=0.3] (3.504, 0.655) .. controls (3.57, 0.588) and (3.57, 0.521) .. (3.504, 0.455) .. controls (3.437, 0.388) and (3.371, 0.388) .. (3.304, 0.455) -- (2.905, 0.854) .. controls (2.848, 0.903) and (2.793, 0.903) .. (2.738, 0.853) -- (2.357, 0.472) .. controls (2.291, 0.405) and (2.224, 0.405) .. (2.158, 0.472) .. controls (2.091, 0.538) and (2.091, 0.605) .. (2.158, 0.671) -- (2.657, 1.171) .. controls (2.767, 1.242) and (2.877, 1.242) .. (2.988, 1.171) -- cycle;
    \node[Mark, Mark_disk] at (0.564, 2.258) {};
    \draw[shift={(0, 2.258)}, yscale=-1] (0, 0) rectangle (2.258, -2.258);
    \draw[shift={(0, 3.387)}, yscale=-1] (0, 0) -- (2.258, 0);
    \draw[shift={(1.129, 2.258)}, yscale=-1] (0, 0) -- (0, -2.258);
    \draw[shift={(0, 2.258)}, yscale=-1] (0, 0) -- (0, 1.129) -- (1.129, 1.129) -- (1.129, 0);
    \draw[shift={(1.129, 1.129)}, yscale=-1] (0, 0) -- (1.129, 0) -- (1.129, -1.129);
    \draw[shift={(2.258, 1.129)}, yscale=-1] (0, 0) -- (1.129, 0) -- (1.129, -1.129) -- (0, -1.129);
    \draw[shift={(3.387, 2.258)}, yscale=-1] (0, 0) -- (0, -1.129) -- (-1.129, -1.129);
    \draw[shift={(3.387, 3.387)}, yscale=-1] (0, 0) -- (0, -1.129) -- (-1.129, -1.129);
    \node[Mark, Mark_disk, BrickRed!90] at (1.129, 1.693) {};
    \node[Mark, Mark_disk] at (2.822, 3.387) {};
    \node[Mark, Mark_disk] at (0.564, 1.129) {};
    \node[Mark, Mark_disk] at (1.693, 1.129) {};
    \draw (1.129, 1.136) -- (1.129, 0.007);
    \draw (2.258, 1.136) -- (2.258, 0.007);
    \draw (3.387, 1.136) -- (3.387, 0.007);
    \node[Mark, Mark_disk, black!50] at (3.951, 0) {};
    \node[Mark, Mark_disk] at (3.951, 2.258) {};
    \node[Mark, Mark_disk] at (3.951, 3.387) {};
    \node[Mark, Mark_disk] at (3.951, 4.516) {};
    \node[Mark, Mark_disk] at (0, 0.564) {};
    \node[Mark, Mark_disk] at (1.129, 0.564) {};
    \node[Mark, Mark_disk] at (2.258, 0.564) {};
    \node[Mark, Mark_disk] at (3.387, 0.564) {};
    \node[Mark, Mark_disk] at (1.693, 2.258) {};
    \node[Mark, Mark_disk, BrickRed!90] at (2.258, 2.822) {};
    \node[Mark, Mark_disk, BrickRed!90] at (0, 2.822) {};
    \node[Mark, Mark_disk] at (0.564, 3.387) {};
    \node[Mark, Mark_disk, BrickRed!90] at (0, 3.951) {};
    \node[Mark, Mark_disk, BrickRed!90] at (0.564, 4.516) {};
    \node[Mark, Mark_disk, BrickRed!90] at (1.693, 4.516) {};
    \node[Mark, Mark_disk, BrickRed!90] at (1.129, 3.951) {};
    \node[Mark, Mark_disk] at (1.693, 3.387) {};
    \node[Mark, Mark_disk, BrickRed!90] at (2.258, 3.951) {};
    \node[Mark, Mark_disk, BrickRed!90] at (0, 1.693) {};
    \node[Mark, Mark_disk, BrickRed!90] at (2.258, 1.693) {};
    \node[Mark, Mark_disk, BrickRed!90] at (2.822, 4.516) {};
    \node[Mark, Mark_disk] at (2.822, 2.258) {};
    \node[Mark, Mark_disk] at (2.822, 1.129) {};
    \node[Mark, Mark_disk, BrickRed!90] at (3.387, 1.693) {};
    \node[Mark, Mark_disk, BrickRed!90] at (3.387, 2.822) {};
    \node[Mark, Mark_disk, BrickRed!90] at (3.387, 3.951) {};
    \node[Mark, Mark_disk] at (3.951, 1.129) {};
    \draw (3.387, 4.516) -- (4.516, 4.516) -- (4.516, 0.007) -- (0, 0.007) -- (0, 1.136);
    \draw (3.387, 1.129) -- (4.516, 1.128);
    \draw (3.387, 2.258) -- (4.516, 2.256);
    \draw (3.387, 3.387) -- (4.516, 3.386);
    \begin{scope}[shift={(-1.693, -5.637)}]
        \node[Mark, Mark_disk, BrickRed!40] at (2.258, 5.637) {};
        \node[Mark, Mark_disk, BrickRed!40] at (3.387, 5.637) {};
        \node[Mark, Mark_disk, BrickRed!40] at (4.516, 5.637) {};
    \end{scope}
    \node[Mark, Mark_disk, BrickRed!40] at (4.516, 1.693) {};
    \node[Mark, Mark_disk, BrickRed!40] at (4.516, 2.822) {};
    \node[Mark, Mark_disk, BrickRed!40] at (4.516, 3.951) {};
    \node[Mark, Mark_disk, black!50] at (4.516, 0.564) {};
    \node[Mark, Mark_disk, BrickRed!90] at (1.129, 2.822) {};
    \filldraw[fill=RoyalBlue!80, fill opacity=0.3] 
    (0.569, 3.385) circle[radius=0.211];
    \filldraw[fill=RoyalBlue!80, fill opacity=0.3] 
    (1.698, 3.385) circle[radius=0.211];
    \filldraw[fill=RoyalBlue!80, fill opacity=0.3] 
    (2.827, 3.385) circle[radius=0.211];
    \filldraw[shift={(0.448, 1.054)}, rotate=45, fill=RoyalBlue!80, fill opacity=0.3] (0, 0) .. controls (0, 0.094) and (0.047, 0.141) .. (0.141, 0.141) .. controls (0.235, 0.141) and (0.282, 0.094) .. (0.282, 0) -- (0.282, -0.847) .. controls (0.282, -0.941) and (0.235, -0.988) .. (0.141, -0.988) .. controls (0.047, -0.988) and (0, -0.941) .. (0, -0.847) -- cycle;
    \filldraw[shift={(3.27, 0.655)}, rotate=-45, yscale=-1, fill=RoyalBlue!80, fill opacity=0.3] (0, 0) .. controls (0, 0.094) and (0.047, 0.141) .. (0.141, 0.141) .. controls (0.235, 0.141) and (0.282, 0.094) .. (0.282, 0) -- (0.282, -0.847) .. controls (0.282, -0.941) and (0.235, -0.988) .. (0.141, -0.988) .. controls (0.047, -0.988) and (0, -0.941) .. (0, -0.847) -- cycle;
    \filldraw[fill=RoyalBlue!80, fill opacity=0.3] 
    (3.956, 3.385) circle[radius=0.211];
\end{tikzpicture}

%% file: Figures/latticetoric3Dperiodic.tikz
\begin{tikzpicture}[scale=1]
    \filldraw[fill=VioletRed!30, opacity=0.75] (1.411, 3.81) -- (1.27, 3.81) -- (1.27, 3.528) -- (1.411, 3.669) -- (1.411, 3.81);
    \filldraw[fill=RoyalBlue!80, opacity=0.75] (0.141, 0.141) -- (0.141, 0) -- (0.282, 0.141) -- (0.282, 0.282) -- (0.141, 0.141);
    \filldraw[fill=VioletRed!30] (0.282, 1.411) rectangle (1.411, 0.282);
    \filldraw[fill=VioletRed!30] (1.411, 1.411) rectangle (2.54, 0.282);
    \filldraw[fill=VioletRed!30] (0.282, 2.54) rectangle (1.411, 1.411);
    \filldraw[fill=VioletRed!30] (1.411, 2.54) rectangle (2.54, 1.411);
    \draw (0.847, 3.104) -- (1.411, 3.669) -- (3.669, 3.669) -- (3.104, 3.104);
    \filldraw[fill=RoyalBlue!80] (2.54, 1.411) -- (3.104, 1.976) -- (3.104, 3.104) -- (2.54, 2.54) -- (2.54, 1.411);
    \filldraw[fill=RoyalBlue!80] (2.54, 0.282) -- (3.104, 0.847) -- (3.104, 1.976) -- (2.54, 1.411) -- (2.54, 0.282);
    \filldraw[fill=RoyalBlue!80] (3.104, 0.847) -- (3.669, 1.411) -- (3.669, 2.54) -- (3.104, 1.976);
    \filldraw[fill=RoyalBlue!80] (3.669, 2.54) -- (3.669, 3.669) -- (3.104, 3.104) -- (3.104, 1.976) -- (3.669, 2.54);
    \filldraw[fill=Goldenrod, opacity=0.75] (2.54, 0.141) -- (3.951, 1.552) -- (3.669, 1.552) -- (3.669, 1.411) -- (2.54, 0.282) -- (0.282, 0.282) -- (0.141, 0.141) -- (2.54, 0.141);
    \filldraw[fill=Goldenrod, opacity=0.75] (0.141, 0.141) -- (0, 0.141) -- (0.141, 0.282);
    \filldraw[fill=Goldenrod] (0.847, 3.104) -- (1.976, 3.104) -- (1.411, 2.54) -- (0.282, 2.54) -- (0.847, 3.104);
    \filldraw[fill=Goldenrod] (1.411, 3.669) -- (2.54, 3.669) -- (1.976, 3.104) -- (0.847, 3.104) -- (1.411, 3.669);
    \filldraw[fill=Goldenrod] (2.54, 3.669) -- (3.669, 3.669) -- (3.104, 3.104) -- (1.976, 3.104) -- (2.54, 3.669);
    \filldraw[fill=Goldenrod] (1.976, 3.104) -- (3.104, 3.104) -- (2.54, 2.54) -- (1.411, 2.54) -- (1.976, 3.104);
    \node[Mark, Mark_disk, black!50] at (0.847, 2.54) {};
    \node[Mark, Mark_disk, black!50] at (1.976, 2.54) {};
    \node[Mark, Mark_disk] at (0.847, 1.411) {};
    \node[Mark, Mark_disk] at (1.411, 1.976) {};
    \node[Mark, Mark_disk] at (1.976, 1.411) {};
    \node[Mark, Mark_disk] at (0.847, 0.282) {};
    \node[Mark, Mark_disk] at (1.976, 0.282) {};
    \node[Mark, Mark_disk, black!50] at (2.54, 1.976) {};
    \node[Mark, Mark_disk, black!50] at (3.104, 2.399) {};
    \node[Mark, Mark_disk, black!50] at (1.27, 3.104) {};
    \node[Mark, Mark_disk, black!50] at (2.54, 3.104) {};
    \node[Mark, Mark_disk] at (1.411, 0.847) {};
    \node[Mark, Mark_disk, black!50] at (2.54, 0.847) {};
    \node[Mark, Mark_disk, black!50] at (3.104, 1.411) {};
    \node[Mark, Mark_disk, black!50] at (3.387, 3.387) {};
    \node[Mark, Mark_disk, black!50] at (2.822, 2.822) {};
    \node[Mark, Mark_disk, black!50] at (3.387, 2.258) {};
    \node[Mark, Mark_disk, black!50] at (2.822, 1.693) {};
    \node[Mark, Mark_disk, black!50] at (3.387, 1.129) {};
    \node[Mark, Mark_disk, black!50] at (2.822, 0.564) {};
    \filldraw[fill=VioletRed!30] (3.81, 1.411) -- (3.81, 1.27) -- (3.669, 1.27) -- (3.81, 1.411);
    \filldraw[fill=RoyalBlue!80, opacity=0.75] (0.282, 0.282) -- (0.141, 0.141) -- (0.141, 2.681) -- (1.552, 4.092) -- (1.552, 3.669) -- (1.411, 3.669) -- (0.282, 2.54) -- (0.282, 0.282) -- (0.847, 0.282);
    \node[Mark, Mark_disk] at (0.282, 0.847) {};
    \node[Mark, Mark_disk] at (0.282, 1.976) {};
    \node[Mark, Mark_disk, black!50] at (0.564, 2.822) {};
    \node[Mark, Mark_disk, black!50] at (1.129, 3.387) {};
    \filldraw[fill=VioletRed!30, opacity=0.75] (1.411, 3.81) -- (3.81, 3.81) -- (3.81, 1.411) -- (3.669, 1.411) -- (3.669, 3.669) -- (1.411, 3.669) -- (1.411, 3.81);
    \node[Mark, Mark_disk, black!50] at (1.834, 3.669) {};
    \node[Mark, Mark_disk, black!50] at (2.963, 3.669) {};
    \node[Mark, Mark_disk, black!50] at (3.669, 2.963) {};
    \node[Mark, Mark_disk, black!50] at (3.669, 1.834) {};
\end{tikzpicture}

%% file: Figures/latticetoric3D.tikz
\begin{tikzpicture}[scale=1.2]
    \draw (0, 1.129) -- (0, 0) -- (1.129, 0) -- (1.129, 1.129) -- (0, 1.129);
    \draw (0, 1.129) -- (0.564, 1.552);
    \draw (1.129, 0) -- (1.693, 0.423);
    \draw (0.564, 1.552) -- (1.693, 1.552);
    \draw (1.693, 1.552) -- (1.693, 0.423);
    \draw[dotted] (0, 0) -- (0.564, 0.423);
    \draw[dotted] (0.564, 0.423) -- (1.693, 0.423);
    \draw[dotted] (0.564, 1.552) -- (0.564, 0.423);
    \node[Mark, Mark_disk] at (0.564, 0) {};
    \node[Mark, Mark_disk] at (0, 0.564) {};
    \node[Mark, Mark_disk] at (1.129, 1.552) {};
    \node[Mark, Mark_disk] at (1.693, 0.988) {};
    \node[Mark, Mark_disk] at (1.129, 0.564) {};
    \node[Mark, Mark_disk] at (1.058, 0.423) {};
    \node[Mark, Mark_disk] at (0.564, 0.847) {};
    \node[Mark, Mark_disk] at (0.282, 1.341) {};
    \draw (1.693, 1.552) -- (1.129, 1.129);
    \node[Mark, Mark_disk] at (1.411, 1.341) {};
    \node[Mark, Mark_disk] at (0.564, 1.129) {};
    \node[Mark, Mark_disk] at (1.411, 0.212) {};
    \node[Mark, Mark_disk] at (0.282, 0.212) {};
\end{tikzpicture}

%% file: Figures/3Dtoric.tikz
\begin{tikzpicture}[scale=1]
    \node[anchor=center] at (1.147, 2.096) {$\sigma_x$};
    \draw (0.724, 0.685) -- (1.57, 1.531);
    \draw (1.147, 1.955) -- (1.147, 0.261);
    \draw (0.3, 1.108) -- (1.994, 1.108);
    \node[anchor=center] at (0.159, 1.108) {$\sigma_x$};
    \node[anchor=center] at (1.712, 1.672) {$\sigma_x$};
    \node[anchor=center] at (2.276, 1.108) {$\sigma_x$};
    \node[anchor=center] at (0.583, 0.543) {$\sigma_x$};
    \node[anchor=center] at (1.112, 0.12) {$\sigma_x$};
    \node[anchor=center] at (3.405, 1.108) {$\sigma_z$};
    \node[anchor=center] at (3.969, 1.672) {$\sigma_z$};
    \node[anchor=center] at (4.534, 1.108) {$\sigma_z$};
    \node[anchor=center] at (3.969, 0.543) {$\sigma_z$};
    \draw (3.405, 1.39) -- (3.405, 1.672) -- (3.687, 1.672);
    \draw (4.252, 1.672) -- (4.534, 1.672) -- (4.534, 1.39);
    \draw (3.405, 0.826) -- (3.405, 0.543) -- (3.687, 0.543);
    \draw (4.252, 0.543) -- (4.534, 0.543) -- (4.534, 0.826);
    \draw (5.804, 1.249) -- (5.945, 1.39) -- (6.227, 1.39);
    \draw (5.522, 0.967) -- (5.38, 0.826) -- (5.663, 0.826);
    \draw (6.651, 1.39) -- (7.074, 1.39) -- (6.933, 1.249);
    \draw (6.651, 0.967) -- (6.509, 0.826) -- (6.086, 0.826) -- (6.086, 0.826);
    \node[anchor=center] at (5.663, 1.108) {$\sigma_z$};
    \node[anchor=center] at (6.439, 1.39) {$\sigma_z$};
    \node[anchor=center] at (5.874, 0.826) {$\sigma_z$};
    \node[anchor=center] at (6.792, 1.108) {$\sigma_z$};
    \draw (7.779, 0.826) -- (7.779, 0.402) -- (7.921, 0.543);
    \draw (8.203, 0.826) -- (8.344, 0.967) -- (8.344, 1.249);
    \draw (7.779, 1.249) -- (7.779, 1.531) -- (7.921, 1.672);
    \draw (8.203, 1.955) -- (8.344, 2.096) -- (8.344, 1.672);
    \node[anchor=center] at (7.779, 1.037) {$\sigma_z$};
    \node[anchor=center] at (8.062, 1.813) {$\sigma_z$};
    \node[anchor=center] at (8.344, 1.461) {$\sigma_z$};
    \node[anchor=center] at (8.062, 0.685) {$\sigma_z$};
\end{tikzpicture}

%% file: Figures/latticehoneycomb.tikz
\begin{tikzpicture}[scale=1.2]
    \draw[shift={(0.804, 0.282)}, rotate=90] (0, 0) -- (0.423, 0) -- (0.635, 0.367) -- (0.423, 0.733) -- (0, 0.733) -- (-0.212, 0.367) -- (0, 0);
    \node[Mark, Mark_disk] at (0.071, 0.282) {};
    \node[Mark, Mark_disk] at (0.437, 0.07) {};
    \node[Mark, Mark_disk] at (0.804, 0.282) {};
    \node[Mark, Mark_disk] at (0.804, 0.705) {};
    \node[Mark, Mark_disk] at (0.437, 0.917) {};
    \node[Mark, Mark_disk] at (0.071, 0.705) {};
    \draw (1.537, 0.282) -- (1.537, 0.705) -- (1.17, 0.917) -- (0.804, 0.705)(0.804, 0.282) -- (1.17, 0.07) -- (1.537, 0.282);
    \node[Mark, Mark_disk] at (1.17, 0.07) {};
    \node[Mark, Mark_disk] at (1.537, 0.282) {};
    \node[Mark, Mark_disk] at (1.537, 0.705) {};
    \node[Mark, Mark_disk] at (1.17, 0.917) {};
    \draw (1.17, 0.917) -- (1.17, 1.34) -- (0.804, 1.552) -- (0.437, 1.34) -- (0.437, 0.917)(1.17, 0.917) -- (1.17, 0.917);
    \node[Mark, Mark_disk] at (1.17, 1.34) {};
    \node[Mark, Mark_disk] at (0.804, 1.552) {};
    \node[Mark, Mark_disk] at (0.437, 1.34) {};
    \draw (0.437, 1.341) -- (0.071, 1.552);
    \node[Mark, Mark_disk] at (0.071, 1.552) {};
    \node[Mark, Mark_disk] at (1.537, 0.282) {};
    \node[Mark, Mark_disk] at (1.537, 0.705) {};
    \draw (2.27, 0.282) -- (2.27, 0.705) -- (1.904, 0.917) -- (1.537, 0.705)(1.537, 0.282) -- (1.904, 0.07) -- (2.27, 0.282);
    \node[Mark, Mark_disk] at (1.904, 0.07) {};
    \node[Mark, Mark_disk] at (2.27, 0.282) {};
    \node[Mark, Mark_disk] at (2.27, 0.705) {};
    \node[Mark, Mark_disk] at (1.904, 0.917) {};
    \node[Mark, Mark_disk] at (2.27, 0.282) {};
    \node[Mark, Mark_disk] at (2.27, 0.705) {};
    \draw (3.004, 0.282) -- (3.004, 0.705) -- (2.637, 0.917) -- (2.27, 0.705)(2.27, 0.282) -- (2.637, 0.07) -- (3.004, 0.282);
    \node[Mark, Mark_disk] at (2.637, 0.07) {};
    \node[Mark, Mark_disk] at (2.637, 0.917) {};
    \node[Mark, Mark_disk] at (1.904, 0.917) {};
    \node[Mark, Mark_disk] at (2.637, 0.917) {};
    \draw (2.637, 0.917) -- (2.637, 1.34) -- (2.27, 1.552) -- (1.904, 1.34) -- (1.904, 0.917)(2.637, 0.917) -- (2.637, 0.917);
    \node[Mark, Mark_disk] at (2.637, 1.34) {};
    \node[Mark, Mark_disk] at (2.27, 1.552) {};
    \node[Mark, Mark_disk] at (1.904, 1.34) {};
    \draw (1.904, 1.34) -- (1.537, 1.552);
    \draw[shift={(1.537, 1.552)}, yscale=-1] (0, 0) -- (-0.367, 0.212);
    \draw[shift={(3.004, 1.552)}, yscale=-1] (0, 0) -- (-0.367, 0.212);
    \node[Mark, Mark_disk] at (0.071, 1.552) {};
    \node[Mark, Mark_disk] at (0.804, 1.552) {};
    \draw (0.804, 1.552) -- (0.804, 1.976) -- (0.438, 2.187) -- (0.071, 1.976) -- (0.071, 1.552)(0.804, 1.552) -- (0.804, 1.552);
    \node[Mark, Mark_disk] at (0.804, 1.976) {};
    \node[Mark, Mark_disk] at (0.438, 2.187) {};
    \node[Mark, Mark_disk] at (0.071, 1.976) {};
    \node[Mark, Mark_disk] at (2.27, 1.552) {};
    \draw (3.004, 1.552) -- (3.004, 1.975) -- (2.637, 2.187) -- (2.27, 1.975) -- (2.27, 1.552)(3.004, 1.552) -- (3.004, 1.552);
    \node[Mark, Mark_disk] at (2.27, 1.975) {};
    \draw (2.27, 1.975) -- (1.904, 2.187);
    \draw[shift={(1.171, 2.187)}, yscale=-1] (0, 0) -- (-0.367, 0.212);
    \draw[shift={(1.904, 2.187)}, yscale=-1] (0, 0) -- (-0.367, 0.212);
    \draw (1.537, 1.975) -- (1.17, 2.187);
    \draw (1.536, 1.976) -- (1.537, 1.552) -- (1.537, 1.552);
    \node[Mark, Mark_disk] at (1.536, 1.976) {};
    \node[Mark, Mark_disk] at (1.537, 1.552) {};
    \draw (2.637, 2.187) -- (2.637, 2.61)(1.904, 2.61) -- (1.904, 2.187)(2.637, 2.187) -- (2.637, 2.187);
    \draw (1.171, 2.187) -- (1.171, 2.61)(0.438, 2.61) -- (0.438, 2.187)(1.171, 2.187) -- (1.171, 2.187);
    \draw[shift={(0.071, 0.547)}, rotate=60] (0, 0) -- (-0.141, 0);
    \draw[shift={(0.141, 0.425)}, rotate=-60, yscale=-1] (0, 0) -- (-0.141, 0);
    \draw[shift={(0.071, 1.822)}, rotate=60] (0, 0) -- (-0.141, 0);
    \draw[shift={(0.141, 1.7)}, rotate=-60, yscale=-1] (0, 0) -- (-0.141, 0);
    \draw[shift={(3.004, 1.82)}, rotate=60] (0, 0) -- (-0.141, 0);
    \draw[shift={(3.074, 1.697)}, rotate=-60, yscale=-1] (0, 0) -- (-0.141, 0);
    \draw[shift={(3.004, 0.547)}, rotate=60] (0, 0) -- (-0.141, 0);
    \draw[shift={(3.074, 0.425)}, rotate=-60, yscale=-1] (0, 0) -- (-0.141, 0);
    \draw[dotted] (0.438, 2.61) -- (2.637, 2.61);
    \draw[dotted] (0.437, 0.07) -- (2.637, 0.071);
    \draw[shift={(0.895, 0.071)}, rotate=-30] (0, 0) -- (-0.141, 0);
    \draw[shift={(0.773, 0)}, rotate=-150, yscale=-1] (0, 0) -- (-0.141, 0);
    \draw[shift={(0.825, 0.071)}, rotate=-30] (0, 0) -- (-0.141, 0);
    \draw[shift={(0.702, 0)}, rotate=-150, yscale=-1] (0, 0) -- (-0.141, 0);
    \draw[shift={(2.307, 0.071)}, rotate=-30] (0, 0) -- (-0.141, 0);
    \draw[shift={(2.185, 0)}, rotate=-150, yscale=-1] (0, 0) -- (-0.141, 0);
    \draw[shift={(2.236, 0.071)}, rotate=-30] (0, 0) -- (-0.141, 0);
    \draw[shift={(2.114, 0)}, rotate=-150, yscale=-1] (0, 0) -- (-0.141, 0);
    \draw[shift={(2.305, 2.611)}, rotate=-30] (0, 0) -- (-0.141, 0);
    \draw[shift={(2.183, 2.54)}, rotate=-150, yscale=-1] (0, 0) -- (-0.141, 0);
    \draw[shift={(2.235, 2.611)}, rotate=-30] (0, 0) -- (-0.141, 0);
    \draw[shift={(2.113, 2.54)}, rotate=-150, yscale=-1] (0, 0) -- (-0.141, 0);
    \draw[shift={(0.895, 2.61)}, rotate=-30] (0, 0) -- (-0.141, 0);
    \draw[shift={(0.773, 2.54)}, rotate=-150, yscale=-1] (0, 0) -- (-0.141, 0);
    \draw[shift={(0.824, 2.61)}, rotate=-30] (0, 0) -- (-0.141, 0);
    \draw[shift={(0.702, 2.54)}, rotate=-150, yscale=-1] (0, 0) -- (-0.141, 0);
    \node[Mark, Mark_disk, black!50] at (3.004, 0.282) {};
    \node[Mark, Mark_disk, black!50] at (3.004, 0.705) {};
    \node[Mark, Mark_disk, black!50] at (3.004, 1.552) {};
    \node[Mark, Mark_disk, black!50] at (3.004, 1.552) {};
    \node[Mark, Mark_disk, black!50] at (3.004, 1.975) {};
    \node[Mark, Mark_disk, black!50] at (0.438, 2.61) {};
    \node[Mark, Mark_disk, black!50] at (1.171, 2.61) {};
    \node[Mark, Mark_disk, black!50] at (1.904, 2.61) {};
    \node[Mark, Mark_disk, black!50] at (2.637, 2.61) {};
\end{tikzpicture}

%% file: Figures/color.tikz
\begin{tikzpicture}[scale=1]
    \draw (4.557, 1.412) -- (4.557, 0.791);
    \draw (2.983, 1.413) -- (2.983, 0.792);
    \draw[shift={(3.895, 1.939)}, rotate=60] (0, 0) -- (0, -0.621);
    \draw[shift={(3.107, 0.577)}, rotate=60] (0, 0) -- (0, -0.621);
    \draw[shift={(3.108, 1.629)}, rotate=120] (0, 0) -- (0, -0.621);
    \draw[shift={(3.894, 0.265)}, rotate=120] (0, 0) -- (0, -0.621);
    \node[anchor=center] at (2.94, 1.611) {$\sigma_z$};
    \node[anchor=center] at (4.621, 1.611) {$\sigma_z$};
    \node[anchor=center] at (3.765, 2.087) {$\sigma_z$};
    \node[anchor=center] at (2.94, 0.583) {$\sigma_z$};
    \node[anchor=center] at (4.621, 0.583) {$\sigma_z$};
    \node[anchor=center] at (3.765, 0.102) {$\sigma_z$};
    \draw (1.806, 1.412) -- (1.806, 0.791);
    \draw (0.232, 1.413) -- (0.232, 0.792);
    \draw[shift={(1.145, 1.938)}, rotate=60] (0, 0) -- (0, -0.621);
    \draw[shift={(0.356, 0.577)}, rotate=60] (0, 0) -- (0, -0.621);
    \draw[shift={(0.357, 1.629)}, rotate=120] (0, 0) -- (0, -0.621);
    \draw[shift={(1.143, 0.265)}, rotate=120] (0, 0) -- (0, -0.621);
    \node[anchor=center] at (0.189, 1.611) {$\sigma_x$};
    \node[anchor=center] at (1.87, 1.611) {$\sigma_x$};
    \node[anchor=center] at (1.014, 2.087) {$\sigma_x$};
    \node[anchor=center] at (0.189, 0.583) {$\sigma_x$};
    \node[anchor=center] at (1.87, 0.583) {$\sigma_x$};
    \node[anchor=center] at (1.014, 0.102) {$\sigma_x$};
\end{tikzpicture}

%% file: Figures/lassoising.tikz
\begin{tikzpicture}[scale=1.5]
    \node[Mark, Mark_disk] at (0.195, 0.218) {};
    \node[Mark, Mark_disk] at (0.76, 0.218) {};
    \node[Mark, Mark_disk] at (1.324, 0.218) {};
    \node[Mark, Mark_disk] at (1.889, 0.218) {};
    \node[Mark, Mark_disk] at (2.112, 0.714) {};
    \node[Mark, Mark_disk] at (1.163, 0.706) {};
    \draw 
    (0.478, 0.218) ellipse[x radius=0.478, y radius=0.148];
    \draw 
    (1.046, 0.219) ellipse[x radius=0.478, y radius=0.148];
    \draw 
    (1.602, 0.222) ellipse[x radius=0.478, y radius=0.148];
    \draw 
    (1.992, 0.445) ellipse[x radius=0.478, y radius=0.148, rotate=67.388];
    \draw 
    (1.891, 0.92) ellipse[x radius=0.478, y radius=0.148, rotate=137.569];
    \draw 
    (0.96, 0.457) ellipse[x radius=0.478, y radius=0.148, rotate=-131.936];
    \draw 
    (1.398, 0.896) ellipse[x radius=0.478, y radius=0.148, rotate=-141.298];
    \node[Mark, Mark_disk] at (1.657, 1.111) {};
\end{tikzpicture}

%% file: Figures/rotatedsurfacelattice.tikz
\begin{tikzpicture}[scale=1]
    \filldraw[fill=RoyalBlue!80, opacity=0.5] (1.352, 3.62) rectangle (2.481, 2.492);
    \filldraw[fill=RoyalBlue!80, opacity=0.5] (0.223, 0.234) arc[start angle=-135, end angle=-45, radius=0.798];
    \filldraw[fill=RoyalBlue!80, opacity=0.5] (2.481, 0.234) arc[start angle=-135, end angle=-45, radius=0.798];
    \filldraw[fill=RoyalBlue!80, opacity=0.5] (0.223, 3.622) arc[start angle=-135, end angle=-45, x radius=0.798, y radius=-0.798];
    \filldraw[fill=RoyalBlue!80, opacity=0.5] (2.483, 3.629) arc[start angle=-135, end angle=-45, x radius=0.798, y radius=-0.798];
    \filldraw[fill=BrickRed!90, opacity=0.5] (0.234, 2.481) arc[start angle=-225, end angle=-135, radius=0.798];
    \filldraw[fill=BrickRed!90, opacity=0.5] (3.598, 1.352) arc[start angle=-45, end angle=45, radius=0.798];
    \filldraw[fill=BrickRed!90, opacity=0.5] (0.223, 1.363) rectangle (1.352, 0.234);
    \node[Mark, Mark_disk] at (0.223, 0.234) {};
    \filldraw[fill=BrickRed!90, opacity=0.5] (2.481, 3.62) rectangle (3.609, 2.492);
    \filldraw[fill=RoyalBlue!80, opacity=0.5] (1.352, 1.363) rectangle (2.481, 0.234);
    \filldraw[fill=BrickRed!90, opacity=0.5] (2.481, 1.363) rectangle (3.609, 0.234);
    \filldraw[fill=BrickRed!90, opacity=0.5] (1.352, 2.492) rectangle (2.481, 1.363);
    \filldraw[fill=RoyalBlue!80, opacity=0.5] (2.481, 2.492) rectangle (3.609, 1.363);
    \node[Mark, Mark_disk] at (2.481, 1.363) {};
    \node[Mark, Mark_disk] at (2.481, 0.234) {};
    \node[Mark, Mark_disk] at (1.352, 0.234) {};
    \node[Mark, Mark_disk] at (2.481, 2.492) {};
    \filldraw[fill=BrickRed!90, opacity=0.5] (0.223, 3.62) rectangle (1.352, 2.492);
    \node[Mark, Mark_disk] at (1.352, 3.62) {};
    \node[Mark, Mark_disk] at (2.481, 3.62) {};
    \node[Mark, Mark_disk] at (3.609, 2.492) {};
    \node[Mark, Mark_disk] at (3.609, 3.62) {};
    \node[Mark, Mark_disk] at (3.609, 0.234) {};
    \node[Mark, Mark_disk] at (3.609, 1.363) {};
    \node[Mark, Mark_disk] at (0.223, 3.62) {};
    \filldraw[fill=RoyalBlue!80, opacity=0.5] (0.223, 2.492) rectangle (1.352, 1.363);
    \node[Mark, Mark_disk] at (0.223, 2.492) {};
    \node[Mark, Mark_disk] at (0.223, 1.363) {};
    \node[Mark, Mark_disk] at (1.352, 1.363) {};
    \node[Mark, Mark_disk] at (1.352, 2.492) {};
\end{tikzpicture}

%% file: Figures/rotatedsurfaceinter.tikz
\begin{tikzpicture}[scale=1]
    \draw[RoyalBlue!80] (2.417, 2.66) -- (2.417, 3.225);
    \draw[RoyalBlue!80] (2.699, 3.507) -- (3.264, 3.507);
    \draw[RoyalBlue!80] (3.546, 3.225) -- (3.546, 2.66);
    \draw[RoyalBlue!80] (2.699, 2.378) -- (3.264, 2.378);
    \draw[BrickRed!90] (0.442, 3.507) -- (1.006, 3.507);
    \draw[BrickRed!90] (1.288, 3.225) -- (1.288, 2.66);
    \draw[BrickRed!90] (1.006, 2.378) -- (0.442, 2.378);
    \draw[BrickRed!90] (0.159, 2.66) -- (0.159, 3.225);
    \node[anchor=center] at (2.417, 2.378) {$\sigma_z$};
    \node[anchor=center] at (2.417, 3.507) {$\sigma_z$};
    \node[anchor=center] at (3.546, 3.507) {$\sigma_z$};
    \node[anchor=center] at (3.546, 2.378) {$\sigma_z$};
    \node[anchor=center] at (0.159, 2.378) {$\sigma_x$};
    \node[anchor=center] at (0.159, 3.507) {$\sigma_x$};
    \node[anchor=center] at (1.288, 3.507) {$\sigma_x$};
    \node[anchor=center] at (1.288, 2.378) {$\sigma_x$};
    \draw[BrickRed!90] (0.724, 0.402) -- (0.724, 0.967);
    \node[anchor=center] at (0.724, 0.12) {$\sigma_x$};
    \node[anchor=center] at (0.724, 1.249) {$\sigma_x$};
    \draw[RoyalBlue!80] (2.699, 0.685) -- (3.264, 0.685);
    \node[anchor=center] at (2.417, 0.685) {$\sigma_z$};
    \node[anchor=center] at (3.546, 0.685) {$\sigma_z$};
\end{tikzpicture}

%% file: Figures/latticehaah.tikz
\begin{tikzpicture}[scale=1.2]
    \draw (0.494, 1.129) rectangle (0.494, 1.129);
    \draw (0.494, 0.564) rectangle (0.494, 0.564);
    \draw (0.071, 1.129) rectangle (1.199, 0);
    \draw (0.071, 1.129) -- (0.635, 1.693);
    \draw (1.199, 1.129) -- (1.764, 1.693);
    \draw (1.764, 0.564) -- (1.199, 0);
    \draw[dotted] (0.635, 0.564) -- (0.071, 0);
    \node[Mark, Mark_disk] at (0, 1.129) {};
    \node[Mark, Mark_disk] at (0.141, 1.129) {};
    \node[Mark, Mark_disk] at (0.564, 1.693) {};
    \node[Mark, Mark_disk] at (0.706, 1.693) {};
    \node[Mark, Mark_disk] at (1.693, 1.693) {};
    \node[Mark, Mark_disk] at (1.834, 1.693) {};
    \node[Mark, Mark_disk] at (0.564, 0.564) {};
    \node[Mark, Mark_disk] at (0.706, 0.564) {};
    \node[Mark, Mark_disk] at (0, 0) {};
    \node[Mark, Mark_disk] at (0.141, 0) {};
    \node[Mark, Mark_disk] at (1.129, 0) {};
    \node[Mark, Mark_disk] at (1.27, 0) {};
    \node[Mark, Mark_disk] at (1.129, 1.129) {};
    \node[Mark, Mark_disk] at (1.27, 1.129) {};
    \node[Mark, Mark_disk] at (1.693, 0.564) {};
    \node[Mark, Mark_disk] at (1.834, 0.564) {};
    \draw (0.635, 1.693) -- (1.764, 1.693);
    \draw[dotted] (0.635, 1.693) -- (0.635, 0.564);
    \draw[dotted] (0.635, 0.564) -- (1.764, 0.564);
    \draw (1.764, 0.564) -- (1.764, 1.693);
\end{tikzpicture}

%% file: Figures/haah.tikz
\begin{tikzpicture}[scale=1]
    \node[anchor=center] at (1.013, 0.711) {$\mathds{1}\mathds{1}$};
    \node[anchor=center] at (0.978, 1.84) {$\mathds{1}\sigma_z$};
    \node[anchor=center] at (0.272, 1.275) {$\sigma_z\mathds{1}$};
    \node[anchor=center] at (1.719, 1.24) {$\sigma_z\sigma_z$};
    \draw (0.449, 1.134) -- (0.449, 0.288);
    \draw (0.59, 0.146) -- (1.295, 0.146);
    \draw (0.59, 1.275) -- (1.295, 1.275);
    \draw (1.578, 1.134) -- (1.578, 0.288);
    \draw (0.59, 1.416) -- (0.872, 1.699);
    \draw (1.719, 1.416) -- (2.001, 1.699);
    \draw (1.719, 0.288) -- (2.001, 0.57);
    \draw (2.142, 1.699) -- (2.142, 0.852);
    \draw (1.154, 1.84) -- (2.001, 1.84);
    \draw[dotted] (0.59, 0.288) -- (0.872, 0.57);
    \draw[dotted] (1.013, 1.699) -- (1.013, 0.852);
    \draw[dotted] (1.154, 0.711) -- (2.001, 0.711);
    \node[anchor=center] at (3.835, 0.711) {$\sigma_x\sigma_x$};
    \node[anchor=center] at (3.13, 1.275) {$\sigma_x\mathds{1}$};
    \node[anchor=center] at (4.4, 1.275) {$\mathds{1}\mathds{1}$};
    \node[anchor=center] at (5.105, 0.711) {$\mathds{1}\sigma_x$};
    \draw (3.412, 0.146) -- (4.118, 0.146);
    \draw (3.412, 1.275) -- (4.118, 1.275);
    \draw (4.4, 1.134) -- (4.4, 0.288);
    \draw (3.412, 1.416) -- (3.694, 1.699);
    \draw (4.541, 1.416) -- (4.823, 1.699);
    \draw (4.541, 0.288) -- (4.823, 0.57);
    \draw (3.977, 1.84) -- (4.823, 1.84);
    \draw[dotted] (3.412, 0.288) -- (3.694, 0.57);
    \draw[dotted] (3.835, 1.699) -- (3.835, 0.852);
    \draw[dotted] (3.977, 0.711) -- (4.823, 0.711);
    \draw (3.271, 1.134) -- (3.271, 0.288);
    \draw (4.964, 1.699) -- (4.964, 0.852);
    \node[anchor=center] at (2.248, 0.711) {$\mathds{1}\sigma_z$};
    \node[anchor=center] at (2.354, 1.84) {$\sigma_z\mathds{1}$};
    \node[anchor=center] at (0.414, 0.146) {$\mathds{1}\sigma_z$};
    \node[anchor=center] at (1.684, 0.146) {$\sigma_z\mathds{1}$};
    \node[anchor=center] at (5.105, 1.84) {$\sigma_x\mathds{1}$};
    \node[anchor=center] at (4.47, 0.146) {$\sigma_x\mathds{1}$};
    \node[anchor=center] at (3.694, 1.84) {$\mathds{1}\sigma_x$};
    \node[anchor=center] at (3.13, 0.146) {$\mathds{1}\sigma_x$};
\end{tikzpicture}

%% file: Figures/Xcube.tikz
\begin{tikzpicture}[scale=1]
    \draw (0.159, 0.967) -- (0.159, 1.249) -- (0.724, 1.249);
    \draw (1.006, 1.249) -- (1.288, 1.249) -- (1.288, 0.967);
    \draw (1.006, 0.12) -- (1.288, 0.12) -- (1.288, 0.402);
    \draw (1.429, 0.261) -- (1.288, 0.12) -- (0.865, 0.12) -- (0.865, 0.12);
    \draw (0.159, 0.967) -- (0.159, 1.249) -- (0.3, 1.39);
    \draw (1.57, 1.813) -- (1.853, 1.813) -- (1.853, 1.531);
    \draw (1.288, 0.826) -- (1.288, 1.249) -- (1.429, 1.39);
    \draw (1.712, 1.672) -- (1.853, 1.813) -- (1.853, 1.39);
    \node[anchor=center] at (0.159, 0.755) {$\sigma_x$};
    \node[anchor=center] at (0.442, 1.531) {$\sigma_x$};
    \node[anchor=center] at (1.288, 1.813) {$\sigma_x$};
    \node[anchor=center] at (1.535, 1.531) {$\sigma_x$};
    \node[anchor=center] at (0.9, 1.249) {$\sigma_x$};
    \node[anchor=center] at (0.724, 1.108) {$\sigma_x$};
    \node[anchor=center] at (0.653, 0.12) {$\sigma_x$};
    \node[anchor=center] at (0.406, 0.402) {$\sigma_x$};
    \node[anchor=center] at (1.218, 0.685) {$\sigma_x$};
    \node[anchor=center] at (1.606, 0.402) {$\sigma_x$};
    \node[anchor=center] at (1.853, 1.178) {$\sigma_x$};
    \node[anchor=center] at (1.288, 0.543) {$\sigma_x$};
    \draw (5.522, 0.967) -- (5.239, 0.685) -- (5.239, 0.685);
    \draw (5.098, 0.826) -- (5.663, 0.826);
    \draw (7.215, 1.108) -- (7.215, 0.543);
    \draw (7.074, 0.685) -- (7.356, 0.967);
    \draw (3.405, 1.108) -- (3.405, 0.543);
    \draw (3.123, 0.826) -- (3.687, 0.826);
    \draw (1.006, 1.813) -- (0.724, 1.813);
    \draw (0.583, 1.672) -- (0.724, 1.813);
    \draw[densely dotted] (0.724, 1.813) -- (0.724, 1.39);
    \draw[densely dotted] (0.583, 0.543) -- (0.724, 0.685);
    \draw[densely dotted] (0.724, 0.685) -- (1.006, 0.685);
    \draw[densely dotted] (0.724, 0.685) -- (0.724, 0.967);
    \draw (0.159, 0.543) -- (0.159, 0.12) -- (0.442, 0.12);
    \draw[densely dotted] (0.159, 0.12) -- (0.3, 0.261);
    \draw (1.853, 0.967) -- (1.853, 0.685) -- (1.712, 0.543);
    \draw[densely dotted] (1.853, 0.685) -- (1.429, 0.685);
    \node[anchor=center] at (3.934, 0.826) {$\sigma_z$};
    \node[anchor=center] at (3.405, 1.249) {$\sigma_z$};
    \node[anchor=center] at (3.405, 0.402) {$\sigma_z$};
    \node[anchor=center] at (2.911, 0.826) {$\sigma_z$};
    \node[anchor=center] at (4.922, 0.826) {$\sigma_z$};
    \node[anchor=center] at (5.874, 0.826) {$\sigma_z$};
    \node[anchor=center] at (5.592, 1.108) {$\sigma_z$};
    \node[anchor=center] at (5.204, 0.543) {$\sigma_z$};
    \node[anchor=center] at (6.933, 0.579) {$\sigma_z$};
    \node[anchor=center] at (7.215, 1.249) {$\sigma_z$};
    \node[anchor=center] at (7.215, 0.332) {$\sigma_z$};
    \node[anchor=center] at (7.532, 1.002) {$\sigma_z$};
\end{tikzpicture}

%% file: Figures/latticesubsystemtoric.tikz
\begin{tikzpicture}[scale=1.2]
    \draw (0, 0) -- (0, 1.129) -- (1.129, 1.129) -- (1.129, 0) -- (0, 0);
    \draw (0, 1.129) -- (0.282, 1.411) -- (1.411, 1.411) -- (1.129, 1.129);
    \draw (1.411, 1.411) -- (1.411, 0.282) -- (1.129, 0);
    \draw (1.411, 1.411) -- (1.693, 1.693) -- (1.693, 0.564) -- (1.411, 0.282);
    \draw (0.282, 1.411) -- (0.564, 1.693) -- (1.693, 1.693);
    \draw (0.564, 1.129) -- (1.129, 1.693);
    \draw (0, 0.564) -- (1.129, 0.564);
    \draw (0.564, 1.129) -- (0.564, 0);
    \draw (1.129, 0.564) -- (1.693, 1.129);
    \draw[dotted] (0.564, 1.693) -- (0.564, 1.129);
    \draw[dotted] (0.282, 1.411) -- (0.282, 0.282);
    \draw[dotted] (0.282, 0.847) -- (1.411, 0.847);
    \draw[dotted] (1.129, 1.693) -- (1.129, 1.129);
    \draw[dotted] (1.693, 1.129) -- (1.129, 1.129);
    \draw[dotted] (0, 0.564) -- (0.564, 1.129);
    \draw[dotted] (0.564, 0.564) -- (1.129, 1.129);
    \draw[dotted] (0, 0) -- (0.564, 0.564);
    \draw[dotted] (1.693, 0.564) -- (1.129, 0.564);
    \draw[dotted] (0.282, 0.282) -- (1.411, 0.282);
    \draw[dotted] (0.564, 0) -- (1.129, 0.564);
    \draw[dotted] (0.847, 1.411) -- (0.847, 0.282);
    \fill[RoyalBlue!80, opacity=0.2] (0.282, 0.282) -- (0.564, 0.564) -- (1.129, 0.564) -- (0.847, 0.282) -- (0.282, 0.282);
    \fill[RoyalBlue!80, opacity=0.2] (0.282, 0.282) -- (0.282, 0.847) -- (0.564, 1.129) -- (0.564, 0.564) -- (0.282, 0.282);
    \fill[RoyalBlue!80, opacity=0.2] (0.564, 1.129) -- (1.129, 1.129) -- (1.129, 0.564) -- (0.564, 0.564) -- (0.564, 1.129);
    \fill[RoyalBlue!80, opacity=0.2] (1.129, 1.129) -- (0.847, 0.847) -- (0.847, 0.282) -- (1.129, 0.564) -- (1.129, 1.129);
    \fill[RoyalBlue!80, opacity=0.2] (0.847, 0.847) -- (0.282, 0.847) -- (0.282, 0.282) -- (0.847, 0.282) -- (0.847, 0.847);
    \fill[RoyalBlue!80, opacity=0.2] (0.282, 0.847) -- (0.564, 1.129) -- (1.129, 1.129) -- (0.847, 0.847) -- (0.282, 0.847);
    \fill[BrickRed!90, opacity=0.2] (0.847, 0.282) -- (1.129, 0.564) -- (1.693, 0.564) -- (1.411, 0.282) -- (0.847, 0.282);
    \fill[BrickRed!90, opacity=0.2] (0.847, 0.282) -- (0.847, 0.847) -- (1.129, 1.129) -- (1.129, 0.564) -- (0.847, 0.282);
    \fill[BrickRed!90, opacity=0.2] (1.129, 1.129) -- (1.693, 1.129) -- (1.693, 0.564) -- (1.129, 0.564) -- (1.129, 1.129);
    \fill[BrickRed!90, opacity=0.2] (1.693, 1.129) -- (1.411, 0.847) -- (1.411, 0.282) -- (1.693, 0.564) -- (1.693, 1.129);
    \fill[BrickRed!90, opacity=0.2] (1.411, 0.847) -- (0.847, 0.847) -- (0.847, 0.282) -- (1.411, 0.282) -- (1.411, 0.847);
    \fill[BrickRed!90, opacity=0.2] (0.847, 0.847) -- (1.129, 1.129) -- (1.693, 1.129) -- (1.411, 0.847) -- (0.847, 0.847);
    \fill[BrickRed!90, opacity=0.2] (0.282, 0.847) -- (0.564, 1.129) -- (1.129, 1.129) -- (0.847, 0.847) -- (0.282, 0.847);
    \fill[BrickRed!90, opacity=0.2] (0.282, 0.847) -- (0.282, 1.411) -- (0.564, 1.693) -- (0.564, 1.129) -- (0.282, 0.847);
    \fill[BrickRed!90, opacity=0.2] (0.564, 1.693) -- (1.129, 1.693) -- (1.129, 1.129) -- (0.564, 1.129) -- (0.564, 1.693);
    \fill[BrickRed!90, opacity=0.2] (1.129, 1.693) -- (0.847, 1.411) -- (0.847, 0.847) -- (1.129, 1.129) -- (1.129, 1.693);
    \fill[BrickRed!90, opacity=0.2] (0.847, 1.411) -- (0.282, 1.411) -- (0.282, 0.847) -- (0.847, 0.847) -- (0.847, 1.411);
    \fill[RoyalBlue!80, opacity=0.2] (0.847, 1.411) -- (0.282, 1.411) -- (0.564, 1.693) -- (1.129, 1.693) -- (0.847, 1.411);
    \fill[RoyalBlue!80, opacity=0.2] (0.847, 0.847) -- (1.129, 1.129) -- (1.693, 1.129) -- (1.411, 0.847) -- (0.847, 0.847);
    \fill[RoyalBlue!80, opacity=0.2] (0.847, 0.847) -- (0.847, 1.411) -- (1.129, 1.693) -- (1.129, 1.129) -- (0.847, 0.847);
    \fill[RoyalBlue!80, opacity=0.2] (1.129, 1.693) -- (1.693, 1.693) -- (1.693, 1.129) -- (1.129, 1.129) -- (1.129, 1.693);
    \fill[RoyalBlue!80, opacity=0.2] (1.693, 1.693) -- (1.411, 1.411) -- (1.411, 0.847) -- (1.693, 1.129) -- (1.693, 1.693);
    \fill[RoyalBlue!80, opacity=0.2] (1.411, 1.411) -- (0.847, 1.411) -- (0.847, 0.847) -- (1.411, 0.847) -- (1.411, 1.411);
    \fill[RoyalBlue!80, opacity=0.2] (0.847, 1.411) -- (1.129, 1.693) -- (1.693, 1.693) -- (1.411, 1.411) -- (0.847, 1.411);
    \fill[BrickRed!90, opacity=0.2] (0.282, 0.282) -- (0, 0) -- (0.564, 0) -- (0.847, 0.282) -- (0.282, 0.282);
    \fill[BrickRed!90, opacity=0.2] (0, 0) -- (0, 0.564) -- (0.282, 0.847) -- (0.282, 0.282) -- (0, 0);
    \fill[BrickRed!90, opacity=0.2] (0.282, 0.847) -- (0.847, 0.847) -- (0.847, 0.282) -- (0.282, 0.282) -- (0.282, 0.847);
    \fill[BrickRed!90, opacity=0.2] (0.847, 0.847) -- (0.564, 0.564) -- (0.564, 0) -- (0.847, 0.282) -- (0.847, 0.847);
    \fill[BrickRed!90, opacity=0.2] (0.564, 0.564) -- (0, 0.564) -- (0, 0) -- (0.564, 0) -- (0.564, 0.564);
    \fill[BrickRed!90, opacity=0.2] (0.564, 0.564) -- (0.847, 0.847) -- (0.282, 0.847) -- (0, 0.564) -- (0.564, 0.564);
    \fill[RoyalBlue!80, opacity=0.2] (0.564, 0) -- (0.847, 0.282) -- (1.411, 0.282) -- (1.129, 0) -- (0.564, 0);
    \fill[RoyalBlue!80, opacity=0.2] (0.564, 0) -- (0.564, 0.564) -- (0.847, 0.847) -- (0.847, 0.282) -- (0.564, 0);
    \fill[RoyalBlue!80, opacity=0.2] (0.847, 0.847) -- (1.411, 0.847) -- (1.411, 0.282) -- (0.847, 0.282) -- (0.847, 0.847);
    \fill[RoyalBlue!80, opacity=0.2] (1.411, 0.847) -- (1.129, 0.564) -- (1.129, 0) -- (1.411, 0.282) -- (1.411, 0.847);
    \fill[RoyalBlue!80, opacity=0.2] (1.129, 0.564) -- (0.564, 0.564) -- (0.564, 0) -- (1.129, 0) -- (1.129, 0.564);
    \fill[RoyalBlue!80, opacity=0.2] (0.564, 0.564) -- (0.847, 0.847) -- (1.411, 0.847) -- (1.129, 0.564) -- (0.564, 0.564);
    \fill[RoyalBlue!80, opacity=0.2] (0, 0.564) -- (0.282, 0.847) -- (0.847, 0.847) -- (0.564, 0.564) -- (0, 0.564);
    \fill[RoyalBlue!80, opacity=0.2] (0, 0.564) -- (0, 1.129) -- (0.282, 1.411) -- (0.282, 0.847) -- (0, 0.564);
    \fill[RoyalBlue!80, opacity=0.2] (0.282, 1.411) -- (0.847, 1.411) -- (0.847, 0.847) -- (0.282, 0.847) -- (0.282, 1.411);
    \fill[RoyalBlue!80, opacity=0.2] (0.564, 0.564) -- (0.564, 1.129) -- (0.847, 1.411) -- (0.847, 0.847) -- (0.564, 0.564);
    \fill[RoyalBlue!80, opacity=0.2] (0.564, 0.564) -- (0.564, 1.129) -- (0, 1.129) -- (0, 0.564) -- (0.564, 0.564);
    \fill[RoyalBlue!80, opacity=0.2] (0, 1.129) -- (0.282, 1.411) -- (0.847, 1.411) -- (0.564, 1.129) -- (0, 1.129);
    \fill[BrickRed!90, opacity=0.2] (0.564, 0.564) -- (0.847, 0.847) -- (1.411, 0.847) -- (1.129, 0.564) -- (0.564, 0.564);
    \fill[BrickRed!90, opacity=0.2] (0.564, 0.564) -- (0.564, 1.129) -- (0.847, 1.411) -- (0.847, 0.847) -- (0.564, 0.564);
    \fill[BrickRed!90, opacity=0.2] (0.847, 0.847) -- (0.847, 1.411) -- (1.411, 1.411) -- (1.411, 0.847) -- (0.847, 0.847);
    \fill[BrickRed!90, opacity=0.2] (1.129, 0.564) -- (1.129, 1.129) -- (1.411, 1.411) -- (1.411, 0.847) -- (1.129, 0.564);
    \fill[BrickRed!90, opacity=0.2] (1.129, 0.564) -- (0.564, 0.564) -- (0.564, 1.129) -- (1.129, 1.129) -- (1.129, 0.564);
    \fill[BrickRed!90, opacity=0.2] (0.564, 1.129) -- (0.847, 1.411) -- (1.411, 1.411) -- (1.129, 1.129) -- (0.564, 1.129);
\end{tikzpicture}

%% file: Figures/stabsubsystoric.tikz
\begin{tikzpicture}[scale=1.2]
    \draw[BrickRed!90] (0.159, 0.967) -- (0.159, 1.249) -- (0.724, 1.249);
    \draw[BrickRed!90] (1.006, 1.249) -- (1.288, 1.249) -- (1.288, 0.967);
    \draw[BrickRed!90] (1.006, 0.12) -- (1.288, 0.12) -- (1.288, 0.402);
    \draw[BrickRed!90] (1.429, 0.261) -- (1.288, 0.12) -- (0.865, 0.12) -- (0.865, 0.12);
    \draw[BrickRed!90] (0.159, 0.967) -- (0.159, 1.249) -- (0.3, 1.39);
    \draw[BrickRed!90] (1.57, 1.813) -- (1.853, 1.813) -- (1.853, 1.531);
    \draw[BrickRed!90] (1.288, 0.826) -- (1.288, 1.249) -- (1.429, 1.39);
    \draw[BrickRed!90] (1.712, 1.672) -- (1.853, 1.813) -- (1.853, 1.39);
    \node[anchor=center] at (0.159, 0.755) {$\sigma_x$};
    \node[anchor=center] at (0.442, 1.531) {$\sigma_x$};
    \node[anchor=center] at (1.288, 1.813) {$\sigma_x$};
    \node[anchor=center] at (1.535, 1.531) {$\sigma_x$};
    \node[anchor=center] at (0.9, 1.249) {$\sigma_x$};
    \node[anchor=center] at (0.724, 1.178) {$\sigma_x$};
    \node[anchor=center] at (0.653, 0.12) {$\sigma_x$};
    \node[anchor=center] at (0.406, 0.402) {$\sigma_x$};
    \node[anchor=center] at (1.182, 0.685) {$\sigma_x$};
    \node[anchor=center] at (1.606, 0.402) {$\sigma_x$};
    \node[anchor=center] at (1.853, 1.178) {$\sigma_x$};
    \node[anchor=center] at (1.288, 0.543) {$\sigma_x$};
    \draw[RoyalBlue!80] (2.982, 0.967) -- (2.982, 1.249) -- (3.546, 1.249);
    \draw[RoyalBlue!80] (3.828, 1.249) -- (4.11, 1.249) -- (4.11, 0.967);
    \draw[RoyalBlue!80] (3.828, 0.12) -- (4.11, 0.12) -- (4.11, 0.402);
    \draw[RoyalBlue!80] (4.252, 0.261) -- (4.11, 0.12) -- (3.687, 0.12) -- (3.687, 0.12);
    \draw[RoyalBlue!80] (2.982, 0.967) -- (2.982, 1.249) -- (3.123, 1.39);
    \draw[RoyalBlue!80] (4.393, 1.813) -- (4.675, 1.813) -- (4.675, 1.531);
    \draw[RoyalBlue!80] (4.11, 0.826) -- (4.11, 1.249) -- (4.252, 1.39);
    \draw[RoyalBlue!80] (4.534, 1.672) -- (4.675, 1.813) -- (4.675, 1.39);
    \node[anchor=center] at (2.982, 0.755) {$\sigma_z$};
    \node[anchor=center] at (3.264, 1.531) {$\sigma_z$};
    \node[anchor=center] at (4.11, 1.813) {$\sigma_z$};
    \node[anchor=center] at (4.357, 1.531) {$\sigma_z$};
    \node[anchor=center] at (3.722, 1.249) {$\sigma_z$};
    \node[anchor=center] at (3.546, 1.178) {$\sigma_z$};
    \node[anchor=center] at (3.475, 0.12) {$\sigma_z$};
    \node[anchor=center] at (3.229, 0.402) {$\sigma_z$};
    \node[anchor=center] at (4.005, 0.685) {$\sigma_z$};
    \node[anchor=center] at (4.428, 0.402) {$\sigma_z$};
    \node[anchor=center] at (4.675, 1.178) {$\sigma_z$};
    \node[anchor=center] at (4.11, 0.543) {$\sigma_z$};
    \draw[BrickRed!90] (1.006, 1.813) -- (0.724, 1.813);
    \draw[BrickRed!90] (0.583, 1.672) -- (0.724, 1.813);
    \draw[BrickRed!90, densely dotted] (0.724, 1.813) -- (0.724, 1.39);
    \draw[BrickRed!90, densely dotted] (0.583, 0.543) -- (0.724, 0.685);
    \draw[BrickRed!90, densely dotted] (0.724, 0.685) -- (1.006, 0.685);
    \draw[BrickRed!90, densely dotted] (0.724, 0.685) -- (0.724, 0.967);
    \draw[BrickRed!90] (0.159, 0.543) -- (0.159, 0.12) -- (0.442, 0.12);
    \draw[BrickRed!90, densely dotted] (0.159, 0.12) -- (0.3, 0.261);
    \draw[BrickRed!90] (1.853, 0.967) -- (1.853, 0.685) -- (1.712, 0.543);
    \draw[BrickRed!90, densely dotted] (1.853, 0.685) -- (1.429, 0.685);
    \draw[RoyalBlue!80] (3.828, 1.813) -- (3.546, 1.813);
    \draw[RoyalBlue!80] (3.405, 1.672) -- (3.546, 1.813);
    \draw[RoyalBlue!80, densely dotted] (3.546, 1.813) -- (3.546, 1.39);
    \draw[RoyalBlue!80, densely dotted] (3.405, 0.543) -- (3.546, 0.685);
    \draw[RoyalBlue!80, densely dotted] (3.546, 0.685) -- (3.828, 0.685);
    \draw[RoyalBlue!80, densely dotted] (3.546, 0.685) -- (3.546, 0.967);
    \draw[RoyalBlue!80] (2.982, 0.543) -- (2.982, 0.12) -- (3.264, 0.12);
    \draw[RoyalBlue!80, densely dotted] (2.982, 0.12) -- (3.123, 0.261);
    \draw[RoyalBlue!80] (4.675, 0.967) -- (4.675, 0.685) -- (4.534, 0.543);
    \draw[RoyalBlue!80, densely dotted] (4.675, 0.685) -- (4.252, 0.685);
\end{tikzpicture}

%% file: Figures/triangsubsystoric.tikz
\begin{tikzpicture}[scale=1]
    \draw[BrickRed!90] (0, 0.875) -- (0, 1.158) -- (0.282, 1.158);
    \draw[BrickRed!90] (0.847, 1.158) -- (1.129, 1.158) -- (1.129, 0.875);
    \draw (0, 0.311) -- (0, 0.029) -- (0.282, 0.029);
    \draw[BrickRed!90] (0.423, 0.452) -- (0.564, 0.593) -- (0.847, 0.593);
    \draw[BrickRed!90] (0.141, 0.17) -- (0, 0.029) -- (0.282, 0.029);
    \draw[BrickRed!90] (1.27, 0.593) -- (1.693, 0.593) -- (1.552, 0.452);
    \draw[BrickRed!90] (1.27, 0.17) -- (1.129, 0.029) -- (0.706, 0.029) -- (0.706, 0.029);
    \draw[BrickRed!90] (0, 0.452) -- (0, 0.029) -- (0.141, 0.17);
    \draw[BrickRed!90] (0.423, 0.452) -- (0.564, 0.593) -- (0.564, 0.875);
    \draw[BrickRed!90] (0, 0.875) -- (0, 1.158) -- (0.141, 1.299);
    \draw[BrickRed!90] (0.423, 1.581) -- (0.564, 1.722) -- (0.564, 1.299);
    \draw[BrickRed!90] (0.564, 1.44) -- (0.564, 1.722) -- (0.847, 1.722);
    \draw[BrickRed!90] (1.411, 1.722) -- (1.693, 1.722) -- (1.693, 1.44);
    \draw[BrickRed!90] (1.411, 0.593) -- (1.693, 0.593) -- (1.693, 0.875);
    \draw[BrickRed!90] (1.129, 0.875) -- (1.129, 1.158) -- (1.27, 1.299);
    \draw[BrickRed!90] (1.552, 1.581) -- (1.693, 1.722) -- (1.693, 1.299);
    \draw[RoyalBlue!80] (4.516, 0.875) -- (4.516, 1.158) -- (4.798, 1.158);
    \draw (4.516, 0.311) -- (4.516, 0.029) -- (4.798, 0.029);
    \draw[RoyalBlue!80] (5.362, 0.029) -- (5.644, 0.029) -- (5.644, 0.311);
    \draw[RoyalBlue!80] (4.939, 0.452) -- (5.08, 0.593) -- (5.362, 0.593);
    \draw[RoyalBlue!80] (4.657, 0.17) -- (4.516, 0.029) -- (4.798, 0.029);
    \draw[RoyalBlue!80] (5.786, 0.593) -- (6.209, 0.593) -- (6.068, 0.452);
    \draw[RoyalBlue!80] (5.786, 0.17) -- (5.644, 0.029) -- (5.221, 0.029) -- (5.221, 0.029);
    \draw[RoyalBlue!80] (4.516, 0.452) -- (4.516, 0.029) -- (4.657, 0.17);
    \draw[RoyalBlue!80] (4.939, 0.452) -- (5.08, 0.593) -- (5.08, 0.875);
    \draw[RoyalBlue!80] (4.516, 0.875) -- (4.516, 1.158) -- (4.657, 1.299);
    \draw[RoyalBlue!80] (4.939, 1.581) -- (5.08, 1.722) -- (5.08, 1.299);
    \draw[RoyalBlue!80] (5.08, 1.44) -- (5.08, 1.722) -- (5.362, 1.722);
    \draw[RoyalBlue!80] (5.927, 1.722) -- (6.209, 1.722) -- (6.209, 1.44);
    \draw[RoyalBlue!80] (5.927, 0.593) -- (6.209, 0.593) -- (6.209, 0.875);
    \draw[RoyalBlue!80] (6.068, 1.581) -- (6.209, 1.722) -- (6.209, 1.299);
    \filldraw[fill=BrickRed!90, opacity=0.5] (1.129, 1.722) -- (1.411, 1.44) -- (1.693, 1.016) -- (1.129, 1.722);
    \filldraw[fill=BrickRed!90, opacity=0.5] (0.282, 1.44) -- (0.564, 1.158) -- (0, 0.593) -- (0.282, 1.44);
    \filldraw[fill=BrickRed!90, opacity=0.5] (0.282, 0.311) -- (1.129, 0.593) -- (0.564, 1.016) -- (0.282, 0.311);
    \filldraw[fill=BrickRed!90, opacity=0.5] (1.135, 0.795) -- (0.523, 0.035) -- (1.416, 0.311) -- (1.141, 0.785);
    \draw[BrickRed!90] (0.847, 0.029) -- (1.129, 0.029) -- (1.129, 0.311);
    \filldraw[fill=RoyalBlue!80, opacity=0.5] (4.798, 1.44) -- (5.644, 1.722) -- (5.08, 1.016) -- (4.798, 1.44);
    \filldraw[fill=RoyalBlue!80, opacity=0.5] (5.08, 1.158) -- (5.927, 1.44) -- (5.644, 0.734) -- (5.08, 1.158);
    \draw[RoyalBlue!80] (5.644, 0.875) -- (5.644, 1.158) -- (5.786, 1.299);
    \draw[RoyalBlue!80] (5.362, 1.158) -- (5.644, 1.158) -- (5.644, 0.875);
    \filldraw[fill=RoyalBlue!80, opacity=0.5] (5.503, 0.593) -- (5.927, 0.311) -- (6.209, 1.158) -- (5.503, 0.593);
    \filldraw[fill=RoyalBlue!80, opacity=0.5] (4.851, 0.359) -- (5.015, 0.022) -- (4.512, 0.66) -- (4.851, 0.353);
    \draw (3.123, 1.315) -- (3.722, 0.278);
    \draw[shift={(2.422, 0.135)}, rotate=60] (0, 0) -- (0.599, -1.038);
    \draw[shift={(2.327, 0.278)}, rotate=120] (0, 0) -- (0.599, -1.038);
    \node[anchor=center] at (2.258, 0.12) {$\sigma_x$};
    \node[anchor=center] at (3.046, 1.457) {$\sigma_x$};
    \node[anchor=center] at (3.77, 0.12) {$\sigma_x$};
    \draw (7.819, 1.315) -- (8.418, 0.278);
    \draw[shift={(7.118, 0.135)}, rotate=60] (0, 0) -- (0.599, -1.038);
    \draw[shift={(7.023, 0.278)}, rotate=120] (0, 0) -- (0.599, -1.038);
    \node[anchor=center] at (6.954, 0.12) {$\sigma_z$};
    \node[anchor=center] at (7.742, 1.457) {$\sigma_z$};
    \node[anchor=center] at (8.465, 0.12) {$\sigma_z$};
\end{tikzpicture}

%% file: Figures/lasso_plain.tikz
\begin{tikzpicture}[scale=1]
    \draw[shift={(1.182, 0.82)}, rotate=60.63] (0, 0) -- (0, -0.847) .. controls (0, -0.941) and (0.047, -0.988) .. (0.141, -0.988) .. controls (0.235, -0.988) and (0.282, -0.941) .. (0.282, -0.847) -- (0.282, 0) .. controls (0.282, 0.094) and (0.235, 0.141) .. (0.141, 0.141) .. controls (0.047, 0.141) and (0, 0.094) .. (0, 0) -- cycle;
    \node[Mark, Mark_disk] at (1.922, 0.567) {};
    \draw[shift={(1.795, 0.473)}, rotate=62.561] (0, 0) -- (0, -0.847) .. controls (0, -0.941) and (0.047, -0.988) .. (0.141, -0.988) .. controls (0.235, -0.988) and (0.282, -0.941) .. (0.282, -0.847) -- (0.282, 0) .. controls (0.282, 0.094) and (0.235, 0.141) .. (0.141, 0.141) .. controls (0.047, 0.141) and (0, 0.094) .. (0, 0) -- cycle;
    \node[Mark, Mark_disk, Green!80] at (2.545, 0.236) {};
    \draw[shift={(2.501, 0.095)}, rotate=91.615] (0, 0) -- (0, -0.847) .. controls (0, -0.941) and (0.047, -0.988) .. (0.141, -0.988) .. controls (0.235, -0.988) and (0.282, -0.941) .. (0.282, -0.847) -- (0.282, 0) .. controls (0.282, 0.094) and (0.235, 0.141) .. (0.141, 0.141) .. controls (0.047, 0.141) and (0, 0.094) .. (0, 0) -- cycle;
    \node[Mark, Mark_disk, BrickRed!90] at (3.333, 0.253) {};
    \draw[shift={(3.376, 0.089)}, rotate=129.289] (0, 0) -- (0, -0.847) .. controls (0, -0.941) and (0.047, -0.988) .. (0.141, -0.988) .. controls (0.235, -0.988) and (0.282, -0.941) .. (0.282, -0.847) -- (0.282, 0) .. controls (0.282, 0.094) and (0.235, 0.141) .. (0.141, 0.141) .. controls (0.047, 0.141) and (0, 0.094) .. (0, 0) -- cycle;
    \node[Mark, Mark_disk, Green!80] at (3.887, 0.69) {};
    \draw[shift={(3.973, 0.552)}, rotate=147.109] (0, 0) -- (0, -0.847) .. controls (0, -0.941) and (0.047, -0.988) .. (0.141, -0.988) .. controls (0.235, -0.988) and (0.282, -0.941) .. (0.282, -0.847) -- (0.282, 0) .. controls (0.282, 0.094) and (0.235, 0.141) .. (0.141, 0.141) .. controls (0.047, 0.141) and (0, 0.094) .. (0, 0) -- cycle;
    \node[Mark, Mark_disk, BrickRed!90] at (4.287, 1.275) {};
    \draw[shift={(4.416, 1.172)}, rotate=164.809] (0, 0) -- (0, -0.847) .. controls (0, -0.941) and (0.047, -0.988) .. (0.141, -0.988) .. controls (0.235, -0.988) and (0.282, -0.941) .. (0.282, -0.847) -- (0.282, 0) .. controls (0.282, 0.094) and (0.235, 0.141) .. (0.141, 0.141) .. controls (0.047, 0.141) and (0, 0.094) .. (0, 0) -- cycle;
    \node[Mark, Mark_disk, Green!80] at (4.499, 1.967) {};
    \draw[shift={(4.626, 1.959)}, rotate=-162.238] (0, 0) -- (0, -0.847) .. controls (0, -0.941) and (0.047, -0.988) .. (0.141, -0.988) .. controls (0.235, -0.988) and (0.282, -0.941) .. (0.282, -0.847) -- (0.282, 0) .. controls (0.282, 0.094) and (0.235, 0.141) .. (0.141, 0.141) .. controls (0.047, 0.141) and (0, 0.094) .. (0, 0) -- cycle;
    \node[Mark, Mark_disk, BrickRed!90] at (4.233, 2.709) {};
    \draw[shift={(4.365, 2.795)}, rotate=-136.812] (0, 0) -- (0, -0.847) .. controls (0, -0.941) and (0.047, -0.988) .. (0.141, -0.988) .. controls (0.235, -0.988) and (0.282, -0.941) .. (0.282, -0.847) -- (0.282, 0) .. controls (0.282, 0.094) and (0.235, 0.141) .. (0.141, 0.141) .. controls (0.047, 0.141) and (0, 0.094) .. (0, 0) -- cycle;
    \node[Mark, Mark_disk, Green!80] at (3.723, 3.282) {};
    \draw[shift={(3.745, 3.429)}, rotate=-79.075] (0, 0) -- (0, -0.847) .. controls (0, -0.941) and (0.047, -0.988) .. (0.141, -0.988) .. controls (0.235, -0.988) and (0.282, -0.941) .. (0.282, -0.847) -- (0.282, 0) .. controls (0.282, 0.094) and (0.235, 0.141) .. (0.141, 0.141) .. controls (0.047, 0.141) and (0, 0.094) .. (0, 0) -- cycle;
    \node[Mark, Mark_disk, BrickRed!90] at (2.952, 3.12) {};
    \draw[shift={(2.899, 3.276)}, rotate=-47.537] (0, 0) -- (0, -0.847) .. controls (0, -0.941) and (0.047, -0.988) .. (0.141, -0.988) .. controls (0.235, -0.988) and (0.282, -0.941) .. (0.282, -0.847) -- (0.282, 0) .. controls (0.282, 0.094) and (0.235, 0.141) .. (0.141, 0.141) .. controls (0.047, 0.141) and (0, 0.094) .. (0, 0) -- cycle;
    \node[Mark, Mark_disk, Green!80] at (2.422, 2.648) {};
    \draw[shift={(2.305, 2.734)}, rotate=-31.659] (0, 0) -- (0, -0.847) .. controls (0, -0.941) and (0.047, -0.988) .. (0.141, -0.988) .. controls (0.235, -0.988) and (0.282, -0.941) .. (0.282, -0.847) -- (0.282, 0) .. controls (0.282, 0.094) and (0.235, 0.141) .. (0.141, 0.141) .. controls (0.047, 0.141) and (0, 0.094) .. (0, 0) -- cycle;
    \node[Mark, Mark_disk, BrickRed!90] at (2.005, 1.978) {};
    \draw[shift={(1.854, 2.061)}, rotate=-10.968] (0, 0) -- (0, -0.847) .. controls (0, -0.941) and (0.047, -0.988) .. (0.141, -0.988) .. controls (0.235, -0.988) and (0.282, -0.941) .. (0.282, -0.847) -- (0.282, 0) .. controls (0.282, 0.094) and (0.235, 0.141) .. (0.141, 0.141) .. controls (0.047, 0.141) and (0, 0.094) .. (0, 0) -- cycle;
    \node[Mark, Mark_disk, Green!80] at (1.845, 1.282) {};
    \draw[shift={(1.674, 1.315)}, rotate=8.288] (0, 0) -- (0, -0.847) .. controls (0, -0.941) and (0.047, -0.988) .. (0.141, -0.988) .. controls (0.235, -0.988) and (0.282, -0.941) .. (0.282, -0.847) -- (0.282, 0) .. controls (0.282, 0.094) and (0.235, 0.141) .. (0.141, 0.141) .. controls (0.047, 0.141) and (0, 0.094) .. (0, 0) -- cycle;
    \node[Mark, Mark_disk, BrickRed!90] at (1.926, 0.568) {};
    \draw[shift={(0.583, 1.302)}, rotate=52.009] (0, 0) -- (0, -0.847) .. controls (0, -0.941) and (0.047, -0.988) .. (0.141, -0.988) .. controls (0.235, -0.988) and (0.282, -0.941) .. (0.282, -0.847) -- (0.282, 0) .. controls (0.282, 0.094) and (0.235, 0.141) .. (0.141, 0.141) .. controls (0.047, 0.141) and (0, 0.094) .. (0, 0) -- cycle;
    \node[Mark, Mark_disk, Green!80] at (1.305, 0.907) {};
    \draw[shift={(0.096, 1.831)}, rotate=42.998] (0, 0) -- (0, -0.847) .. controls (0, -0.941) and (0.047, -0.988) .. (0.141, -0.988) .. controls (0.235, -0.988) and (0.282, -0.941) .. (0.282, -0.847) -- (0.282, 0) .. controls (0.282, 0.094) and (0.235, 0.141) .. (0.141, 0.141) .. controls (0.047, 0.141) and (0, 0.094) .. (0, 0) -- cycle;
    \node[Mark, Mark_disk, BrickRed!90] at (0.723, 1.373) {};
    \node[Mark, Mark_disk, Green!80] at (0.247, 1.875) {};
\end{tikzpicture}

%% file: Figures/lasso_superposition.tikz
\begin{tikzpicture}[scale=1]
    \filldraw[shift={(1.182, 1.047)}, rotate=60.63, fill=RoyalBlue!80, fill opacity=0.3] (0, 0) -- (0, -0.847) .. controls (0, -0.941) and (0.047, -0.988) .. (0.141, -0.988) .. controls (0.235, -0.988) and (0.282, -0.941) .. (0.282, -0.847) -- (0.282, 0) .. controls (0.282, 0.094) and (0.235, 0.141) .. (0.141, 0.141) .. controls (0.047, 0.141) and (0, 0.094) .. (0, 0) -- cycle;
    \filldraw[shift={(1.795, 0.7)}, rotate=62.561, fill=RoyalBlue!80, fill opacity=0.3] (0, 0) -- (0, -0.847) .. controls (0, -0.941) and (0.047, -0.988) .. (0.141, -0.988) .. controls (0.235, -0.988) and (0.282, -0.941) .. (0.282, -0.847) -- (0.282, 0) .. controls (0.282, 0.094) and (0.235, 0.141) .. (0.141, 0.141) .. controls (0.047, 0.141) and (0, 0.094) .. (0, 0) -- cycle;
    \filldraw[shift={(2.501, 0.323)}, rotate=91.615, fill=RoyalBlue!80, fill opacity=0.3] (0, 0) -- (0, -0.847) .. controls (0, -0.941) and (0.047, -0.988) .. (0.141, -0.988) .. controls (0.235, -0.988) and (0.282, -0.941) .. (0.282, -0.847) -- (0.282, 0) .. controls (0.282, 0.094) and (0.235, 0.141) .. (0.141, 0.141) .. controls (0.047, 0.141) and (0, 0.094) .. (0, 0) -- cycle;
    \filldraw[shift={(3.376, 0.317)}, rotate=129.289, fill=Mulberry!80, fill opacity=0.5] (0, 0) -- (0, -0.847) .. controls (0, -0.941) and (0.047, -0.988) .. (0.141, -0.988) .. controls (0.235, -0.988) and (0.282, -0.941) .. (0.282, -0.847) -- (0.282, 0) .. controls (0.282, 0.094) and (0.235, 0.141) .. (0.141, 0.141) .. controls (0.047, 0.141) and (0, 0.094) .. (0, 0) -- cycle;
    \filldraw[shift={(3.973, 0.779)}, rotate=147.109, fill=Mulberry!80, fill opacity=0.5] (0, 0) -- (0, -0.847) .. controls (0, -0.941) and (0.047, -0.988) .. (0.141, -0.988) .. controls (0.235, -0.988) and (0.282, -0.941) .. (0.282, -0.847) -- (0.282, 0) .. controls (0.282, 0.094) and (0.235, 0.141) .. (0.141, 0.141) .. controls (0.047, 0.141) and (0, 0.094) .. (0, 0) -- cycle;
    \filldraw[shift={(4.416, 1.399)}, rotate=164.809, fill=Lavender, fill opacity=0.5] (0, 0) -- (0, -0.847) .. controls (0, -0.941) and (0.047, -0.988) .. (0.141, -0.988) .. controls (0.235, -0.988) and (0.282, -0.941) .. (0.282, -0.847) -- (0.282, 0) .. controls (0.282, 0.094) and (0.235, 0.141) .. (0.141, 0.141) .. controls (0.047, 0.141) and (0, 0.094) .. (0, 0) -- cycle;
    \filldraw[shift={(4.626, 2.187)}, rotate=-162.238, fill=Lavender, fill opacity=0.3] (0, 0) -- (0, -0.847) .. controls (0, -0.941) and (0.047, -0.988) .. (0.141, -0.988) .. controls (0.235, -0.988) and (0.282, -0.941) .. (0.282, -0.847) -- (0.282, 0) .. controls (0.282, 0.094) and (0.235, 0.141) .. (0.141, 0.141) .. controls (0.047, 0.141) and (0, 0.094) .. (0, 0) -- cycle;
    \filldraw[shift={(4.365, 3.023)}, rotate=-136.812, fill=Lavender, fill opacity=0.3] (0, 0) -- (0, -0.847) .. controls (0, -0.941) and (0.047, -0.988) .. (0.141, -0.988) .. controls (0.235, -0.988) and (0.282, -0.941) .. (0.282, -0.847) -- (0.282, 0) .. controls (0.282, 0.094) and (0.235, 0.141) .. (0.141, 0.141) .. controls (0.047, 0.141) and (0, 0.094) .. (0, 0) -- cycle;
    \filldraw[shift={(3.745, 3.656)}, rotate=-79.075, fill=Lavender, fill opacity=0.3] (0, 0) -- (0, -0.847) .. controls (0, -0.941) and (0.047, -0.988) .. (0.141, -0.988) .. controls (0.235, -0.988) and (0.282, -0.941) .. (0.282, -0.847) -- (0.282, 0) .. controls (0.282, 0.094) and (0.235, 0.141) .. (0.141, 0.141) .. controls (0.047, 0.141) and (0, 0.094) .. (0, 0) -- cycle;
    \filldraw[shift={(2.899, 3.503)}, rotate=-47.537, fill=Lavender, fill opacity=0.3] (0, 0) -- (0, -0.847) .. controls (0, -0.941) and (0.047, -0.988) .. (0.141, -0.988) .. controls (0.235, -0.988) and (0.282, -0.941) .. (0.282, -0.847) -- (0.282, 0) .. controls (0.282, 0.094) and (0.235, 0.141) .. (0.141, 0.141) .. controls (0.047, 0.141) and (0, 0.094) .. (0, 0) -- cycle;
    \filldraw[shift={(2.305, 2.962)}, rotate=-31.659, fill=Mulberry!80, fill opacity=0.5] (0, 0) -- (0, -0.847) .. controls (0, -0.941) and (0.047, -0.988) .. (0.141, -0.988) .. controls (0.235, -0.988) and (0.282, -0.941) .. (0.282, -0.847) -- (0.282, 0) .. controls (0.282, 0.094) and (0.235, 0.141) .. (0.141, 0.141) .. controls (0.047, 0.141) and (0, 0.094) .. (0, 0) -- cycle;
    \filldraw[shift={(1.854, 2.289)}, rotate=-10.968, fill=Mulberry!80, fill opacity=0.5] (0, 0) -- (0, -0.847) .. controls (0, -0.941) and (0.047, -0.988) .. (0.141, -0.988) .. controls (0.235, -0.988) and (0.282, -0.941) .. (0.282, -0.847) -- (0.282, 0) .. controls (0.282, 0.094) and (0.235, 0.141) .. (0.141, 0.141) .. controls (0.047, 0.141) and (0, 0.094) .. (0, 0) -- cycle;
    \filldraw[shift={(1.674, 1.542)}, rotate=8.288, fill=RoyalBlue!80, fill opacity=0.3] (0, 0) -- (0, -0.847) .. controls (0, -0.941) and (0.047, -0.988) .. (0.141, -0.988) .. controls (0.235, -0.988) and (0.282, -0.941) .. (0.282, -0.847) -- (0.282, 0) .. controls (0.282, 0.094) and (0.235, 0.141) .. (0.141, 0.141) .. controls (0.047, 0.141) and (0, 0.094) .. (0, 0) -- cycle;
    \filldraw[shift={(0.583, 1.53)}, rotate=52.009, fill=RoyalBlue!80, fill opacity=0.3] (0, 0) -- (0, -0.847) .. controls (0, -0.941) and (0.047, -0.988) .. (0.141, -0.988) .. controls (0.235, -0.988) and (0.282, -0.941) .. (0.282, -0.847) -- (0.282, 0) .. controls (0.282, 0.094) and (0.235, 0.141) .. (0.141, 0.141) .. controls (0.047, 0.141) and (0, 0.094) .. (0, 0) -- cycle;
    \filldraw[shift={(0.096, 2.058)}, rotate=42.998, fill=RoyalBlue!80, fill opacity=0.3] (0, 0) -- (0, -0.847) .. controls (0, -0.941) and (0.047, -0.988) .. (0.141, -0.988) .. controls (0.235, -0.988) and (0.282, -0.941) .. (0.282, -0.847) -- (0.282, 0) .. controls (0.282, 0.094) and (0.235, 0.141) .. (0.141, 0.141) .. controls (0.047, 0.141) and (0, 0.094) .. (0, 0) -- cycle;
    \filldraw[shift={(3.973, 1.884)}, rotate=-22.406, densely dashed, fill=Lavender] (0, 0) -- (0.847, 0);
    \filldraw[shift={(2.961, 0.798)}, rotate=-15.595, densely dashed, fill=Lavender] (0, 0) -- (0.564, -0.564);
    \node[anchor=center, text=Mulberry!80] at (4.393, 3.79) {$A$};
    \node[anchor=center, text=RoyalBlue!80] at (1.605, 0.12) {$B$};
    \filldraw[shift={(2.233, 3.314)}, rotate=-44.133, densely dashed, fill=Lavender] (0, 0) -- (0.847, 0);
    \filldraw[shift={(1.488, 1.252)}, rotate=3.637, densely dashed, fill=Lavender] (0, 0) -- (0.847, 0);
    \draw[<->] (3.419, 0.136) .. controls (4.172, 0.324) and (4.595, 0.795) .. (4.689, 1.547);
    \draw[shift={(2.144, 3.166)}, rotate=-159.891, <->] (0, 0) .. controls (0.753, 0.188) and (1.176, 0.659) .. (1.27, 1.411);
    \node[anchor=center] at (1.29, 2.57) {$l$};
    \node[anchor=center] at (4.521, 0.499) {$l$};
    \node[Mark, Mark_disk, BrickRed!90] at (0.723, 1.6) {};
    \node[Mark, Mark_disk, Green!80] at (0.247, 2.102) {};
    \node[Mark, Mark_disk, BrickRed!90] at (1.926, 0.795) {};
    \node[Mark, Mark_disk, Green!80] at (1.845, 1.509) {};
    \node[Mark, Mark_disk, Green!80] at (1.846, 1.515) {};
    \node[Mark, Mark_disk, BrickRed!90] at (2.005, 2.206) {};
    \node[Mark, Mark_disk, Green!80] at (2.422, 2.875) {};
    \node[Mark, Mark_disk, BrickRed!90] at (2.952, 3.347) {};
    \node[Mark, Mark_disk, Green!80] at (3.723, 3.509) {};
    \node[Mark, Mark_disk, BrickRed!90] at (4.233, 2.937) {};
    \node[Mark, Mark_disk, Green!80] at (4.499, 2.194) {};
    \node[Mark, Mark_disk, BrickRed!90] at (4.287, 1.503) {};
    \node[Mark, Mark_disk, Green!80] at (3.887, 0.918) {};
    \node[Mark, Mark_disk, BrickRed!90] at (3.333, 0.481) {};
    \node[Mark, Mark_disk, Green!80] at (1.313, 1.135) {};
    \node[Mark, Mark_disk, Green!80] at (2.545, 0.464) {};
\end{tikzpicture}

%% file: Figures/lasso_blocking.tikz
\begin{tikzpicture}[scale=1]
    \draw[shift={(1.457, 3.295)}, rotate=60.63] (0, 0) -- (0, -0.847) .. controls (0, -0.941) and (0.047, -0.988) .. (0.141, -0.988) .. controls (0.235, -0.988) and (0.282, -0.941) .. (0.282, -0.847) -- (0.282, 0) .. controls (0.282, 0.094) and (0.235, 0.141) .. (0.141, 0.141) .. controls (0.047, 0.141) and (0, 0.094) .. (0, 0) -- cycle;
    \node[Mark, Mark_disk] at (2.197, 3.043) {};
    \draw[shift={(2.07, 2.948)}, rotate=62.561] (0, 0) -- (0, -0.847) .. controls (0, -0.941) and (0.047, -0.988) .. (0.141, -0.988) .. controls (0.235, -0.988) and (0.282, -0.941) .. (0.282, -0.847) -- (0.282, 0) .. controls (0.282, 0.094) and (0.235, 0.141) .. (0.141, 0.141) .. controls (0.047, 0.141) and (0, 0.094) .. (0, 0) -- cycle;
    \node[Mark, Mark_disk, Green!80] at (2.82, 2.711) {};
    \draw[shift={(2.775, 2.571)}, rotate=91.615] (0, 0) -- (0, -0.847) .. controls (0, -0.941) and (0.047, -0.988) .. (0.141, -0.988) .. controls (0.235, -0.988) and (0.282, -0.941) .. (0.282, -0.847) -- (0.282, 0) .. controls (0.282, 0.094) and (0.235, 0.141) .. (0.141, 0.141) .. controls (0.047, 0.141) and (0, 0.094) .. (0, 0) -- cycle;
    \node[Mark, Mark_disk, BrickRed!90] at (3.607, 2.728) {};
    \draw[shift={(3.651, 2.565)}, rotate=129.289] (0, 0) -- (0, -0.847) .. controls (0, -0.941) and (0.047, -0.988) .. (0.141, -0.988) .. controls (0.235, -0.988) and (0.282, -0.941) .. (0.282, -0.847) -- (0.282, 0) .. controls (0.282, 0.094) and (0.235, 0.141) .. (0.141, 0.141) .. controls (0.047, 0.141) and (0, 0.094) .. (0, 0) -- cycle;
    \node[Mark, Mark_disk, Green!80] at (4.162, 3.165) {};
    \draw[shift={(4.248, 3.027)}, rotate=147.109] (0, 0) -- (0, -0.847) .. controls (0, -0.941) and (0.047, -0.988) .. (0.141, -0.988) .. controls (0.235, -0.988) and (0.282, -0.941) .. (0.282, -0.847) -- (0.282, 0) .. controls (0.282, 0.094) and (0.235, 0.141) .. (0.141, 0.141) .. controls (0.047, 0.141) and (0, 0.094) .. (0, 0) -- cycle;
    \node[Mark, Mark_disk, BrickRed!90] at (4.562, 3.751) {};
    \node[Mark, Mark_disk, BrickRed!90] at (2.28, 4.453) {};
    \draw[shift={(2.129, 4.536)}, rotate=-10.968] (0, 0) -- (0, -0.847) .. controls (0, -0.941) and (0.047, -0.988) .. (0.141, -0.988) .. controls (0.235, -0.988) and (0.282, -0.941) .. (0.282, -0.847) -- (0.282, 0) .. controls (0.282, 0.094) and (0.235, 0.141) .. (0.141, 0.141) .. controls (0.047, 0.141) and (0, 0.094) .. (0, 0) -- cycle;
    \node[Mark, Mark_disk, Green!80] at (2.119, 3.757) {};
    \draw[shift={(1.949, 3.79)}, rotate=8.288] (0, 0) -- (0, -0.847) .. controls (0, -0.941) and (0.047, -0.988) .. (0.141, -0.988) .. controls (0.235, -0.988) and (0.282, -0.941) .. (0.282, -0.847) -- (0.282, 0) .. controls (0.282, 0.094) and (0.235, 0.141) .. (0.141, 0.141) .. controls (0.047, 0.141) and (0, 0.094) .. (0, 0) -- cycle;
    \node[Mark, Mark_disk, BrickRed!90] at (2.2, 3.043) {};
    \draw[shift={(0.857, 3.778)}, rotate=52.009] (0, 0) -- (0, -0.847) .. controls (0, -0.941) and (0.047, -0.988) .. (0.141, -0.988) .. controls (0.235, -0.988) and (0.282, -0.941) .. (0.282, -0.847) -- (0.282, 0) .. controls (0.282, 0.094) and (0.235, 0.141) .. (0.141, 0.141) .. controls (0.047, 0.141) and (0, 0.094) .. (0, 0) -- cycle;
    \node[Mark, Mark_disk, Green!80] at (1.58, 3.382) {};
    \draw[shift={(0.371, 4.306)}, rotate=42.998] (0, 0) -- (0, -0.847) .. controls (0, -0.941) and (0.047, -0.988) .. (0.141, -0.988) .. controls (0.235, -0.988) and (0.282, -0.941) .. (0.282, -0.847) -- (0.282, 0) .. controls (0.282, 0.094) and (0.235, 0.141) .. (0.141, 0.141) .. controls (0.047, 0.141) and (0, 0.094) .. (0, 0) -- cycle;
    \node[Mark, Mark_disk, BrickRed!90] at (0.998, 3.848) {};
    \node[Mark, Mark_disk, Green!80] at (0.522, 4.35) {};
    \draw[shift={(1.182, 0.82)}, rotate=60.63] (0, 0) -- (0, -0.847) .. controls (0, -0.941) and (0.047, -0.988) .. (0.141, -0.988) .. controls (0.235, -0.988) and (0.282, -0.941) .. (0.282, -0.847) -- (0.282, 0) .. controls (0.282, 0.094) and (0.235, 0.141) .. (0.141, 0.141) .. controls (0.047, 0.141) and (0, 0.094) .. (0, 0) -- cycle;
    \node[Mark, Mark_disk, Mark_large, Peach] at (1.922, 0.567) {};
    \draw[shift={(1.795, 0.473)}, rotate=62.561] (0, 0) -- (0, -0.847) .. controls (0, -0.941) and (0.047, -0.988) .. (0.141, -0.988) .. controls (0.235, -0.988) and (0.282, -0.941) .. (0.282, -0.847) -- (0.282, 0) .. controls (0.282, 0.094) and (0.235, 0.141) .. (0.141, 0.141) .. controls (0.047, 0.141) and (0, 0.094) .. (0, 0) -- cycle;
    \node[Mark, Mark_disk, Green!80] at (2.545, 0.236) {};
    \draw[shift={(2.501, 0.095)}, rotate=91.615] (0, 0) -- (0, -0.847) .. controls (0, -0.941) and (0.047, -0.988) .. (0.141, -0.988) .. controls (0.235, -0.988) and (0.282, -0.941) .. (0.282, -0.847) -- (0.282, 0) .. controls (0.282, 0.094) and (0.235, 0.141) .. (0.141, 0.141) .. controls (0.047, 0.141) and (0, 0.094) .. (0, 0) -- cycle;
    \node[Mark, Mark_disk, BrickRed!90] at (3.333, 0.253) {};
    \draw[shift={(3.376, 0.089)}, rotate=129.289] (0, 0) -- (0, -0.847) .. controls (0, -0.941) and (0.047, -0.988) .. (0.141, -0.988) .. controls (0.235, -0.988) and (0.282, -0.941) .. (0.282, -0.847) -- (0.282, 0) .. controls (0.282, 0.094) and (0.235, 0.141) .. (0.141, 0.141) .. controls (0.047, 0.141) and (0, 0.094) .. (0, 0) -- cycle;
    \node[Mark, Mark_disk, Green!80] at (3.887, 0.69) {};
    \draw[shift={(3.973, 0.552)}, rotate=147.109] (0, 0) -- (0, -0.847) .. controls (0, -0.941) and (0.047, -0.988) .. (0.141, -0.988) .. controls (0.235, -0.988) and (0.282, -0.941) .. (0.282, -0.847) -- (0.282, 0) .. controls (0.282, 0.094) and (0.235, 0.141) .. (0.141, 0.141) .. controls (0.047, 0.141) and (0, 0.094) .. (0, 0) -- cycle;
    \node[Mark, Mark_disk, BrickRed!90] at (4.287, 1.275) {};
    \draw[shift={(0.583, 1.302)}, rotate=52.009] (0, 0) -- (0, -0.847) .. controls (0, -0.941) and (0.047, -0.988) .. (0.141, -0.988) .. controls (0.235, -0.988) and (0.282, -0.941) .. (0.282, -0.847) -- (0.282, 0) .. controls (0.282, 0.094) and (0.235, 0.141) .. (0.141, 0.141) .. controls (0.047, 0.141) and (0, 0.094) .. (0, 0) -- cycle;
    \node[Mark, Mark_disk, Green!80] at (1.305, 0.907) {};
    \draw[shift={(0.096, 1.831)}, rotate=42.998] (0, 0) -- (0, -0.847) .. controls (0, -0.941) and (0.047, -0.988) .. (0.141, -0.988) .. controls (0.235, -0.988) and (0.282, -0.941) .. (0.282, -0.847) -- (0.282, 0) .. controls (0.282, 0.094) and (0.235, 0.141) .. (0.141, 0.141) .. controls (0.047, 0.141) and (0, 0.094) .. (0, 0) -- cycle;
    \node[Mark, Mark_disk, BrickRed!90] at (0.723, 1.373) {};
    \node[Mark, Mark_disk, Green!80] at (0.247, 1.875) {};
    \draw[very thick, ->] (2.334, 2.474) -- (2.334, 1.486);
    \node[Mark, Mark_disk, Green!80] at (2.696, 5.123) {};
    \draw[shift={(2.58, 5.209)}, rotate=-31.659] (0, 0) -- (0, -0.847) .. controls (0, -0.941) and (0.047, -0.988) .. (0.141, -0.988) .. controls (0.235, -0.988) and (0.282, -0.941) .. (0.282, -0.847) -- (0.282, 0) .. controls (0.282, 0.094) and (0.235, 0.141) .. (0.141, 0.141) .. controls (0.047, 0.141) and (0, 0.094) .. (0, 0) -- cycle;
    \draw 
    (1.922, 0.567) circle[radius=0.272];
\end{tikzpicture}

%% file: Figures/lasso_complementary.tikz
\begin{tikzpicture}[scale=1]
    \filldraw[shift={(1.237, 0.873)}, rotate=60.63, fill=RoyalBlue!80, fill opacity=0.3] (0, 0) -- (0, -0.847) .. controls (0, -0.941) and (0.047, -0.988) .. (0.141, -0.988) .. controls (0.235, -0.988) and (0.282, -0.941) .. (0.282, -0.847) -- (0.282, 0) .. controls (0.282, 0.094) and (0.235, 0.141) .. (0.141, 0.141) .. controls (0.047, 0.141) and (0, 0.094) .. (0, 0) -- cycle;
    \filldraw[shift={(1.85, 0.526)}, rotate=62.561, fill=RoyalBlue!80, fill opacity=0.3] (0, 0) -- (0, -0.847) .. controls (0, -0.941) and (0.047, -0.988) .. (0.141, -0.988) .. controls (0.235, -0.988) and (0.282, -0.941) .. (0.282, -0.847) -- (0.282, 0) .. controls (0.282, 0.094) and (0.235, 0.141) .. (0.141, 0.141) .. controls (0.047, 0.141) and (0, 0.094) .. (0, 0) -- cycle;
    \filldraw[shift={(4.681, 2.013)}, rotate=-162.238, fill=Lavender, fill opacity=0.3] (0, 0) -- (0, -0.847) .. controls (0, -0.941) and (0.047, -0.988) .. (0.141, -0.988) .. controls (0.235, -0.988) and (0.282, -0.941) .. (0.282, -0.847) -- (0.282, 0) .. controls (0.282, 0.094) and (0.235, 0.141) .. (0.141, 0.141) .. controls (0.047, 0.141) and (0, 0.094) .. (0, 0) -- cycle;
    \filldraw[shift={(4.42, 2.849)}, rotate=-136.812, fill=Lavender, fill opacity=0.3] (0, 0) -- (0, -0.847) .. controls (0, -0.941) and (0.047, -0.988) .. (0.141, -0.988) .. controls (0.235, -0.988) and (0.282, -0.941) .. (0.282, -0.847) -- (0.282, 0) .. controls (0.282, 0.094) and (0.235, 0.141) .. (0.141, 0.141) .. controls (0.047, 0.141) and (0, 0.094) .. (0, 0) -- cycle;
    \filldraw[shift={(3.8, 3.483)}, rotate=-79.075, fill=Lavender, fill opacity=0.3] (0, 0) -- (0, -0.847) .. controls (0, -0.941) and (0.047, -0.988) .. (0.141, -0.988) .. controls (0.235, -0.988) and (0.282, -0.941) .. (0.282, -0.847) -- (0.282, 0) .. controls (0.282, 0.094) and (0.235, 0.141) .. (0.141, 0.141) .. controls (0.047, 0.141) and (0, 0.094) .. (0, 0) -- cycle;
    \filldraw[shift={(0.638, 1.356)}, rotate=52.009, fill=RoyalBlue!80, fill opacity=0.3] (0, 0) -- (0, -0.847) .. controls (0, -0.941) and (0.047, -0.988) .. (0.141, -0.988) .. controls (0.235, -0.988) and (0.282, -0.941) .. (0.282, -0.847) -- (0.282, 0) .. controls (0.282, 0.094) and (0.235, 0.141) .. (0.141, 0.141) .. controls (0.047, 0.141) and (0, 0.094) .. (0, 0) -- cycle;
    \filldraw[shift={(0.151, 1.884)}, rotate=42.998, fill=RoyalBlue!80, fill opacity=0.3] (0, 0) -- (0, -0.847) .. controls (0, -0.941) and (0.047, -0.988) .. (0.141, -0.988) .. controls (0.235, -0.988) and (0.282, -0.941) .. (0.282, -0.847) -- (0.282, 0) .. controls (0.282, 0.094) and (0.235, 0.141) .. (0.141, 0.141) .. controls (0.047, 0.141) and (0, 0.094) .. (0, 0) -- cycle;
    \draw[<->] (0, 1.931) .. controls (0.565, 0.797) and (1.441, 0.153) .. (2.63, 0);
    \node[anchor=center, text=RoyalBlue!80] at (0.841, 0.403) {$X$};
    \node[anchor=center, text=Mulberry!80] at (4.639, 3.96) {$Y$};
    \draw[<->] (2.838, 3.348) .. controls (4.244, 4.28) and (4.903, 3.808) .. (4.817, 1.932);
    \draw[<->] (2.767, 2.807) -- (2.049, 1.104);
    \draw[<->] (4.372, 1.649) -- (3.114, 0.36);
    \node[rotate=65.845, anchor=center] at (2.593, 1.855) {$l+1$};
    \node[rotate=44.144, anchor=center] at (3.587, 1.179) {$l+1$};
    \node[Mark, Mark_disk, BrickRed!90] at (0.778, 1.426) {};
    \node[Mark, Mark_disk, Green!80] at (0.302, 1.929) {};
    \node[Mark, Mark_disk, Green!80] at (3.778, 3.336) {};
    \node[Mark, Mark_disk, BrickRed!90] at (4.288, 2.763) {};
    \node[Mark, Mark_disk, Green!80] at (1.368, 0.962) {};
    \node[Mark, Mark_disk, Green!80] at (2.6, 0.29) {};
    \node[Mark, Mark_disk, BrickRed!90] at (3.388, 0.307) {};
    \node[Mark, Mark_disk, Green!80] at (3.942, 0.744) {};
    \node[Mark, Mark_disk, BrickRed!90] at (4.342, 1.329) {};
    \node[Mark, Mark_disk, Green!80] at (4.554, 2.021) {};   
    \node[Mark, Mark_disk, BrickRed!90] at (3.007, 3.173) {};
    \node[Mark, Mark_disk, Green!80] at (2.477, 2.702) {};
    \node[Mark, Mark_disk, BrickRed!90] at (2.06, 2.032) {};
    \node[Mark, Mark_disk, Green!80] at (1.9, 1.336) {};
    \node[Mark, Mark_disk, Green!80] at (1.901, 1.341) {};
    \node[Mark, Mark_disk, BrickRed!90] at (1.981, 0.621) {};
\end{tikzpicture}